\documentclass{article}
\usepackage[margin=.8in,left=.8in]{geometry}
\usepackage{amsthm}
\usepackage{amssymb}
\usepackage{amsmath}

\usepackage{adjustbox}

\usepackage{mathtools}
\usepackage{xcolor}
\usepackage{colortbl}    
\usepackage{stmaryrd}    
\usepackage[safe]{tipa} 

\usepackage{hyperref}
\hypersetup{
  colorlinks = true,
 linkcolor= blue, citecolor=blue           
}

\hypersetup{
        colorlinks=true,
        linkcolor=blue,
        filecolor=blue,      
        urlcolor=blue,
        citecolor=blue,
        linktoc=page
    }

\let\PLAINthebibliography\thebibliography
\renewcommand\thebibliography[1]{
  \PLAINthebibliography{#1}
  \setlength{\parskip}{0.5pt}
  \setlength{\itemsep}{0.5pt plus .3ex}
}

\newcommand{\glb}
{\mathrm{glb}}

\newcommand{\CartesianSpaces}{\mathrm{CartSp}}

\newcommand{\FrechetManifolds}{\mathrm{FrMfd}}

\newcommand{\SmoothManifolds}{\mathrm{SmthMfd}}

\newcommand{\SmoothSets}{\mathrm{SmthSet}}

\newcommand{\ThickenedSmoothSets}{\mathrm{Th}\SmoothSets}

\newcommand{\ThickenedCartesianSpaces}{\mathrm{Th}\CartesianSpaces}

\newcommand{\ThickenedSmoothManifolds}{\mathrm{Th}\SmoothManifolds}

\newcommand{\ev}
{\mathrm{ev}}

\newcommand{\loc}
{\mathrm{loc}}

\newcommand{\var}
{\mathrm{var}}

\newcommand{\op}
{\mathrm{op}}
\usepackage{enumitem} \setlist{nosep} 

\usepackage{tikz}
\usetikzlibrary{cd}

\usetikzlibrary{shapes.geometric,calc}

\usepackage{pgfplots}
\usepackage{amsmath,amssymb}
\pgfplotsset{compat=1.15}

\DeclareMathAlphabet{\mathpzc}{OT1}{pzc}{m}{it} 

\newcommand\mathscr[1]{\scalebox{1.1}{$\mathpzc{#1}$}}

\definecolor{darkblue}{rgb}{0.05,0.25,0.65}
\definecolor{darkgreen}{RGB}{20,140,10}
\definecolor{lightgray}{rgb}{0.9,0.9,0.9}
\definecolor{darkorange}{RGB}{200,100,5}
\definecolor{darkyellow}{rgb}{.91,.91,0}
\definecolor{lightolive}{RGB}{189,183,107}

\definecolor{greenii}{RGB}{20,140,10}
\definecolor{orangeii}{RGB}{200,100,5}

\usepackage{multirow}

\usepackage{enumerate} 

\makeatletter
\newcommand\makebig[2]{%
  \@xp\newcommand\@xp*\csname#1\endcsname{\bBigg@{#2}}%
  \@xp\newcommand\@xp*\csname#1l\endcsname{\@xp\mathopen\csname#1\endcsname}%
  \@xp\newcommand\@xp*\csname#1r\endcsname{\@xp\mathclose\csname#1\endcsname}%
}
\makeatother

\makebig{biggg} {3.0}
\makebig{Biggg} {3.5}
\makebig{bigggg}{4.0}
\makebig{Bigggg}{4.5}

\usepackage{mathptmx}
\usepackage{graphicx}
\DeclareRobustCommand{\coprod}{\mathop{\text{\fakecoprod}}}
\newcommand{\fakecoprod}{%
  \sbox0{$\prod$}%
  \smash{\raisebox{\dimexpr.9625\depth-\dp0}{\scalebox{1}[-1]{$\prod$}}}%
  \vphantom{$\prod$}%
}

\DeclareFontFamily{OMX}{MnSymbolE}{}
\DeclareSymbolFont{mnomx}{OMX}{MnSymbolE}{m}{n}
\SetSymbolFont{mnomx}{bold}{OMX}{MnSymbolE}{b}{n}
\DeclareFontShape{OMX}{MnSymbolE}{m}{n}{
    <-6>  MnSymbolE5
   <6-7>  MnSymbolE6
   <7-8>  MnSymbolE7
   <8-9>  MnSymbolE8
   <9-10> MnSymbolE9
  <10-12> MnSymbolE10
  <12->   MnSymbolE12}{}

\usepackage[new]{old-arrows}   

\usepackage{cleveref}

\crefformat{section}{\S#2#1#3} 
\crefformat{subsection}{\S#2#1#3}
\crefformat{subsubsection}{\S#2#1#3}

\theoremstyle{italics}
\newtheorem{theorem}{Theorem}[section]
\newtheorem{lemma}[theorem]{Lemma}
\newtheorem{proposition}[theorem]{Proposition}
\newtheorem{corollary}[theorem]{Corollary}
\theoremstyle{definition}
\newtheorem{definition}[theorem]{Definition}

\newtheorem{example}[theorem]{Example}

\newtheorem{remark}[theorem]{Remark}

\newtheorem{literature}[theorem]{Literature}

\renewcommand{\emph}{\textit}

\usepackage{latexsym}
\usepackage[all]{xy}
\usepackage{color}

\usepackage{bbm} 

\newcommand{\id}{\mathrm{id}}   			

\newcommand{\CB}{\mathcal{B}}

\newcommand{\CF}{\mathcal{F}}

\newcommand{\CG}{\mathcal{G}}

\newcommand{\CH}{\mathcal{H}}

\newcommand{\CK}{\mathcal{K}}

\newcommand{\CL}{\mathcal{L}}

\newcommand{\CO}{\mathcal{O}}

\newcommand{\CP}{\mathcal{P}}

\newcommand{\CU}{\mathcal{U}}
\newcommand{\CV}{\mathcal{V}}

\newcommand{\CX}{\mathcal{X}}

\newcommand{\CZ}{\mathcal{Z}}

\newcommand{\CE}{\mathcal{E}}
\newcommand{\frg}{\mathfrak{g}}				

\newcommand{\FR}{\mathbbm{R}}     			
\newcommand{\NN}{\mathbbm{N}}     			
\newcommand{\DD}{\mathbbm{D}}     			

\newcommand{\dd}{\mathrm{d}}     			

\newcommand{\pr}{\mathrm{pr}}     			

\newcommand{\comment}[1]{}     				
\def\tyng(#1){\hbox{\tiny$\yng(#1)$}}			
\def\tyoung(#1){\hbox{\tiny$\young(#1)$}}			

\newcommand{\beq}{\begin{eqnarray}}
	\newcommand{\eeq}{\end{eqnarray}}

\definecolor{outrageousorange}{rgb}{1.0, 0.43, 0.29}

\newcommand{\om}{\omega}
\newcommand{\epsi}{\epsilon}
\newcommand{\nn}{\nonumber}
\newcommand{\frJ}{\mathfrak{J}}

\usepackage{eulervm}

\usepackage{tocloft}

\newcommand{\Hom}{\mathrm{Hom}}

\newcommand{\Alg}{\mathrm{Alg}}

\newcommand{\frT}{\mathfrak{T}}

\begin{document}


\setlength{\abovedisplayskip}{3pt}
\setlength{\belowdisplayskip}{3pt}
\setlength{\abovedisplayshortskip}{-3pt}
\setlength{\belowdisplayshortskip}{3pt}


\title{\Large\bf 
Field Theory via Higher 
Geometry II:
\\
Thickened Smooth Sets as Synthetic Foundations 
}
\author{Grigorios Giotopoulos${}^\ast$  \quad Hisham Sati${}^{\ast, \dagger}$   
}
\date{}

\maketitle

\begin{abstract}

\medskip
This is the second in a series of papers that aim to develop rigorous and most encompassing foundations for field theory, where in 
the first installment \cite{GS23} we laid out the natural formulation of bosonic variational field theory via the ``functorial geometry'' of smooth sets, namely in the topos over the site of spaces with smooth maps between them.
Here, we extend this to the category  $\ThickenedSmoothSets$
of {\it infinitesimally thickened smooth sets}.
We first describe the Cahiers topos in a simplified, but still fully rigorous, $\FR$-algebraic setting -- which should serve as a more accessible introduction to the theory of Synthetic Differential Geometry to both physicists and mathematicians.
Then, we formulate local Lagrangian field theory in this setting in which infinitesimal spaces exist 
and interact correctly with the field-theoretic spaces of infinite jet bundles, off-shell and on-shell spaces of fields etc.

\smallskip 
 This discussion subsumes previous constructions (see \cite{GS23}) and further recovers all the relevant tangent bundles of traditional (off-shell and on-shell) field theory considerations via the synthetic tangent bundle construction, i.e., as ``infinitesimal curves'' in those spaces, which were previously defined only in an indirect manner.
Beyond finally putting such aspects of the theory on a firm foundation, this approach recognizes the variational principle of local Lagrangian field theory, equivalently, as an intersection
of thickened smooth sets. Lastly, as we will show in detail, it permits a mathematical formalization (and recovery) of perturbative field theory as the actual restriction to the infinitesimal neighborhood $\DD_\phi \hookrightarrow \CF$ around a field configuration -- a statement which thus far had remained only in the realm of intuition. Crucially, all of these results are obtained with no use of infinite dimensional manifold theory, topology or functional analysis of mapping spaces.

\end{abstract}

 \begin{center}
 \begin{minipage}{11.5cm}
\small \tableofcontents
\end{minipage}
 \end{center}

\vfill

\hrule
\vspace{5pt}

{
\footnotesize
\noindent
\def\arraystretch{1}
\tabcolsep=0pt
\begin{tabular}{ll}
${}^*$\,
&
Mathematics, Division of Science; and
\\
&
Center for Quantum and Topological Systems,
\\
&
NYUAD Research Institute,
\\
&
New York University Abu Dhabi, UAE.  
\end{tabular}
\hfill
{
\includegraphics[width=3cm]{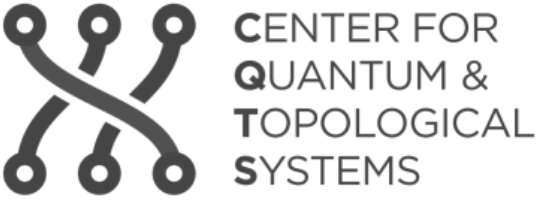}
}

\vspace{1mm} 
\noindent ${}^\dagger$The Courant Institute for Mathematical Sciences, NYU, NY

\vspace{.2cm}

\noindent
The authors acknowledge the support by {\it Tamkeen} under the 
{\it NYU Abu Dhabi Research Institute grant} {\tt CG008}.
}

\newpage

\section{Introduction} 

This is the second in a series of papers that aim to lay rigorous foundations for field theories. 
In the first installment \cite{GS23} we set up a natural formulation of variational field theory via the ``functorial geometry'' of {\it smooth sets}, namely in the topos over the site of Cartesian spaces with smooth maps between them.
Here, we further refine this formulation to thickened smooth sets, which are in addition equipped with infinitesimal extension.
We show,  by a list of constructive results and examples, that this is an even more natural setting 
to accommodate further intuitive notions used in (local) classical field theory. More explicitly, this forms a well-adapted topos
for ``\textit{Synthetic Differential Geometry}'' (SDG), defined as sheaves over a `thickened' version of Cartesian spaces that includes test probes with `infinitesimal directions', also known as the \textit{Cahiers topos}. 
That is, we consider a category of smooth geometrical objects that are furthermore probe-able by certain well-defined infinitesimal spaces. 
In other words, these ``infinitesimally thickened smooth sets'' will have, in addition to a smooth structure, a well-defined notion of infinitesimal structure, both essential for a proper treatment of field theory.  
 Schematically, 
 
\begin{center}
\fbox{Smooth structure}
+
\fbox{Infinitesimal structure}
$\leadsto$ Cahiers topos 

\end{center} 
hence containing objects of finite (smooth) and infinitesimal extent, e.g., 
 
\begin{center}
 \begin{tabular}{cc} 
\begin{tikzpicture}
  \shade[shading=radial,
         inner color=black!40,
         outer color=black!5]
        (0,0) circle (1);

  \fill (0,0) circle (1.5pt);

  \node at (-1.65,.2) {\footnotesize Thickened} ;
  \node at (-1.55, -0.2) {\footnotesize point, $\mathbb{D}$};
\end{tikzpicture}
&\qquad 
\adjustbox{raise=18pt}{
\begin{tikzpicture}
  \shade[shading=axis,
         bottom color=black!5,
         middle color=black!80,
         top color=black!5]
        (-2.5, -.5) rectangle (2.5, .5);

  \fill (0,0) circle (1.5pt);

  \draw[thick] (-2.5,0) -- (2.5,0);

  \node[right] at (2.65,0) {\footnotesize Thickened $\mathbb{R}$};
\end{tikzpicture}
}
\end{tabular} 
\end{center} 

\vspace{-4mm} 
\paragraph{The setting.}
In \cite{GS23}, we used the category  $\mathrm{CartSp}$ as probes, Cartesian spaces with smooth maps between them,
which yielded the same notion of generalized smooth spaces (smooth sets) as probing with the whole category  $\SmoothManifolds$ of smooth manifolds instead. 
A most useful fact towards further generalized notions of geometrical spaces is to work dually with the corresponding algebras of smooth functions, captured in 
the (opposite) category of commutative algebras $\mathrm{CAlg}_\FR^{op}$, giving a  fully faithful 
composite $\mathrm{CartSp} \longhookrightarrow\SmoothManifolds \longhookrightarrow \mathrm{CAlg}_\FR^{op}$. Indeed, this allows to enlarge our original site of probe spaces so as to include infinitesimal points and their products with Cartesian spaces, leading to the category of (infinitesimally) thickened Cartesian spaces
$\ThickenedCartesianSpaces$. We intuitively think of objects, denoted by $\FR^{k}\times \DD$, 
as the usual Cartesian spaces, supplied with an extra `halo of infinitesimal directions' around each point $x\in \FR^k$. This category becomes a site by defining the (Grothendieck) coverage to be that of (differentiably) good open covers,
extended trivially along infinitesimal directions. 
Considering maps between these objects as formally dual to `pullbacks' between their function algebras in CAlg$_{\FR}$, leads 
as above to the fully faithful composite 
$\mathrm{ThCartSp} \longhookrightarrow \ThickenedSmoothManifolds   \longhookrightarrow \mathrm{CAlg}_\FR^{op}$, 
where $\mathrm{ThCartSp}$ is the thickened version of $\mathrm{CartSp}$, and similarly $\ThickenedSmoothManifolds$ of $\SmoothManifolds$. Overall, we have the following diagram of fully faithful inclusions between sites of probes, along with corresponding projections that forget the infinitesimal extensions
\begin{center}
  \begin{tikzcd}[row sep=-1pt, 
    column sep=20pt
  ]
   \mathrm{CartSp} 
    \ar[rr, hook]
    \ar[dd, hook]
    &&
    \SmoothManifolds 
    \ar[rrd, hook]
    \ar[dd, hook]
    \\
    &&&&
    \mathrm{CAlg}_\FR^{op}.
    \\
  \ThickenedCartesianSpaces  
  \ar[rr, hook]
   \ar[uu, shift left=1.5ex] 
  && 
    \ThickenedSmoothManifolds
    \ar[rru, hook]
    \ar[uu, shift left=1.5ex] 
  \end{tikzcd}
\end{center} 
\smallskip 
We will show that the category of spaces probe-able by such thickened smooth probes is the proper mathematical (and intuitive) setting to naturally define the notions of tangent vectors in all sorts of field theoretic spaces: manifolds with boundaries, infinite jet bundles, off-shell and on-shell field spaces. These are all constructed in such a way that immediately recovers the objects introduced in \cite{GS23} and hence in turn those already used in the literature. Moreover, this setting allows us to formalize the intuition of ``infinitesimal neighborhoods'', hence allowing to view jets of fields as sections over infinitesimal neighborhoods of points, and initial conditions as sections over infinitesimal neighborhoods of codimension-1 submanifolds of spacetime. 

\begin{literature}[\bf Synthetic differential geometry: foundations]
The idea of axiomatizing differential geometry using inspiration from topos theory 
goes back to Lawvere \cite{Law67} (see also a modernized perspective in \cite{Law98}). 
Further developments, especially in terms of refining the model-building
were originally provided by Dubuc \cite{Dubuc79} in terms of sheaf toposes on the category of smooth loci, certain formal duals of $C^\infty$-algebras (more commonly known as $C^\infty$-rings).
Detailed discussion on infinitesimals is in \cite{Penon85}. 
The topic has been taken up within differentially cohesive homotopy type theory
\cite{Wellen17}\cite{Myers22}\cite{Cher24}.
Surveys in the form of books include an overview of geometric structures in \cite{Lavendhomme96}, 
calculations and applications in \cite{Bell98},
axiomatics in \cite{Kock06}\cite{Kock10}, and further concrete constructions of well-adapted models using the technology of $C^\infty$-algebras in \cite{BD86}\cite{MoerdijkReyes}. Shorter expositions include 
\cite{Reyes86}\cite{Kock16}\cite{Shulman06}\cite{Bell00}\cite{Kostecki09}\cite{Urs23}.
See also \cite{BGS18} for further extension to differential topology. 
\end{literature}

\newpage 
\begin{literature}[\bf Synthetic differential geometry: applications]
For the original motivation of SDG with an aim to applications to motion and continuum mechanics, see \cite{Law80}\cite{Law97}\cite{Law02}. 
A framework in which Ehresmann's theory of jets can be treated from the viewpoint of synthetic differential geometry, i.e., in the categorical study of a well-adapted topos is in \cite{Kock80}. 
An alternative attempt at a formulation of jet bundles is given in \cite{Nish01}. 
Lie's geometric theory of first-order differential equations is treated in \cite{Kock25}. A modern perspective of PDE theory in the Cahiers topos is provided in \cite{KS17}. Synthetic reasoning has been used to provide a derivation of Einstein's equation for general relativity in
\cite{Reyes11}. 
Scalar field theories are treated in the setting of SDG by \cite{BS17}, and our treatment here 
can be viewed as a vast generalization of their observation and for all Lagrangian field theories 
in full detail, including connection to their perturbative approximations. 

\end{literature}

\noindent {\bf Our approach.} We highlight the following: 

\vspace{1mm} 
\begin{itemize}[leftmargin=4.5mm, itemsep=2pt] 
\item We are not working  in the setting of synthetic differential geometry in the axiomatic 
sense, but rather \textit{externally} via a particularly approachable and utilizable sheaf topos model, namely the  Cahiers topos of Dubuc, in our simpler and more down-to-earth incarnation of it as infinitesimally thickened smooth sets. In other words, we will often be phrasing our discussion in a \textit{plot-wise} manner (also known as level-wise or object-wise), following the intuition arising from physical considerations.

\item We employ the fact that the Cahiers topos of Dubuc is, in fact, $\FR$-algebraic in disguise. This has been implicitly assumed in earlier works by Urs Schreiber, and is now proven explicitly in the accompanying technical Appendix (Cor. \ref{CahiersToposIsRalgebraic}). Indeed, the perhaps abstract and technically heavier machinery of $C^\infty$-algebras might have been an obstacle in the dissemination of the ideas of (and models of) synthetic differential geometry; we hope this result and its application here will make the subject more accessible, by guaranteeing that one may safely bypass familiarity with the notions of $C^\infty$-algebras (at least at the beginner's level). 

\item Generally, our approach highlights a streamlined presentation with  
pedagogy in mind to the extent possible. 
Our main goal here is to lay out a foundation that is both mathematically sound and practical, convenient, and staying 
close to the established practice/intuition of constructions using ``infinitesimals'' in physics.

\end{itemize}

\vspace{1mm} 
\noindent While it is our goal to develop Lagrangian field theory in the above rigorous sense, we highlight 
some of the previous, several being more traditional, approaches: 

\begin{literature}[{\bf Classical Lagrangian field theory}]\label{lit-cft}
The reader is referred to \cite{GS23} for extensively cited literature.
Here, we highlight the ones more directly relevant to the current
development
\cite{GMS09}\cite{Kr15}\cite{MH16}.
The approach to field theory in terms of toposes of sheaves, which inspired this series, is due to  
\cite{dcct}\cite{Schreiber24}.
Another approach close to ours as a formalization of (bosonic) field theories 
 \cite{Blohmann23b}.
The calculus of variations for field theory prominently involves the variational bicomplex.
Also for this topic, extensive literature is cited in \cite{GS23}, while here we highlighted the most relevant ones, 
which include the cohomological properties within various 
approaches in 
\cite{Takens79}\cite{AD80}\cite{BDK}\cite{AF97}\cite{GMS00}.
 A detailed exposition is given in \cite{An91}.
\end{literature}

 \vspace{-5mm} 
\paragraph{Structure and outline of results.}

\vspace{1mm} 
\begin{itemize}[leftmargin=4.5mm] 
\item In \cref{Sec-inf}, we motivate and develop in detail infinitesimally thickened smooth sets as the main backdrop for 
where bosonic field theories ought to be defined. 

\vspace{1mm}
\begin{itemize}[leftmargin=6mm, itemsep=3pt]  
\item[--]
In \cref{Subsec-inf},  we consider infinitesimal probes living alongside Cartesian spaces, 
by recalling how Cartesian spaces and smooth manifolds embed into commutative $\FR$-algebras (Prop. \ref {CartSptoAlgebras}). We then 
describe tangent vectors on smooth manifolds via infinitesimals (Lem. \ref{TangentBundleSet}), 
define infinitesimal disks and more generally infinitesimally thickened points,
and show how the latter are subspaces of the former (Lem. \ref{InfinitesimalPointsAsSubspacesOfInfinitesimalDisks}).

\item[--] In \cref{Subsec-thickSet}, we consider sheaves over infinitesimally thickened smooth probes. 
We enlarge our original site of probe spaces CartSp from \cite{GS23} within the (opposite) category
of commutative algebras, so as to include infinitesimally thickened points and their products with Cartesian spaces; hence thickened
versions of Cartesian spaces. We demonstrate how thickened manifolds, manifolds with boundaries, and smooth closed subspaces thereof, 
are naturally examples of thickened smooth sets. We \textit{algebraically} extend the definition of moduli space of de Rham forms to this thickened setting, and come back to the matter of its actual full classifying property in \cref{ModuliOfDifferentialForms}.

\item[--] In \cref{Subsec-thickmap},
we define thickened smooth mapping spaces and thickened smooth sets of sections, allowing us, in particular, to consider the ``synthetic tangent bundle'' of a manifold (Prop. \ref{TangBundlesCoincide}).
We then consider the synthetic tangent bundle of infinitesimal disks and manifold mapping spaces (Lem. \ref{ManifoldMappingSpaceTangentBundle}), as well as that of field spaces of sections (Prop.  \ref{SyntheticTangentBundleOfFieldSpace}). 
We show how to represent plots of tangent vectors in field spaces via paths of plots (Lem. \ref{LinePlotsRepresentTangentVectors}), 
define the variational cotangent bundle, and exhibit the intuitively desired property of tangent vectors within thickened smooth sets, namely that path derivations depend only on tangent vectors (Lem. \ref{DerivativesAlongLinesDependOnTangentVectors}). 
We define thickened infinite jet bundles, recover their tangent bundles synthetically (Lem. \ref{SyntheticInfiniteJetTangentBundle}), 
and then describe how to explicitly represent plots of $J^\infty_M F$ and of its tangent bundle.

\item[--] In \cref{JetBundlesSyntheticallySec}, we consider the natural relationship of infinitesimal neighborhoods and jets, by 
defining them for manifolds and identifying jets as bona fide sections over these
(Lem. \ref{JetsofSections=InfinitesimalJets}). We define the infinitesimal shape
and synthetic infinitesimal neighborhoods of arbitrary thickened smooth sets (Def. \ref{SyntheticInfinitesimalNeighborhood}), and show how the latter recovers the more traditional notion for manifolds 
(Lem. \ref{SyntheticInfinitesimalNeighborhoodOfManifold}). We motivate how the restriction of local field-theoretic functionals to these neighborhoods in field spaces ought to encode the traditional practice of perturbative field theory, by showing this explicitly in the case of field theory over a point-spacetime, with the general case being the subject of \cref{PerturbationSection}.
We describe $\infty$-jets of sections as sections over infinitesimal neighborhoods and define synthetic infinitesimal neighborhoods of arbitrary submanifolds. Finally, we describe jets over submanifolds as sections over infinitesimal neighborhoods of submanifolds  (Prop. \ref{SectionsOverInfinitesimalNeighborhoodOfSumbanifolds}), hence formalizing the intuition of initial conditions as fields defined over infinitesimal neighborhoods of (codimension 1) submanifolds. 

\end{itemize}

\item In \cref{Sec-FieldtheorySynth}, we undertake our description of local Lagrangian field theory in the above setting. 

\vspace{1mm} 
\begin{itemize}[leftmargin=6mm, itemsep=3pt]  
\item[--]
In \cref{Sec-jetSyth}, we study the geometry of the infinite jet bundle, 
considering horizontal and vertical splittings, characterizing the vertical jet bundle 
(Lem. \ref{VerticalJetBundleAsJetVerticalBundle}), the splitting of the infinite jet tangent bundle
(Prop. \ref{SmoothSplittingProp}), and the synthetic pushforward of jet prolongated section
(Cor. \ref{SyntheticPushfowardOfJetProlongatedSection}).
Then, we describe how the variational bicomplex arises 
completely within thickened smooth sets, along with the corresponding results on vertical and
horizontal cohomologies.

\item[--] This leads us naturally, in \cref{Subsec-bicomplexSynth}, to consider the transgression and the local bicomplex. 
We describe a direct extension to ThSmthSet for all the definitions and related results regarding local Lagrangians and local currents/functionals.
For this, we require the crucial ingredient of the thickened infinite jet prolongation (Lem. \ref{ThickenedInfiniteJetProlongation}), which 
as a concrete example 
we apply explicitly to the Euler-Lagrange operator. 
We then describe the synthetic tangent bundle of jet bundle sections (Lem. \ref{SyntheticTangentBundleOfJetBundleSections})
and the synthetic pushforward of jet prolongation (Prop. \ref{SyntheticPushforwardOfJetProlongation}). This leads us
to consider the pushforward of the evaluation map of sections $\ev : \CF^\infty\times M \rightarrow J^\infty_M F$ of the infinite jet bundle $J^\infty_M F\rightarrow M$ (Prop. \ref{PushfowardOfEvaluationMap}), together with an explicit description of the pushforward of corresponding prolongated evaluation
 (Cor. \ref{PushforwardOfProlongatedEvaluation}). 
We end by showing how the bicomplex of local forms and its Cartan calculus thus appear naturally within ThSmthSet.

\item[--] In \cref{Subsec-onshellSynth}, we describe on-shell fields as a thickened smooth critical set. 
We define critical plots of a thickened smooth map and establish 
functoriality of the thickened critical set (Thm. \ref{FunctorialityOfTheThickenedCriticalSet}) for local field theories, by exhibiting its equivalence with the Euler--Lagrange locus. 
This then naturally yields the identification of the synthetic tangent bundle of on-shell fields as the space of Jacobi fields
(Cor. \ref{OnshellSyntheticTangentBundle}). 

\item[--] Finally in \cref{PerturbationSection}, we describe perturbative field theory as \textit{actual restriction} to infinitesimal neighborhoods of  field configurations $\DD_\phi \hookrightarrow \CF$. Starting with Prop. \ref{PerturbationPropMain}, we show that the (abstract) Taylor data of any thickened smooth map between sections of vector (or affine) bundles computes exactly the restriction to the infinitesimal neighborhood. By calculating this abstract Taylor data explicitly to coincide with the traditional formulas in the case of local field theories (Prop. \ref{PerturbationPropLocal}), we show that the familiar physics-style formal perturbative expansion of a local action functional/Lagrangian represents exactly the actual restriction to the corresponding infinitesimal neighborhood (Cor. \ref{PertubativeActionForVectorAffineBundles}). We then treat the case of nonlinear field bundles, which requires prescribing (finite-dimensional) fiberwise coordinates $\phi^*VF \rightarrow F$ in order to specify an addition structure on the infinitesimal neighborhoods inside the now nonlinear field space (Prop. \ref{PerturbationLemChart}). This extends the earlier results on perturbative field theory to the nonlinear field bundle case, culminating in Thm. \ref{PerturbationCorNonlinear} and Cor. \ref{PerturbationCorEulerLagrange}. Throughout this section, we accompany these technical concepts with plenty of concrete examples coming from physics. We close this final section by indicating the relation of our current results to the perturbative BV/BRST formalism, and the extensions required to formulate a fully nonperturbative version, which when restricted to (derived) infinitesimal neighborhoods ought to recover the local $L_\infty$-algebraic formalism. 
\end{itemize}

\medskip 
\item Our Appendix \cref{App-Synth}
serves as a technical, purely mathematical backbone, and could be of independent interest to 
those seeking deeper perspectives on the synthetic nature of infinitesimal neighborhoods, jets, and the classifying space of de Rham forms. 

\vspace{1mm} 
\begin{itemize}[leftmargin=6mm, itemsep=3pt]
\item[--] Firstly, in \cref{App-CahierCorners}, we justify our usage of the $\FR$-algebraic version of Dubuc's Cahiers topos comprised of sheaves over infinitesimally thickened Cartesian probes $\ThickenedCartesianSpaces\hookrightarrow \mathrm{CAlg}^\op$, rather than the original $C^\infty$-algebraic version from \cite{Dubuc79}, via Thm. \ref{ThickenedCartesianSpacesAreSmoothLoci} and Cor. \ref{CahiersToposIsRalgebraic}. 
We prove a generalized version of Hadamard's lemma for partial derivatives (Lem. \ref{PartialHadamardsLemma}) and extend to manifolds 
and Cartesian products (Cor. \ref{HadamardOnManifoldCartesianProduct}), which imply that the function algebras of our thickened probes  are all quotients of smooth
Cartesian algebras (Cor. \ref{ThickenedFunctionAlgebrasAsQuotientAlgebras}). 
We establish that $\mathbb{R}$-algebras map into thickened Cartesian algebras are automatically $C^\infty$-algebraic (Prop. \ref{RAlgebraMapsAreCinfty}), which implies that our 
  thickened Cartesian spaces qualify as smooth loci (Thm. \ref{ThickenedCartesianSpacesAreSmoothLoci}). The latter immediately implies that the Cahiers topos is $\FR$-algebraic (Cor. \ref{CahiersToposIsRalgebraic}), being equivalent to our $\ThickenedSmoothSets$.

\item [--] In \cref{App-SyntheticJetProlongation}, we define the infinitesimal neighborhood along 
a morphism (Lem. \ref{InfinitesimalNeighborhoodOfDiagonal}) of arbitrary thickened smooth sets, 
and use it to define the 
synthetic jet bundle of any such map, which we prove is a thickened smooth set (Lem. \ref{SyntheticJetbundleisThickenedSmoothSet}).
We establish the adjunction between the diagonal neighborhood functor $\frT_\CF$ and the synthetic $\infty$-jet bundle $J^\infty_\CF$ functor 
(Prop. \ref{FormaldiskbundleJetbundleAdjunction}). 
We then define the synthetic jet prolongation, 
and show how this recovers the traditional infinite jet prolongation in the case of  fiber bundles of manifolds (Lem. \ref{SyntheticJetprolongforFiberBundle}).

\item[--] Lastly, in \cref{ModuliOfDifferentialForms}, we
 define the fully classifying version of the moduli space of de Rham forms.
Our approach here 
allows for a direct extraction of its classifying property, which has remained widely underappreciated due to its somewhat obscure realization in the original axiomatic SDG approaches. 
For instance, it is immediate in our setting that the moduli space of 1-forms $\widehat{\mathbold{\Omega}}^1 \in \CE$ classifies a subset of 
smooth maps $T\CF \rightarrow \FR$ out of the synthetic tangent bundle, which identifies the former as a linear subobject 
(Cor. \ref{1FormModuliSpaceIsSubobject}) of a larger classifier of all (nonlinear) functions out of tangent bundles. 
The full (linear) classifying nature of form moduli is established in Prop. \ref{ClassifyingNatureOfFormModuli}.
The algebraic moduli space is shown to naturally be a subspace of the fully classifying moduli space 
(Lem. \ref{RelationOfTwoModuliSpaces}) and they become isomorphic in a larger topos where tangent bundles are 
representable (Cor. \ref{ModuliSpacesAreIsomorphic}). 
We finally briefly describe how the above extends to differential $n$-form classifying spaces.

\end{itemize} 

\end{itemize}



\newpage

\section{Infinitesimally thickened smooth sets}
\label{Sec-inf}

Recall that in \cite{GS23} we presented the natural formulation of variational field theory via the ``functorial geometry'' of {\it smooth sets}, namely in the topos over the site of Cartesian spaces with smooth maps between them.
In this section, we further refine the latter topos to a setting for  ``synthetic differential geometry'' via \textit{thickened smooth sets}, which are, in addition, equipped with infinitesimal extension.

\subsection{Infinitesimally thickened probes}
\label{Subsec-inf}

Our first goal is to identify a category of \textit{infinitesimal probe spaces} $\DD$, such that generalized spaces $\CF$ may be defined as being probe-able by them, in such a manner that $\CF(\DD)$ may be reasonably interpreted as the \textit{infinitesimal plots} in $\CF$. This is in direct analogy to $\CF(\FR^k)$ being interpreted as the set of smooth $k$-dimensional plots of finite (or infinite) extent in $\CF$. 

\smallskip 
None of the Cartesian probe spaces in our original site $\CartesianSpaces$ serve such a purpose, 
and so we must enlarge the site by introducing infinitesimal probe spaces. To that end, we recall the following crucial but possibly underappreciated result (e.g. \cite[35.10]{KMS93}), which follows by “Milnor’s exercise” (e.g. \cite[§35.9]{KMS93}).
\begin{proposition}[\bf Smooth manifolds embed into Algebras]
\label{CartSptoAlgebras}
The functor 
\begin{align*}
	C^{\infty}(-) \;:\; \SmoothManifolds & \; \xhookrightarrow{\quad \quad } \; \mathrm{CAlg}_{\FR}^{op} \\
	M& \; \xmapsto{\quad \quad} \; C^{\infty}(M) \, ,\nn 
\end{align*}
sending a finite-dimensional smooth (second countable and Hausdorff) manifold to its function algebra is fully faithful, in that for any pair $N,M \in \SmoothManifolds$ the smooth functions $f : N \xrightarrow{\;} M$ biject onto the algebra homomorphisms\footnote{Here we mean the category $\mathrm{CAlg}_\FR$ to have objects \textit{unital} commutative algebras, with the \textit{unit preserving} $\FR$-algebra morphisms between these.} $f^\ast : C^\infty(M) \xrightarrow{\;} C^\infty(N)$.
Hence also the composite
\begin{align*}
C^\infty(-)\, :\,\mathrm{CartSp} \xhookrightarrow{\quad \quad} \SmoothManifolds \longhookrightarrow \mathrm{CAlg}_\FR^{op}
\end{align*}
is fully faithful.
\end{proposition}

In the spirit of algebraic geometry, it follows that the objects of $\mathrm{CAlg}_{\FR}^{op}$ may be thought of as spaces generalizing smooth manifolds, being formally dual to some algebra of functions. Indeed, among these generalized smooth spaces we may also find smooth manifolds (with \textit{boundaries}), infinitesimal spaces, and furthermore smooth manifolds equipped with \textit{infinitesimal thickenings}.
The archetypical example of an infinitesimal space  (e.g. \cite[p. 218]{Mumford88}) is the one dual to the ``algebra of dual numbers'' (\cite[pp. xi, 4]{Kock81}\cite[p. 19]{MoerdijkReyes} following \cite{Gruenwald09})
$$
  \CO\big(\DD^1(1)\big)
  \;:=\;
  \FR[\epsi]/\epsi^2
  \,,
$$
which may be thought of as the smooth functions on the first-order infinitesimal line (or neighborhood) $\DD^1(1) \hookrightarrow \FR$ of the origin $0\in \FR$ in the real line. One thinks of $\epsi\in \CO\big(\DD^1(1)\big)$ as the linear coordinate function on $\DD^1(1)$, 
which squares to zero since `all points are infinitesimally close to the origin'. Such a quotient captures the idea of dropping higher-order terms of infinitesimally small quantities, 
as often informally practiced in the physics literature. 
Following the formal dual intuition, the algebra projection map $C^\infty(\FR) \rightarrow \CO\big(\DD^1(1)\big)$, given by the first-order 
Taylor expansion $f\mapsto f(0)+\epsi \cdot f'(0)$, may be thought of as dually encoding the 
`embedding' 
$$ 
  \iota_0
  \;:\; 
  \DD^1(1) \xhookrightarrow{\quad \quad} \FR^1\, ,
$$ 
of this infinitesimal neighborhood into the real line. The suggestive naming for $\DD^1(1)$ is justified by the following classical result (e.g. \cite[\S I.7]{Kock81}).

\begin{lemma}[\bf Tangent vectors via infinitesimals]\label{TangentBundleSet}
For any smooth manifold $M$, potentially with boundary and corners\footnote{Recall, a function $f:M\rightarrow \FR$ on a manifold with boundary / corners is smooth if it is locally smooth with respect to any chart. Namely, if at any point $x\in M$ and any chart $\psi:\FR^{k,l}\xrightarrow{\sim} U_x \hookrightarrow M$ around it, the corresponding composition $f\circ \psi : \FR^{k,l}\rightarrow \FR$ extends to a smooth function $\widetilde{f}:\FR^{k+l}\rightarrow \FR$ between Cartesian spaces, in the usual analytical sense; 
 see \cite{Joyce19}\cite{FSJ24}\cite{Melrose} for more details on these concepts.} $M\in \SmoothManifolds^\mathrm{cor}$, there is a canonical bijection 
\begin{align*}
	TM 
    \; \cong_{\mathrm{Set}} \;
    \mathrm{Hom}_{\mathrm{CAlg}_{\FR}}\Big(C^{\infty}(M),\CO\big(D^1(1)\big)\!\Big) 
\end{align*}
of tangent vectors on $M$ with (formally dual) maps of $\DD^1(1)$ into $M$.

\begin{proof}
An algebra homomorphism $C^{\infty}(M)\rightarrow \FR[\epsi]/\epsi^2$ 
has components
$$
\big(p,X_p \big) \;:\; f \; \longmapsto \; p(f) + \epsi \cdot X_p(f)\, .
$$
Using the homomorphism property and $\epsi^2=0$, it follows for $f_1, f_2 \in C^\infty(M)$ that
$$
(f_1 \cdot f_2) \; \longmapsto \; p(f_1)\cdot p(f_2) + \epsi  
\big( p(f_1)\cdot X_p (f_2) + X_p (f_1) \cdot p(f_2) \big)\,
$$
and therefore
\begin{itemize}[
  leftmargin=.4cm,
  topsep=1pt,
  itemsep=2pt
]
\item 
the first component is an algebra homomorphism $p : C^\infty(M)\rightarrow \FR\cong C^\infty(*)$, and hence necessarily coincides with the evaluation at a point $p\in M$, by Milnor's exercise;\footnote{Milnor's exercise was originally proven for manifolds without boundary, but extends easily via Whitney's embedding to the case of boundaries and corners (see e.g. \cite{Gi26}).}
\item the second component is a derivation $X_p: C^{\infty}(M) \rightarrow \FR$ at $p\in M$ 
\footnote{More precisely, a derivation of the germs of functions $C^\infty_p(M) \rightarrow \FR$ at the point $p\in M$.}, i.e., a tangent vector at $p\in M$.\qedhere 
\end{itemize}
\end{proof} 
\end{lemma}

The intuitive interpretation of this result is that a (formal dual) map 
$$
  \DD^1(1) \xrightarrow{\quad \quad} M 
  \,,
$$ 
is an (first-order) infinitesimal curve segment in $M$, or, put differently, a (first-order) infinitesimal line plot.
\begin{remark}[\bf Infinitesimal vs finite curves]\label{InfinitesimalVsFiniteCurves}
The above, almost trivial, example already hints at the power of considering infinitesimal spaces as bona fide geometrical objects. 

\begin{itemize}[leftmargin=20pt]
\item[\bf (i)] Indeed, recall that for manifolds without boundary, tangent vectors are more often introduced as \textit{equivalence classes\footnote{Two curves with $\gamma_1(0)=x=\gamma_2(0)$ are considered as equivalent $\gamma_1\sim \gamma_2$ if $\partial_t(f\circ \gamma_1)|_t = \partial_t(f\circ \gamma_2)|_{t=0}$ for all $f\in C^\infty(M)$.} of smooth curves}
$$
\gamma \, : \, \FR^1_t \longrightarrow M\, .
$$
While this is an equivalent characterization for manifolds without boundary, which captures the `infinitesimal part' of the curves via the imposed equivalence relation, this is no longer true when considering manifolds with boundary and corners. Indeed, it is easy to see that in the latter cases, this approach only recovers part of the full tangent bundle.

\item[\bf (ii)] For instance, consider the case of the upper half plane 
$
\FR^{1,1}_{x,y} := \FR^1_x \times [0,\infty)_{y}
$, the model space for 2-dimensional manifolds with boundary. By Lem. \ref{TangentBundleSet}, it follows that the set of infinitesimal curves $\DD^1(1)\rightarrow \FR^{1,1}_{x,y}$ recovers the full tangent bundle
$$
T\FR^{1,1} \,  \cong \, \bigsqcup_{(x_0,y_0)\in \FR^{1,1}}\mathrm{Span}_\FR\Big\{\partial_{x}|_{(x_0,y_0)},\, \partial_y|_{(x_0,y_0)}\Big\} \, .
$$

Notice, this identification includes outward-pointing tangent vectors at the boundary, corresponding to derivations of the form $-c\cdot \frac{\partial}{\partial_y}|_{(x_0,0)}$ where $c\in \FR_{>0}$. However, this part of the tangent bundle cannot be represented by (equivalence classes) of smooth curves\footnote{Strictly speaking, one considers ``$1$-sided'' curves $\gamma_t : [0,\infty) \rightarrow M$ in manifolds with boundaries/corners to define tangent vectors passing through boundary and corner points $\gamma_t(0)\in M$. Still, however, these necessarily only point inwards.} $\gamma : \FR^1_t \rightarrow  \FR^{1,1}$, since there is no smooth curve inside $\FR^{1,1}$ pointing outwards at the boundary.

\item[\bf (iii)] This phenomenon persists when considering the appropriate notion of tangent vectors (and bundles) over more complicated spaces, such as intersections of manifolds, infinite-dimensional field theoretic spaces, or even the Euler--Lagrange locus of a classical field theory.  As we shall see, the tangent bundles of the aforementioned are all correctly identified as maps out of the infinitesimal line, in an appropriate sense (Def. \ref{SyntheticTangentBundle}), whereas (equivalence classes) of curves might miss several components of their tangent bundles (cf. \cite[Rem. 7.14]{GS23}).
\end{itemize}
\end{remark}

Proceeding analogously, we may define $m$-dimensional disks of $l$-infinitesimal order $\DD^{m}(l)$, for any $m,l\in \NN$, 
and consider the corresponding (formally dual) maps $\DD^{m}(l)\rightarrow M$:
\begin{definition}[\bf Infinitesimal disks]\label{InfinitesimalDisk}
For $m,l\in \NN$, the \textit{infinitesimal m-disk of order l}, denoted by $\DD^m(l)$, is defined as the formal dual
of 
\vspace{-2mm} 
\begin{align*}
\CO\big(\DD^m(l)\big)\, := \, \FR[\epsi^1,\cdots, \epsi^m]/(\epsi^1,\cdots,\epsi^m)^{l+1} \quad \in \quad \mathrm{CAlg}_{\FR} \, .
\end{align*}
\end{definition}
 Just as with the infinitesimal line $\DD^1(1)$, every infinitesimal disk $\DD^m(l)$ also `embeds' into the corresponding $m$-dimensional Cartesian space $\FR^m$, around its origin,
\begin{align}\label{InfDiskIntoCartSp} \iota_0 \, : \, \DD^{m}(l) \xhookrightarrow{\quad \quad} \FR^m \, ,
\end{align}
in the formal dual sense of the algebra projection map 
$$
C^\infty(\FR^m)\xrightarrow{\qquad} \CO\big(\DD^m(l)\big)$$ defined by the $l$-order Taylor expansion
$$
f \;\; \xmapsto{\qquad} \;\; f(0)+ \epsi^i \cdot \partial_i f(0) 
\;+\; 
\tfrac{1}{2}\epsi^i\epsi^j\cdot \partial_i\partial_jf(0) \;+\; \cdots 
\;+\; \tfrac{1}{l!}\sum_{i_{1},\cdots,i_{l}=1}^m \epsi^{i_{1}}\cdots \epsi^{i_{l}}\cdot \partial_{i_{1}}\cdots \partial_{i_{l}}f(0)\, .
$$
\begin{remark}[\bf On nomenclature of thickened points]In algebro-geometric terms, one would say $$\DD^m(l)\; := \;  \mathrm{Spec}\big(\FR[\epsi^1,\cdots, \epsi^m]/(\epsi^1,\cdots,\epsi^m)^{l+1} \big)$$ 
is the \textit{spectrum} of $\FR[\epsi^1,\cdots, \epsi^m]/(\epsi^1,\cdots,\epsi^m)^{l+1} $, i.e., the (locally) ringed space 
\vspace{1mm} 
$$
\Big(*\,,\, \FR[\epsi^1,\cdots, \epsi^m]/(\epsi^1,\cdots,\epsi^m)^{l+1} \Big)
$$
over a single point. From this viewpoint, every infinitesimal disk $\DD^m(l)$ has one actual geometrical point,
and hence one may also refer to it 
as an example of an \textit{infinitesimally thickened point} (cf. the more general case of Def. \ref{InfinitesimallyThickenedPoints}).
\end{remark}

The relationship of infinitesimal disks and Cartesian spaces is, in fact, tighter. To witness this, we recall a standard result in real analysis. A proof of (a slightly generalized version of) the following is found in Lem. \ref{PartialHadamardsLemma}.
\begin{lemma}[\bf Hadamard's Lemma]\label{HadamardsLemma}
For any smooth function $f\in C^\infty(\FR^m)$ and any $l\in \NN$, there exist smooth functions 
$\big\{h_{i_{1}\cdots i_{l+1}}\big\}_{i_1,\cdots, i_{l+1}=1,\cdots, m}\subset C^\infty(\FR^m)$ such that 
\vspace{-1mm} 
\begin{align*}
f=  f(0)+ x^i \cdot \partial_i f(0) &+ \tfrac{1}{2}x^i x^j\cdot \partial_i\partial_jf(0) +\cdots +\; \tfrac{1}{l!}  \sum_{i_{1},\cdots,i_{l}=1}^m x^{i_{1}}\cdots x^{i_{l}}\cdot \partial_{i_{1}}\cdots \partial_{i_{l}}f(0)
\;+ \sum_{i_{1},\cdots,i_{l+1}=1}^m x^{i_{1}}\cdots x^{i_{l+1}}\cdot h_{i_{1}\cdots i_{l+1}}\, ,
\end{align*}
where $\{x^i\}_{i=1,\cdots,m}\subset C^\infty(\FR^m)$ are the canonical (linear) coordinate functions.
\end{lemma}
It follows that the function algebra on the infinitesimal disk $\DD^m(l)$ (Def. \ref{InfinitesimalDisk}) is \textit{canonically} identified as the quotient
\begin{equation}\label{InfinitesimalDiskAlgebraAsQuotientOfCartesian}
\CO\big(\DD^m(l)\big)\;\; \cong \;\; C^{\infty}(\FR^m)/(x^1,\cdots, x^m )^{l+1}
\end{equation}
of smooth functions on $\FR^m$ by (the ideal generated by) $(l+1)$-fold products of the linear coordinate functions. It is in this sense that the embedding of an infinitesimal disk at the origin of the corresponding Cartesian space \eqref{InfDiskIntoCartSp} is \textit{canonical}, factoring (dually) through
$$
C^\infty(\FR^m) \longrightarrow C^\infty(\FR^m)/(x^1,\cdots, x^m )^{l+1} \xrightarrow{\quad \sim \quad } \CO\big(\DD^m(l)\big) \, .
$$
We will see that mapping 
appropriately out of these infinitesimal spaces detects higher-order variants of tangent vectors (Rem. \ref{infinitesimalMappingExamples}) and jets of sections of fiber bundles (Prop. \ref{JetsofSections=InfinitesimalJets}). 

\smallskip 
More generally, then, one may reasonably posit that (formal duals of) further quotients of the above nilpotent algebras 
$$
\CO\big(\DD^m(l)\big)\, / \,I
$$ 
should also be considered as valid infinitesimal spaces, since the latter algebras are also nilpotent in nature. The most general such consideration along these lines would be to consider any abstract, finite-dimensional algebra $W$ with nilpotent elements apart from the constants $\FR\hookrightarrow W$ as a function algebra on an infinitesimal point $\DD$. Roughly speaking, such an infinitesimal point $\DD$ should be any space with finitely many ``infinitesimal coordinates''.
\begin{definition}[\bf Infinitesimally thickened points]\label{InfinitesimallyThickenedPoints} An infinitesimally thickened point, denoted by $\DD$, is defined as the formal dual of a \textit{``Weil algebra''} (also known as a \textit{``local Artinian algebra''}). More explicitly, it is a space with function algebra of the form
$$
\CO(\DD)\; \cong_{\mathrm{Vect}_\FR}  \FR \oplus V\, ,
$$ 
where $V$ is finite-dimensional and contains only nilpotent elements. 
\end{definition}

It turns out that any Weil algebra $\CO(\DD) \cong \FR\oplus V$  may be identified  \cite[Thm. 3.17]{MoerdijkReyes}\cite[Prop. 4.43]{CarchediRoytenberg13} with a quotient of some infinitesimal disk function algebra $\CO\big(\DD^m(l)\big)$ of Def. \ref{InfinitesimalDisk} -- canonically up to an automorphism of $\CO\big(\DD^m(l)\big)$ \cite[Lem. 1.7]{Kolar08}. Interpreted dually, this yields the following embeddings of general infinitesimal points into infinitesimal disks. 

\begin{lemma}[\bf Infinitesimal points as subspaces of infinitesimal disks]\label{InfinitesimalPointsAsSubspacesOfInfinitesimalDisks}
Let $\DD$ be an infinitesimally thickened point with Weil function algebra $\CO(\DD) \cong \FR \oplus V$ of nilpotency order $l$, i.e., $V^{l+1}=\{0\}$. There exists a subspace embedding into an infinitesimal disk
$$
\iota\,: \, \DD \xhookrightarrow{\quad \quad} \DD^m(l)\, ,
$$
with the embedding map being the formal dual map of a quotient projection $\CO\big(\DD^m(l)\big)\longrightarrow \CO\big(\DD^m(l)\big) / I\cong \CO(\DD) $, and 
any other such embedding being equivalent up to an automorphism of $\DD^m(l)$
\[ 
\xymatrix@R=1.4em@C=3em  { &&  \DD^m(l)\ar[d]^{\sim}
	\\ 
	\DD \ar@{^{(}->}[rru]^{\iota} \ar@{_{(}->}[rr]^-{\hat{\iota}} && \DD^m(l)\, .
}   
\]
\end{lemma}
\begin{proof}
Let $\{v^{i}\}_{i=1,\cdots, m}$ be a \textit{minimal} generating set for the ideal of nilpotent elements $V\hookrightarrow \CO(\DD)$. It is not hard to see that, equivalently, this corresponds to a basis $\{v^i + V^2\}_{i=1,\cdots , m}$ for the vector space $V/V^2$ (cf. \cite[Prop. 1.5]{Kolar08}). Define a surjective algebra morphism via the assignment on linear ``generators''  
\begin{align*}
 C^\infty(\FR^m) & \longrightarrow \CO(\DD) \\
x^i &\longmapsto v^i
\end{align*}
extended to all of $C^\infty(M)$ via Lem. \ref{HadamardsLemma}. By the $l$-order nilpotency of $\CO(\DD)$, the kernel of this map contains 
$ (x^1,\cdots, x^m)^{l+1}$, and so this projection factors through the infinitesimal disk function algebra \eqref{InfinitesimalDiskAlgebraAsQuotientOfCartesian}
\begin{align*}
 C^\infty(\FR^m)  \longrightarrow \CO\big(\DD^m(l)\big) &\xrightarrow{\quad \pi \quad } \CO(\DD)\\
\epsi^i &\xmapsto{\quad \quad} v^i\, .
\end{align*}
In turn, this yields the sought-after isomorphism of $\FR$-algebras 
$$
\CO\big(\DD^m(l)\big)\longrightarrow
\CO(\DD) \; \cong\;  \CO\big(\DD^m(l)\big) / \mathrm{ker}(\pi) \, .
$$

Next, choosing any other minimal generating set $\{\,\widehat{v}^i\}_{i=1,\cdots, m}$ for $V$, or equivalently a basis 
$\{\,\widehat{v}^i+V^2\}_{i=1,\cdots, m}$ for $V/V^2$ corresponds to a linear map 
$\big\{v^i\mapsto \widehat{v}^i = a^i{}_j v^j +P^i(v)  \big\}_{i,j=1,\cdots , m}$ where $(a^{i}{}_j)$ is an invertible linear map and $(P^{i})$ is some polynomial map of degree higher that 1. This yields a different isomorphism 
\begin{align*}
\CO\big(\DD^m(l)\big)&\xrightarrow{\quad \widehat{\pi}\quad}
\CO(\DD) \; \widehat{\cong}\;  \CO\big(\DD^m(l)\big) / \mathrm{ker}(\widehat{\pi}) \\
\epsi^i &\xmapsto{\quad \quad} \widehat{v}^i =a^i{}_j v^j +P^i(v)\, ,
\end{align*}
but which is related to the former by defining the morphism of infinitesimal disks
\begin{align*}
\CO\big(\DD^m(l)\big) &\xrightarrow{\quad \Phi \quad} \CO\big(\DD^m(l)\big)\\
\epsi^i &\xmapsto{\quad \quad} (a)^{i}{}_j \epsi^j + P^i(\epsi) 
\end{align*}
which is invertible, since it is invertible to linear order and the algebras are nilpotent, thus
making the following diagram commute
\[ 
\xymatrix@C=2.2em@R=.2em  {\CO\big(\DD^m(l)\big) \ar[rd] \ar[rr]^{\Phi} &   & \CO\big(\DD^m(l)\big)\,.
\ar[ld]
	\\ 
& \CO\big(\DD^m(l)\big) / \mathrm{ker}(\pi) \, \cong \, \, \CO(\DD) \, \widehat{\cong} 
\, \, \CO\big(\DD^m(l)\big) /\mathrm{ker}({\widehat{\pi}})   & 
}  
\]
Interpreting the above algebra morphisms in the formal dual sense completes the proof.
\end{proof}
The dimension $m$ of the vector space $V/V^2$, hence the minimal number of nilpotent generators of any Weil algebra, is known as the \textit{width} of the algebra. From our dual perspective, given that this corresponds to the dimension of the corresponding ambient infinitesimal disk $\DD^m(l)$, we should think of this invariant as the (infinitesimal) dimension of the corresponding thickened point $\DD$. Similarly, any generator $v^i$ for $V$ contained in a minimal such set should be thought of as a kind of ``linear'' coordinate function on $\DD$. Thus, it is instructive to think of the embeddings of Lem. \ref{InfinitesimalPointsAsSubspacesOfInfinitesimalDisks} being related by an infinitesimal invertible transformation inside the corresponding disks, descending from the ambient finite polynomial transformation $\Phi:\FR^m \longrightarrow \FR^m$ with invertible linear order $a\in \mathrm{GL}(m,\FR)$.

\begin{remark}[\bf Infinitesimal disks vs. points] 
\label{OmittingGeneralInfinitesimalPoints}
For simplicity of presentation, we will often not explicitly consider these extra infinitesimal points in some of this text, and focus instead only on infinitesimal disks. All results and proofs developed herein may be straightforwardly generalized to this general case of infinitesimal points by recalling that any such  embeds into an infinitesimal disk via Lem. \ref{InfinitesimalPointsAsSubspacesOfInfinitesimalDisks}, and the subsequent constructions extend since their formulas are compatible with the additional quotient relations. 
\end{remark}

\subsection{Sheaves over thickened probes}
\label{Subsec-thickSet}

We are now sufficiently equipped to enlarge our original site of probe spaces $\CartesianSpaces$ \cite{GS23} within the (opposite) category of commutative algebras, so as to include these infinitesimal points and their products with Cartesian spaces (hence 
{\it thickened} versions of Cartesian spaces).
\begin{definition}[\bf Thickened Cartesian spaces]\label{ThickenedCartesianSpaces} 
The category of \textit{(infinitesimally) thickened Cartesian spaces}
$$
\ThickenedCartesianSpaces \xhookrightarrow{\quad \quad} \mathrm{CAlg}_\FR^{op}
$$
is defined as the opposite category to the full subcategory of CAlg$_{\FR}$ consisting of objects of the form 
	$$
 \CO\big(\FR^k\times \DD\big)\; := \; C^{\infty}(\FR^k) \otimes \CO(\DD)\, , 
 $$
	for $\DD$ any infinitesimally thickened point. 
\end{definition}
In summary, we denote objects in ThCartSp by $\FR^{k}\times \DD$ and consider maps between them as formally dual to `pullbacks' between their function algebras in CAlg$_{\FR}$. Intuitively,
one thinks of $\FR^{k}\times \DD$ as the usual Cartesian space, supplied with an extra `halo of infinitesimal directions' around each point $x\in \FR^k$. This category becomes a site by defining the (Grothendieck) coverage to be that of (differentiably) good open covers $\big\{U_i\xhookrightarrow{\iota_i}  \FR^k \big\}_{i\in I}$, extended trivially along infinitesimal directions, and hence of the form
\begin{align}\label{GoodOpenCoversonThCartSp}
\Big\{ U_i\times \DD \xhookrightarrow{\;\iota_i \times \id\;}  \FR^k \times \DD \Big\}_{i\in I} \, .
\end{align}
There is an obvious relation between the two site categories $\CartesianSpaces$ and $\ThickenedCartesianSpaces$. That is, the embedding of sites
\begin{align}\label{CartToFCart}
	\iota: \mathrm{CartSp} &\longhookrightarrow \mathrm{ThCartSp} \\
\FR^{k} &\longmapsto \FR^k\times \DD^{0}(0)\cong \FR^k \times \{*\} \cong \FR^k \ \nn 
\end{align}
and the projection that forgets the infinitesimal halo
\begin{align*}
 p\;:\;\mathrm{ThCartSp}&\longrightarrow \mathrm{CartSp}\\ 
   \FR^k\times \DD \;&\longmapsto \; \FR^k \, , \nn 
\end{align*}
in fact, form an \textit{adjunction}
$$
\iota \quad  \dashv \quad  p \, .
$$Namely
there is a canonical bijection 
\begin{align*}
\mathrm{Hom}_{\ThickenedCartesianSpaces}\big( \iota(\FR^k), \FR^l\times \DD \big) \;\; \cong \;\;  \mathrm{Hom}_{\CartesianSpaces}\big(\FR^k,\, p(\FR^l\times \DD) \big) 
\end{align*}
natural in both $\FR^k$ and $\FR^l\times \DD$. This is easily seen since $p(\FR^l\times \DD)= \FR^l$, and due to the fact that $C^\infty(\FR^k)$ has no nilpotent elements other than the zero-function $0$. Indeed, the latter implies that 
\begin{align*}
\mathrm{Hom}_{\ThickenedCartesianSpaces}\big( \iota(\FR^k), \FR^l\times \DD \big)  & \; \cong \; \mathrm{Hom}_{\mathrm{CAlg}_\FR}\big(C^\infty(\FR^l)\otimes \CO(\DD), \, C^\infty(\FR^k) \big)\,
\\
& \; \cong\;  \mathrm{Hom}_{\mathrm{CAlg}_\FR}\big(C^\infty(\FR^l), \, C^\infty(\FR^k) \big)  ,
\end{align*}
since any nilpotent element of $C^\infty(\FR^l\otimes \DD) \cong C^\infty(\FR^l)\oplus \big(C^\infty(\FR^l)\otimes W\big)$ is necessarily mapped to $0\in C^\infty(\FR^k)$, and hence the adjunction bijection result follows. See for instance \cite{KS17} for more implications of this adjunction.

\smallskip 
Having described our new notion of thickened and smooth probe-spaces, and following the lines of logic of defining generalized spaces by the plots carved out by probing them with simpler spaces (see \cite[\S 2.1]{GS23}\cite{Schreiber24}\cite{Gi25}),
 we immediately arrive at our definition of infinitesimally thickened smooth sets.

\begin{definition}[\bf{Thickened smooth sets}]\label{ThickenedSmoothSets}
	The category of \textit{(infinitesimally) thickened smooth sets}, or \textit{thickened generalized smooth spaces}, is defined to be the (gros) sheaf topos
	\begin{align*}
		\mathrm{ThSmthSet}\; :=\; \mathrm{Sh}(\mathrm{ThCartSp})\, , 
	\end{align*}
	i.e., sheaves on thickened Cartesian spaces with respect to (differentiably) good open covers \eqref{GoodOpenCoversonThCartSp}.
\end{definition}
As with the case of smooth sets, thickened smooth sets are defined entirely in an \textit{operational} manner, by consistently answering the question:

\hspace{2cm}
\adjustbox{
  margin=3pt,
  bgcolor=lightgray
}{
\def\arraystretch{1.2}
\begin{tabular}{l}
\textit{`What are the ways we  can probe the would-be space smoothly  with Cartesian spaces, }
\\
\textit{ (infinitesimally) thickened points and more generally thickened Cartesian spaces?'}
\end{tabular}
}

 \vspace{1mm}
\noindent To stress the point further: Sheaves over plain Cartesian spaces \textit{define} the notion of smoothness of the would-be smooth generalized spaces, by giving meaning to what it means to smoothly probe them. Similarly, sheaves over thickened Cartesian spaces
\textit{define} the notion of smoothness \textit{and} their infinitesimal nature, by giving meaning to what it means to (smoothly) probe them with simple test spaces of finite and infinitesimal extent.

 \begin{remark}[\bf Petit sheaf characterization]\label{PetitSheafCharacterization} 
 The categorical sheaf condition under the coverage \eqref{GoodOpenCoversonThCartSp} is related to the usual topological sheaf condition in the following sense: A presheaf $\CF:\mathrm{ThCartSp}^{op}\rightarrow \mathrm{Set}$ is a thickened smooth set if and only if $\CF(\FR^k\times \DD)$ defines a sheaf on $\FR^k$, in the topological sense, for each $\FR^k\in \mathrm{CartSp}$ and every infinitesimal point $\DD$.
\end{remark}

Of course, the Yoneda embedding (see, e.g.,  \cite[Prop. 2.4]{GS23}) now applies for our thickened site and sheaves over it, yielding a fully faithful embedding\footnote{For a generic sheaf category $\mathrm{Sh}(C)$ with respect to a chosen coverage on $C$, it is plausible that under the Yoneda embedding not every object of the site $C$ satisfies the sheaf (`gluing') condition. In other words, it is not generally guaranteed that the Yoneda embedding factors through the corresponding full sheaf subcategory $y: C \hookrightarrow \mathrm{Sh}(C)\hookrightarrow \mathrm{PreSh}(C)$. Site categories with the property of representables being sheaves, and not just presheaves, are called \textit{``subcanonical''}. Both $\CartesianSpaces$ and $\ThickenedCartesianSpaces$ with their corresponding good open coverages form examples of subcanonical sites.}
\begin{align}\label{ThCartSpYonedaEmbedding}
y \; : \; \ThickenedCartesianSpaces & \; \xhookrightarrow{\quad \quad} \; \ThickenedSmoothSets 
\\
\FR^k\times \DD & \; \xmapsto{\quad \quad} \; \mathrm{Hom}_{\ThickenedCartesianSpaces}(-, \, \FR^k\times \DD) := \mathrm{Hom}_{\mathrm{CAlg}_{\FR}}\big(\CO(\FR^k \times \DD),\, \CO(-)\big)\big|_{\mathrm{ThCartSp}}  \nn \, ,
\end{align}
which allows one to view thickened Cartesian spaces as special instances of thickened smooth sets.

\medskip
\noindent {\bf Examples of thickened smooth sets.}
Following Def. \ref{ThickenedCartesianSpaces}, one may define an analogous category of (infinitesimally) thickened smooth manifolds \footnote{There is a more general definition of thickened manifolds which are only (petit) \textit{locally} isomorphic to the product form $\FR^k\times \DD$ (cf. \cite{Kock06}, therein termed ``formal manifolds''). This is in complete analogy to the case of super-manifolds being only locally products of super-Cartesian spaces $\FR^{k,q}=\FR^{k,0}\times \FR^{0,q}$. It can be shown in a similar manner that sheaves over these ``formal manifolds'' also yield an equivalent topos.} 
 $$
 \ThickenedSmoothManifolds\xhookrightarrow{\quad \quad} \mathrm{CAlg}_\FR^{op}\, ,
 $$
 with objects denoted by $M\times \DD$ being formal duals of the corresponding algebras $C^\infty(M)\otimes \CO(\DD)$. Hence, we could as well consider thickened generalized smooth spaces as the sheaf category $\mathrm{Sh}(\ThickenedSmoothManifolds)$ over thickened smooth manifolds with the coverage of being that of (differentiably) good open covers \eqref{GoodOpenCoversonThCartSp}. However, this category is precisely equivalent to that of thickened smooth sets (Def. \ref{ThickenedSmoothSets}), in a manner completely analogous to how sheaves over plain smooth manifolds are equivalent to smooth sets (see discussion around  \cite[Def. 2.1]{GS23}). In more detail, the pullback along the inclusion $\iota: \ThickenedCartesianSpaces \hookrightarrow \ThickenedSmoothManifolds $ defines an equivalence of categories 
\begin{align}\label{SheavesOverManifoldsAndSheavesOverCart}
 \iota^* \, : \, \mathrm{Sh}(\ThickenedSmoothManifolds)   \xrightarrow{\quad \sim \quad} \mathrm{Sh}(\ThickenedCartesianSpaces)  \, . 
 \end{align}
This is essentially because any (thickened) manifold $M\times \DD$ is always covered by a (differentiably good) open cover 
$\{U_i\times \DD \hookrightarrow M\times \DD \}_{i\in I}$ of (thickened) Cartesian spaces. In topos theoretic language, $\ThickenedCartesianSpaces$ is a \textit{dense sub-site} of $\ThickenedSmoothManifolds$.
 
\begin{example}[\bf Thickened manifolds as thickened smooth sets]\label{ManifoldsAsThickenedSmoothSets} It follows that the composition which views a thickened smooth manifold as a thickened smooth set $\ThickenedSmoothManifolds \xrightarrow{y} \mathrm{Sh}(\ThickenedSmoothManifolds) \xrightarrow{i^*} \mathrm{Sh}(\ThickenedCartesianSpaces)$, abbreviated by 
\begin{align}\label{ManifoldasFormalsSmoothset}
	  	y \;: \;  \ThickenedSmoothManifolds& \; \xhookrightarrow{\quad \quad}\; \ThickenedSmoothSets 
        \\[-2pt]
	  	 M\times \DD &\;\xmapsto{\quad \quad} \; \mathrm{Hom}_{\ThickenedSmoothManifolds}(-,\, M\times \DD)|_{\mathrm{ThCartSp}} := \mathrm{Hom}_{\mathrm{CAlg}_{\FR}}\big(C^{\infty}(M)\otimes \CO(\DD),\, \CO(-)\big)\big|_{\mathrm{ThCartSp}} \nn \, ,
	  \end{align}
   is also fully faithful.
\end{example}

The same formula can be used to view smooth (second countable and Hausdorff) manifolds \textit{with boundary}\footnote{In fact, this result extends to the case of manifolds with corners. This extension will appear in \cite{Gi26}.} as thickened smooth sets. 
\begin{example}[\bf Manifolds with boundary]\label{ManifoldsWithCornersExample}
For any manifold with boundary $M\in \SmoothManifolds^{\mathrm{bnd}}$, locally modelled on $\FR^{k,1}:= \FR^k \times [0,\infty)$, there is a corresponding thickened smooth set defined by
$$
y(M) \, := \, \mathrm{Hom}_{\mathrm{CAlg}_\FR}\big( C^\infty(M),\, \CO(-)\big) \big|_{\ThickenedCartesianSpaces} 
$$
where $C^\infty(M)$ is algebra of \textit{smooth} functions 
on $M$. The proof that this presheaf satisfies the sheaf condition with respect to \eqref{GoodOpenCoversonThCartSp} can be found in \cite{Reyes07}, streamlined in \cite[Thm. 9.6]{Kock06}, where it is also shown that the functor
$$
y\, : \, \SmoothManifolds^{\mathrm{bnd}}\xhookrightarrow{\quad \quad} \ThickenedSmoothSets
$$
is, in fact, \textit{fully faithful}.\footnote{Therein, the Hom set is over $C^\infty$-algebra morphisms. But by Whitney's embedding theorem (for manifolds /w boundary) and our results Thm. \ref{ThickenedCartesianSpacesAreSmoothLoci}, Cor. \ref{CahiersToposIsRalgebraic}, one can see that these are again exhausted by $\FR$-algebra morphisms.} Hence, manifolds with boundaries are also represented ``from the outside''. 
\end{example}
Moreover, the same formula yields a thickened smooth set corresponding to any \textit{closed} subset $|K| \subset \FR^n$, viewed as a formal smooth space $K$ dual to the subalgebra of functions on $|K|$ which extend to smooth functions on $\FR^n$.
\begin{example}[\bf Smooth closed subspaces]\label{SmoothClosedSubspacesExample} 
For any formal ``smooth space'' $K\in \mathrm{CAlg}^\mathrm{op}$ dual to the algebra of functions 
$$
C^\infty(K) \, \cong \, C^\infty(\FR^n)/I_K \, ,
$$
where $|K|\subset \FR^n$ is a closed subset and $I_K$ is the ideal of smooth functions on $\FR^n$ which vanish along $|K|$, there is a corresponding thickened smooth set defined by
$$
y(K) \, := \, \mathrm{Hom}_{\mathrm{CAlg}_\FR}\big( C^\infty(K),\, \CO(-)\big) \big|_{\ThickenedCartesianSpaces}\,. 
$$
The sheaf condition for the presheaf $y(K)$ holds precisely as in the proof of the previous example.
\end{example}

\begin{remark}[\bf On Yoneda notation]To ease the heavy notation, we will often omit the Yoneda embedding symbols from Eqs. \eqref{ThCartSpYonedaEmbedding} \eqref{ManifoldasFormalsSmoothset} and Examples \ref{ManifoldsWithCornersExample}, \ref{SmoothClosedSubspacesExample}, as it will be clear from the context in which category we consider these objects. We will explicitly use the Yoneda symbols only when they are necessary and/or instructive.
\end{remark}

Another immediate example of a thickened smooth set is the ``moduli space of de Rham $n$-forms'' $\mathbold{\Omega}^n$, which we may naturally extend from its purely smooth set incarnation (\cite[Def. 2.29]{GS23}) to a working definition in the current context. 
In algebro-geometric terms, the infinitesimal $\DD$-plots of $\mathbold{\Omega}^n$ are given by the ``K\"{a}hler differential forms'' \footnote{One can readily check explicitly that this is the same as K\"{a}hler forms on $\CO(\DD)$ as a $C^\infty$-algebra (Rem. \ref{VariantsOfThickenedSites}, ftn. 3, Def. \ref{FinitelyGeneratedCinftyAlgebras}), as it is finitely generated (cf. \cite[\S 2]{DK84}\cite[\S 5]{Joyce19}). This latter notion is formally the correct one in our context, as it also recovers the traditional differential forms on $\FR^k$ (while the $\FR$-algebraic one does not). See Appendix \ref{ModuliOfDifferentialForms} for more on these details.} on the finitely generated $\FR$-algebra $\CO(\DD)$. Explicitly, 
\begin{align*}
\mathbold{\Omega}^{1}\big(\DD^{m}(l)\big)\; :&= \; \CO\big(\DD^m(l)\big)\cdot(\dd \epsi^1,\cdots, \dd \epsi^m) 
\big/ \dd (\epsi^1,\cdots , \epsi^m)^{l+1}\, ,
\end{align*}
with the extra variables $\{\dd \epsi^i\}_{i=1,\cdots,m}$ being of degree 1, where the equivalence relation is the induced one $0=\dd (\epsi^{i_1}\cdots \epsi^{i_{l+1}}) = \dd \epsi^{i_1}\cdot \epsi^{i_2}\cdots \epsi^{i_{l+1}} +\cdots + \epsi^{i_1}\cdots \epsi^{i_{l}}\cdot \dd \epsi^{i_{l+1}}$. More generally, this means its $\DD$-plots 
are given by the $\CO(\DD)$-module freely generated in the formal variables $\{ \dd v\}_{v\in \CO(\DD)}$ subject to the induced algebra relations via $\dd(v_1 \cdot v_2)= \dd v_1 \cdot v_2 + v_1 \cdot \dd v_2$, $\dd (v_1 + v_2) = \dd v_1 + \dd v_2$, $\dd r =0$ for $r\in \FR \hookrightarrow \CO(\DD)$. That is,
\begin{align}\label{AlgebraicFormsOnInfinitesimalPoint}
\mathbold{\Omega}^1 (\DD)\, := \, \mathrm{Span}_{\CO(\DD)}\big\{ \dd v \, | \, v \in \CO(\DD)\big\} 
\big/
\dd \big( p(v^1,\cdots,v^n)\big) = \sum_i \frac{\partial p}{\partial x^i} (v^1,\cdots, v^n) \cdot \dd v^i  ,
\end{align}
where $p\in \mathrm{Pol}(\FR^n)$ is any polynomial in $n$-variables. As expected, this reduces to the above simplified form for $\DD = \DD^m(l)$. Along the lines of Lem. \ref{InfinitesimalPointsAsSubspacesOfInfinitesimalDisks}, this can take a more explicit form for a general infinitesimal point $\DD$, if necessary, via any identification $\CO(\DD) \cong \CO\big(\DD^m(l)\big) / I\cong C^\infty(\FR^m) / \widetilde{I}$ for $\widetilde{I}\subset C^\infty(\FR^m)$ the full lifted ideal, whereby
\begin{align}\label{KahlerFormsInQuotientIdentification}
\mathbold{\Omega}^{1}\big(\DD\big)\; \cong \; C^\infty(\FR^m)\cdot(\dd x^1,\cdots, \dd x^m)\,  / \, \dd \widetilde{I} \, .
\end{align}
Combining appropriately with the smooth $\FR^k$-plots of $\mathbold{\Omega^1}$ yields the definition of the thickened moduli space of de Rham forms.
\begin{example}[\bf Thickened moduli space of de Rham forms]\label{ThickenedModuliSpaceOfdeRhamForms}
The \textit{thickened moduli space of de Rham $1$-forms} $\mathbold{\Omega}^1$ is the thickened smooth set with $\FR^k\times \DD$-plots
 $$ \mathbold{\Omega}^1(\FR^k\times \DD)\; := \; 
 \mathbold{\Omega}^1(\FR^k)\otimes \CO(\DD) \; \oplus \; C^\infty(\FR^k) \otimes \mathbold{\Omega}^1(\DD) \, ,$$
Similarly, the \textit{moduli space of $n$-forms} $\mathbold{\Omega}^n$ is the thickened smooth set with $\FR^k\times \DD$-plots 
$$
\mathbold{\Omega}^{n}(\FR^k\times \DD)\; := \;  \bigwedge^{n}_{\CO\left(\FR^k\times \DD\right)} \mathbold{ \Omega}^1\big(\FR^k\times\DD\big) \, .
$$
\end{example}

By construction, as with the smooth set case \cite[Lem. 2.31]{GS23} these thickened smooth sets yield a DGCA structure internal to $\ThickenedSmoothSets$
$$
\big(\mathbold{\Omega}^\bullet, \, \dd_\mathrm{dR},\, \wedge\big) \;\; \in \;\; \mathrm{DGCA}_\FR(\ThickenedSmoothSets)\, .
$$
which can be, yet again, used to define a notion of \textit{de Rham forms} on any \textit{thickened} smooth set $\CF\in \ThickenedSmoothSets$ via (cf. \cite[Def. 2.32]{GS23})
\begin{align}\label{deRhamFormsOnThickenedSmoothSet}
\Omega^\bullet_\mathrm{dR}(\CF)\, : = \, \Hom_{\ThickenedSmoothSets}(\CF,\, \mathbold{\Omega}^\bullet)\, . 
\end{align}

\begin{remark}[\bf Variants of thickened sites]\label{VariantsOfThickenedSites}
$\,$

\begin{itemize}[leftmargin=20pt]
\item[\bf (i)]
There are several different variants of sites of smooth and infinitesimally thickened probe-spaces one may consider, \footnote{Namely these choices correspond to full subcategories of the category of ``\textit{smooth loci}\,'', the formal dual category of the full subcategory of finitely generated $C^\infty$-algebras (Def. \ref{FinitelyGeneratedCinftyAlgebras}), i.e., of the form $C^\infty(\FR^k) / I$ for some ideal $I\subset C^\infty(\FR^k)$ and $k\in \NN$. The choices are determined by conditions imposed on the allowed ideals $I$, but also by different choices of coverages one may supply the probe-category with \cite{MoerdijkReyes}.} and hence many different categories of sheaves that serve as ``well-adapted models'' for synthetic differential geometry. 

\item[\bf (ii)]
This means that the category of smooth manifolds $\SmoothManifolds$ embeds fully faithfully in any such well-adapted sheaf topos $\iota: \SmoothManifolds \hookrightarrow \mathcal{E}$, while preserving transversal intersections/pullbacks among a couple further technical axioms
(see e.g. \cite{Kock06}). The existence of different models for synthetic differential geometry guarantees that several abstract arguments and universal constructions may be equally interpreted in any of these categories of generalized spaces. Thus, these categories share several properties, but may in fact differ in more subtle technical details and not be actually equivalent. A detailed analysis of several such sites, their categories of sheaves, and their interrelations has been carried out in \cite{MoerdijkReyes}. 

\item[\bf (iii)]
In this manuscript, we focus on arguably the simplest of these sites, with the least amount of probe spaces, and the corresponding sheaf topos $\ThickenedSmoothSets$ being equivalent to the ``Cahiers topos'' (Thm. \ref{ThickenedCartesianSpacesAreSmoothLoci}, Cor. \ref{CahiersToposIsRalgebraic}) of Dubuc \cite[\S 4]{Dubuc79}\cite{KockReyes87}, since this seems to be sufficiently rich for most purposes of (local) Lagrangian field theory. An instance where a more general choice of thickened site might be more appropriate, from a purely mathematical point of view, is to make manifest the full ``classifying'' property of $\mathbold{\Omega}^n$ as defined above, \footnote{Alternatively, there exists a different, and more abstract, extension of the classifying space from smooth sets. Namely, one may assign instead to any probe $\FR^k\times \DD $ the set of ``fiberwise linear maps'' of thickened smooth sets, $\mathrm{Hom}_{\ThickenedSmoothSets}^{\mathrm{fib.lin.}}\big(T(\FR^k\times \DD),\, y(\FR)\big)$, from its synthetic tangent bundle (Def. \ref{SyntheticTangentBundle}, Ex. \ref{TangetBundleOfInfinitesimalPoint}) into the real numbers. The resulting thickened smooth set satisfies the full classifying property over $\ThickenedCartesianSpaces$. See Appendix \ref{ModuliOfDifferentialForms} for more on this.} along the lines 
of the original sources \cite[\S 3]{Law80}\cite[\S 8]{DK84} (described therein in more abstract terms, hence perhaps somewhat underappreciated). We review the technicalities of this matter in a brief manner in the Appendix 
(\cref{ModuliOfDifferentialForms}), as it is not directly relevant to the purposes of local Lagrangian field theory; This being due to the fact that local forms on $\CF\times M$ may still be identified with certain maps into $\mathbold{\Omega}^n$ from Def. \ref{ThickenedModuliSpaceOfdeRhamForms} (see Rem. \ref{LocalFormsViaModuliSpace}).
\end{itemize}
\end{remark}

\subsection{Thickened mapping spaces and the synthetic tangent bundle}
\label{Subsec-thickmap}

Since our thickened smooth sets form a (pre-)sheaf category, there exists a natural mapping space construction for any pair of thickened smooth sets. This is witnessed by the internal hom functor
$$
\big[-,-\big] \;:\; \ThickenedSmoothSets^{op}\times \ThickenedSmoothSets \;\xrightarrow{\quad \quad} \;\ThickenedSmoothSets\, ,
$$ 
which is defined exactly by the same formula as in that of smooth sets \cite{GS23}, and in fact as for any (pre)sheaf category \cite{MacLaneMoerdijk}. We recall its defining formula for completeness, now in thickened smooth sets. 

\begin{definition}[\bf Thickened smooth mapping space]\label{ThickenedSmoothMappingSet} Let $\CG,\CH \in \ThickenedSmoothSets$, the \textit{(thickened) smooth mapping space} $[\CG,\CH]\in \ThickenedSmoothSets$ is defined by
$$
\big[\CG,\CH\big]\big(\FR^k\times \DD\big) \; := \; \mathrm{Hom}_{\ThickenedSmoothSets}\big( y(\FR^k \times \DD)\times\CG, \, \CH\big)\, .
$$
\end{definition}
Crucially, thickened mapping spaces satisfy the corresponding internal hom property (or Exponential Law) property, 
\begin{align}\label{InternalHomProperty}
\mathrm{Hom}_{\ThickenedSmoothSets}\big(\CX, \,[ \CG, \CH]\big)
\;  \cong \;
\mathrm{Hom}_{\ThickenedSmoothSets
}\big( \CX \times \CG , \, \CH\big) \, ,
\end{align}
naturally in $\CX \in \ThickenedSmoothSets$. This too, is proved using the fact that any thickened smooth set may be written as a colimit of thickened Cartesian spaces (representables), verbatim as in the case of smooth sets (\cite[p. 12]{GS23}), and in fact as in any (pre)sheaf category (\cite[p. 42]{MacLaneMoerdijk}). The internal hom construction allows us to consider infinitesimally thickened mapping spaces between any two manifolds.

\begin{example}[\bf Thickened mapping space of manifolds]\label{ThickenedMappingSpaceOfManifolds}
Let $M, N \in \SmoothManifolds \hookrightarrow \ThickenedSmoothSets$. The \textit{thickened (smooth) space of maps} from $M$ to $N$ given by
\begin{align*}
[M,\, N]\big(\FR^k\times \DD\big)&\; :=\; \mathrm{Hom}_{\ThickenedSmoothSets}\big((\FR^k\times\DD)\times M, \, N\big)
\\ 
& \; \cong \;  \mathrm{Hom}_{\mathrm{CAlg}_{\FR}}\big(C^{\infty}(N),\, C^{\infty}(M\times \FR^k)\otimes \CO(\DD)\big)\, ,
\end{align*}
where $\FR^k\times \DD\in \mathrm{ThCartSp}$ is any infinitesimally thickened Cartesian probe-space. 
\end{example}
In field theoretic terms, this neatly encodes the smooth \textit{and} infinitesimal nature of the (off-shell) space of fields for a ``$\sigma$-model'' with domain $M$ and target $N$. Similarly, more general field spaces consisting of sections of an arbitrary field fiber bundle $F\rightarrow M$, viewed originally simply as smooth sets in 
 \cite[Def. 2.12]{GS23}, are promoted to thickened smooth sets by reinterpreting their defining diagram within $\ThickenedSmoothSets$ .
\begin{definition}[\bf Thickened smooth set of sections]\label{ThickenedSmoothSetOfSections}
Let $\pi: F\rightarrow M$ be a fiber bundle of smooth manifolds. The thickened smooth set of sections $\CF=\mathbold{\Gamma}_M(F)\in \ThickenedSmoothSets$ is defined by
\begin{align*}
	\CF(\FR^k\times \DD)\; :=\; \big\{\phi^{k,\epsi}:\FR^k\times \DD \times M \rightarrow F \; | \; \pi\circ \phi^{k,\epsi} = p_M \big\}\, ,
\end{align*}
where $\FR^k\times\DD\in \mathrm{ThCartSp}$ and $p_M :\FR^k\times\DD\times M\rightarrow M$ is the projection onto $M$. Explicitly, these consist of (formal duals of) algebra maps $(\phi^{k,\epsi})^*: C^{\infty}(F) \rightarrow  C^{\infty}(M\times \FR^k) \otimes \CO(\DD) $ such that
the following diagram commutes 
\[ 
\xymatrix@R=1.4em  { &&  C^{\infty}(F) \ar[lld]_<<<<<<<<<{(\phi^{k,\epsi})^*}
	\\ 
	C^{\infty}(M\times \FR^k) \otimes \CO(\DD) && C^{\infty}(M)\;.  \ar[ll]_<<<<<{p^*_M} \ar[u]_{\pi^*}
}   
\]

\end{definition}
The thickened bosonic field space above may be equivalently identified as a pullback / fiber product construction (cf. discussion above
Rem 2.13 of \cite{GS23}), computed in thickened smooth sets 
\begin{align} \label{FieldSpaceAsPullback}
\xymatrix@=1.6em 
{\CF  \ar[d] \ar[rr] &&  
[M,F] \ar[d]^{\pi_*} 
	\\ 
	\mathrm{id}_M\cong* \; \ar@{^{(}->}[rr]  && [M,M]
	\, , } 
\end{align}
since we have already extended the three corners of the diagram to thickened smooth sets. In more detail, the mapping spaces on the right column are via Ex. \ref{ThickenedMappingSpaceOfManifolds}, while the bottom embedding is defined probe-wise by sending the (unique) $(\FR^k\times \DD)$-plot of the point $\id_M\cong *$ to the projection-plot  $p_M : \FR^k\times \DD \times M \rightarrow M$ inside 
$[M,\, M](\FR^k\times \DD)$, for each $\FR^k\times \DD \in \ThickenedCartesianSpaces$. 

\medskip 
As with any (pre)sheaf category, this is computed as the probe-wise fiber product in the category of sets, for each $\FR^k\times \DD\in \ThickenedCartesianSpaces$, and immediately recovers the plots of Def. \ref{ThickenedSmoothSetOfSections}. The top arrow then exhibits the field space of sections as a (thickened smooth) subspace of the mapping space
$$\CF \xhookrightarrow{\quad \quad} [M,F]\, ,$$
as expected by the point-set inclusion $\Gamma_M(F) \hookrightarrow C^\infty(M,F)$. We note that, as with \cite[Rem. 2.14]{GS23}, field spaces of sections should really be viewed as also being (petit) sheaves over $\mathrm{Open}(M)$, hence forming overall a sheaf 
\begin{align*}
\mathbold{\Gamma}_{(-)} (F)\; : \;  \big(\mathrm{Open}(M)\times \ThickenedCartesianSpaces\big)^{\op} &\xrightarrow{\quad \quad}  \mathrm{Set} \, \\  
(U,\, \FR^k\times \DD) &\xmapsto{\quad \quad} \mathbold{\Gamma}_U(F)(\FR^k \times \DD)        \,
\end{align*}
with respect to the product coverage of sites, but we shall suppress this point throughout in favor of ease of exposition.

\medskip
\noindent
{\bf Synthetic tangent bundles.} 
 The fruitful -- and quite intuitive -- interaction of infinitesimal spaces with the internal hom functor of thickened smooth sets is one of the main reasons for the thickening of our Cartesian site. For instance, 
by Lem. \ref{TangentBundleSet} and the Yoneda embedding of Ex. \ref{ManifoldsAsThickenedSmoothSets}, the following \textit{sets} are in canonical bijection
\begin{align*}
	TM\, &\cong \,\mathrm{Hom}_{\mathrm{CAlg}_{\FR}}\Big(C^\infty(M), \CO\big(\DD^1(1)\big)\!\Big)\\
 \; &\cong \;
\mathrm{Hom}_{\ThickenedSmoothSets}\big(\DD^1(1),\, M\big)\, .
\end{align*}
That is, maps of thickened smooth sets from the infinitesimal line $\DD^1(1)$ to a manifold $M$ detect precisely its tangent vectors. By the same formula, however, one may define two thickened smooth set structures on the tangent bundle set $TM$, which could (a priori) be different: the first by embedding it as a manifold via Ex. \ref{ManifoldsAsThickenedSmoothSets} and \ref{ManifoldsWithCornersExample}, denoted by $y(TM)$, and the second by the internal hom construction $\big[\DD^1(1),\,  M\big]$ given as in Def. \ref{ThickenedSmoothMappingSet}. The latter defines the \textit{``synthetic tangent bundle''} of $M$, and in fact recovers the smooth structure of the traditional tangent bundle $TM$ in the following sense.
\begin{proposition}[\bf Synthetic tangent bundle of manifold]\label{TangBundlesCoincide}
Let $M\in \SmoothManifolds^{\mathrm{bnd}}$ be a manifold with boundary\footnote{In fact, the case of corners follows with minimal modifications \cite{Gi26}.} and $TM \in \SmoothManifolds^{\mathrm{bnd}}$ its traditional tangent bundle. Then $y(TM)\in \ThickenedSmoothSets$ coincides with the synthetic tangent bundle of $M$, 
$$
y(TM)\, \cong\, \big[\DD^1(1), \,  M\big] \;\; \in \;\; \ThickenedSmoothSets\, .
$$
\begin{proof}
We only show the isomorphism at the level of smooth sets here, i.e., restricting the site to Cartesian spaces, as it is more straightforward to spell out. The thickened probes require some more argumentation, which has been worked out for the boundaryless (\cite{Dubuc79}\cite[Prop. 1.12]{MoerdijkReyes}) and boundary case in  (\cite[Prop. 2.8]{Reyes07}). The $\FR^k$-plots of the synthetic tangent bundle are given by 
	$$
 \big[D^1(1),\, M\big]\big(\FR^k\big)= \big[D^1(1)\times\FR^k,\, M\big]\cong \mathrm{Hom}_{\mathrm{CAlg}_{\FR}}\Big(C^{\infty}(M),\,  \CO\big(D^{1}(1)\big)\otimes C^{\infty}\big(\FR^k\big)\!\Big)\, ,
 $$
An algebra morphism $C^{\infty}(M)\rightarrow \FR[\epsi]/\epsi^2\otimes C^{\infty}(\FR^k)\cong_{\mathrm{Vect}_{\FR}}C^{\infty}(\FR^k)\oplus \epsi \cdot C^{\infty}(\FR^k)$ is given, suggestively, by 
$$
(p,X_{p})\;:\; f \;\longmapsto \; p(f) + \epsi \cdot X_p(f)\, .
$$
	Expanding the algebra morphism property, and using $\epsi^2=0$,
	$$
 (f_1 \cdot f_2)\;\longmapsto \; p(f_1)\cdot p(f_2) + \epsi  \big( p(f_1)\cdot X_p (f_2) + X_p (f_1) \cdot p(f_2) \big)\, .
 $$
	The first component defines an algebra map $p: C^\infty(M)\rightarrow C^{\infty}(\FR^k)$, and so dually a smooth map of manifolds $\FR^k\rightarrow M$ by Prop. \ref{CartSptoAlgebras}. For each $x\in \FR^k$, we may compose with evaluation at x (also an algebra map) to land in $\FR \oplus \epsi \cdot \FR$. Then the second component is identified as a tangent vector $X_{p(x)}$ at $p(x)\in M$ for each $x\in \FR^k$, as in Lem. \ref{TangentBundleSet}. In particular, these fit in a diagram as
		\[ 
	\xymatrix@R=1.5em@C=2.6em  { &&  TM \ar[d]^{\pi}
		\\ 
		\FR^k \ar[rru]^{X_p} \ar[rr]^>>>>>>>>{p} && M
	}   
	\]
	of smooth manifolds. Equivalently, these are sections of the pullback bundle, $X\in \Gamma_{\FR^k}(p^*TM)$, acting as derivations $X_p: C^\infty(M) \rightarrow C^\infty(\FR^k)$ relative to the algebra morphism $p:C^\infty(M)\rightarrow C^\infty(\FR^k)$, and hence are \textit{smooth} sections -- by the same reasoning of derivations  $C^\infty(M)\rightarrow C^\infty(M)$ being in bijection with \textit{smooth} sections of $TM$ (a fact which again relies on Hadamard's Lemma \ref{HadamardsLemma}); for a proof see \cite[\S 4]{Flanders}.     
    On the other hand, a plot in $y(TM)$ is given simply by
	$$
    X\;\;\in\;\; y(TM)\big(\FR^k\big)= \mathrm{Hom}_{\SmoothManifolds^{\mathrm{bnd}}}\big(\FR^k,TM\big)\, ,
    $$
	i.e., a smooth map of manifolds $X: \FR^k \rightarrow TM$. By composing with the projection $\pi:TM\rightarrow M$, this is equivalent to a triangle as above. Hence the two sets of plots are in bijection, which is easily seen to be functorial under pullbacks of probe Cartesian spaces.
\end{proof}
\end{proposition}

\begin{remark}[\bf Mapping out of infinitesimal points]\label{infinitesimalMappingExamples} Working along the same lines  it is not hard to see that maps out of infinitesimal disks of higher order and dimension detect higher order tangent bundles and variants thereof. 

 \begin{itemize}[leftmargin=20pt]
\item[{\bf (i)}] For instance, for $M\in \SmoothManifolds^\mathrm{bnd}$, expanding explicitly the defining conditions gives 
\begin{align}\label{ManifoldTangentBundleFiberProduct}
\big[\DD^2(1),\, M\big]\, \cong \, y(TM\times_{M} TM)\, \cong \, y(TM) \times_M y(TM)
\end{align} 
as the fiber product of $TM$ with itself. For instance, at the level of $*$-plots, this isomorphism is given by sending a pair of tangent vectors $(p,X^1_p): C^\infty(M) \rightarrow \FR[\epsi^1]/(\epsi^1)^2$ and $(p,X^2_p): C^\infty(M) \rightarrow \FR[\epsi^2]/(\epsi^2)^2$ to (the formal dual) of 
\begin{align*}
\big(p,X^1_p,X^2_p\big) \; : \; C^\infty(M) &\; \longrightarrow \; \FR[\epsi^1,\epsi^2]\big/(\epsi^1,\epsi^2)^2 \\
f &\; \longmapsto \; p(f) + \epsi^1 \cdot X_p^1(f) + \epsi^2 \cdot X_p^2(f) \;.
\end{align*}
Notice that this assignment defines an algebra map if and only if $\epsi^1 \cdot \epsi^2 =0$; namely, it \textit{cannot extend} a map from the product $\DD^1(1)\times \DD^1(1)$ into $M$ (unless either of $X^{1,2}_p$ vanishes).

\item[{\bf (ii)}]
Similarly, $\big[\DD^1(1)\times \DD^1(1),\, M\big] \cong y\big(T(TM)\big)$ is the second-order tangent bundle, which may also be deduced abstractly by repeatedly applying the internal hom property \eqref{InternalHomProperty} (along with Prop. \ref{TangBundlesCoincide}). Along the same lines, $\big[\DD^1(2),\, M\big]\cong y\big(J^2_0(\FR,M)\big)$ is the manifold of second order jets of curves in $M$ and, analogously, $\big[\DD^m(l),\, M]\cong y\big(J^l_0(\FR^m,M)\big)$ is the manifold of $l$-jets of $\FR^m$-plots in $M$. 

More generally, for any infinitesimally thickened point $\DD$ (Def. \ref{InfinitesimallyThickenedPoints}), the mapping space into any manifold $M$ is itself a manifold 
\cite{Gi26}\cite[Thm. 3]{Kock81}\cite[Prop. 2.8]{Reyes07} 
\begin{align}\label{WeilBundleMappingSpace}
\big[\DD,\, M\big] \; \cong \; y(T_\DD M)\, ,
\end{align}
where
$T_\DD M \in \SmoothManifolds^{\mathrm{bnd}}$
is the ``\textit{Weil bundle}'' over $M$ with respect to the Weil algebra $\CO(\DD)$, originally due to A. Weil\footnote{Weil showed that a natural smooth fiber bundle structure exists on the set of algebra morphisms 
$$
T_{\DD} M \, := \,  \Hom_{\mathrm{CAlg}_\FR}\big(C^\infty(M),\, \CO(\DD)\big) \, \cong \, \Hom_{\ThickenedSmoothManifolds}(\DD, M)\, ,
$$
for each Weil algebra $\CO(\DD)$ and a manifold (originally without boundary). Indeed, it is not hard to see that for $M=\FR^{n,1}$ and $\CO(\DD)\cong \FR\oplus V$ that $T_\DD \FR^{n,1}\cong \FR^{n,1}\times V^{\times(n+1)}$. Hence any choice of basis for $V\cong \FR^w$ yields a bijective chart on $T_\DD \FR^{n,1}\xrightarrow{\sim} \FR^{n+w(n+1)}\times [0,\infty)$. The construction is functorial, hence for a cover of charts on the set $M$, yields a compatible cover of charts on $T_\DD M$. In fact, the construction easily extends to the case of manifolds with corners \cite{Gi26}.} \cite{Weil53} (see also \cite{Kolar08} for a modern recount), to be thought of as a kind of generalized tangent (or jet prolongation) bundle.

\item[{\bf (iii)}] Furthermore, as we shall see in \cref{JetBundlesSyntheticallySec}, mapping appropriately out of special infinitesimal (neighborhoods of) points canonically attached to points of manifolds $M$ detects (or \textit{defines}) smooth spaces of jets of sections of arbitrary fiber bundles over $M$. 
\end{itemize} 
\end{remark}

The upshot in the identification of the tangent bundle $TM$ with the mapping space $\big[\DD^1(1),M\big]$ is that the latter construction may be applied for an arbitrary target $\CF\in \ThickenedSmoothSets$. Hence, following \cite{Law80}\cite{Law98}, we may extend it to a definition of the ``\textit{synthetic tangent bundle}'' of any thickened smooth set. The discussion and examples that follow will show that this is indeed a reasonable definition, making into reality the intuition of tangent vectors being `first-order infinitesimal curves', as often invoked in particular in the field theory literature (e.g.,  \cite[Ex. 15]{GS23}).
\begin{definition}[\bf Synthetic tangent bundle]\label{SyntheticTangentBundle}
The \textit{synthetic tangent bundle} $T \CF$ of a thickened smooth set $\CF \in \ThickenedSmoothSets$ is defined as the mapping thickened smooth set
$$
T\CF\,  :=\,  \big[\DD^1(1),\,  \CF\big] \;\; \in \;\; \ThickenedSmoothSets \, .
$$
\end{definition}
The set of tangent vectors on $\CF$ is then given simply as the corresponding $*$-plots 
$$
\big[\DD^1(1), \, \CF\big](*)=\mathrm{Hom}_{\ThickenedSmoothSets}\big(\DD^1(1),\, \CF\big)\, .
$$ 
Crucially, this construction does not yield a bare set of tangent vectors but a fully-fledged (thickened) smooth set structure. The (unique) embedding map  $\iota: *\cong \FR^0 \hookrightarrow \DD^1(1)$ of the point into the infinitesimal line defines a natural projection 
\begin{align}\label{SynthBundleProj}
\xymatrix@=1.6em  {   T\CF :=\big[\DD^1(1),\, \CF\big] \ar[d]^{(-)\circ \iota} 
	\\ 
\CF \cong  [*, \,\CF]
	\, , } 
\end{align}
given, under the Yoneda embedding, by precomposition with $\iota: * \rightarrow  \DD^1(1)$. In the case of $M\in \SmoothManifolds\hookrightarrow \ThickenedSmoothSets$ being a smooth manifold, it is easy to see that under the isomorphism of Prop. \ref{TangBundlesCoincide}, this projection corresponds to the usual bundle projection $TM\rightarrow M$. 

\smallskip 
Furthermore, the construction of the synthetic tangent bundle forms a functor $$T:=\big[\DD^1(1),\, - \big] \;\; : \;\;  \ThickenedSmoothSets \longrightarrow \ThickenedSmoothSets\, ,$$ with smooth maps $f:\CF\rightarrow \CG$ mapped to
\begin{align}\label{SyntheticPushforward} 
Tf\,: \,T\CF & \xrightarrow{\quad \quad} T\CG 
\\[-2pt]
v^k_\epsi&\xmapsto{\quad \quad}  f\circ v^k_\epsi \nn \, ,
\end{align} 
where $v^k_\epsi$ is a plot of tangent vectors, i.e., $v^k_\epsi: \DD^1(1)\times \FR^k\times \DD \rightarrow \CF$ making the composition above manifest. As such, it automatically satisfies the expected functoriality property 
$$
T(g\circ f) \equiv T g\circ T f \;\; : \;\; T\CF \longrightarrow T\CG \longrightarrow T\CH\, , 
$$
for any maps $f:\CF \rightarrow \CG$ and $g: \CG \rightarrow \CH$. Hence $Tf : T\CF \rightarrow T\CG$ may be naturally interpreted as the differential (pushforward) of any map $f: \CF\rightarrow \CG$, and so we will often also denote it as $\dd f: T\CF \rightarrow T\CG$. Indeed, in the case of $\CF=M,\, \CG=N$ being two finite-dimensional manifolds, this recovers precisely the traditional pushforward 
\begin{align}\label{SyntheticPushforwardRecoversTraditional}
Tf \equiv \dd f \; : \; TM \longrightarrow TN \, , 
\end{align}
as seen immediately (plot-wise) by the identification\footnote{For instance, at the point-set level and using the derivation point of view, $Tf$ sends a derivation $(p,X_p)\in T_p M \cong \mathrm{Der}_p\big(C^\infty(M)\big)$ to the derivation $f\circ (p,X_p)\in T_f(p) N \cong \mathrm{Der}_{f(p)} \big(C^\infty(N)\big)$ acting as $h\mapsto h \circ f \mapsto (p,X_p) (h\circ f) = \big(h\circ f(p), X_p(h\circ f)\big)$, which is precisely the definition of the traditional pushforward.} used in the proof of Prop. \ref{TangBundlesCoincide}. 

Let us also note a crucial categorical property of the synthetic tangent functor that will be very  useful in applications. Namely by its very  definition (Def. \ref{SyntheticTangentBundle}) and the internal hom property \eqref{InternalHomProperty}, it enjoys a left-adjoint functor 
\begin{align}\label{TangentFunctorHasLeftAdjoint}
\big(\DD^1(1)\times -\big ) \quad \dashv \quad  T(-)\, .
\end{align}
In other words, there is a canonical bijection 
\begin{align*} 
\mathrm{Hom}_{\ThickenedSmoothSets}\big(\DD^1(1)\times \CF,\, \CG \big) \; \cong \; \mathrm{Hom}_{\ThickenedSmoothSets}(\CF,\, T\CG)  
\end{align*}
natural in both $\CF$ and $\CG$.

\begin{remark}[\bf Infinitesimally linear spaces]\label{InfinitesimallyLinearSpaces} 
Note that any synthetic tangent bundle $T \CF = [\DD^1(1),\, \CF]$ enjoys a fiberwise scalar $\FR$-action structure via a canonical `scaling' map
\begin{align}\label{TangentBundleScalingMap}
\cdot \; : \; \FR \times T\CF \longrightarrow T\CF
\end{align}
induced\footnote{In the ``\textit{internal language}'' of the sheaf topos (see e.g. \cite{Kock06}\cite{MoerdijkReyes}) this action is often denoted simply as $(r,v)\mapsto r\cdot v$, where $r$ and $v$ are arbitrary (generalized) elements (i.e. plots) of $\FR$ and $T\CF$, and $r\cdot v(d) := v(r \cdot d)$ for $d$ any plot of $\DD$. More explicitly, this is defined as the adjunct of the composition 
$
\ev \circ (m \times \id_{T\CF}) \, : \, y\big(\FR \times \DD^1(1)\big) \times T\CF \longrightarrow y(\DD^1(1)) \times T\CF \longrightarrow \CF  
$
under \eqref{TangentFunctorHasLeftAdjoint}, where in turn $\ev: y(\DD^1 (1)) \times T\CF \rightarrow \CF $ is the adjunct of the identity $\id_{T\CF} : T\CF\rightarrow T\CF$ under \eqref{TangentFunctorHasLeftAdjoint}. Spelling out the above plot-wise yields the map 
\vspace{-1mm} 
$$
\big(r_\epsi^k,\, d_\epsi^k, \, \CZ_\epsi^k\big) \longmapsto 
\CZ_\epsi^k \circ \Big(\big(m(r_\epsi^k, d^k_\epsi)\times  \id_{\FR^k\times \DD}\big) \circ \Delta_{\FR^k \times \DD} 
\Big)
$$

\vspace{-2mm} 
where $\Delta_{\FR^k\times \DD} :\FR^k\times \DD\rightarrow (\FR^k\times \DD)\times (\FR^k\times \DD )$ is the diagonal map, which justifies the internal language notation.
} by the (Yoneda embedded) scaling map $m : \FR\times \DD^1(1) \rightarrow \DD^1(1)$ of the real line on the infinitesimal line. 

\begin{itemize} [leftmargin=6mm] 
\item[\bf (i)]  This is given dually by
\begin{align*}
\CO\big(\DD^1(1)\big) &\longrightarrow C^\infty(\FR)\otimes \CO\big(\DD^1(1)\big)
\\
\epsi &\longmapsto t\cdot \epsi
\end{align*}
where $t\in C^\infty(\FR)$ is the  coordinate (identity) function on $\FR$. Of course, the latter scaling map descends precisely from the traditional finite scaling map 
$$
m \, : \, \FR_t \times \FR_x \longrightarrow \FR_x \, ,
$$ 
acting on points as $(t_0,x_0)\mapsto t_0\cdot x_0$, and hence dually on the coordinate function (by pulling back) as $x \mapsto t\cdot x$.

\item[\bf (ii)]   In full generality, however, one is not able to supply a fiberwise $\FR$-linear structure $+ : T\CF \times_\CF T\CF \rightarrow T\CF$. Nevertheless, this is possible whenever the space $\CF$ satisfies a natural categorical property regarding the behavior of mapping spaces of infinitesimal disks $\DD^n(1)$ into it, aptly called ``infinitesimal linearity'' (see e.g. \cite[Prop. 7.2]{Kock06}). In a bit more detail, recall the two canonical embedding maps 
$$
\iota_{1/2} \, : \, \DD^1(1)\longhookrightarrow \DD^2(1)
$$
given dually by the coordinate projections $\epsi_1 \mapsto \epsi, \, \epsi_2 \mapsto 0$ and $\epsi_1\mapsto 0, \, \epsi_2 \mapsto \epsi$, respectively. These yield the \textit{cocone} diagram
$$
\xymatrix@R=1.4em@C=3em
{ {*} \ar[d] \ar[r] &  
\DD^1(1)\ar[d]^{\iota_2 } 
	\\ 
\DD^1(1) \ar[r]^{\iota_1} & \DD^2(1)
	\, , } 
$$
which, in turn, under the internal hom functor into $\CF$ yields the \textit{cone} diagram
$$
\xymatrix@=1.6em@C=2.5em 
{ [\DD^2(1),\,\CF] \ar[d]_{(-)\circ \iota_1} \ar[rr]^-{(-)\circ \iota_2} &&  
T \CF \ar[d] 
	\\ 
T \CF \ar[rr] && \CF  
	\, . } 
$$

\item[\bf (iii)]   Now, if the latter diagram is actually a \textit{limit} cone, then the corresponding (unique) universal map into the fiber product of tangent bundles
$$
k\, : \, [\DD^2(1),\,\CF] \; \xlongrightarrow{\sim} T\CF\times_\CF T\CF 
$$
is an \textit{isomorphism} (cf. Rem. \ref{infinitesimalMappingExamples}{\bf(i)}). This map may be interpreted as a canonical way to uniquely combine any pair of (plots of) tangent vectors $\DD^1(1)\rightarrow \CF$ to (plots of) maps $\DD^2(1)\rightarrow \CF$, and vice-versa.\footnote{Recall the explicit form of  \eqref{ManifoldTangentBundleFiberProduct} as a further motivation for the appearance of $\DD^2(1)$, rather than $\DD^1(1)\times \DD^1(1)$.} In such a situation, one defines a natural fiberwise ``addition'' map of thickened smooth sets over $\CF$ via
$$
+ \;\; : \;\; T\CF \times_\CF T\CF \xlongrightarrow{\;\;k^{-1}\;\;} [\DD^2(1),\, \CF] 
\xlongrightarrow{\quad (-)\circ \mathrm{Diag}\quad } [\DD^1(1),\, \CF ] = T\CF \, , 
$$
where $\mathrm{Diag}: \DD^1(1) \hookrightarrow \DD^2(1)$ is the ``diagonal embedding'' given dually by $\epsi_1 \mapsto \epsi, \, \epsi_2 \mapsto \epsi$. 
The above addition map can then be easily shown to be $\FR$-linear with respect to the natural scaling action \eqref{TangentBundleScalingMap} (cf. \cite[Prop. 7.2]{Kock06}\cite[Prop. 1.3]{MoerdijkReyes}). Its associativity, however, requires further that the cocone diagrams formed by all the higher-dimensional inclusions $\iota_{1/\cdots/n}: \DD^1(1) \hookrightarrow \DD^n(1)$, also  becomes limit diagrams upon (internal) mapping into $\CF$. The totality of the latter conditions is the definition of the ``\textit{infinitesimal linearity}'' property of a thickened space $\CF$. \footnote{Infinitesimal linearity is subsumed by a more general condition, called ``\textit{microlinearity}''  \cite[\S 2.3]{Lavendhomme96}\cite[\S V.1]{MoerdijkReyes} 
(the notion is originally due to \cite{Bergeron}), which states that a larger class of cocone diagrams should turn into limit diagrams upon internal mapping into such a ``microlinear'' space $\CF$. This allows for categorically natural formulations of concepts from classical differential geometry, beyond the fiberwise linearity of tangent bundles.}
We will not need to employ the abstract details of these conditions in this text, since in the 
field-theoretic spaces of interest this fiberwise linear structure exists and is immediately apparent, as we shall see.\footnote{In particular, all smooth manifolds and infinitesimal points are
microlinear spaces. Moreover, the class of microlinear spaces is furthermore closed under limits and the formation of mapping spaces via the internal hom functor \cite[\S V.1]{MoerdijkReyes}. These abstract results guarantee that the field-theoretic tangent bundles constructed in this text (Lem. \ref{ManifoldMappingSpaceTangentBundle}, Prop. \ref{SyntheticTangentBundleOfFieldSpace}, Lem. \ref{SyntheticInfiniteJetTangentBundle}, Lem. \ref{SyntheticTangentBundleOfJetBundleSections}, Cor. \ref{OnshellSyntheticTangentBundle}) will enjoy a fiberwise linear structure.}
\end{itemize} 
\end{remark}

The following is an exotic application of the synthetic tangent bundle construction (from the traditional differential geometric point of view). Namely, it provides a well-behaved (and intuitive) notion of a tangent bundle on infinitesimal disks and points, even if they have only one actual geometric point. Note, any attempt using finite smooth curves would instead, necessarily, yield a trivial tangent bundle. 
\begin{example}[\bf Tangent bundle of infinitesimal disks]\label{TangetBundleOfInfinitesimalPoint}
Viewing the thickened point $\DD^m(l)\in \ThickenedCartesianSpaces$ as a thickened smooth set, its synthetic tangent bundle is defined as
$$T\DD^m(l) \, := \,  \big[\DD^1(1), \, \DD^m(l) \big]\, ,$$
i.e., the thickened smooth space with $(\FR^k\times \DD)$-plots
\begin{align*} 
T \DD^m(l) (\FR^k\times \DD)\, :&= \, \mathrm{Hom}_{\ThickenedSmoothSets}\big(\DD^1(1)\times \FR^k\times \DD , \, 
\DD^m(l) \big)\, 
\\
&\cong \, \mathrm{Hom}_{\ThickenedCartesianSpaces}\big(\DD^1(1)\times \FR^k\times \DD\, , \, 
\DD^m(l)\big)
\\ \, &\cong \, 
\mathrm{Hom}_{\mathrm{CAlg}_\FR}\big( \CO(\DD^m(l)) \, , \, \CO(\DD^1(1))\otimes C^\infty(\FR^k) \otimes \CO(\DD) \big)\, , 
\end{align*}
by the Yoneda embedding. By an application of Lem. \ref{HadamardsLemma}, it is not hard to explicitly check\footnote{This bijection is shown in \cite[Prop. 1.11]{MoerdijkReyes} where they consider only $C^\infty$-algebra maps. Explicitly, it is determined by sending the map which acts on generators as $\epsi^i \mapsto g^i(h,x,\epsi)$ to that acting as $\epsi^i\mapsto g(0,x,\epsi)$ and $y^i \mapsto \partial_h g^i(0,x,\epsi)$. By 
 Prop. \ref{RAlgebraMapsAreCinfty}${\bf (ii)}$, these $C^\infty$-algebra hom-sets actually coincide with those of all $\FR$-algebra maps between the given algebras.} that 
$$
\mathrm{Hom}_{\mathrm{CAlg}_\FR}\big( \CO(\DD^m(l)) \, , \, \CO(\DD^1(1))\otimes C^\infty(\FR^k) \otimes \CO(\DD) \big) \, \cong \,  \mathrm{Hom}_{\mathrm{CAlg}_\FR}\Big( \CO\big(T D^m(l) \big)  \, , \, C^\infty(\FR^k)\otimes \CO(\DD) \Big) \, , 
$$
naturally in $\FR^k\otimes \DD$, where 
$$
\CO\big(T D^m(l)\big)\, := \,\CO\big(\DD^m(l)\big)\otimes C^\infty(\FR^m)\big/ 
\big\{0=y^{i_1}\epsi^{i_2}\cdots \epsi^{i_{l+1}} +\epsi^{i_1}y^{i_2}\cdots \epsi^{i_{l+1}} +\dots+ \epsi^{i_1}\epsi^{i_2}\cdots y^{i_{l+1}}\big\} \, ,
$$
with $\{y^i\}$ being the linear coordinate functions on $\FR^m$, written suggestively as ``functions on the tangent bundle $T \DD^m(l)$". In terms of thickened smooth sets, the above equation says that
$$
T \DD^m(l) \;\; \cong_{\, \ThickenedSmoothSets}  
\;\; \mathrm{Hom}_{\mathrm{CAlg}_\FR}\Big( \CO\big(T D^m(l) \big)  \, , \, \CO(-) \Big) \Big|_{\ThickenedCartesianSpaces} \, , 
$$
and so the tangent bundle to any infinitesimally thickened point $\DD^m(l)$ is ``represented from outside''\footnote{$T \DD^m(l)$ is not an object in our site $\ThickenedCartesianSpaces$, and so not representable in the strict sense.} by the (formal dual) space with function algebra 
$\CO\big(T D^m(l)\big)$. 

Intuitively, the $\CO\big(\DD^m(l)\big)$ factor of the algebra represents the function algebra on the base of $T \DD^m(l) \rightarrow \DD^m(l)$, with its nilpotency representing the infinitesimal nature. On the other hand, the $C^\infty(\FR^m)$ factor represents the functions on the fiber, signaling that even though the base is infinitesimal and has only one actual point, its tangent fiber is not infinitesimal, having infinitely many points. The quotient relation describes how the infinite extent fiber $\FR^m$ is `attached' to the infinitesimal base. Furthermore, the fiber possesses an obvious $\FR$-linear structure, witnessed dually as the map of function algebras\footnote{Strictly speaking, the tensor product on the right is the completed projective one.}
\begin{align}\label{LinearityOfTangentInfinitesimalPoint}
+^* \; : \; \CO\big(T D^m(l)\big) &\xrightarrow{\quad \quad} \CO\big(T D^m(l)\big)\otimes_{\CO(\DD^m(l))} \CO\big(T D^m(l)\big) \\
\epsi^i &\xmapsto{\qquad} \epsi^i \nn \\
f(y^i) & \xmapsto{\qquad} \; f(y_1^i + y_2^i) \nn \, ,
\end{align}
where $\{y_1^i\}$ and $\{y_2^i\}$ correspond to the linear coordinates of each of the (different) copies of $C^\infty(\FR^m)$ within 
the codomain.

Since $T=[\DD^1 \, , \, -]$ has a left adjoint, it preserves limits and, in particular, direct products, and so 
$
T\big(\FR^k\times \DD^m(l)\big)\cong T(\FR^k)\times T\big(\DD^m(l)\big)\in \ThickenedSmoothSets . 
$ Working as above, it follows that the synthetic tangent bundle of an arbitrary probe $\FR^k\times \DD^m(l)\in \ThickenedCartesianSpaces$ is represented by the formal dual of  
$$
\CO\Big(T\big(\FR^k\times \DD^m(l)\big)\!\Big):=\CO\big(\DD^m(l)\big)\otimes C^\infty\big(\FR^m_y\times \FR^k\times \FR^k\big)
\!\big/ \!
\big\{0=y^{i_1}\epsi^{i_2}\cdots \epsi^{i_{l+1}} +\epsi^{i_1}y^{i_2}\cdots \epsi^{i_{l+1}} +\dots+ \epsi^{i_1}\epsi^{i_2}\cdots y^{i_{l+1}}\big\}\, , 
$$ 
with the two extra $C^\infty(\FR^k)$-factors corresponding to the usual smooth functions on $T\FR^k\cong\FR^k\times \FR^k \in \SmoothManifolds$. Furthermore, the formula from Eq. \eqref{LinearityOfTangentInfinitesimalPoint} straightforwardly generalizes to a fiberwise $\FR$-linear structure on $T\big(\FR^k\times \DD^m(l)\big)$. Moreover, one can work similarly, under any identification of Lem. \ref{InfinitesimalPointsAsSubspacesOfInfinitesimalDisks}, to show that the tangent bundle $T (\FR^k\times \DD)$ of thickened Cartesian space has a fiberwise linear structure, described dually by a further quotient of the above algebras. In fact, it is further true that all objects in our site are actually microlinear (Rem. \ref{InfinitesimallyLinearSpaces}) \cite[Prop 17.6]{Kock06}.
\end{example}

We now show how the synthetic tangent bundle construction immediately recovers the correct notion of tangent bundles on the infinite-dimensional field spaces arising in physics. We first show this for the $\sigma$-model field spaces $[M,N]$ from Ex. \ref{ThickenedMappingSpaceOfManifolds}.
\begin{lemma}[\bf Manifold mapping space tangent bundle]\label{ManifoldMappingSpaceTangentBundle}
For any smooth manifolds $M,N\in \SmoothManifolds$,  the synthetic tangent bundle of the thickened smooth mapping space $[M,\, N]\in \ThickenedSmoothSets$ is naturally isomorphic to $[M,\, TN]$,
$$
T\big[M,\, N\big] := \big[\DD^1(1), [M,N] \big] \;\; \cong  \;\; \big[M,\, TN\big]\, .
$$
\begin{proof}
We provide two proofs, the first being abstract and resting on universal properties and Prop. \ref{TangBundlesCoincide}, with the second being more explicit by working plot-wise. Firstly, recall the internal hom property \eqref{InternalHomProperty}
$$
\big[\DD^1(1), \, [M,N]\big] \cong \big[\DD^1(1)\times M,\, N\big]\, ,
$$
by which the identification follows via the sequence of (natural) isomorphisms
\begin{align*}
T\big[M,N\big]:&= \big[\DD^1(1), \, [M,N]\big] \cong \big[\DD^1(1)\times M, \, N\big] \\
&\cong \big[M\times \DD^1(1), \, N\big] \cong \big[M, \, [\DD^1(1),\, N] \big] \\
&\cong \big[M,TN\big]\, ,
\end{align*}
where we used the internal hom property \eqref{InternalHomProperty}, the symmetric monoidal structure of $\ThickenedSmoothManifolds$ (or $\ThickenedSmoothSets$) via the direct product, followed again by the internal hom property, and finally the canonical identification of Prop. \ref{TangBundlesCoincide}.

Alternatively, being more explicit, by the internal hom property the $*$-plots of $T[M,N]$ are given by
\begin{align*}
T[M,N](*)&\cong \mathrm{Hom}_{\ThickenedSmoothSets} \big(\DD^1(1)\times M, \, N\big) \\
&\cong \mathrm{Hom}_{\mathrm{CAlg}_\FR}\Big(C^\infty(N), \, C^\infty(M)\otimes \CO\big(\DD^1(1)\big)\!\Big) \, .
\end{align*}
Now, precisely as in the proof of Prop. \ref{TangBundlesCoincide}, an algebra map 
\begin{align*}
C^\infty(N) &\longrightarrow C^\infty(M) \otimes \CO\big(\DD^1(1)\big) \\
f &\longmapsto f\circ \phi + \epsi\cdot \CZ_\phi(f)\, ,
\end{align*}
i.e., determined by a pair of maps $(\phi,\CZ_\phi):M\rightarrow N\times TN$ 
such that the diagram
		\[ 
\xymatrix@R=1.4em@C=2.6em  { &&  TN \ar[d]^{\pi}
	\\ 
	M \ar[rru]^{\CZ_\phi} \ar[rr]^>>>>>>>>{\phi} && N
}   
\]
commutes. Equivalently, this is simply a section of the pullback bundle $\CZ_{\phi}\in \Gamma_M(\phi^*TN)$, recovering the standard notion of tangent vectors on mapping spaces of manifolds while avoiding analytical technicalities. Similarly, the $\FR^k\times \DD$-plots are given by
\begin{align*}
T\big[M,\,N\big](\FR^k \times \DD)&\cong \big[\DD^1(1)\times M,\, N\big] (\FR^k\times \DD)\\
& \cong \mathrm{Hom}_{\ThickenedSmoothSets} \big((\FR^k\times \DD) \times \DD^1(1)\times M, \, N\big) \\
&= \mathrm{Hom}_{\mathrm{CAlg}_\FR}\Big(C^\infty(N), \, C^\infty(M\times \FR^k)\otimes \CO(\DD)\otimes  \CO\big(\DD^1(1)\big)\!\Big) \, .
\end{align*}
Computing analogously, algebra maps $C^\infty(N) \rightarrow C^\infty(M\times \FR^k)\otimes \CO(\DD)\otimes  \CO\big(\DD^1(1)\big)$ can be checked to correspond (dually) to pairs of $(\FR^k\times \DD)$-parametrized maps $(\phi^{k,\epsi}, \, \CZ_{\phi^{k,\epsi}}): \FR^k\times \DD \times M \longrightarrow N \times TN$ such that the following diagram commutes 
		\[ 
\xymatrix@R=1.6em@C=2.6em  { &&  TN \ar[d]^{\pi}
	\\ 
	\FR^k\times \DD \times M \ar[rru]^{\CZ_{\phi^{k,\epsi}}} \ar[rr]^>>>>>>>>>{\phi^{k,\epsi}} && N
    \mathrlap{\,.}
}   
\]
Such diagrams are completely determined by the diagonal maps $\CZ_{\phi^{k,\epsi}}:\FR^k\times \DD \times M\rightarrow TN$, with the bottom maps always recovered as the compositions $\phi^{k,\epsi}=\pi \circ \CZ_{\phi^{k,\epsi}}$. But these diagonal maps are precisely the defining plots of $[M,TN]$
\begin{align*}
\big[M,TN\big](\FR^k \times \DD):&= \mathrm{Hom}_{\ThickenedSmoothSets}\big(\FR^k\times \DD\times M, \, TN\big)
\\ &\cong \mathrm{Hom}_{\ThickenedSmoothManifolds}\big(\FR^k\times \DD\times M, \, TN\big) \, ,
\end{align*} 
which completes the proof.
\end{proof}
\end{lemma}

In this case, the fiberwise $\FR$-linear structure of $T[M,N]\cong [M,TN]$ becomes immediately apparent as that inherited by the target $TN$. The synthetic tangent bundle construction also immediately recovers the tangent bundle to field spaces of sections (\cite[Def. 2.19]{GS23}), generalizing the above result.

\begin{proposition}[\bf Synthetic tangent bundle of field space]\label{SyntheticTangentBundleOfFieldSpace}
The synthetic tangent bundle $T\CF:= [\DD^1(1),\, \CF] $ of a thickened smooth field space $\CF=\mathbold{\Gamma}_M(F)$ is canonically isomorphic to $\mathbold{\Gamma}_{M}(VF)$,
$$
T\CF\, \cong\, \mathbold{\Gamma}_M(VF) \;\; \in \;\; \ThickenedSmoothSets \, ,
$$
where $VF\rightarrow F\xrightarrow{\pi} M$ is the vertical fiber bundle corresponding to the field bundle $\pi:F\rightarrow M$.
\begin{proof} 
This follows by universal properties. Namely, recall the pullback (limit) characterization of the field space from \eqref{FieldSpaceAsPullback},
\begin{align*} 
\xymatrix@=1.6em 
{\CF  \ar[d] \ar[rr] &&  
[M,F] \ar[d]^{\pi_*} 
	\\ 
	\mathrm{id}_M\cong* \; \ar@{^{(}->}[rr]  && [M,M]
	\, . } 
\end{align*}
Since the tangent functor $T:= [\DD^1(1),\, - ]$ has a left adjoint \eqref{TangentFunctorHasLeftAdjoint}, namely $\DD^1(1)\times (-)$, it preserves all (small) limits and in particular pullbacks. Thus, the synthetic tangent bundle $T\CF$
is identified as the pullback
\begin{align*} 
\xymatrix@=1.6em 
{T\CF  \ar[d] \ar[rr] &&  
T[M,F] \ar[d]^{T(\pi_*)} 
	\\ 
	T(\mathrm{id}_M)\cong T(*)\cong *\; \ar@{^{(}->}[rr]  && T[M,M]
	\, . } 
\end{align*}
By Lem. \ref{ManifoldMappingSpaceTangentBundle}, the right vertical arrow is canonically identified with $[M,\, TF]\xlongrightarrow{(T\pi)_*} [M,TM]$, while the bottom arrow is canonically identified with the embedding of the point $* \hookrightarrow [M,TM]$ using its zero-section incarnation. Explicitly, the latter is defined by sending the (unique) $\FR^k\times \DD$-plot of $0_M\cong*$ to the `zero-section' $\FR^k\times \DD \times M \xrightarrow{\pr_M} M \xrightarrow{0_M} TM$, which can be seen by applying the tangent functor adjunction plot-wise to the embedding $\id_M\cong *\hookrightarrow [M,M]$ from diagram \eqref{FieldSpaceAsPullback}.

Thus, the tangent bundle $T\CF$ is equivalently the pullback 
\begin{align*} 
\xymatrix@=1.6em 
{T\CF  \ar[d] \ar[rr] &&  
[M,\, TF] \ar[d]^{(T\pi)_*} 
	\\ 
	0_M\cong *\; \ar@{^{(}->}[rr]  && [M,\, TM]
	\, , } 
\end{align*}
where by computing object-wise this says that the $\FR^k\times \DD$-plots of $T\CF$ are given by 
$$
T\CF(\FR^k\times \DD) = \big\{ \CZ_{\phi^{k,\epsi}}\in [M,\,TF](\FR^k\times \DD) \; \big| \;  T\pi\circ \CZ_{\phi^{k,\epsi}} = 0_M \circ \pr_M  \big\}\, .
$$
These are maps $\CZ_{\phi^{k,\epsi}}:\FR^k\times \DD \times M\rightarrow TF$ making the diagram
\[ 
\xymatrix@C=4pc @R=2pc{ && TF \ar[d]^{T\pi}
	\\ 
	\FR^k\times \DD \times M  \ar[r]^-{\;\;\; \pr_M} \ar[rru]^-{\CZ_{\phi^{k,\epsi}}} &M \ar[r]^{0_M} &TM \, 
}   
\]
commute, and hence in particular factoring through the vertical subbundle $VF\hookrightarrow TF$. By further composing on the right with the projection $TM\rightarrow M$, one sees that maps $\CZ_{\phi^k_\epsi}$ consist precisely of sections of the form
\[ 
\xymatrix@C=2.6pc @R=1.3pc { &&  VF \ar[d]^{\pi}
	\\ 
	\FR^k\times \DD \times M \ar[rru]^-{\CZ_{\phi^{k,\epsi}}} \ar[rr]^-{\pr_M} && M
}   
\]
over $M$, which are exactly the defining $\FR^k\times \DD$-plots of $\mathbold{\Gamma}_M(VF)$. 
\end{proof}
\end{proposition}
Of course, composing a plot $\CZ_{\phi^{k,\epsi}}$ on the right with the projection $VF\rightarrow F$ yields $\FR^k\times \DD$-parametrized sections of $F\rightarrow M$,
$$ 
\phi^{k,\epsi} := \pi_{F} \circ \CZ_{\phi^{k,\epsi}} \, .
$$
That is, any plot of tangent vectors $\CZ_{\phi^{k,\epsi}}$ covers a (unique) plot of fields $\phi^{k,\epsi}$, which realizes precisely the expected canonical projection of thickened smooth spaces from \eqref{SynthBundleProj},
$$
T\CF\cong \mathbold{\Gamma}_M(VF)\; \xlongrightarrow{\;\; (\pi_F)_*\;\; } \; \mathbold{\Gamma}_M(F)=\CF  \, ,
$$
and hence justifying the notation for plots of the former. Furthermore, the content of \cite[Lem. 2.17, ftn. 12]{GS23} generalizes to the case of $\FR^k\times \DD$-parametrized plots, namely:
\begin{lemma}[{\bf Representing plots of  tangent vectors via paths of plots}]\label{LinePlotsRepresentTangentVectors}
For any plot of tangent vectors $\CZ_{\phi^{k,\epsi}} \in T\CF(\FR^k\times \DD)$ covering a plot of fields $\phi^{k,\epsi}=\pi_F \circ \CZ_{\phi^{k,\epsi}} \in \CF(\FR^k\times \DD) $
\begin{align}\label{TangentFieldPlotDiagram}
\xymatrix@R=1.3em{ && VF\ar[d]^{\pi_F} \\&&  F \ar[d]
	\\ 
	\FR^k\times \DD \times M\ar[rruu]^{\CZ_{\phi^{k,\epsi}} } \ar[rru]^>>>>>>>>{{\phi^{k,\epsi}}} \ar[rr]^>>>>>>>>>{\pr_1}&&  M \, ,  
}   
\end{align}
\noindent there exists a parametrized section $\phi^{k,\epsi}_t
:\FR^k\times \DD \times M\times \FR^1_t\rightarrow F$ over $M$, such that $\phi^{k,\epsi}_0=\phi^{k,\epsi}$ and $\partial_t \phi^{k,\epsi}_t\big|_{t=0} = \CZ_{\phi^{k,\epsi}}$, with the derivative computed pointwise over $M$. That is, the map 
\begin{align*} 
\mathbold{\Gamma}_M(F)(\FR^k\times \DD \times \FR^1_t) &\longrightarrow \mathbold{\Gamma}_M(VF)(\FR^k\times \DD)
\\
\phi^{k,\epsi}_t &\longmapsto \partial_t \phi^{k,\epsi}_t\big|_{t=0} 
\end{align*}
\noindent is surjective, for each probe $\FR^k\times \DD\in \ThickenedCartesianSpaces$.
\end{lemma}
\begin{proof}
The case of (non-thickened) $\FR^k$-plots of tangent vectors to field space is precisely the content of \cite[Lem. 2.17, ftn. 12]{GS23}. The case of an arbitrary thickened $\FR^k\times \DD$-plots follows by reducing it to the former case.

In more detail, recall from Rem. \ref{infinitesimalMappingExamples} and Eq. \eqref{WeilBundleMappingSpace} (cf. Prop. \ref{TangBundlesCoincide}) that the mapping space out of any infinitesimal point into the total space of the field bundle $F\rightarrow M$ is, in fact, canonically isomorphic to another fiber bundle \textit{of manifolds} over $M$
$$
[\DD, F] \, \cong \, T_\DD F \longrightarrow F \rightarrow M \, .
$$
Similarly, applying the same mapping space construction to the (synthetic) tangent bundle $TF\cong [\DD^1(1),F]$ over $F$ (Prop. \ref{TangBundlesCoincide}), and using the internal hom property we have
\begin{align*}
T_\DD(TF) \, &\cong \, 
\big[\DD, [\DD^1(1), F] \big] \cong [\DD\times \DD^1(1), F]
\\
& \cong \, \big[\DD^1(1), [\DD, F]\big]\, \cong\, T (T_\DD F) \, ,
\end{align*}
and since these are functorial with respect to the base projection maps onto $M$ (and $T_\DD M$), the corresponding vertical subbundle $VF\hookrightarrow TF$ must be preserved under the functor $T_\DD: \ThickenedSmoothManifolds \rightarrow \ThickenedSmoothManifolds$ (and hence as thickened smooth sets)
$$
T_\DD (VF) \, \cong \, V(T_\DD F) \xhookrightarrow{\quad} T(T_\DD F),
$$
where on the right hand side verticality is with respect to $T_\DD F \rightarrow T_\DD M$ (not the composite).

Applying the internal hom property with respect to the $\DD$-factor of the $\FR^k\times \DD$-probe to the diagram \eqref{TangentFieldPlotDiagram}, this uniquely corresponds to an analogous diagram of \textit{non-thickened} $\FR^k$-plots
\begin{align*}
\xymatrix@R=1.6em{ && V(T_\DD F)\ar[d]^{\pi_{T_\DD F}} \\&& T_\DD F \ar[d]
	\\ 
	\FR^k \times M\ar[rruu]^{\CZ_{\phi^{k,\DD}} } \ar[rru]^>>>>>>{{\phi^{k,\DD}}} \ar[rr]^>>>>>>>{\iota^\DD_{\pr_1}}&&  T_\DD M \, ,  
}   
\end{align*}
where the bottom map necessarily maps $\FR^k\times M$ to an embedded submanifold of $T_\DD M$, canonically diffeomorphic to the base $M\cong \iota^\DD_{\pr_1}(\FR^k\times M)\hookrightarrow T_\DD M$.\footnote{This can be seen analogously to the case of the tangent bundle $T M\equiv T_{\DD^1(1)} M$ and the canonical embedding given by the zero  section $s_0:M\hookrightarrow TM$.} 
Applying now the original result \cite[Lem. 2.17, ftn. 12]{GS23} yields a map $\phi^{k,\DD}_{t
}:\FR^k\times \DD \times M\times \FR^1_t\rightarrow T_\DD F$ such that $\phi^{k,\DD}_{0}=\phi^{k,\DD}$ and $\partial_t \phi_{t}^{k,\DD} |_{t=0} = \CZ_{\phi^{k,\DD}}$. Finally, applying the 
internal hom adjunction again, in reverse with respect to $T_\DD = [\DD,-]$, yields the desired result. 
\end{proof}

In other words, since $[\FR^1_t,\CF](\FR^k\times \DD) := \mathrm{Hom}_{\ThickenedSmoothSets}(\FR^1_t \times \FR^k\times \DD , \CF) \cong \CF(\FR^k\times \DD\times \FR^1_t)$, this exhibits the following sequence of projections (plot-wise surjections, hence epimorphisms)
\begin{align*}
\mathbold{P}(\CF):= \big[\FR^1_t,\, \CF\big] & \;\xrightarrow{\quad \quad}\qquad T\CF \quad \quad \longrightarrow \;\; \CF\\
\phi^{k,\epsi}_t \;\; &\;\xmapsto{\quad \quad} \; \partial_{t} \phi^{k,\epsi}_t\big|_{t=0} \; \longmapsto \; \phi^{k,\epsi}_{t=0}\nn
\end{align*}
where $\phi^{k,\epsi}_t \in \mathbold{P}(\CF)(\FR^k\times \DD)$ is an $\FR^k\times \DD$-plot in the (thickened) path space of fields. Due to the first plot-wise surjection, one may represent arbitrary plots $\CZ_{\phi^{k,\epsi}}$ of the tangent bundle $T\CF$ by the derivative of a (non-unique) plot in paths of fields
$$
\CZ_{\phi^{k,\epsi}}\, = \, \partial_t \phi^{k,\epsi}_t \big|_{t=0} \, ,
$$
just as with tangent vectors of finite-dimensional manifolds. Importantly, let us also note that although by the explicit plot-wise (and pointwise over M!) construction of the projection map $\partial_t |_{t=0} :
\mathbold{P}(\CF)\rightarrow T\CF 
$ seems rather ad hoc from the categorical point of view of (purely) smooth sets (cf. \cite[Ex. 2.16]{GS23}), here it is immediately seen to coincide with the (internal) restriction of (plots of) paths along $\iota_0: \DD^1(1)\hookrightarrow \FR^1_t$, i.e., 
\begin{align}\label{DerivativeOfPathsIsRestriction}
\partial_{t}|_{t=0} \, \equiv \, \iota_0^* \;\; : \;\; \mathbold{P}(\CF) \,:=\, [\FR^1_t, \, \CF] \xrightarrow{\quad \quad} T\CF \, := \,  \big[ \DD^1(1),\, \CF \big] \, ,
\end{align}
just as with the case of finite-dimensional manifolds.

Indeed, to see this notice  that restricting any path of plots $\phi^{k,\epsi}_t \in \mathbold{P}\CF(\FR^k\times \DD) \cong \CF(\FR^k\times \DD \times \FR^1_t)$ 
\[ 
\xymatrix@C=2.6pc @R=1.3pc { &&  F \ar[d]^{\pi}
	\\ 
	\FR^k\times \DD\times \FR^1_t \times M \ar[rru]^>>>>>>>>>>>>{\phi^{k,\epsi}_t} \ar[rr]^-{\pr_M} && M
}   
\]
along $\FR^k\times \DD \times \DD^1(1)\xhookrightarrow{\;\id_{\FR^k\times \DD}\times \iota_0\;} \FR^k\times \DD \times \FR^1_t$
yields an $\FR^k\times \DD\times \DD^1(1)$ parametrized section of $F$ over M. Equivalently by the internal hom property \eqref{InternalHomProperty} with respect to the $\DD^1(1)$, this is a section of $TF\cong [\DD^1(1), F]\rightarrow TM\cong [\DD^1(1), M]$ (Prop. \ref{TangBundlesCoincide}) which, moreover, by the induced section condition, necessarily factors through $VF$ 
\[ 
\xymatrix@C=2.6pc @R=1.5pc { &&  VF \ar[d]^{\pi}
	\\ 
	\FR^k\times \DD \times M \ar[rru]^-{\CZ_{\phi^{k,\epsi}}} \ar[rr]^-{\pr_M} && M
    \mathrlap{\,.}
}   
\]
Computing the form of the dual map, as e.g., in the proof of Lem. \ref{ManifoldMappingSpaceTangentBundle}, then shows that it coincides precisely with the intuitive derivative of $\phi^{k,\epsi}_t$ at $t=0$.

\begin{remark}[\bf General mapping space tangent bundles]\label{GeneralMappingSpaceTangentBundles}
The diagrammatic categorical arguments in the proofs of Lem. \ref{ManifoldMappingSpaceTangentBundle} and Prop. \ref{SyntheticTangentBundleOfFieldSpace} apply for arbitrary objects in $\ThickenedSmoothSets$. That is, it follows identically that for any two thickened smooth spaces $\CG,\CH$, there is a canonical isomorphism
$$
T[\CG,\CH]\; \cong \; [\CG, T\CH]\, .
$$ Similarly, for any `bundle' map of thickened smooth spaces $\pi: \CK\rightarrow \CG$, the synthetic tangent bundle to the space of sections $\mathbold{\Gamma}_\CG(\CK)$ is canonically identified as the pullback $T \mathbold{\Gamma}_{\CG}(\CK)\cong \mathbold{\Gamma}_\CG(V\CK)$, with the latter being defined as the pullback
\begin{align*} 
\xymatrix@=1.6em 
{\mathbold{\Gamma}_\CG(V\CK) \ar[d] \ar[rr] &&  
[\CG,\, T\CK] \ar[d]^{(T\pi)_*} 
	\\ 
	0_\CG\cong * \ar@{^{(}->}[rr]  && [\CG,\, T\CG]
	\, . } 
\end{align*}
The catch here is that in this generality, these constructions are quite formal in nature. $T\CH$ is defined simply as $\big[ \DD^1(1), \, \CH \big]$, 
which might not have a more `concrete' description, and might even lack a fiberwise linear structure. 
Similarly, $T \mathbold{\Gamma}_\CG (\CK)$ might lack a fiberwise linear structure, and hence the bottom map may not necessarily 
be an actual \textit{zero}-section, but instead is defined simply as the (unique) adjunct $0_M: \CG \rightarrow T\CG$ of the projection $p_\CG: \DD^1(1)\times \CG \rightarrow \CG$.
\end{remark}

Having identified the appropriate notion of a tangent bundle on field spaces, we may now import the correct notion of a cotangent bundle from \cite[Def. 5.35]{GS23}, now defined directly within thickened smooth sets. 
\begin{definition}[\bf Variational cotangent bundle]\label{VariationalCotangentBundle}
The \textit{variational cotangent bundle} $\pi_\CF: T^*_{\mathrm{var}} \CF\rightarrow \CF$ of a thickened field space 
$\CF\in \ThickenedSmoothSets$ is defined as the (smooth, thickened) vector bundle 
 
\[
\xymatrix@=1em  { T^*_{\mathrm{var}}\CF\,:= \, \mathbold{\Gamma}_M(\wedge^d T^*M \otimes V^*F)  
\ar[d] & 
	\\ 
\CF\,:= \, \mathbold{\Gamma}_M(F)
	\, , } 
\]
 
\noindent where the projection $\pi_\CF$ is given by postcomposition of (plots of) sections of $\wedge^d T^*M \otimes V^*F:= \pi_M^*(\wedge^d T^*M)\otimes_F V^*F$
with the (manifold) vector bundle projection $\pi_F: \wedge^d T^*M \otimes V^*F\rightarrow F$.  
\end{definition}
The crucial property of the tangent and variational cotangent bundle is the existence of a pairing on fibers over (plots of) field configurations
\begin{align}\label{VariationalCotangentTangentPairing}
\langle - \, , \, - \rangle \;\; : \;\;   T\CF  \times_\CF & T^*_\mathrm{var}\CF \quad \xrightarrow{\quad \quad} \; \Omega^{d}_\mathrm{vert}(M) \,:=\, \mathbold{\Gamma}_M(\wedge^d T^*M)
\\
\big(\CZ_{\phi^{k,\epsi}} \, &, \, \mathcal{B}_{\phi^{k,\epsi}}\big)
\; \longmapsto  \quad 
\big\langle \CZ_{\phi^{k,\epsi}} \, , \, \mathcal{B}_{\phi^{k,\epsi}} \big\rangle \, , \nn
\end{align}
induced by composing pairs of plots tangent $\CZ_{\phi^{k,\epsi}}: \FR^k\times \DD \times M \rightarrow VF $ and cotangent vectors $\mathcal{B}_{\phi^{k,\epsi}}$ over a common plot of fields $\phi^{k,\epsi} \in \CF$, pointwise over $M$, with the canonical fiberwise duality pairing  
$$
V F\times \big(V^*F \otimes \wedge^d T^*M\big) \xrightarrow{\quad \quad} \wedge^d T^*M\, .
$$
When further composed with integration\footnote{When well-defined, i.e. for $M$ compact, or more generally for plots of fields with compact support along $M$.} \eqref{ThickenedIntegration}, this yields a real-valued non-degenerate\footnote{In the sense of a pairing of $C^\infty(\FR^k)\otimes \CO(\DD)$-modules, plot-wise for each $\FR^k\times \DD\in \ThickenedCartesianSpaces$. See the proof of Thm. \ref{FunctorialityOfTheThickenedCriticalSet} for more details.} pairing of tangent vectors and (variational) cotangent vectors on field space within $\ThickenedSmoothSets$, in direct analogy to the case of the tangent and cotangent bundle of a finite-dimensional smooth manifold. As naturally expected, this pairing appears precisely in the dynamical \textit{variational} principle -- which, as we shall see (Thm. \ref{FunctorialityOfTheThickenedCriticalSet}), still holds in $\ThickenedSmoothSets$ and hence serves as the core motivation for the above definition.

\medskip 
We include a further intuitively desired property of tangent vectors within thickened smooth sets. In \cite[Rem. 2.18]{GS23}, we pointed out that tangent vectors to field spaces do not \textit{necessarily} define $\FR$-valued derivations of the full algebra of smooth functions $C^\infty(\CF)$, if defined in the category of smooth sets as $\mathrm{Hom}_{\SmoothSets}\big(\CF,\, \FR\big)$. This peculiarity is automatically (and generally) cured in the category of thickened smooth sets.
\begin{lemma}[\bf Path derivations depend only on tangent vectors]\label{DerivativesAlongLinesDependOnTangentVectors}
Let $\CF$ be a thickened smooth set and $\phi_t\in \CF(\FR)\cong \mathrm{Hom}_{\ThickenedSmoothSets}(\FR,\, \CF)$ a line plot in $\CF$. Denote the smooth real-valued functions on $\CF$ by $C^\infty(\CF) :=\mathrm{Hom}_{\ThickenedSmoothSets}(\CF,\, \FR)$. Then the induced derivation
\begin{align*}
C^\infty(\CF)&\xrightarrow{\quad \quad} \FR \\
f&\xmapsto{\quad \quad} {\partial_t} (f\circ \phi_t) \big|_{t=0}
\end{align*}
depends only on the corresponding synthetic tangent vector $\phi_\epsi := \phi_t \circ \iota_0 : \DD^1(1)\hookrightarrow \FR^1 \rightarrow \CF$.
\begin{proof}
Suppressing the mention of the Yoneda embedding symbol as before, we have the commutative diagram
	\[ 
\xymatrix@R=1.3em@C=3em{ &&  \CF  \ar[d]^{f}
	\\ 
	\FR \ar[rru]^{\phi_t} \ar[rr]^>>>>>>>{f\circ {\phi}_t} && \FR \, ,
}   
\]
where by precomposing with the canonical inclusion $\iota_0: \DD^1(1) \hookrightarrow \FR$ we get 
	\[ 
\xymatrix@R=1.3em@C=3em{ & &&  \CF  \ar[d]^{f}
	\\ 
\DD^1(1)\ar[r]^-{\iota_0}&	\FR \ar[rru]^{\phi_t} \ar[rr]^>>>>>>>>{f\,\circ \, {\phi}_t} && \FR \, .
}   
\]
The bottom composition $(f\circ \phi_t) \circ \iota_0:\DD^1(1)\rightarrow \FR$ expands dually as 
\begin{align*}
 C^\infty(\FR) & \; \longrightarrow \;  \CO\big(\DD^1(1)\big) \\
 x& \; \longmapsto \; f(\phi_0) + \epsi \cdot {\partial_t}(f\circ \phi_t) \big|_{t=0} \; ,
\end{align*}
while the upper composition $f\circ \phi_\epsi = f \circ (\phi_t\circ \iota_0)$ manifestly depends only on 
the corresponding tangent vector, and the result follows. 
\end{proof}
\end{lemma}
This means that in the particular case of $\CF=\mathbold{\Gamma}_{M}(F)$ being a thickened smooth field space, any tangent vector 
$$
\CZ_\phi=\partial_t \phi_t |_{t=0} \quad \in  \quad T\CF(*)\cong \Gamma_M(VF)\, 
$$ 
defines a unique real-valued derivation, for any choice of representative line plot $\phi_t$ (which exists by \cite[Lem. 2.17]{GS23}). This is in stark contrast with \cite[Rem. 2.18]{GS23} where the field space is viewed simply as a smooth set, ignoring its infinitesimal structure. Indeed, maps between \textit{thickened} smooth sets `\textit{respect the infinitesimal structure}' and not only the smooth structure, thus guaranteeing the intuitive sought-after property that derivatives along smooth curves should depend only on the first order class of the curve. It follows similarly that vector fields on $\CZ\in \CX(\CF)$ on the field space $\CF$, defined in the thickened smooth setting as sections (cf.  \cite[Def. 2.20]{GS23}) 
\[ 
\xymatrix@R=1.2em@C=3em{ &&  T\CF  \ar[d]^{\pi}
	\\ 
	\CF \ar[rru]^{\CZ} \ar[rr]^>>>>>>>{\id_\CF} && \CF \, ,
}   
\] define smooth derivations 
$$
\CZ: C^\infty(\CF) \xrightarrow{\quad \quad} C^\infty(\CF) \, ,
$$ in contrast to the purely smooth setting of \cite[(33)]{GS23}. 

\medskip
\noindent {\bf Infinite jet bundles.}
 We close this section by importing the infinite jet bundle in the thickened smooth setting, and hence recovering its tangent bundle, too, via the synthetic tangent bundle construction. In \cite[\S 3.1]{GS23}, we have considered the infinite jet bundle $J^\infty_M F$ as a limit in Fr\`{e}chet manifolds and its smooth set incarnation $\mathrm{Hom}_{\FrechetManifolds}(-,J^\infty_M F)\in \mathrm{SmSet}$, which is equivalently  the limit of the smooth set incarnations of the finite-order jet bundles $J^n_M F\in \SmoothManifolds$, i.e., 
$$
\lim_{\SmoothSets} y\big(J^n_M F\big)\;\; \in \;\; \SmoothSets . 
$$
Following the above discussion, it is straightforward to promote the infinite jet bundle to a thickened smooth set by considering instead each
finite-order jet bundle as a thickened smooth set (Ex. \ref{ManifoldsAsThickenedSmoothSets} ), and computing the limit in thickened smooth sets.
\begin{definition}[\bf Thickened infinite jet bundle]\label{InfiniteJetBundleFormalSmoothLimit} The thickened infinite jet bundle $J^\infty_M F$ of a fiber bundle $F\rightarrow M$ of manifolds is given by 
$$
J^\infty_M F\, := \,  \lim_{\ThickenedSmoothSets} y(J^n_M F) \quad \in \;\;  \ThickenedSmoothSets \, .
$$
\end{definition}
By construction, this limit comes supplied with the (universal) projection maps 
\begin{align}\label{JetBundleProjections}
\pi^\infty_n \;:\; J^\infty_M F \xrightarrow{\quad \quad} J^n_M F
\end{align}
of thickened smooth sets. It is straightforward to see that the synthetic tangent bundle 
$$
T\big(J^\infty_MF\big) := \big[\DD^1(1), \, J^\infty_M F \big]
$$
recovers that of \cite[Def. 4.1]{GS23}, however, computed in thickened smooth sets and hence also providing it with an infinitesimal structure.

\begin{lemma}[\bf Infinite jet synthetic tangent bundle]\label{SyntheticInfiniteJetTangentBundle}
The synthetic tangent bundle $T J^\infty_M F := \big[\DD^1(1),\,  J^\infty_M F\big]$ of an infinite jet bundle is canonically isomorphic to the limit of tangent bundles of finite jet bundles
$$
TJ^\infty_M F \; \cong \;\;  \lim_{\ThickenedSmoothSets}y(T J^n_M F) \, . 
$$
\begin{proof}Once again, this follows since the tangent functor $T:= \big[\DD^1(1),\, - \big]$ has a left adjoint and hence preserves (small) limits. More explicitly, we can see the result by standard properties of limits and (internal) hom functors. Suppressing the Yoneda embedding, the $(\FR^k\times \DD)$-plots of the two thickened smooth sets are canonically (and functorially) identified as 
\begin{align*}
\hspace{-2mm} 
\big[\DD^1(1), J^\infty_M F \big](\FR^k\times \DD)\; &= \; \mathrm{Hom}_{\ThickenedSmoothSets}\big(\FR^k\times \DD\times \DD^1(1) , J^\infty_M F\big) \;=\; \mathrm{Hom}_{\ThickenedSmoothSets}\Big(\FR^k\times \DD\times \DD^1(1) , \, \lim_{\ThickenedSmoothSets} J^n_M F\Big)
\\
&\;\cong\; \lim_{Set} \mathrm{Hom}_{\ThickenedSmoothSets}\big(\FR^k\times \DD\times \DD^1(1) , J^n_M F\big)
\;\cong \; \lim_{Set} \mathrm{Hom}_{\ThickenedSmoothSets}\Big(\FR^k\times \DD, \, \big[\DD^1(1), J^n_M F\big] \Big)
\\
&\;\cong\; \lim_{Set}  \mathrm{Hom}_{\ThickenedSmoothSets}\big(\FR^k\times \DD, \, T J^n_M F \big)
\;\cong\; \mathrm{Hom}_{\ThickenedSmoothSets}\Big(\FR^k\times \DD, \, \lim_{\ThickenedSmoothSets}  T J^n_M F \Big) \\
&\;\cong\; \lim_{\ThickenedSmoothSets}  T J^n_M F \big( \FR^k\times \DD\big) \,,
\end{align*}
where we used the fact that limits commute with the hom functor, the internal hom property and Prop. \ref{TangBundlesCoincide}. 
\end{proof}
\end{lemma}
It follows similarly that the corresponding limit of the (pushforward) projection maps 
\begin{align}\label{TangentJetBundleProjections}
 \dd \pi^\infty_n \;:\; TJ^\infty_M F \xrightarrow{\quad \quad} TJ^n_M F 
 \end{align}
coincide with the synthetic pushforwards \eqref{SyntheticPushforward}, i.e., acting via (internal) composition plot-wise $\DD^1(1)\times (\FR^k\times \DD)  \rightarrow J^\infty_M F \xrightarrow{\pi^\infty_n} J^n_M F$.
We note that from this thickened point of view, the result of \cite[Lem. 4.2]{GS23}, which identifies tangent vectors on $J^\infty_M F$  with first order classes of lines in $J^\infty_M F$, follows immediately as a corollary of Lem. \ref{DerivativesAlongLinesDependOnTangentVectors}. We will see in  \cref{JetBundlesSyntheticallySec} that there is a more natural way, categorically and intuitively speaking, to define the infinite jet bundle of \textit{any} (thickened) smooth map $\pi:\CG\rightarrow \CF$ -- which does indeed recover the above definition in the case of a finite-dimensional fiber bundle.

\medskip
\noindent {\bf Representing plots of $J^\infty_M F$ and its tangent bundle.}
Note that the limit characterization of both $J^\infty_M F$ and $T(J^\infty_M F)$ means that their plots may be formally represented and manipulated precisely as in the purely smooth setting, essentially by carrying along the dependence on nilpotent elements corresponding to the infinitesimal thickening of the Cartesian probes. More explicitly, by definition, $\FR^k\times \DD$-plots of $J^\infty_M F$ are compatible families of plots of $\{ J^n_M F\}_{0\leq |I|} $
$$
J^\infty_M F(\FR^k\times \DD)\; \cong\; \Big\{ \big\{s^{k,n}_{\epsi}:\FR^k\times \DD \rightarrow J^n_M F\;\;  \big|  \;\;   \pi^{n}_{n-1}\circ  s^{k,n}_{\epsi}
= s^{k,n-1}_{\epsi}\big\}_{n\in \NN} \Big\}\, .
$$
That is, any plot $s^k_\epsi: \FR^k\times \DD \rightarrow J^\infty_M F$ is dually completely determined by its (compatible) action on (a cover of) the local coordinate functions $\{x^\mu,\, u^a_I\}_{0\leq |I|\leq n}$ on each $J^n_M F$, hence by an infinite list of compatible assignments
\begin{align*}
x^\mu &\xmapsto{\quad \quad} s^\mu(c^k,\epsi)  \\ 
u^a_I &\xmapsto{\quad \quad} s^a_I(c^k,\epsi) \;\; \in \;\; C^\infty(\FR^k)\otimes \CO(\DD) \, . \nn 
\end{align*}
In this local representation, the limit projections $\pi^\infty_n: J^\infty_M F\rightarrow J^n_M F$ from \eqref{JetBundleProjections} act plot-wise by `forgetting' the $|I|> n$ members of the family, i.e., by forgetting assignment on coordinates $u^a_I$ with $|I|>n$.   

\medskip 
Analogously, by Lem. \ref{SyntheticInfiniteJetTangentBundle}, $\FR^k\times \DD$-plots of $T J^\infty_M F$ are compatible families of plots of $\{TJ^n_M F\}_{0\leq |I|}$, and hence any plot 
\[ 
\xymatrix@R=1.2em@C=2.5em  { &&  TJ^\infty_M F \ar[d]
	\\ 
	\FR^k\times \DD  \ar[rru]^{X_{s^k_\epsi}} \ar[rr]^>>>>>>>{s^k_\epsi} && J^\infty_M F
}   
\]
covering a plot $s^k_{\epsi}$ of $J^\infty_M F$ may be represented by an infinite formal sum
\begin{align}
\label{InfinityJetTangentPlotsFormalSum}
X_{s^k_\epsi} \;=\; X^\mu_{s^k_\epsi} \, \frac{\partial}{\partial x^\mu}\Big\vert_{s^k_\epsi} \;+\; 
\sum_{|I|=0}^\infty Y^a_{I,s^k_\epsi} \frac{\partial}{\partial u_I^a} \Big\vert_{s^k_\epsi} \, ,
\end{align}
where  $\big\{X^\mu_{s^k_\epsi}, Y^a_{I,s^k_\epsi}\big\}_{0\leq |I|} \subset C^\infty(\FR^k)\otimes \CO(\DD)$, which stands for the infinite list of compatible assignments on the corresponding coordinates on the tangent bundles $TJ^n_M F$
\begin{align*}
x^\mu &\xmapsto{\quad \quad} s^\mu(c^k,\epsi) \nn \\ 
\dot{x}^\mu &\xmapsto{\quad \quad} X^\mu_{s^k_\epsi} (c^k,\epsi) \\
u^a_I &\xmapsto{\quad \quad} s^a_I(c^k,\epsi) \nn \\
\dot{u}^a_I &\xmapsto{\quad \quad} Y^a_{I,s^k_\epsi} (c^k,\epsi) \quad \in \quad C^\infty(\FR^k)\otimes \CO(\DD) \, . \nn 
\end{align*}
In this local representation, the pushforward limit projections $\dd \pi^\infty_n: TJ^\infty_M F\rightarrow TJ^n_M F$ from \eqref{TangentJetBundleProjections} act plot-wise by `forgetting' the $|I|> n$ members of the family, i.e., by truncating the formal sum at $|I|=n$. The formal sum representation \eqref{InfinityJetTangentPlotsFormalSum} also makes the fiberwise (and plot-wise) $\FR$-linear structure of $TJ^\infty_M F$ manifest, hence constituting a well-defined map of thickened smooth sets
$$
+ \; : \; 
TJ^\infty_M F \times_{J^\infty_MF} TJ^\infty_M F \xrightarrow{\quad \quad} TJ^\infty_M F\, ,
$$
as intuitively expected.

\subsection{Infinitesimal neighborhoods and jets of sections}
\label{JetBundlesSyntheticallySec}

We have seen (Lem. \ref{TangentBundleSet}) that tangent vectors on a finite-dimensional manifold M are detected by maps out of the 
infinitesimal line $\DD^1(1)$. Another equivalent point of view of tangent vectors on $M\in \mathrm{Man}$, is as 1-jets of maps $\FR\rightarrow M$ at $0\in \FR$, 
or yet equivalently, jets of sections of the trivial bundle $\FR\times M \rightarrow \FR$ at $0\in \FR$. By canonical 
embedding of the infinitesimal line into the real line $ 
\DD^1(1) \xhookrightarrow{\iota_0} \FR^1$ (see \eqref{InfDiskIntoCartSp}),  Lem. \ref{TangentBundleSet} may be equivalently 
be phrased as the bijection between jets of sections $J^1_{0\in \FR} (\FR\times M)$ at $0\in \FR$ of the base, with the set of maps
$\DD^1(1)\rightarrow \FR\times M$ such that the diagram 
\[ 
\xymatrix@R=1.4em@C=3em  { &&  \FR\times M\ar[d]^{\pr_1}
	\\ 
	\DD^1(1) \ar[rru]^{X_p} \ar@{^{(}->}[rr]^-{\iota_0} && \FR
}   
\]
commutes. The interpretation is in line with intuition: 1-jets of sections of $\FR\times M\rightarrow \FR$ at the origin are the same as sections over the first order 
infinitesimal neighborhood of the origin. It follows similarly that 
the analogous diagrams 
\[ 
\xymatrix@R=1.4em@C=3em   { &&  \FR\times M\ar[d]^{\pr_1}
	\\ 
	\DD^1(k) \ar[rru]^{X_p} \ar@{^{(}->}[rr]^-{\iota_0} && \FR
}   
\]
with $\DD^1(k)$ instead will detect $k$-jets $J^k_{0\in \FR} (\FR\times M)$ of curves in $M$, for any $k\in \NN \cup \{\infty\}$.
That is, $k$-order jets of sections of $\FR\times M\rightarrow \FR$ at $0\in \FR$ are precisely sections over the $k$-order
infinitesimal neighborhood at $0\in \FR$. We will show how this intuition naturally generalizes to the case of any 
fiber bundle $F\rightarrow M$ of manifolds, allowing us to detect (or define) jets of sections as sections over infinitesimal neighborhoods 
at a point in the base. In fact, a natural notion of an infinitesimal neighborhood in \textit{any} thickened smooth sets allows one to define a notion of jets for maps between any two thickened smooth sets (Def. \ref{syntheticJets}). Moreover, the latter yields another way to recover the infinite jet bundle (Thm. \ref{SyntheticJetBundleOfFiberBundle}, due to \cite{KS17})  with its (thickened) smooth structure from Def. \ref{InfiniteJetBundleFormalSmoothLimit}, directly in $\ThickenedSmoothSets$, once again avoiding 
any mention of infinite-dimensional manifold theory.

\smallskip 
\smallskip 
Before stating the abstract definition of the infinitesimal neighborhood of a point (or plot) of an arbitrary thickened smooth set, we identify the
notion of an infinitesimal neighborhood at a point $p$ in an arbitrary smooth manifold $M\in \mathrm{Man}$ via the embedding $\mathrm{Man}\hookrightarrow\mathrm{CAlg}^{op}_{\FR}$of Lem. \ref{CartSptoAlgebras}, extending the case of those of Cartesian spaces. For $p\in M$ a point in a manifold, 
we are after natural `infinitesimal spaces' attached to $p\in M$ that should serve as infinitesimal neighborhoods 
in at $p\in M$. 

\smallskip 
As discussed in \cref{Subsec-inf}, these ought to be the formal duals of 
certain nilpotent algebras, now naturally associated to the point $p\in M$. Of course, there is a well-known collection of such algebras associated to the pair $(p,M)$.

\begin{definition}[\bf Infinitesimal neighborhoods of manifolds]
\label{InfitesimalNbdinManTraditional}
Let $p\in M$ be a point in a finite-dimensional manifold. For any $k\in \NN \cup \{\infty\}$, 
\begin{itemize}[leftmargin=6mm] 
 \item[{\bf (i)}] the \textit{$k$-order 
infinitesimal neighborhood at $p$}, denoted by $\DD_{p,M}(k)$, is the formal dual of the jet algebra denoted by
\begin{align*}
J^k_p(M):=  \Big\{j^k_p f=[f] \, \big{|}\, f \sim f' \in C^\infty(M) \iff \partial_{I} f (p) 
 = \partial_{I} f' (p) \hspace{0.3cm} \forall \,I \hspace{0.3cm}  \mathrm{s.t.}\hspace{0.3cm}  k\geq |I| \geq0 \Big\} \quad  \in \quad \mathrm{CAlg}_\FR \, ,
\end{align*}
based at $p\in M$, with the partial derivatives being computed in any local chart;

\item[{\bf (ii)}] the \textit{$\infty$-order neighborhood at $p$} is given by the formal dual of the limit of $k$-order jet algebras in $\mathrm{CAlg}_\FR$, or dually as the colimit of $k$-order infinitesimal neighborhoods in $\mathrm{CAlg}_\FR^{op}$. \footnote{
Strictly speaking, $J^\infty_p (M)$ is infinite-dimensional and not with nilpotent maximal ideal, hence not an example of a traditional Weil algebra (Def. \ref{InfinitesimallyThickenedPoints}).
}
\end{itemize} 
\end{definition}

Tautologically, of course, $J^k_p(M)$ coincides with the $k$-jets of sections of the trivial bundle $ M\times \FR \rightarrow M$. The algebra structure is inherited by that of the fibers, that is $j^k_p(f)\cdot j^k_p(f'):= j^k_p(f\cdot f')$ for any representatives $f,f'\in C^\infty(M)$ of the corresponding jet classes. For the case of the Cartesian space $\FR^d$, there is a canonical isomorphism $ \DD^d(k)\cong \DD_{0,\FR^d}(k)$ given dually by the (formal) Taylor expansion to order $k$, with respect to the canonical coordinates on $\FR^d$,
\begin{align*}
J^k_0(\FR^d) &\; \xlongrightarrow{\;\;\sim\;\;} \; \CO\big(\DD^d(k)\big)=\FR[\epsi^1,\cdots, \epsi^d]/(\epsi^1,\cdots,\epsi^d)^{k+1} \\ 
j^k_0 f&\; \xmapsto{\quad} \; \sum_{I=0}^{k} \frac{1}{I!}\partial_{|I|}f(0)\cdot \epsi^{I}= f(0) + \partial_i f(0) \cdot \epsi^i +\cdots 
\nn \, ,
\end{align*}
by which we shall identify $\DD_{0,\FR^d}(k)$ with $\DD^d(k)$ from now on. Given a local chart $\phi:\FR^d \xrightarrow{\sim} U_p\hookrightarrow M$ centered at $p\in M$, there is an induced isomorphism on jet algebras given by
\begin{align}\label{infNbdIsomorphism}
j^k\phi^*\;:\; J^k_p(M)& \;\xlongrightarrow{\sim}\; J^k_0(\FR^d) \cong \CO\big(\DD^d(k)\big) \\
j^k_p f&\;\longmapsto \; j^k_0(f\circ \phi)=f\circ \phi (0) + \partial_{i} (f\circ \phi) (0) \cdot \epsi^i +\cdots \, , \nn
\end{align}

\vspace{1mm} 
\noindent which is, in fact, already implicit in Def. \ref{InfitesimalNbdinManTraditional} of $J^k_p(M)$. In particular, any such isomorphism makes the nilpotent nature of $J^k_p(M)$ apparent, giving an explicit splitting $J^k_p(M)\cong_{\mathrm{Vect}}\FR\oplus V$ where $V$ contains only nilpotent elements. Of course, this means that $j^k_p(f)\in J^k_p(M)$ is nilpotent if and only if $f(p)=0$. Denoting the formal dual of this isomorphism by 
$$
j^k\phi \, :\, \DD^d(k)\xlongrightarrow{\sim} \DD_{p,M}(k) \, ,
$$
the nilpotency of the latter algebra, together with the canonical `embedding' 
$$
\DD_{p,M}(k)\longhookrightarrow M\, ,
$$
given dually by the projection 
$$
f\longmapsto j^k_p f \, ,
$$ 
justify its interpretation as an infinitesimal neighborhood in $M$. The above somewhat tautological observations are summarized in the following commutative diagram 
\begin{equation}\label{infNbdIsoDiagram}
\begin{gathered}
\xymatrix@C=2em@R=1.5em  { \DD_{p,M}(k) \; \ar@{^{(}->}[rr] &&  M 
	 \\ 
  \DD^d(k)\ar[u]_{\sim}^{j^k\phi} \; \ar@{^{(}->}[rr]^{\iota_0}  &&  \FR^d  \ar@{^{(}->}[u]_{\phi}
}
\end{gathered}
\end{equation}
for each local chart $\phi:\FR^d\rightarrow M$ centered at $p\in M$. Hence, one may also think of the infinitesimal neighborhood at a point, up to isomorphism, as an infinitesimal disk attached to it via a local chart. \footnote{This latter -- up to isomorphism -- point of view is taken as a definition of infinitesimal neighborhoods in manifolds by \cite{KS17}.}

\medskip 
We may now generalize the interpretation of jets $j^k_p f$ of any bundle $F\in M$ over $M$ as sections over infinitesimal neighborhoods $\DD_{p,M}(k)\hookrightarrow M$.

\begin{lemma}[\bf Jets as sections over infinitesimal neighborhoods]\label{JetsofSections=InfinitesimalJets}
Let $F\rightarrow M$ be a fiber bundle over an $d$-dimensional manifold $M\in \mathrm{Man}$. Then: 

\begin{itemize}
\item[{\bf (i) }] 
The set of $k$-jets $J^k_p(F)$ of sections at a point $p\in M$ are in 1-1 correspondence with sections over the $k$-order infinitesimal neighborhood of $p\in M$, i.e., maps $\DD_{p,M}(k)\rightarrow F$ such that the following diagram commutes 
\[ 
\xymatrix@R=1em  { &&  F \ar[d]^{\pi}
	\\ 
	\DD_{p,M}(k)\ar[rru] \ar@{^{(}->}[rr] && M
    \mathrlap{\,.}
}   
\]

\item[{\bf (ii) }]
Furthermore, for any choice of local chart $\phi :\FR^d\hookrightarrow U_p\subset M$ centered at $p\in M$, they are in 1-1 correspondence with maps $\DD^d(k)\rightarrow F$ such that the following diagram commutes 
\vspace{-2mm} 
\[
\xymatrix@C=4pc @R=1.2pc { &&  F \ar[d]^{\pi}
	\\ 
	\DD^d(k)\ar[rru] \ar@{^{(}->}[r]^>>>>>>{\iota_0} &\FR^d \ar[r]^{\phi} & M
     \mathrlap{\,.}
}   
\]
\end{itemize} 
\begin{proof}
Fix a local chart $\phi:\FR^d\rightarrow U_p\subset M$ centered at p. A section $\sigma$ of F over $\DD^d(k)\xhookrightarrow{\iota_0}\FR^d\xrightarrow{\phi} M$ fits in the commutative diagram
\vspace{-2mm} 
\[
\xymatrix@C=4pc @R=1.2pc  { &&  F \ar[d]^{\pi}
	\\ 
	\DD^d(k)\ar[rru]^\sigma \ar@{^{(}->}[r]^>>>>>>{\iota_0} &\FR^d \ar[r]^{\phi} & M \, .
}\,   
\]
Pasting the commutative diagram \eqref{infNbdIsoDiagram} given by the isomorphism $j^k\phi :\DD^d(k) \rightarrow \DD_{p,M}(k)$, we get 
\[
\xymatrix@C=3pc @R=1.2pc  { &&  F \ar[d]^{\pi}
	\\ 
	\DD^d(k)  \ar[rru]^\sigma \; \ar[d]_{j^k\phi} \ \ar@{^{(}->}[r] &\FR^d \ar[r] & M \\
 \DD_{p,M}(k) \ar@{^{(}->}[rru]
}\,    
\]
which commutes since the internal diagrams do. Since $j^k\phi$ is an isomorphism, we may define $\sigma \circ (j^k\phi)^{-1}:\DD_{p,M}(k)\rightarrow F$, which forms a section of F over the infinitesimal
neighborhood at $p$, since 
\[ 
\xymatrix@R=1em@C=2.6em  { &&  F \ar[d]^{\pi}
	\\ 
	\DD_{p,M}(k)\ar[rru]^-{\sigma\circ (j^k \phi)^{-1}} \ar@{^{(}->}[rr] && M
}   
\]
commutes by the previous diagram. Similarly, any section $\tilde{\sigma}$ over the infinitesimal neighborhood 
gives a section over $\DD^{d}(k)\xhookrightarrow{\iota_0}\FR^d \xrightarrow{\phi} M$ defined by $\tilde{\sigma}\circ j^k\phi$.
This proves the {\bf (ii)} bijection of the statement.

Thus, it remains to prove the correspondence with jets of sections $J^k_p(F)$ of $F\rightarrow M$ in a local chart of $\tilde{\phi}:\FR^d\times \FR^m\xrightarrow{~} \ F|_{U_p} $, compatible with the chart on the base $\phi:\FR^d\rightarrow M$. 
Denote the corresponding coordinates functions on $\FR^d\times \FR^m$ by $\{x^i,y^j\}$ with $i=1,\cdots,d$ and $j=1,\cdots, m$ respectively.
By Hadamard's Lemma \ref{HadamardsLemma}, a map $\sigma: \DD^d(k)\rightarrow F|_{Up} \xrightarrow{\phi^{-1}} \FR^d\times \FR^m$ 
is completely determined, dually by the image of the linear coordinate functions $\{x^i,y^j\} \subset C^\infty(\FR^d\times \FR^m)$. 
The commutativity of the diagram of a section $\sigma: \DD^d(k)\rightarrow F|_{Up} \xrightarrow{\phi^{-1}} \FR^d\times \FR^m$ over $\FR^d$
then fixes the dual map to be of the form
\begin{align*}
\sigma^*: \CO(\FR^d\times \FR^m) &\;\longrightarrow \; \CO\big(\DD^d(k)\big) \\
x^i&\; \longmapsto  \;  \epsi^i \\[-3pt] 
y^j &\;\longmapsto \; c^j_{i} \cdot \epsi^i + c^j_{i_1 i_2} \cdot \epsi^{i_1}\epsi^{i_2} +\cdots + c^{j}_{i_1 \cdots i_k} \cdot \epsi^{i_1}\cdots \epsi^{i_k} \,,
\end{align*}
for some collection of real numbers $\big\{c^j_i,\,  c^{j}_{i_1 i_2},\,  \cdots ,\,  c^{j}_{i_1\cdots i_k} \big\}$ for $j=1,\cdots, m$, $i=1,\cdots, d$, symmetric under exchange of indices. Note that we have tacitly assumed, without loss of generality, that the 
chart on $F|_{U_p}$ is further centered around the image of the point  $\tilde{\pi}:*\hookrightarrow \DD^d(k)\rightarrow F$ in the fiber $\pi^{-1}(p)$. Such a map polynomially extends along $\DD^d(k)\xhookrightarrow{\iota_0} \FR^d $ to a map 
given dually by
\begin{align*}
\tilde{\sigma}^*: \CO(\FR^d\times \FR^m) &\;\longrightarrow \; \CO(\FR^d) \\
x^i&\;\longmapsto \;   x^i \\[-3pt] 
y^j &\;\longmapsto \; c^j_{i} \cdot x^i + c^j_{i_1 i_2} \cdot x^{i_1}x^{i_2}
+\cdots + c^{j}_{i_1 \cdots i_k} \cdot x^{i_1}\cdots x^{i_k} \, ,
\end{align*} 
or explicitly in the manifold picture, 
\begin{align*}
\tilde{\sigma}: \FR^d & \;\longrightarrow \; \FR^d\times \FR^m \\[-2pt]
\tilde{x}=\tilde{x}^i\cdot e_{i} &\; \longmapsto \; \big(\tilde{x}^i\cdot e_i\, , \; c^j_i\cdot \tilde{x}^i\cdot e_j 
+ \cdots + c^j_{i_1,\cdots i_k} \cdot \tilde{x}^{i_1}\cdots\tilde{x}^{i_k} \cdot e_j \big)\, ,  
\end{align*}
where ${\tilde{x}^i}$ denote the 
standard coordinates of a point $\tilde{x}\in \FR^d$ and $\{e_{i},e_{j}\}$ are the standard basis vectors 
for $\FR^d$ and $\FR^m$ respectively. For $k=\infty$, the fact that there exists a smooth function with Taylor series corresponding to the former infinite formal power series it is the content of Borel's lemma. By construction, $\tilde{\sigma}$ fits in the commutative diagram
\[ 
\xymatrix@C=3pc @R=1pc  { &\quad \FR^d\times \FR^m \ar[r]^-{\tilde{\phi}}  & F|_{U_p} \ar[d]^{\pi}
	\\ 
	\FR^d \ar[rr]^{\phi} \ar[ru]^{\tilde{\sigma}} &&U_p \, ,
}   
\]
and so the map $\tilde{\phi} \circ \tilde{\sigma}\circ \phi^{-1}: U_p\rightarrow F$ defines a local section of $F$, 
whose jet class is completely determined by $\sigma: \DD^d(k) \rightarrow F$, by definition. Conversely, the representative 
of any jet class $j^k \tilde{\sigma} \in J^k_p(F)$, defines a map $\DD^d(k)\rightarrow F|_{U_p}$ as its (formal) Taylor 
expansion in the given compatible trivialization as above. 
\end{proof}
\end{lemma}
We have thus characterized the set of jets of sections $J^k_p(F)$ of a bundle at a point $p\in M$, as sections over the 
infinitesimal neighborhood $\DD_{p,M}(k)\hookrightarrow M$, hence dispensing of \textit{almost} all reference to local charts.
Of course, the original definition (Def. \ref{InfitesimalNbdinManTraditional}) of $J^k_p(M)$ still implicitly employs local charts. 

\smallskip 
Working within thickened smooth sets, 
there is a completely natural and intuitive definition of infinitesimal neighborhoods in any thickened smooth set, which does not refer 
to any choice of local chart, and hence actually allows one to define jets of sections of any map between thickened smooth sets. 
To that end, let us first declare following definition.
\begin{definition}[\bf Infinitesimal Shape]\label{infinitesimalShape}
Let $\CF\in\ThickenedSmoothSets$ be a thickened smooth set. The \textit{infinitesimal shape}, 
or \textit{de Rham shape}, $\mathfrak{J} \CF\in \ThickenedSmoothSets$ of $\CF$ is defined by
\begin{align*}
    \frJ  \CF \;:\; \mathrm{ThCartSp}^\op &\longrightarrow \mathrm{Set} \\
    \FR^k\times \DD &\longmapsto \CF(\FR^k) 
\nn 
\end{align*}
for any $\FR^k \times \DD \in \mathrm{ThCartSp}$.
\end{definition}
In words, the infinitesimal shape of a space $\CF$ has the same $\FR^k$-plots, and its infinitesimal plots carry no additional information. Notice there is a natural smooth map
\begin{align*}
\eta_\CF \;:\; \CF &\longrightarrow \frJ \CF \\
\phi^k_{\epsi} &\longmapsto \phi^k_{\epsi=0}
\end{align*}
which simply sends every $\FR^k\times \DD$-plot to its pullback $\phi^k_{\epsi=0}:=\iota^* \phi^k_{\epsi}$  
along the canonical embedding $\iota: \FR^k\times \{*\} \hookrightarrow \FR^k \times \DD$. 
Thinking of $(\FR^k\times \DD)$-plots in $\CF$ as smooth $\FR^k$-plots with an extra thickening given by an infinitesimal extension, 
then $\eta_\CF$ exhibits $\frJ \CF$ as the space obtained by identifying / collapsing the infinitesimal extension, 
since any $\phi^k_{\epsi},\psi^k_{\epsi} \in \CF(\FR^k\times \DD) $ with $\phi^k_{\epsi=0}=\psi^k_{\epsi=0}$ map to precisely the 
same $(\FR^k\times \DD)$-plot in $\frJ \CF$. Colloquially, one says that the operation $\frJ$ identifies `infinitesimal neighbors'. 
It is immediate to see that the infinitesimal shape forms a functor \footnote{
Incidentally, there exists a left adjoint $\mathfrak{R} : \ThickenedSmoothSets\rightarrow \ThickenedSmoothSets$ called reduction, which reduces a thickened smooth set 
to its underlying body smooth set, in a precise sense, rather than 
reducing its plots. These, in fact, arise from a deeper structure involving a quadruple of adjoints, termed ``\textit{Differential Cohesion}''. For full details on these aspects, see \cite{dcct}\cite{KS17}. }
$$
\frJ \;:\; \ThickenedSmoothSets \longrightarrow \ThickenedSmoothSets\, ,
$$
with the corresponding maps $\eta_{\CF}: \CF\rightarrow \frJ \CF$ comprising a natural transformation
$$
\eta \;:\; \mathrm{id} \longrightarrow \frJ \, .
$$
With the infinitesimal shape at hand, we may form an intuitive thickened smooth set that deserves the name of the ``$\infty$-order infinitesimal 
neighborhood'' of a point $p:{*}\rightarrow \CF$.

\begin{definition}[\bf Synthetic infinitesimal neighborhood]\label{SyntheticInfinitesimalNeighborhood}
Let $p:* \rightarrow \CF$ be a point in a thickened smooth set $\CF\in \ThickenedSmoothSets$. The  \textit{(synthetic) infinitesimal neighborhood} 
 $\DD_{p,\CF}$ of $\CF$ at $p$ is defined as the subspace given by the pullback
 \[
\xymatrix@C=1.8em@R=1.4em  { \DD_{p,\CF}  \ar[rr] \ar[d] && {*} \ar[d]^{p} 
	 \\ 
  \CF \ar[rr]^{\eta_\CF}  &&  \frJ \CF \, .
}
\]
\end{definition}
Explicitly, this means that the $\FR^k\times \DD$-plots of $\DD_{p,\CF}$ are given by
\begin{align*} 
\DD_{p,\CF}(\FR^k\times \DD)&=
\Big\{ \big(\phi^k_{\epsi},{*}\big)\in \CF(\FR^k\times \DD)\times {*}(\FR^k \times \DD)  
\, \big\vert \; \eta_\CF \big(\phi^k_{\epsi}\big) = p(*)  \Big\} \nn \\
&\cong
\Big\{\phi^k_{\epsi}\in \CF(\FR^k \times \DD) \, \big\vert \; \phi^k_{\epsi=0} = \mathrm{const}_{p} \in \CF(\FR^k) \Big\} \, .
\end{align*}
Said otherwise, the plots of $\DD_{p,\CF}$ are constant onto the point $p\in \CF(*)$ along their purely finite directions, 
but free to probe all infinitesimally thickened directions inside $\CF$ at p. The left vertical map of the pullback defines 
a natural embedding 
$$
\iota_{p} :\DD_{p,\CF} \longhookrightarrow \CF\, ,
$$ realizing the infinitesimal neighborhood as a thickened (smooth) subspace of $\CF$. Thus, the \textit{synthetic} infinitesimal neighborhood captures precisely the intuitive properties 
of a reasonable notion of a ``$\infty$-order infinitesimal neighborhood'', applicable for any ambient thickened smooth set. Of course, this definition recovers 
that the traditional notion of an infinitesimal neighborhood of a finite-dimensional manifold (Def. \ref{InfitesimalNbdinManTraditional}).

\begin{lemma}[\bf Synthetic infinitesimal neighborhood at point in manifold]
\label{SyntheticInfinitesimalNeighborhoodOfManifold}
Let $y(M)\in \mathrm{ThickenedSmoothSet}$ be a smooth $d$-dimensional manifold viewed as a thickened smooth set and $p:{*}\rightarrow y(M)$ a 
point in $M$. There is a canonical isomorphism 
$$
\DD_{p,M} \cong y\big(\DD_{p,M}(\infty)\big)\, 
$$
between the synthetic and the traditional $\infty$-order infinitesimal neighborhoods at $p\in M$.
\begin{proof}
It suffices to show that their plots are in natural bijection. We prove this for the cases of $\big(\FR^k\times \DD^m(l)\big)$-plots, with the more general cases of thickened probes $\FR^k \times \DD$ following similarly via any subspace identification of from Lem. \ref{InfinitesimalPointsAsSubspacesOfInfinitesimalDisks}. On the one hand, we have
 \begin{align*}   
\DD_{p,M}\big(\FR^k\times \DD^m(l)\big)&\cong \Big\{\phi^k_{\epsi}\in y(M)\big((\FR^k \times \DD^m(l)\big) \; \big|\; \phi^k_{\epsi=0} = \mathrm{const}_{p} \in y(M)(\FR^k) \Big\} \\
&\cong \Big\{ \phi^k_\epsi \in \mathrm{Hom}_{\mathrm{CAlg}_\FR^{op}}\big(\FR^k\times \DD^m(l), M \big)\; \big| \; \phi^k_{\epsi=0}= \mathrm{const}_p \in \mathrm{Hom}_{\mathrm{Man}}(\FR^k, M) \Big\} \, .
\end{align*}
Now any algebra morphism $C^\infty(M)\rightarrow C^\infty(\FR^k)\otimes \CO\big(\DD^m(l)\big)$ corresponds precisely to a higher order differential  operator $X_p: C^\infty(M)\rightarrow C^\infty(\FR^k)$ relative to the evaluation pullback algebra map $p: C^\infty(M)\rightarrow C^\infty(\FR^k)$, given by ``constant evaluation'' at $p\in M$ (cf. proof of Lem. \ref{TangBundlesCoincide}). This means that the value of the differential operator $X_p$ depends only on germs of functions, hence factors through any local chart subalgebra 
$$
X_p \, : \, C^\infty(M)\xlongrightarrow{\;\iota^*_{U_p}\;} C^\infty(U_p)\xlongrightarrow{\sim} C^\infty(\FR^d) \longrightarrow C^\infty(\FR^k)
$$
centered at $p\in U_p \hookrightarrow M$.

In other words, any such plot $\phi^k_\epsi$ factors through a local chart of $M$ 
$$
\phi^k_\epsi \, : \, \FR^k\times \DD^m(l) \longrightarrow \FR^{d}\xrightarrow{\;\;\sim\;\;} U_p \longhookrightarrow M \, ,
$$
and so is determined by a map of thickened Cartesian spaces $\tilde{\phi}^k_{\epsi} \in \mathrm{Hom}_{\ThickenedCartesianSpaces}\big(\FR^k \times \DD^m(l), \,\FR^d\big)$. By Prop. \ref{RAlgebraMapsAreCinfty}, these are dually 
determined by the image of the linear coordinates $\{x^i\} \in C^\infty(\FR^d)$ 
$$
x^i \;\; \xmapsto{\qquad} \;\;  c^i_{a} \epsi^a + c^i_{ab} \cdot \epsi^a\epsi^b + \cdots + c^{i}_{a_1\cdots a_l} \cdot \epsi^{a_1}\cdots \epsi^{a_{l}} \quad  \in \quad  C^\infty(\FR^k)\otimes\CO\big(\DD^m(l)\big),
$$
 with $\{c^i_{a},\cdots c^{i}_{a_1\cdots a_l}\}\subset C^\infty(\FR^k)$, with no term valued purely in $C^\infty(\FR^k)\hookrightarrow C^\infty(\FR^k)\otimes\CO\big(\DD^m(l)\big)$ since this vanishes by the constancy assumption along $p:\FR^k\hookrightarrow \FR^k\times \DD^m(l)\rightarrow M$ (and the fact that the chart is chosen to be centered at $p\in M$). 
 
 On the other hand, 
\begin{align*}
y\big(\DD_{p,M}(\infty)\big)\big(\FR^k \times \DD^{m}(l)\big)&\cong 
\mathrm{Hom}_{\mathrm{CAlg}_{\FR}}\big(J^\infty_pM\, ,\,  C^\infty(\FR^k)\otimes\CO( \DD^m(l)) \big) \, , 
\end{align*}
and in the given local chart centered  p, 
$J^\infty_p M\cong \CO(\DD^d(\infty))\cong \FR[[\epsi^1,\cdots, \epsi^d]]$ by \eqref{infNbdIsomorphism}, 
so that an algebra map $J^\infty_p M \rightarrow C^\infty(\FR^k)\otimes\CO\big(\DD^m(l)\big)$ is again completely determined by the image of the 
generators $\{\epsi^i\}$, as
$$
\epsi^i \;\; \xmapsto{\qquad}  \;\; c^i_{a} \epsi^a + c^i_{ab}\cdot \epsi^a\epsi^b + \cdots +
c^{i}_{a_1\cdots a_l}\cdot \epsi^{a_1}\cdots \epsi^{a_l} \quad \in \quad C^\infty(\FR^k)\otimes\CO\big(\DD^m(l)\big)\, .
$$
Thus, we see there is a canonical correspondence between plots in $\DD_{p,M}$ and those in 
$y\big(\DD_{p,M}(\infty)\big)$, which is moreover manifestly functorial under pullbacks of probe spaces.
\end{proof}
\end{lemma}

\paragraph{Formalizing perturbative field theory: A toy example.} 
In the traditional practice of perturbative field theory, one commonly considers the perturbation theory of local functionals of the fields (action functionals, currents, Euler--Lagrange operators etc.) by syntactically ``Taylor expanding'' in a ``small perturbation h'' around a fixed field configuration $\phi \in \CF$. Here, we motivate that the restriction of said local functionals to infinitesimal neighborhoods around fixed field configurations 
$$
\DD_{\phi, \CF} \longhookrightarrow \CF 
$$ 
ought to encode this traditional notion of perturbation theory around a fixed field configuration $\phi$. We leave the explicit formulation and proof of this result in the infinite-dimensional setting of field theory for \cref{PerturbationSection}, as it requires more technical background and effort than the following simple motivating example.

\smallskip 
To witness the plausibility of this, we consider the toy-model case of 0-dimensional field theory, that is, with field bundle $M\rightarrow *$ over a 0-dimensional (point) space-time. In that case, the space of fields is $[*,M]\cong M\in \SmoothManifolds \hookrightarrow \ThickenedSmoothSets$, with dynamics described by (the criticality of) a smooth map of manifolds $S:M\rightarrow \FR$. Embedding the picture in thickened smooth sets, or equivalently simply in $\ThickenedSmoothManifolds\hookrightarrow \mathrm{CAlg}_\FR^{op}$ in this finite-dimensional case, we may choose a field configuration $p\in M$ (critical or not) and consider the  restriction on the infinitesimal neighborhood around it
$$
\DD_{p,M}(\infty) \longhookrightarrow M\, .
$$
Now notice that the (defining) function algebra on the infinite-order infinitesimal disk $\DD^d(\infty)\hookrightarrow \FR^d$ is further canonically identified with the (formal completion of) the symmetric algebra of the cotangent space $T_0^* \FR^d$
$$
\CO(\DD^d(\infty)\big) \, := \,  \FR[[\epsi^1,\cdots,\epsi^d]]\;  \cong \;  \FR[[\dd x^1,\cdots , \dd x^d]] 
 \; \cong \; \widehat{\mathrm{Sym}}^\bullet (T^*_0 \FR^d) \, .
$$
Under this identification, by the duality of $T^*_0 \FR^d$ and $T_0 \FR^d$, the elements of this algebra may also be thought of as `formal polynomial maps' on the tangent space $T_0\FR^d$ at $0\in \FR^d$, i.e., formal polynomials 
$$
T_0 \FR^d \longrightarrow \FR \,.
$$ 

The same can be said for the (jet) function algebras of infinitesimal neighborhoods in any manifold $M$ (Def. \ref{InfitesimalNbdinManTraditional}), but the isomorphism to formal polynomials then depends on the choice of chart. Indeed, recall by \eqref{infNbdIsomorphism} that a choice of chart gives an isomorphism 
\begin{align*}
j^\infty\phi^*\;:\; J^\infty_p(M)& \;\xlongrightarrow{\quad \sim\quad}\;  \CO\big(\DD^d(\infty) \big) \cong \widehat{\mathrm{Sym}}^\bullet (T^*_0 \FR^m) \\
j^k_p f&\;\; \xmapsto{\qquad}  \; \; j^k_0(f\circ \phi)=f\circ \phi (0) + \partial_{i} (f\circ \phi) (0) \cdot \dd x^i + \tfrac{1}{2} \partial_i \partial_j (f\circ \phi)(0) \cdot \dd x^i \cdot \dd x^j +\cdots \, , \nn
\end{align*}
which upon pulling back \textit{each} of cotangent vector generators $\{\dd x^i\}$ along the same chart
yields a composite algebra isomorphism
\begin{align}\label{JetAlgebraAsFormalCotangentPolynomials}
\vspace{-2mm} 
\xymatrix@C=1.8em@R=1.4em  { J^\infty_p(M)  \ar[rr]^-{\sim_\phi} \ar[d]^{j^\infty \phi^*} && \widehat{\mathrm{Sym}}^\bullet (T^*_p M) 
	 \\ 
  \CO\big(\DD^d(\infty)\big) \ar[rr]^-{\sim}  &&  \widehat{\mathrm{Sym}}^\bullet (T^*_0 \FR^m) \ar[u]_{(\phi^{-1})^*} \, .
}
\end{align}
Thus any choice of local chart $\phi$, interprets the (canonical jet) function algebra on $\DD_{p,M} \hookrightarrow M$ as formal power series on the tangent space
\begin{align*}
j^\infty_p f |_\phi \;\; : \;\; T_p M \;\;  &\xrightarrow{\qquad} \;\; \FR \\ 
h^i \partial_i|_{p} & \; \xmapsto{\qquad} \; f\circ \phi (0) + \partial_{i} (f\circ \phi) (0) \cdot h^i + \partial_i \partial_j (f\circ \phi)(0) \cdot h^i h^j +\cdots .
\end{align*}
Now, given an `action' functional $S:M\rightarrow \FR$ (or analogously its differential $\dd S : M \rightarrow T^*M$ etc.),
we may restrict to the infinitesimal neighborhood around $p\in M$
$$
S\vert_{\DD_{p,M}} \;:\; \DD_{p,M} \longhookrightarrow M \longrightarrow \FR \, ,
$$
which, dually, corresponds canonically to the $\infty$-jet $j^\infty_p S\in J^\infty_p (M) $, via $\mathrm{Hom}_{\mathrm{CAlg}}\big(C^\infty(\FR),\, J^\infty_p(M)\big) \, \cong \, J^\infty_p(M)$. Thus, under a \textit{choice} of coordinate chart  around p, whereby $J^\infty_p M\cong_\phi \widehat{\mathrm{Sym}}^\bullet (T^*_p M)$ via \eqref{JetAlgebraAsFormalCotangentPolynomials}, this yields exactly 
the corresponding formal Taylor expansion
$$
Tayl_{p,\phi}(S) = S(p) + \sum_{|l|=1}^\infty \tfrac{1}{l!} \partial_{l}S (p) \cdot \dd x^l \, . 
$$
By duality, this may be seen as acting on tangent vectors $h= h^i\partial_i\vert_p \in T_p M$,
hence recovering precisely the formal expansion
$$
``S(p+h)"\equiv \mathrm{Tayl}_{p,\phi}(S)(h) := S(p) + \partial_{i} S(p) \cdot h^i + \partial_i \partial_j S(p) \cdot h^i h^j +\cdots 
$$
in the ``small perturbation'' $h\in T_p M$, with respect to the given chart. Thus, this is only a formal infinite expansion (``function'') on tangent vectors. Evaluation on ``nilpotent tangent parameters'', i.e., by evaluating e.g. (on the nilpotent part of) $T_p M \otimes \CO(\DD^1(l))$, truncates at finite order and requires no convergence assumption.

\smallskip 
\begin{remark}[\bf Equivalent perturbative field theories]
Formulated in this rigorous setting, it is immediate that a choice of different chart $\psi$ gives an equivalent perturbative field theory, in the sense that there exists (by construction) an isomorphism of infinitesimal disks $j^\infty( \phi^{-1} \circ \psi) = j^\infty \phi^{-1} \circ j^\infty \psi$ intertwining the two restricted (i.e., perturbative) action functionals. In this sense, the abstract restriction to the infinitesimal neighborhood $S\vert_{\DD_{p,M}}:\DD_{p,M} \rightarrow \FR$ encodes perturbative field theory in a coordinate-invariant manner. 
\end{remark}

Notice, combining Lem. \ref{JetsofSections=InfinitesimalJets} and Lem. \ref{SyntheticInfinitesimalNeighborhoodOfManifold}, it follows that the set of $\infty$-order jets $J^\infty_p(F)$ of a fiber bundle $F\rightarrow M$ is in bijection with sections over the synthetic infinitesimal neighborhood
\[ 
\xymatrix@C=2pc @R=1pc { &&  y(F) \ar[d]^{\pi}
	\\ 
	\DD_{p,M} \ar@{^{(}->}[rr] \; \ar[rru]  && y(M)
}   
\]
as maps of thickened smooth sets. Formulated as above, the notion of $\infty$-jets of sections may be immediately extended to arbitrary `bundles' of
thickened smooth sets.
\begin{definition}[\bf $\infty$-jets of sections]\label{syntheticJets}
Let $\pi:\CG \rightarrow \CF$ be a map of thickened smooth sets and $p:*\rightarrow \CF$ a point in $\CF$. We define the set $J^\infty_p(\CG)$ of `$\infty$-jets of sections at p' as the set of sections of $\pi$ over the infinitesimal neighborhood at $p$.
That is, morphisms $\DD_{p,\CF} \rightarrow \CG$ of thickened smooth sets such that
\[ 
\xymatrix@C=2pc @R=1pc  { &&  \CG\ar[d]^{\pi}
	\\
	\DD_{p,\CF} \ar[rru] \ar@{^{(}->}[rr] && \CF
}   
\]
commutes. Suggestively, we denote the collection of jets over all points by $J^\infty(\CG)(*):= \underset{p\in \CF(*)}{\bigcup}J^\infty_p(\CG) $.
\end{definition}


Having formulated now the notion of $\infty$-jets at a point purely in terms of maps in thickened smooth sets, there is a natural thickened smooth set structure on $J^\infty_p(F)$
given as in \eqref{FieldSpaceAsPullback}, by considering $\FR^k\times \DD$-parametrized sections over $\DD_{p,M}$. 
This suggests that it is possible to further endow the full collection $\underset{p\in M }{\bigcup}J^\infty_p(F)$ of jets of a fiber
bundle $F\rightarrow M$ with a thickened smooth structure,  by working entirely within $\ThickenedSmoothSets$. This is indeed the case and in fact applies in the generality of arbitrary bundles in $\ThickenedSmoothSets$ and their jets 
(Def. \ref{syntheticJets}), but that requires some further natural constructions which we include in the Appendix \cref{App-SyntheticJetProlongation}, since these are not necessary for the field theoretic purposes of this text.

\paragraph{\bf Infinitesimal neighborhoods of submanifolds.}
Let us instead step back into the more concrete setting of sections of finite-dimensional fiber bundles $F\rightarrow M$ and their jets along embedded submanifolds, a notion that naturally arises in field theory when considering ``initial conditions'' of fields along a submanifold (see e.g. \cite[\S 6.3]{GS23}).
Viewing any point $p:*\hookrightarrow M$ as an embedded submanifold, what we have shown in Lem. \ref{JetsofSections=InfinitesimalJets} and Lem. \ref{SyntheticInfinitesimalNeighborhoodOfManifold} may be phrased as follows: Sections of the jet bundle  $J^\infty F\rightarrow M$ over $p:*\hookrightarrow M$ are precisely sections of the original bundle $F\rightarrow$ over the infinitesimal neighborhood $\DD_{p,M}\hookrightarrow M$

\begin{align}\label{SectionsOverInfinitesimalNbdDiagram}
\left(\!\!\!\!
\raisebox{18pt}{
\xymatrix@R=1.4em@C=3em  { &&  J^\infty_M F \ar[d]^{}
	\\ 
	{*}  \ar[urr]^-{j^\infty_p \phi} 
    \ar@{^{(}->}[rr]^>>>>>>>>>{p} && M
}   
}
\!\right)
\quad 
\xleftrightarrow{\quad 1\mbox{-} 1\quad} 
\quad 
\left(\!\!\!\!
\raisebox{18pt}{
\xymatrix@R=1.4em@C=3em  { &&  F\ar[d]^{}
	\\ 
	\DD_{p,M} \ar[urr]^-{}  
    \ar@{^{(}->}[rr]^>>>>>>>>{\iota_p} && M
}   
}
\!\!\right) .
\end{align}
This suggests that a similar equivalent description for sections of the jet bundle over \textit{any} embedded submanifold $\Sigma \hookrightarrow M$ should exist, for an appropriate notion of an ``infinitesimal neighborhood around $\Sigma$ in $M$''. However, 
working within $\ThickenedSmoothSets$, there is a natural generalization of Def. \ref{SyntheticInfinitesimalNeighborhood} to yield such a notion (cf. Def. \ref{InfinitesimalNeighborhoodOfDiagonal}).

\begin{definition}[\bf Synthetic infinitesimal neighborhood of submanifold]\label{SyntheticInfinitesimalNeighborhoodOfSubmanifold}
Let $\iota_\Sigma:\Sigma \hookrightarrow M$ be an embedding submanifold of $M$, viewed as thickened smooth sets, with infinitesimal shape $\iota_{\frJ \Sigma}: \mathfrak{J}\Sigma \hookrightarrow \frJ M$. The ($\infty$-order)  \textit{synthetic infinitesimal neighborhood} 
 $\DD_{\Sigma,M}$ of $\Sigma$ in $M$ is defined as the thickened smooth subspace given by the pullback
 \[
\xymatrix@C=1.8em@R=1.4em  { \DD_{\Sigma,M}  \ar[rr] \ar[d] && \frJ \Sigma \ar[d]^{ \frJ \iota_\Sigma } 
	 \\ 
  M \ar[rr]^{\eta_M}  &&  \frJ M \, .
}
\]
\end{definition}

Explicitly, this means that the $\FR^k\times \DD$-plots of $\DD_{\Sigma,M}$ are given by
\begin{align*} 
\DD_{\Sigma,M}(\FR^k\times \DD)&=
\Big\{ \big(\phi^k_{\epsi},\sigma^k_{\epsi=0}\big)\, \in \,  M(\FR^k\times \DD)\times \frJ \Sigma(\FR^k \times \DD)  
\, \big\vert \; \eta_M \big(\phi^k_{\epsi}\big) = \iota_{\frJ\Sigma} \big(\sigma^k_{\epsi}\big)  \Big\} \nn \\
&\cong
\Big\{ \big(\phi^k_{\epsi},\sigma^k_{\epsi=0}\big)\, \in \,  M(\FR^k\times \DD)\times \Sigma(\FR^k \times \DD)  
\, \big\vert \; \phi^k_{\epsi=0} = \sigma^k_{\epsi=0} \in \Sigma(\FR^k)  \Big\} \nn \, .
\end{align*}
In words, the plots of $\DD_{\Sigma,M}$ are given by arbitrary plots  $\phi^k_\epsi$ in the ambient $M$, with finite directions $\phi^k_{\epsi=0} = \sigma^{k}_{\epsi=0}$ fixed to be along the submanifold $\Sigma$ but free to probe $M$ in all their infinitesimally thickened directions! Of course, this is precisely what one would intuitively expect from a notion of an infinitesimal neighborhood of a subspace $\Sigma$ in $M$, with the left vertical map giving the expected embedding
$$
\DD_{\Sigma,M} \longhookrightarrow M \, .
$$
Naturally, working explicitly with this thickened smooth set would be greatly facilitated by showing it is representable (from the outside) by the (formal dual) of some algebra $\DD_{\Sigma,M}(\infty) \in \mathrm{CAlg}^{\op}$. This is much like the case of the point, that is, an algebra canonically associated to the pair $(\Sigma, M)$ with an appropriate nilpotency behavior ``along normal directions of $\Sigma$ in $M$''.

\begin{example}[\bf Infinitesimal neighborhood of line in plane]
As motivation for the abstract definition of these algebras, let us consider the toy example of a fiber bundle $F\rightarrow \FR^2$ over the plane and the canonical embedding along the $x$-axis $\iota : \FR^1_x\hookrightarrow \FR^{2}_{x,y}$. In this situation, we intuitively know what should be the infinitesimal neighborhood of $\FR^1_x$ in $\FR^2$, namely the thickened Cartesian subspace
$$
\DD_{\FR^1_x,\FR^2_{x,y}(k)}\,:=\, \FR^1_x \times \DD^1(k) \; \longhookrightarrow \; \FR^2 
$$
given dually by the projection (cf. Cor. \ref{HadamardOnManifoldCartesianProduct})
$$
C^\infty\big(\FR^2_{x,y}\big) \xrightarrow{\qquad} C^\infty\big(\FR^2_{x,y}\big) \big/ \big(y^{k+1}\big) 
\, \cong \,  C^\infty\big(\FR^1_x\big) \otimes \CO\big(\DD^1(k)\big) \, .
$$
Working precisely along the lines\footnote{The only difference being that the coefficients $\big\{c^j_i, c^j_{i_1 i_2}, \cdots, c^j_{i_1 \cdots i_k}\big\}$ appearing therein are now functions on $\FR^1_x$, rather than constants.} of Lem. \ref{JetsofSections=InfinitesimalJets} then yields a canonical bijection between sections of the jet bundle $J^k F\rightarrow \FR^2$ over $\FR^1_x\hookrightarrow \FR^2$ (satisfying the usual holonomicity along $\FR^1_x$, $\partial_x u_I = u_{I+x}$),and sections of $F\rightarrow M$ over $\FR^1_x\times \DD^1(k)\hookrightarrow \FR^2$. Moreover, working along the lines of Lem. \ref{SyntheticInfinitesimalNeighborhoodOfManifold}, one can show that this represents the \textit{synthetic infinitesimal neighborhood of $\FR^1_x$} from Def. \ref{SyntheticInfinitesimalNeighborhoodOfSubmanifold}, i.e., 
$$
\DD_{\FR^1_x,\FR^2_{x,y}} \, \cong \, y\big( \DD_{\FR^1_x,\FR^2_{x,y}}(\infty) \big) \quad \in \quad \ThickenedSmoothSets \, .
$$
\end{example}

Now notice that the latter algebra $C^\infty(\FR^2_{x,y}) / (y^{k+1}) \cong C^\infty(\FR^1_x) \otimes \CO\big(\DD^1(k)\big)$ is also canonically identified with the algebra of $\textit{jets of functions on $\FR^2_{x,y}$ along $\FR^1_x$ }$
\begin{align*}
J^k_{\FR^1_{x}}(\FR^2):=  \Big\{j^k_{\FR^1_x} f=[f] \, \big{|}\, f \sim f' \in C^\infty(M) \iff \partial_{|I|} f |_{\FR^1_x} 
 = \partial_{|I|} f' |_{\FR^1_x} \hspace{0.3cm} \forall k\geq I \geq0 \Big\} \, ,
\end{align*}
where the derivatives are only along the normal $\FR^1_y$-axis in $\FR^2_{x,y}$. This formulation generalizes directly to yield the appropriate notion for an infinitesimal neighborhood along arbitrary \textit{embeddings} of submanifolds.
\begin{definition}[\bf Infinitesimal neighborhoods of submanifold]
\label{InfitesimalNbdofSubManifoldTraditional}
Let $\Sigma\hookrightarrow M$ be an embedded submanifold, closed in $M$. For any $k\in \NN \cup \{\infty\}$, 

\begin{itemize} 
\item[{\bf (i)}] the \textit{$k$-order 
infinitesimal neighborhood around $\Sigma$ in $M$}, denoted by $\DD_{\Sigma,M}(k)$, is the formal dual of the jet algebra of $M$ along $\Sigma$
\begin{align*}
J^k_\Sigma(M)\, :&= \,  \Big\{j^k_\Sigma f=[f] \, \big{|}\, f \sim f' \in C^\infty(M) \iff \partial_{|I|} f |_{\Sigma} 
 = \partial_{|I|} f' |_{\Sigma} \hspace{0.3cm} \forall k\geq I|_\Sigma \geq0 \Big\} \quad  \in \quad \mathrm{CAlg}_\FR \, ,
\end{align*}
with the partial derivatives ranging over \textit{normal} derivatives w.r.t. (any) adapted slice chart $\FR^\sigma\hookrightarrow \FR^{\sigma}\times \FR^{d-\sigma}$ for $\Sigma\hookrightarrow M$;

\item[{\bf (ii)}] the \textit{$\infty$-order neighborhood around $\Sigma$ in M} is given by the formal dual of the limit of $k$-order jet algebras in $\mathrm{CAlg}_\FR$, or dually as the colimit of $k$-order infinitesimal neighborhoods in $\mathrm{CAlg}_\FR^{op}$
$$
\DD_{\Sigma,M}(\infty) \, := \, \mathrm{colim}_{\mathrm{CAlg}^{\op}} \DD_{\Sigma,M}(k) \;\; \in \;\; \mathrm{CAlg}^{\op} \, .
$$
\end{itemize} 
\end{definition}

With this definition at hand, we have the following.

\begin{proposition}[\bf Sections over infinitesimal neighborhoods of submanifolds]\label{SectionsOverInfinitesimalNeighborhoodOfSumbanifolds}
Let $F\rightarrow M$ be a fiber bundle over an $d$-dimensional manifold $M\in \mathrm{Man}$ and $\Sigma\hookrightarrow M$ a embedded submanifold with topologically closed image. Then: 

\begin{itemize} 
\item[{\bf (i)}]
 There is a canonical isomorphism 
$$
\DD_{\Sigma,M} \cong y\big(\DD_{\Sigma,M}(\infty)\big)\, 
$$
between the synthetic and the traditional $\infty$-order infinitesimal neighborhoods around $\Sigma\hookrightarrow M$.

\item[{\bf (ii)}]
Tangentially holonomic sections of the jet bundle $J^\infty_M F$ over $\Sigma\hookrightarrow M$ are in 1-1 correspondence with sections of $F\rightarrow M$ over the $\infty$-order infinitesimal neighborhood of $\Sigma$ in $M$, 

\[ 
\left(\!\!\!\!
\raisebox{18pt}{
\xymatrix@R=1.4em@C=3em  { &&  J^\infty_M F \ar[d]^{}
	\\ 
	\Sigma  \ar[urr]^-{j^\infty_\Sigma \phi} 
    \ar@{^{(}->}[rr]^-{} && M
}   
}
\!\right)_{\!\!\mathrm{Hol.}_{\Sigma}}
\quad 
\xleftrightarrow{\quad 1\mbox{-} 1\quad} 
\quad 
\left(\!\!\!\!
\raisebox{18pt}{
\xymatrix@R=1.4em@C=3em  { &&  F\ar[d]^{}
	\\ 
	\DD_{\Sigma,M} \ar[urr]^-{}  
    \ar@{^{(}->}[rr]^-{} && M
}   
}
\!\!\right) .
\]
\end{itemize} 
\end{proposition}
\begin{proof}
The proof of $\bf(i)$ follows as that of Lem. \ref{SyntheticInfinitesimalNeighborhoodOfManifold}. The proof of $\bf(ii)$ follows as that of Lem. \ref{JetsofSections=InfinitesimalJets} with minimal modifications, by working over (adapted) covers of slice charts for $\Sigma \hookrightarrow M$ and trivialization of $F\rightarrow M$ and essentially carrying along the new dependence of the functions therein along $\Sigma$.

In more detail, let $(z,y)$ be adapted coordinates on $M$, with
$\Sigma=\{y=0\}$ and $y\in\FR^c$, and let $u^a$ be fiber coordinates 
on a trivialization of $F$ containing the (part of the) image of the underlying section over $\Sigma$.
Writing $u^a_{\beta,\alpha}$ for jet coordinates which compute the derivatives
$$
u^{a}_{\beta,\alpha}(s ) \, = \, \partial_z^\beta\partial_y^\alpha (s) \, ,
$$
a tangentially holonomic section $s:\Sigma\rightarrow J^\infty_MF$
satisfies
\[
\partial_{z^i}u^a_{\beta,\alpha}(s(z))
 =u^a_{\beta+e_i,\alpha}(s(z)).
\]
Hence, setting $f^a_\alpha(z):=u^a_{0,\alpha}(s(z))$, its entire
jet is determined as
\[
u^a_{\beta,\alpha}(s(z))=\partial_z^\beta f^a_\alpha(z).
\]
To these coefficients we may associate the formal expression
$$
\widehat\phi^a(z,y)=\sum_\alpha f^a_\alpha(z)y^\alpha/\alpha! \, .
$$

Now let $q$ be an $\FR^k\times\DD$-plot of
$\DD_{\Sigma,M}$ in this chart, where
$\CO(\DD)=\FR\oplus V$ and $V^{N+1}=0$.
Writing its coordinates as $z(q)=z_0+\zeta$, $y(q)=\eta$, with
$\zeta$ and $\eta$ valued in $C^\infty(\FR^k)\otimes V$, we may
define ``the fiber coordinates'' of $\Phi(q)$ by
\[
\Phi^a(q)=
\sum_{|\alpha|+|\beta|\leq N}
\frac{\partial_z^\beta f^a_\alpha(z_0)}{\alpha!\,\beta!}
\zeta^\beta\eta^\alpha,
\]
with corresponding ``base coordinates'' to be those of $q$.
The sum is finite and independent of the chosen upper bound $N$, and
finite Taylor substitution commutes with pullback along probe maps, so 
that this formula defines a map of thickened smooth sets over $M$, i.e. a section of $F$ over $\DD_{\Sigma, M}$.
Its body is the underlying section of $s$, so it takes values in the
chosen fiber chart after the local restriction already made.

Conversely, restrict a section $\Phi:\DD_{\Sigma,M}\rightarrow F$
to the normal thickenings
$U_\Sigma\times\DD^c(N)$, where $U_\Sigma$ is a coordinate open
subset of $\Sigma$. The finite-order Taylor description gives unique
smooth coefficients $f^a_\alpha$ such that
\[
\Phi_N^a(z,y)=\sum_{|\alpha|\leq N}
\frac{f^a_\alpha(z)}{\alpha!}y^\alpha
\quad\bmod (y)^{N+1}.
\]
Setting $u^a_{\beta,\alpha}:=\partial_z^\beta f^a_\alpha$,
these define a section of the inverse limit of the finite jet bundles
and satisfy the tangential holonomicity equations. Naturality, applied
to plots of the normal thickenings, recovers the preceding evaluation
formula for every $q$, including its tangential nilpotent components.
The two constructions are therefore inverse.

Finally, under adapted base-coordinate changes and fiber-coordinate
changes, the jet coordinates and the displayed evaluation formula
transform by the same finite Taylor chain rule on each probe.
Consequently, the local constructions agree on overlaps and glue
uniquely by the sheaf property. Since we used only finite sums on each probe, no actual realization of a formal series by an ordinary smooth
section is required.
\end{proof}

This result finally puts on firm ground the intuitive notion that initial data of field equations (PDEs) are literally certain\footnote{That is, those sections that factor through the prolongated shell of the corresponding differential operator $\CP:= P\circ j^\infty$, $S^\infty_P\hookrightarrow J^\infty_M F$, now formed as a thickened subspace within $\ThickenedSmoothSets$ (cf. \cite[Defs. 5.16, 6.29]{GS23}).} fields defined over an infinitesimal neighborhood of a codimension 1 submanifold (see, e.g., the discussion around \cite[Def. 6.28]{GS23}).

\newpage 
\section{Local 
Lagrangian field theory}
\label{Sec-FieldtheorySynth}
We have defined the very basic building blocks for local field theory as thickened smooth sets, namely, the field spaces $\CF:=\mathbold{\Gamma}_M (F)$ and the infinite jet bundle $J^\infty_M F$ etc., and recovered their tangent bundles via the synthetic tangent bundle construction $T= \big[\DD^1(1), -\big]$. All the constructions performed on $J^\infty_M F$ (cf. \cite{GS23}) that are relevant to field theory (cf. Rem. \ref{CaveatOnJetBundleForms}) can now be interpreted as taking place in $\ThickenedSmoothSets$, and we briefly review how this works in \cref{Sec-jetSyth}. 

\smallskip 
More importantly, however, working in our thickened setting allows for a fully solid consideration of the differential geometry on the actual field space, and its product with spacetime $\CF \times M$. That is achieved by showing how all relevant tangent bundles and the pushforward maps of the prolongation $j^\infty$ and prolongated evaluation $ev^\infty$, which have previously only been defined -- at best -- in an ad hoc manner (cf. \cite{GS23}), are now recovered as applications of the same \textit{synthetic tangent functor} of thickened smooth sets.

\smallskip 
Beyond putting these aspects of the theory on a fully solid ground, this bears the further fruit of recognizing the variational principle of local Lagrangian field theory, equivalently, as an intersection
of thickened smooth sets (Thm. \ref{FunctorialityOfTheThickenedCriticalSet}, Rem. \ref{CriticalityOnGeneralSpacetimes}). This means that the on-shell field space is naturally endowed with the correct infinitesimal structure, whose \textit{synthetic} tangent bundle recovers the intuitive, but somewhat ad hoc notion of ``Jacobi fields'' (Cor. \ref{OnshellSyntheticTangentBundle}). Lastly, it offers a prospect for the mathematical realization of perturbative field theory as \textit{literally} the restriction to a (synthetic) infinitesimal neighborhood (Def. \ref{SyntheticInfinitesimalNeighborhood}) around a field configuration $\DD_\phi \hookrightarrow \CF$.

\subsection{Geometry of the infinite jet bundle}
\label{Sec-jetSyth}

\noindent {\bf Horizontal splitting.}  The discussion spelled out in 
\cite[\S 4.2]{GS23} mostly depends on universal properties of limits, and applies essentially verbatim in $\ThickenedSmoothSets$, as smooth manifolds embed fully faithfully (Ex. \ref{ManifoldsAsThickenedSmoothSets}). Namely, there is a short exact sequence of thickened smooth sets
\begin{align}\label{TangentInfinityJetExactSequence} 0_{J^\infty_M F}\longrightarrow V J^\infty_M F\xhookrightarrow{\quad \quad} T J^\infty_M F 
\xrightarrow{\quad \quad} J^\infty_M F \times_{M} TM \longrightarrow 0_{J^\infty_M F}\, ,
\end{align}
where $V J^\infty_M F$
is the vertical subbundle  \cite[Def. 4.6]{GS23}, defined as the equalizer of $\dd \pi^\infty_M : TJ^\infty_M F \rightarrow TM$ (see \eqref{TangentJetBundleProjections}) and $0_M \circ \pi_M \circ  \dd \pi^\infty_M: TJ^\infty_M F\rightarrow TM$, or equivalently as the pullback
\begin{align}\label{VerticalJetbundleAsPullback} 
\xymatrix@=1.6em 
{ VJ^\infty_M F \ar[d] \ar[rr] &&  
TJ^\infty_M F\ar[d]^{\dd \pi^\infty_M } 
	\\ 
M \ar[rr]^{0_M} && TM
	\, . } 
\end{align}
Yet, equivalently, there is the following useful characterization:

\begin{lemma}[\bf Vertical jet bundle characterization]\label{VerticalJetBundleAsJetVerticalBundle}
The vertical subbundle of $TJ^\infty_M F$ is canonically identified with the infinite jet bundle of $VF\rightarrow F\rightarrow M$, i.e., 
$$
V(J^\infty_M F) \;  \cong  \; J^\infty_M (VF)\, ,
$$
as bundles of thickened smooth sets over $F$, and hence over $M$.
\end{lemma}
\begin{proof}
The pullback characterization \eqref{VerticalJetbundleAsPullback} along with the limit property from Lem. \ref{SyntheticInfiniteJetTangentBundle} implies that 
\begin{align*}
VJ^\infty_M F \cong \lim_{\ThickenedSmoothSets} VJ^n_M F\, ,
\end{align*}
where each $VJ^n_M F\hookrightarrow TJ^n_M F$ is traditional (finite-dimensional) vertical tangent bundle of the $n^{\mathrm{th}}$-order jet bundle. At the finite-dimensional manifold level, it is a standard fact that
$$
J^n_M (VF) \cong V (J^n_M F)
$$
as bundles of manifolds over $F$, canonically, as can be immediately seen in local coordinates (see e.g.
\cite[(6.8)]{Sardanashvily02})
$$
\big(x^\mu,\, u^a_{I},\,  \dot{(u^a_{I})} \big) \longmapsto \big(x^\mu, \,(u^a)_{I},\,  (\dot{u}^a)_I\big)\, . 
$$
Carrying this isomorphism through the fully faithful embedding into $\ThickenedSmoothSets$ (Ex. \ref{ManifoldsAsThickenedSmoothSets}), it follows that 
\begin{align*}
VJ^\infty_M F \;\cong \; \lim_{\ThickenedSmoothSets} VJ^n_M F  \cong  \lim_{\ThickenedSmoothSets} J^n_M VF =: J^\infty_M VF \, .
\end{align*}

\vspace{-4mm} 
\end{proof}

At this point, we recall the canonical smooth pushforward maps of manifolds (\cite[(76)]{GS23})
\begin{align*}
H^n \;:\; J^{n+1}_M F \times_M TM & \; \xrightarrow{\quad \quad}  \;T(J^n_M F) 
\\
\big(\,j^{n+1}_p \phi,\,  X_p\big) &\; \xmapsto{\quad \quad} \; 
\dd \big(j^n \phi \big)_p (X_p) \nn
\end{align*}
\noindent for each $n\in \NN$, where $\phi:U\subset  M\rightarrow F $ on the right-hand side is any representative local section of
$j^{n+1}_p \phi$. These fit into the commutative diagram of smooth manifolds 
\begin{align}\label{CartanDistribCompatibility}
\begin{gathered}
\xymatrix@R=.8em@C=2.6em  {J^{n+1}_M F\times_M TM \ar[dd]_{\pi^{n+1}_n\times \id} \ar[rr]^-{H^n} &&  
T\big(J^n_M F\big) 
\ar[dd]^{\dd \pi^{n}_{n-1} } 
	\\ \\
J^n_M F\times_M TM \ar[rr]^-{H^{n-1}}  && T\big(J^{n-1}_M F\big) 
\, , }
\end{gathered}
\end{align}
which we interpret as a commutative diagram of thickened smooth sets via Ex. \ref{ManifoldsAsThickenedSmoothSets}. This is already sufficient to provide the splitting from $\cite[Prop. 4.8]{GS23}$, however now within $\ThickenedSmoothSets$.

\begin{proposition}[\bf Splitting of infinite jet tangent bundle]
\label{SmoothSplittingProp}
The family of smooth bundle maps $\big\{H^n:J^{n+1}_M F \times_M TM\longrightarrow  T(J^n_M F)\big\}_{n\in \NN}$ determines a map of thickened smooth sets   
\begin{align}\label{SplittingMap}
H\;:\; \,J^\infty_M F\times_{M} TM \; \xrightarrow{\quad \quad} \; T(J^\infty_M F)\, ,
\end{align}
which splits the corresponding exact sequence \eqref{TangentInfinityJetExactSequence}, and hence induces a canonical isomorphism of vector-bundle-thickened smooth sets
\begin{align*}
TJ^\infty_M F\;  \cong \;  VJ^\infty_M F \, \oplus \, H J^\infty_M F\, ,
\end{align*}
\noindent
over $J^\infty_M F$, where the plots of $H J^\infty_M F$ are defined by the (plot-wise) image of the map \eqref{SplittingMap}.
\begin{proof}
By the limit property 
of $T(J^\infty_M F)=\mathrm{lim}_n^{\ThickenedSmoothSets} T(J^n_M F)$ from Lem. \ref{SyntheticInfiniteJetTangentBundle} and commuting limits through the hom functor, we get
\vspace{-1mm} 
\begin{align*}
\mathrm{Hom}_{\ThickenedSmoothSets}\big(J^\infty_M F\times_{M} TM \, , \, T(J^\infty_M F)  \big) 
\; \cong \; 
\mathrm{lim}_n^{\mathrm{Set}}\,  \mathrm{Hom}_{\ThickenedSmoothSets}\big(J^\infty_M F\times_M TM\, , \,T(J^n_M F)  \big)\, . 
\end{align*}
Hence a map $f:J^\infty_M F\times_M TM \rightarrow T(J^\infty_M F) $ corresponds to a family of maps of thickened smooth sets
$\big\{f^n: J^\infty_M F\times_M TM \rightarrow T(J^n_M F)\big\}_{n\in \NN}$ such that $\dd \pi^n_{n-1} \circ f^n = f^{n-1}$, and vice-versa. 
Pulling back the maps from \eqref{CartanDistribCompatibility} along the canonical projections $\pi_{n+1}:J^\infty_M F\rightarrow J^{n+1}_M F$, we get a family
\vspace{0mm} 
$$
(\pi_{n+1}\times \id)^* H^n \;\;: \;\;
J^\infty_M F\times_M TM \; \longrightarrow \; J^{n+1}_M F\times_M TM \longrightarrow T(J^n_M F)
$$ 
which, by the  
commutativity of diagram \eqref{CartanDistribCompatibility}, satisfies 
\vspace{1mm} 
\begin{align*}
\dd\pi^n_{n-1}\circ \big(\pi_{n+1}\times \id)^*H^n&= \dd \pi^{n}_{n-1}\circ H^n \circ (\pi_{n+1}\circ \id) 
\\[-1pt] 
&= H^{n-1}\circ (\pi^{n+1}_n \times \id ) \circ (\pi_{n+1} \times \id) 
\\[-1pt]
&= H^{n-1}\circ (\pi_{n}\times \id)
\\[-1pt]
&= (\pi_n \times \id)^* H^{n-1} \, .
\end{align*}

\vspace{0mm} 
\noindent Thus the family $\big\{(\pi_{n+1}\times \id)^*H^n:J^\infty_M F\times_M TM \rightarrow T(J^n_M F)\big\}_{n\in \NN}$ uniquely corresponds to a 
map $H:J^\infty_M F \times_M TM \rightarrow T(J^\infty_M F)\, .$ Following the formal sum representation from \eqref{InfinityJetTangentPlotsFormalSum}, this map may be represented in local coordinates (plot-wise)
\begin{align}\label{InftyJetBundleHorizontalTangentCoordinate}
\bigg(s^k_\epsi \, ,\,  X^\mu \frac{\partial}{\partial x^\mu}\Big\vert_{\pi^\infty_M(s^k_\epsi)} \bigg) \; \longmapsto \;
X^\mu \bigg(\frac{\partial}{\partial x^\mu} \Big\vert_{s^k_\epsi} + \sum_{|I|=0}^{\infty} s^a_{I+\mu}
\cdot    \frac{\partial}{\partial u^a_I} \Big\vert_{s^k_\epsi}\bigg)\, ,
\end{align}
by which it easily follows that it splits the sequence \eqref{TangentInfinityJetExactSequence}  (cf. \cite[(78),(79)]{GS23}).
\end{proof}
\end{proposition}

Finally, let us also fully characterize the synthetic pushforward \eqref{SyntheticPushforward} 
$$ \dd( j^\infty \phi) \;:\; TM \xrightarrow{\quad \quad} T J^\infty_M F 
$$ 
along the infinite jet prolongation  $j^\infty \phi \,: \, M \rightarrow J^\infty_M F$ of any field configuration $\phi\in \CF(*)$. 

\begin{corollary}[\bf Synthetic  of jet prolongated section]\label{SyntheticPushfowardOfJetProlongatedSection}
$\,$

\begin{itemize} 
\item[{\bf (i)}]
The synthetic pushforward of the infinite jet prolongation $j^\infty \phi : M \rightarrow J^\infty_M F$ of a field configuration $\phi \in \Gamma_M (F)$ is equivalently represented by the family of (traditional) pushforwards 
$$
\big\{ \dd j^n \phi : TM \xrightarrow{\quad \quad} TJ^n_M F\big\}_{n\in \NN}
$$
of smooth manifolds.

\item[{\bf (ii)}]
Furthermore, it is equivalently given by the composition
$$
H\circ (j^\infty \phi , \, \id_{TM}) \circ (\pi_M, \, \id_{TM}) 
\;\; : \;\; 
TM \longrightarrow J^\infty_M F \times_M TM \xrightarrow{\quad \quad}  HJ^\infty_M F \longhookrightarrow TJ^\infty_M F\, ,
$$
and hence in particular factors through the horizontal subbundle of $TJ^\infty_M F$. 

\item[{\bf (iii)}]
It follows that this is represented in local coordinates as 
$$ X^\mu \frac{\partial}{\partial x^\mu}\Big\vert_{p}  \;\; \longmapsto \;\;
X^\mu \bigg(\frac{\partial}{\partial x^\mu} \Big\vert_{j^\infty_p \phi} + \sum_{|I|=0}^{\infty} \frac{\partial \phi^a} {\partial x^{I+\mu}} (p)
\cdot    \frac{\partial}{\partial u^a_I} \Big\vert_{j^\infty_p \phi}\bigg)\, ,
$$
and similarly on arbitrary plots (cf. \eqref{InftyJetBundleHorizontalTangentCoordinate}).
\end{itemize} 
\end{corollary}
\begin{proof}
The first claim follows immediately by limit characterization of $TJ^\infty_M F$ from Lem. \ref{SyntheticInfiniteJetTangentBundle}, and the fact that the synthetic pushforward along maps of
finite-dimensional manifolds recovers the traditional  between their tangent bundles \eqref{SyntheticPushforwardRecoversTraditional}. The second claim follows immediately by the definition of the splitting map \eqref{SplittingMap} (cf. \eqref{CartanDistribCompatibility}), which also yields the local coordinate representation by \eqref{InftyJetBundleHorizontalTangentCoordinate}.
\end{proof}

\medskip
\noindent {\bf Forms and vector fields on the infinite jet bundle.}
 Differential forms on $J^\infty_M F \in \ThickenedSmoothSets$ are defined analogously as maps of thickened smooth sets out of its (synthetic) tangent bundle. The subalgebra of global finite order forms (cf. Rem. \ref{CaveatOnJetBundleForms}) is identified precisely with the corresponding notion appearing within smooth sets (\cite[\S 4.3]{GS23}), and hence the corresponding Cartan calculus and cohomology results follow verbatim. For a more expansive discussion of the standard concepts recalled below, and their relation to the traditional literature, the interested reader may consult \cite{GS23}.

\begin{definition}[\bf Forms on the infinite jet bundle]\label{JetBundleDifferentialForms}
The set of differential m-forms on the infinite jet bundle is defined as
\begin{align*}
\Omega^m(J^\infty_M F) \; := \;  \mathrm{Hom}^{\mathrm{fib.lin. an.}}_{\ThickenedSmoothSets}
\big(T^{\times m}(J^\infty_M F)\,,\, \FR \big) \, , 
\end{align*} 
i.e., real-valued, fiberwise $\FR$-linear antisymmetric maps of thickened smooth sets.
For $p+q=m$, a differential $m$-form $\omega\in \Omega^{m}(J^\infty_M F)$ on the infinite jet bundle is called:
\begin{itemize} 
\item[{\bf (i)}] $p$-\textit{horizontal} if the following restriction is the zero map
\vspace{0mm} 
$$
\om \;\;: \;\;  V^{\times (q+1)} (J^\infty_M F) \times_{J^\infty_M F} T^{\times (p-1)}(J^\infty_M F) \; \longhookrightarrow \; T^{\times m}(J^\infty_M F) \; \longrightarrow \; \FR \;;
$$

\item[{\bf (ii)}] $q$-\textit{vertical} if the following restriction is the zero map
\vspace{0mm} 
$$
\om\;\; : \;\; H^{\times (p+1)}(J^\infty_M F) \times_{J^\infty_M F} T^{\times (q-1)}(J^\infty_M F)\; \longhookrightarrow \; T^{\times m}(J^\infty_M F)
\; \longrightarrow \; \FR \;;
$$ 

\item[{\bf (iii)}] a $(p,q)$-form if it is both $p$-horizontal \textit{and} $q$-vertical.
\end{itemize} 
\end{definition}


\begin{remark}[\bf Point-set characterization of forms on $J^\infty_M F$]\label{CaveatOnJetBundleForms} 
Recall that smooth functions $J^\infty_M F \rightarrow \FR$, defined within pure smooth sets $\SmoothSets$, are completely characterized via their point-set description as those of locally finite order, due to the fully faithful embedding $\FrechetManifolds\hookrightarrow \SmoothSets$ (\cite[Prop. 3.7, Rem. 3.8]{GS23}) and the natural Fr\'{e}chet (locally pro-) manifold structure on $J^\infty_M F$. In particular, these include those of globally finite order, which are the only ones appearing in examples from field theory. 
\begin{itemize} 
\item[\bf (i)]On abstract grounds, there exists a fully faithful functor\footnote{This is given by the left Kan extension along $\iota : \CartesianSpaces \hookrightarrow \ThickenedCartesianSpaces$ from \eqref{CartToFCart}, composed with sheafication $L: \mathrm{PreSh}(\ThickenedCartesianSpaces) \rightarrow \mathrm{Sh}(\ThickenedCartesianSpaces)$.} $\iota_! : \SmoothSets \hookrightarrow \ThickenedSmoothSets $, which hence fully faithfully embeds also Fr\'{e}chet manifolds into thickened smooth manifolds and so produces an -- a priori different -- version of the thickened infinite jet bundle $\iota_! \big( \lim_{\SmoothSets} J^k_M F \big)$.

The fact that this version of the jet bundle is the same as that from Def. \ref{InfiniteJetBundleFormalSmoothLimit}, i.e., that passage from the category of smooth Fr{\'e}chet manifolds to the \emph{Cahiers topos} of \emph{thickened smooth sets} 
 preserves the projective limits that define infinite jet bundles, is proved in \cite{GKSS25}. This implies that smooth functions on the jet bundle defined in $\ThickenedSmoothSets$, and similarly differential forms on it, are again fully characterized as being exhausted by the point-set \textit{locally} finite order functions and differential forms on $J^\infty_M F$.
 
\item[\bf (ii)] Nevertheless, both for our purposes here and for the sake of exposition, the weaker statement that $\mathrm{Hom}_{\ThickenedSmoothSets}(J^\infty_M F, \, \FR) $ (injectively) includes the subalgebra of \textit{globally} finite order functions
$$
C^\infty_{\mathrm{glb}}(J^\infty_M F, \, \FR)  \xhookrightarrow{\quad \quad}  \mathrm{Hom}_{\ThickenedSmoothSets}(J^\infty_M F, \, \FR)\, ,
$$
is sufficient. But this is obviously true, precisely as in $\SmoothSets$, by pulling back finite-order functions $J^n_M F \rightarrow \FR$ along the canonical projections $\pi^\infty_n: J^\infty_M F\rightarrow J^n_M F$, now interpreted as maps of thickened smooth sets (cf. Def. \ref{InfiniteJetBundleFormalSmoothLimit}, Ex. \ref{ManifoldsAsThickenedSmoothSets}). Analogously, the same applies for differential forms defined within thickened smooth sets, but focusing on the globally finite-order forms 
$$
\Omega^m_{\mathrm{glb}}(J^\infty_M F)\xhookrightarrow{\quad \quad} \Omega^m (J^\infty_M F) \, . 
$$
\item[\bf (iii)] From now onwards, in order to ease the descriptive adjectives, we will often abusively speak of `(smooth) $m$-forms and functions' on the thickened infinite jet bundle to mean only those of globally finite order, bearing in mind that in reality these form the larger set of \textit{locally} finite order forms.
\end{itemize} 
\end{remark}

Henceforth, we shall focus on globally finite order forms on $J^\infty_M F$. Any such $m$-form 
$$
\om\, := (\pi^\infty_n)^*\om_n \, \; : \; T^{\times m} (J^\infty_M F) \xlongrightarrow{\, \dd \pi^\infty_n\, } T^{\times m} (J^n_M F) \xrightarrow{\, \, \om_n\, \,} \FR
$$
is represented in local coordinates (of $J^n_M F$) as
\begin{align}\label{mformInfinityJetLocalCoordinates}
\om = \sum_{p+q=m} \, \sum_{I_1,\cdots, I_q=0}^{n} \om_{\mu_1\cdots \mu_p a_1 \cdots a_q}^{I_1\dots I_q} 
\dd x^{\mu_1}\wedge \cdots \wedge \dd x^{\mu_p}\wedge \dd u^{a_1}_{I_1}\wedge\cdots \wedge \dd u^{a_q}_{I_q}\, ,
\end{align}
due to the fully faithful embedding from Ex. \ref{ManifoldsAsThickenedSmoothSets}. Already by this local representation, the wedge product structure 
$$
\wedge\;\; :\;\;  \Omega^m(J^\infty_M F) \; \times \; \Omega^{m'} (J^\infty_M F) \xrightarrow{\quad \quad} \Omega^{m+m'} (J^\infty_M F )
$$
becomes immediate, and similarly for the de Rham differential 
$$
\dd \; :\; \Omega^m(J^\infty_M F) \xrightarrow{\quad \quad} \Omega^{m+1}(J^\infty_M F)
$$ 
acting locally as 
$$
\dd \om \, = \, \bigg(\dd x^\mu \wedge  \frac{\partial}{\partial x^\mu}+
\sum_{|I|=0}^{\infty} \dd u^a_I \wedge  \frac{\partial}{\partial u^a_I}\bigg) \om  \, . 
$$
\begin{remark}[\bf Forms on $J^\infty_M F$ via the moduli space]\label{FormsOnJetBundleViaModuliSpace}
Note that the proof of \cite[Lem. 4.15]{GS23} holds identically in the thickened smooth setting, exactly by the same formal reasons, thus yielding a canonical DGCA injection
$$
\Omega^\bullet_{\mathrm{glb}}(J^\infty_M F) \longhookrightarrow \Omega^\bullet_{\mathrm{dR}}(J^\infty_M F) 
$$
into the de Rham forms \eqref{deRhamFormsOnThickenedSmoothSet} of the infinite jet bundle.
\end{remark}

\begin{definition}[\bf Vector fields on infinite jet bundle]\label{VectorFieldsOnJetBundle}
The set of (smooth, thickened) vector fields on $J^\infty_M F$ is defined as
$$\CX(J^\infty_M F) := \Gamma_{J^\infty_M F} (TJ^\infty_M F)\, ,$$
that is, as sections of its tangent bundle
\[ 
\xymatrix@R=1.7em@C=3em{ &&  TJ^\infty_M F  \ar[d]
	\\ 
	J^\infty_M F \ar[rru]^{X} \ar[rr]^>>>>>>>{\id_{J^\infty_M F}} && J^\infty_M F \, 
}   
\]
within $\ThickenedSmoothSets$. 
\end{definition}

Let us expand on the definition. By the limit property of $TJ^\infty_M F$ from Lem. \ref{SyntheticInfiniteJetTangentBundle}, a vector field $X\in \CX(J^\infty_M F)$ is equivalently given by a 
compatible family of maps $\{X: J^\infty_M F\rightarrow TJ^n_M F\}_{k\in \NN}$, such that the diagrams
	\[ 
\xymatrix@R=1.4em@C=3em{ &&  TJ^n_M F \ar[d]^{\dd \pi^n_{n-1}}
	\\ 
	J^\infty_M F \ar[rru]^{X^n} \ar[rr]^{\qquad X^{n-1}} && TJ^{n-1}_M F 
    \mathrlap{\;,}
}    
\qquad 
\xymatrix@R=1.4em@C=3em{ &&  TJ^n_M F \ar[d]
	\\ 
	J^\infty_M F \ar[rru]^{X^n} \ar[rr]^{\qquad  \pi^{\infty}_{n}} && J^n_M F
}    
\]
\noindent commute for all $n\in \NN$.

\smallskip 
Note Rem. \ref{CaveatOnJetBundleForms} applies verbatim for each member $X^n$ of such a family. As with differential forms, even if these are exhausted by those of locally finite order, we focus on the physically relevant case of globally finite-order vector fields
$$
\CX_{\mathrm{glb}}(J^\infty_M F) \subset \CX(J^\infty_M F)\, ,
$$ where each of the vector fields in the family is actually of finite order, i.e., a pullback 
$$ 
X^n \, \equiv  \, X^n_{j_n} \circ \pi^\infty_{j_n} \;\; : \;\; J^\infty_M F \longrightarrow J^{j_n}_M F \xrightarrow{\quad \quad} TJ^n_M F \, . 
$$
In this case, by the fully faithful embedding of Ex. \ref{ManifoldsAsThickenedSmoothSets}, such a family 
$\big\{X^n: J^\infty_M F\rightarrow T(J^n_M F)\big\}_{n\in \NN} $ 
may be represented in local coordinates by an infinite (formal) sum
\begin{align}\label{InftyJetBundleVectorFieldCoordinates}
X= X^\mu \frac{\partial}{\partial x^\mu} + \sum_{|I|=0}^\infty Y^a_I \frac{\partial}{\partial u_I^a} \, , 
\end{align} 
\noindent for an infinite list of finite-order smooth functions $\big\{X^\mu,\{Y_I^a\}_{0\leq |I|}\big\}\subset C^\infty (J^\infty_M F)$, 
with every $X^n$ corresponding to the case where the sum is terminated at order $|I|=k$ (cf.  \cite[(67)]{GS23}). It follows directly (e.g., \cite[Lem. 4.5]{GS23}) that vector fields are in bijection with derivations of functions on $J^\infty_M F$
$$
\CX(J^\infty_M F)\; \cong \;  \mathrm{Der}\big(C^\infty(J^\infty_M F)\big)\, .
$$ 
Similarly, it follows that (cf.  \cite[Rem 4.12]{GS23})
$$
\Omega^1(J^\infty_M F)
\;\; \cong \;\; \mathrm{Hom}_{C^\infty(J^\infty_M F)\mbox{-}\mathrm{Mod}}\Big(\CX(J^\infty_M F) \, , \, C^\infty(J^\infty_M F) \Big) \, 
$$ 
as modules over $C^\infty(J^\infty_M F)$, and hence further that
\begin{align}
\label{FormsInfinityJetBundleAsVectorFieldMaps}
\Omega^m(J^\infty_M F)\; &\cong \;  \mathrm{Hom}_{C^\infty(J^\infty_M F)\mbox{-}\mathrm{Mod}}^{\mathrm{antis.}}\Big(\CX^{\times m}(J^\infty_M F) \, 
, \, C^\infty(J^\infty_M F) \Big) \,  \  
\\
\; &\cong  \;  \bigwedge^{m}_{C^\infty(J^\infty_M F)} \Omega^1(J^\infty_M F) \, . \nn
\end{align}
The isomorphisms here are witnessed via the contraction operator 
\begin{align*}
\iota_X \, : \,\Omega^1(J^\infty_M F) &\xrightarrow{\quad \quad} C^\infty(J^\infty_M F) \\
\om &\xmapsto{\quad \quad} \om \circ X
\end{align*}
whose form in local coordinates follows exactly as in finite-dimensional manifolds using \eqref{mformInfinityJetLocalCoordinates} and \eqref{InftyJetBundleVectorFieldCoordinates}, and which extends as a (graded) derivation to a map
$$
\iota_X\; :\; \Omega^m(J^\infty_M F) \xrightarrow{\quad \quad} \Omega^{m-1} (J^\infty_M F) \, .
$$
The Lie derivative operator is then defined as
$$
\mathbb{L}_X \, := \, [\dd, \iota_X] \; : \; \Omega^{m}(J^\infty_M F) \xrightarrow{\quad \quad} \Omega^{m}(J^\infty_M F) \, ,
$$
whose form in local coordinates is an infinite-dimensional manifold, and which satisfies the corresponding Cartan calculus identities.

\medskip
\noindent {\bf Variational bicomplex.} 
 By the same arguments as in  \cite[\S 5]{GS23}, the bicomplex of globally (or locally) finite-order forms may now be seen as arising completely within thickened smooth sets, along with the all-important corresponding results on vertical and horizontal cohomologies.
In more detail, by the splitting of Prop. \ref{SmoothSplittingProp}, the module of vector fields splits into vertical and horizontal parts
$$
\CX(J^\infty_M F) \; \cong \;
\CX_H(J^\infty_M F)\, \oplus \, \CX_V(J^\infty_M F)\, ,
$$
which in the local coordinate representation of \eqref{InftyJetBundleVectorFieldCoordinates} translates to the splitting of any vector field $X=X_H+ X_V$ with 
\begin{align*}
X_H= X^\mu \bigg( \frac{\partial}{\partial x^\mu} + \sum_{|I|=0}^{\infty} u^{a}_{I+\mu} \frac{\partial}{\partial u^a_I} \bigg) \, , 
\qquad 
X_V=\sum_{|I|=0}^{\infty}\big(Y_I^a - X^\mu \cdot u^a_{I+\mu}\big) \cdot \frac{\partial}{\partial u^a_I}  \,.
\end{align*}

Similarly, by the identification of $m$-forms from \eqref{FormsInfinityJetBundleAsVectorFieldMaps}, it follows that (globally finite-order) 1-forms split into horizontal and vertical parts $\big( (1,0)-$ and $(0,1)-$forms respectively, as per Def. \ref{JetBundleDifferentialForms}$\big)$
\begin{align*}
\Omega^1(J^\infty_M F) \; \cong \;  \Omega^1_H(J^\infty_M F)\oplus \Omega^1_V(J^\infty_M F)
\end{align*}
which in the local coordinate representation from \eqref{mformInfinityJetLocalCoordinates} translates to the splitting of any $1$-form $\om=\om_H+\om_V$ with
\begin{align*}
\om_H= (\om_H)_\mu \cdot \dd x^\mu \, , \hspace{1cm}  \om_V = \sum_{|I|=0}^{n} (\om_V)_a^{I} \cdot 
(\dd u^a_I - u^{a}_{I+\mu} \dd x^\mu)=:\sum_{|I|=0}^{n} (\om_V)_a^{I} \cdot \theta^a_I \, .
\end{align*}
This further implies that the differential forms $\Omega^\bullet_{\glb}(J^\infty_M F)$ carry a bi-graded algebra structure, as they split into the sum of $(p,q)$-forms (Def. \ref{JetBundleDifferentialForms})
\begin{align*}
\Omega^\bullet(J^\infty_M F)\; \cong \;
\Omega^{\bullet,\bullet}(J^\infty_M F)\; := \; 
\bigoplus_{m\in \NN} \bigoplus_{p+q=m} \Omega^{p,q}(J^\infty_M F )\, ,
\end{align*}
where in local coordinates a $(p,q)$-form $\om \in \Omega^{p,q}(J^\infty_M F)$ takes the form
\vspace{-1mm} 
\begin{align}\label{pqformInfinityJetLocalCoordinates}
\om &= \sum_{I_1,\cdots, I_p=0}^{n} \om_{\mu_1\cdots \mu_p a_1 \cdots a_q}^{I_1\dots I_q} \dd x^{\mu_1}\wedge \cdots \wedge \dd x^{\mu_p}\wedge 
\theta^{a_1}_{I_1}\wedge\cdots \wedge  \theta^{a_q}_{I_q}\, . 
\end{align}
By appropriately projecting into the corresponding subspaces (\cite[Def 5.3]{GS23}), the de Rham differential splits
into a horizontal and vertical differential
$$
\dd = \dd_H + \dd_V 
$$
where 
\begin{align*}
\dd_H: \Omega^{p,q}(J^\infty_M F)\xrightarrow{\quad \quad}  \Omega^{p+1,q}(J^\infty_M F), \hspace{1.5cm} \dd_V:
\Omega^{p,q}(J^\infty_M F)\xrightarrow{\quad \quad} \Omega^{p,q+1}(J^\infty_M F) \, ,
\end{align*}
which are completely determined by their action in local coordinates as:
$$
\dd_H(f)= \bigg(\frac{\partial f}{\partial x^\mu} + \sum_{I=0}^{n}     u^a_{I+\mu} \frac{\partial f}{\partial u_{I}^a} \bigg)
\cdot \dd x^\mu=: D_\mu(f) \cdot \dd x^\mu  \, ,\qquad   \dd_{V}(f) = \sum_{|I|=0}^{n} \frac{\partial f}{\partial u^a_I} \cdot 
\big(\dd u^a_I - u^{a}_{I+\mu} \dd x^\mu\big) =:  \sum_{|I|=0}^{n} \frac{\partial f}{\partial u^a_I} \cdot  \theta^a_I \, ,
$$
for any $f\in C^\infty(J^\infty_M F)$, and
\begin{align*}
\dd_H ( \dd x^\mu)= 0, \hspace{1cm} \dd_V (\dd x^\mu)= 0, \hspace{1cm} \dd_H\big( \dd_V u^a_I\big) = 
- \dd_V u^{a}_{I+\mu} \wedge \dd x^\mu ,  \hspace{1cm} \dd_V \big( \dd_V  u^a_I\big) = 0 \, .
\end{align*}
These satisfy $\dd_H^2 = \dd_V^2 = \dd_H \circ \dd_V + \dd_V \circ \dd_H =0$, thus defining the ``\textit{variational bicomplex}'' (\cite[Def. 5.4]{GS23}) 
\begin{align*}
\big(\Omega^{\bullet,\bullet}(J^\infty_M F), \, \dd_H, \, \dd_V \big) \, .
\end{align*}

It follows that Takens' Acyclicity Theorem on the cohomology of the variational bicomplex applies (\cite[Prop. 5.5]{GS23}), and its augmentation (extension) by functional/source forms using the interior Euler operator (\cite[Def. 5.9, Prop. 5.10]{GS23}) apply verbatim, along with the corresponding results on the Euler--Lagrange subcomplex (\cite[Def. 5.11, Prop. 5.18]{GS23}). For the sake of brevity, we shall not review these here.
\subsection{Transgression and the local bicomplex}
\label{Subsec-bicomplexSynth}

 In the same vein, all the definitions and related results regarding local Lagrangians and local currents/functionals (\cite[\S3.2]{GS23}), the corresponding (shell and) on-shell spaces of fields (\cite[\S 5.2]{GS23}), local vector fields and conserved currents (\cite[\S6.1, \S6.2]{GS23}) follow identically in $\ThickenedSmoothSets$, \textit{provided that} the prolongation map 
$$
j^\infty: \CF \xrightarrow{\quad \quad} \Gamma_M(J^\infty F) 
$$ 
is actually a well-defined map of \textit{thickened} smooth sets, such that its synthetic pushforward \eqref{SyntheticPushforward} yields precisely the `ad hoc' pushforward defined in \cite[Eq. (76) and Ex. 4.7]{GS23}. 
More generally, the local bicomplex on $\CF\times M$ and the corresponding presymplectic structure of local field theories 
(\cite[\S7]{GS23}) follows verbatim, \textit{provided that} the prolongated evaluation map 
$$
\mathrm{ev}^\infty \;:\; \CF\times M \xrightarrow{\quad \quad} J^\infty_M F
$$ is a well-defined map of thickened smooth sets, such that its synthetic pushforward yields precisely the one from 
\cite[Def. 7.4]{GS23}.

The fact that the infinite jet prolongation exists as a map of thickened smooth sets follows by internal considerations (cf. Def. \ref{SyntheticInfiniteJetProlongation} of Appendix \cref{App-SyntheticJetProlongation}), i.e., using the viewpoint of jets as sections over infinitesimal neighborhoods from \eqref{SectionsOverInfinitesimalNbdDiagram}. However, for the sake of making this abstract statement more concrete, we now expand on the definition of infinite jet prolongation in the current setting in a different ``external'' manner. This boils down to defining it locally over trivializing charts and by a further careful application of the universal property of limits.

\smallskip 
To that end, recall (Def. \ref{ThickenedSmoothSetOfSections}) that
any $\FR^k\times \DD$-plot $\phi^{k,\epsi}$ of $\mathbold{\Gamma}_M(F)$ is given dually by a `pullback' map of algebras $(\phi^{k,\epsi})^* : C^\infty(F) \rightarrow \CO(\FR^k\times \DD \times M)$. To define such a map, it is sufficient to declare its (compatible) action on (a cover of) adapted coordinate functions $\{x^\mu,u^a\}$
\begin{align*}
\big(\phi_\epsi^k\big)^* \;:\;  C^\infty(F) &\longrightarrow \CO\big(\FR^k\times \DD \times M\big) \\ 
x^\mu &\longmapsto x^\mu  \\[-3pt]
u^a &\longmapsto \phi^a(x^\mu, c^k ,\epsi)
\mathrlap{\,,}
\end{align*}
where $\phi^a(x^\mu,c^k, \epsi) \in C^\infty(\FR^k\times M) \otimes \CO(\DD)$ is a smooth (locally and compatibly defined) function on $\FR^k\times M$, polynomially extended in $\epsi \in \CO(\DD)$. 
This holds by the (topological/petit) sheaf properties of the given function algebras.

From this point of view, it is easy to define the (first) jet prolongation map, acting on (plots of) sections as 
\begin{align*}
j^1 \;:\quad \CF \;\; &\longrightarrow \mathbold{\Gamma}_M (J^1 F)
\\
\phi^{k,\epsi} &\longmapsto j^1 \phi^{k,\epsi} 
\mathrlap{\,,}
\end{align*}
where $j^1 \phi^{k,\epsi}$ is the $\FR^k\times \DD$-plot of $\mathbold{\Gamma}_M (J^1 F)$ defined dually by its action on any (cover of) coordinate charts
\begin{align*}
\big(j^1\phi^{k,\epsi}\big)^* \; :\; 
C^\infty(J^1_M F) &\longrightarrow \CO\big(\FR^k\times \DD \times M\big) \\ 
x^\mu &\longmapsto x^\mu  \\
u^a &\longmapsto \phi^a(x^\mu, c^k ,\epsi)
\\
u^a_\mu &\longmapsto \partial_\mu \phi^a (x^\mu, c^k, \epsi) 
\mathrlap{\,,}
\nn
\end{align*}
where $\partial_\mu$ denotes the derivations with respect to $x^\mu$, i.e. the partial derivatives along $M$ on the manifold $\FR^k\times M$, extended trivially along the nilpotent algebra elements of $\CO(\DD)$. Since the locally defined target functions are compatible, this glues to a morphism of global functions. This assignment on plots is obviously functorial under maps of probes, and hence defines a map of thickened smooth sets. Completely analogously, the $n^{\mathrm{th}}$-jet prolongation defines a map of thickened smooth sets 
\begin{align}\label{nthJetProlongation}
j^n \;:\quad \CF \;\; &\longrightarrow \; \mathbold{\Gamma}_M (J^n F)
\\
\phi^{k,\epsi}& \longmapsto \; j^n \phi^{k,\epsi} \nn 
\end{align}
where $j^n\phi^{k,\epsi}$ is defined dually (and locally) by
\begin{align}\label{nthJetOfSectionPlotLocally}
\big(j^n\phi^k_\epsi\big)^* \;:\; C^\infty(J^n_M F) & \; \longrightarrow \; \CO\big(\FR^k\times \DD \times M\big) \\ 
x^\mu & \; \longmapsto \; x^\mu \nn \\
\big\{u^a_I\big\}_{|I| \leq n} &\;\longmapsto \; \big\{\partial_I \phi^a\big(x^\mu, c^k ,\epsi\big) \big\}_{ |I| \leq n }. \nn 
\end{align}
The compatibility of finite-order jet prolongations implies the existence of the infinite jet prolongation valued in the thickened smooth set of sections of the infinite jet bundle $J^\infty_M F$, which we shall henceforth denote as
$$
\CF^\infty \, := \, \mathbold{\Gamma}_M(J^\infty F) \quad \in \quad \ThickenedSmoothSets \, .
$$
\begin{lemma}[\bf Thickened infinite jet prolongation]\label{ThickenedInfiniteJetProlongation}
The infinite jet prolongation exists as a map of thickened smooth sets of sections
$$
j^\infty \;:\; \CF \xrightarrow{\quad \quad} \CF^\infty \, .
$$
\end{lemma}
\begin{proof} 
By its defining limit property from Def. \ref{InfiniteJetBundleFormalSmoothLimit}, any map $\psi^{k,\epsi,\, \infty} : \FR^k\times \DD \times M \rightarrow J^\infty_M F$ of thickened smooth sets is equivalently given by a compatible family of maps of thickened smooth manifolds 
$$
\big\{ \psi^{k,\epsi,\, n} : \FR^k\times \DD \times M \longrightarrow J^n_M F \big\}_{n\in \NN} \, .
$$
It follows that a plot of sections in  $\mathbold{\Gamma}_M(J^\infty F)$, i.e., a map $\FR^k\times \DD \times M \rightarrow J^\infty_M F$ that satisfies the section condition 
	\[ 
	\xymatrix@R=1em@C=2.6em  { &&  J^\infty_M F \ar[d]
		\\ 
		\FR^k\times \DD \times M \ar[rru]^-{\psi^{k,\epsi,\, \infty}} \ar[rr]^>>>>>>>>{\pr_M} && M \, ,
	}   
	\]
is equivalently a compatible family of maps of thickened smooth manifolds $
\big\{ \psi^{k,\epsi, \, n} : \FR^k\times \DD \times M \longrightarrow J^n_M F \big\} 
$ each of which satisfies the corresponding section condition over $M$. 

Thus, given any $\FR^k\times \DD$-plot $\phi^{k,\epsi}$ of $\CF$ the compatible family of sections
$$\big\{j^n \phi^{k,\epsi} : \FR^k\times \DD \times M \longrightarrow J^n_M F\big\}_{n\in \NN}\, , 
$$
obtained by applying each finite jet prolongation \eqref{nthJetProlongation},  corresponds precisely to a plot of $\mathbold{\Gamma}_M(J^\infty F)$
$$ 
j^\infty \phi^{k,\epsi} \;:\; \FR^k\times \DD \times M \longrightarrow J^\infty_M F
$$
which is  functorial under pullbacks of probes, hence defining the desired map of thickened smooth sets
$
j^\infty : \CF \longrightarrow \CF^\infty 
$. 
\end{proof}
Notice, by \eqref{nthJetOfSectionPlotLocally}, the infinite jet prolongation of an $\FR^k\times \DD$-plot may be presented dually via its action on the infinite list of local coordinates as
\begin{align}\label{InfinityJetOfSectionPlotLocally}
(j^\infty\phi^{k,\epsi})^* \;:\; C^\infty(J^\infty_M F) &
\; \longrightarrow\; \CO \big(\FR^k\times \DD \times M \big) \\ 
x^\mu &\;\longmapsto\; x^\mu \nn 
\\[-3pt]
\big\{u^a_I\big\}_{0\leq |I|} &\;\longmapsto\; \big\{\partial_I \phi^a(x^\mu, c^k ,\epsi) \big\}_{0 \leq |I|} \nn \, . 
\end{align}

\begin{example}[\bf Differential operators and the Euler--Lagrange equation]
With the infinite jet prolongation at hand, we may plot-wise precompose any bundle map $P:J^\infty F \rightarrow G$ over $M$ to identify differential operators as maps of thickened smooth sets of sections (\cite[Lem. 3.15]{GS23})
$$
\CP \, := \, P \circ j^\infty \, .
$$
Of particular importance to field theory is the Euler--Lagrange operator corresponding to the Euler--Lagrange ``source'' form $EL := \delta_V L \in \Omega^{d,1}_{s}(J^\infty_M F)$ (see \cite[pp. 50-51]{GS23}), i.e., equivalently the corresponding bundle map $EL:J^\infty_M F \rightarrow \wedge^d T^*M \otimes V^*F$ over M, which induces the Euler--Lagrange differential operator 
\begin{align}\label{EulerLagrangeOperator}
\CE \CL \, : = \, EL \circ j^\infty \;\; : \quad \CF \;\; & \xrightarrow{\qquad  \quad} T^*_\var \CF \\
 \phi^{k,\epsi} &\quad \longmapsto \quad  EL\circ j^\infty \phi^{k,\epsi} \nonumber \, ,
\end{align}
defining a section of the variational cotangent bundle (Def. \ref{VariationalCotangentBundle}). In local coordinates where $L=\bar{L}\big(x^\mu,\{u^a_I\}_{0\leq |I|}\big) \cdot \dd x^1\cdots \dd x^d$ it follows that $EL=\delta_V L = \sum (-1)^{|I|} D_I \Big(\! \frac{\partial \bar{L}}{\partial u^a_I} \!\Big) \cdot \dd_F u^a \otimes \dd x^1\cdots \dd x^d $, and hence by Eq. \eqref{InfinityJetOfSectionPlotLocally} the result of applying the Euler--Lagrange operator on a $\FR^k\times \DD$-plot of fields may be represented (dually) as (cf. \cite[(123)]{GS23})
\begin{align}\label{ELOperatorOfSectionPlotLocally}
\big(EL\circ j^\infty\phi^{k,\epsi}\big)^* \;\;:\;\; 
C^\infty\big(\wedge^d T^*M \otimes V^* F\big) & \; \xrightarrow{\qquad} \; \CO\big(\FR^k\times \DD \times M \big) \\ 
x^\mu &\;\;\; \longmapsto \; x^\mu \nn 
\\ 
u^a &\;\;\; \longmapsto \; \phi^a(x^\mu,c^k,\epsi) \nn
\\
\hat{u}^a &\;\;\; \longmapsto \;
\sum_{|I|=0}^{\infty} (-1)^{|I|} \frac{\partial}{\partial x^I} 
\Bigg(\frac{\delta \bar{L}\big(x^\mu ,\{\partial_J \phi^b(x^\mu,c^k,\epsi)\}_{|J|\leq k} \big)}{\delta (\partial_I \phi^a )} \Bigg) \nn \, , 
\end{align}
where $\{\hat{u}^a\}$ are the coordinates along the fibers of $\wedge^d T^*M \otimes V^*F$, adapted to the canonical coordinate associated to $\dd x^1\cdots \dd x^d$ for the chosen base coordinates $\{ x^\mu\}$. Abusing notation (even more), we may denote all of this simply as
\begin{align}\label{ELOperatorOfSectionPlotAbusively}
\CE \CL\big(\phi^{k,\epsi}\big) \, = \,  \sum_{|I|=0}^{\infty} (-1)^{|I|} \frac{\partial}{\partial x^I} 
\Bigg(\frac{\delta \bar{L}\big(x^\mu ,\{\partial_J \phi^b(x^\mu,c^k,\epsi)\}_{|J|\leq k} \big)}{\delta (\partial_I \phi^a )} \Bigg) \cdot \dd_F u^a \otimes \dd x^1\cdots \dd x^d \, ,
\end{align}
so that one may think of the (thickened, smooth) Euler--Lagrange operator acting on $\FR^k\times \DD$-plots of fields as usual along spacetime directions, and carrying along the parametrizing coordinates of $\FR^k\times \DD$.

The vanishing of this map of thickened smooth sets yields the Euler--Lagrange ``on-shell'' equations 
$$
\CE \CL = 0
$$ 
now incarnated within thickened smooth sets. We return to the proper description of this thickened zero-locus in Thm. \ref{FunctorialityOfTheThickenedCriticalSet} as a pullback/intersection of thickened smooth sets, where we furthermore identify it as the expected critical locus of the corresponding local action functional.
\end{example}

We now turn to the synthetic pushforward \eqref{SyntheticPushforward} of the infinite jet prolongation map. To that end, let us note the immediate analog result for the synthetic tangent bundle of $\mathbold{\Gamma}_M(J^\infty F)$. 
\begin{lemma}[\bf Synthetic tangent bundle of jet bundle sections]
\label{SyntheticTangentBundleOfJetBundleSections}
The synthetic tangent bundle $T \CF^\infty\:= \big[\DD^1(1),\, \CF^\infty \big] $ of the space of jet bundle sections is naturally identified with the space of jet bundle sections of the vertical fiber bundle $VF\rightarrow F \rightarrow M$, i.e., 
$$
T\CF^\infty\, \cong\, \mathbold{\Gamma}_M(J^\infty VF)\;\; \in \;\; \ThickenedSmoothSets \, .
$$
\end{lemma}
\begin{proof}
Working precisely as in the proof of Prop. \ref{SyntheticTangentBundleOfFieldSpace} (cf. also Rem. \ref{GeneralMappingSpaceTangentBundles}) yields the tangent bundle in question equivalently as the pullback 
\begin{align*} 
\xymatrix@=1.6em 
{T\CF^\infty  \ar[d] \ar[rr] &&  
[M,\, T J^\infty_M F] \ar[d]^{(\dd \pi^\infty_M)_*} 
	\\ 
	0_M\cong *\; \ar@{^{(}->}[rr]  && [M,\, TM]
	\, . } 
\end{align*}
By definition then, analogously to the case of $T\CF$, the plots of $T\CF^\infty$ correspond to maps $X_{\psi^{k,\epsi,\, \infty}}: \FR^k\times \DD\times M \rightarrow TJ^\infty_M F$ making the diagram 

\[ 
\xymatrix@C=4pc @R=1.8pc{ && TJ^\infty_M F \ar[d]^{\dd \pi^\infty_M}
	\\ 
	\FR^k\times \DD \times M  \ar[r]^-{\;\;\;\pr_M} \ar[rru]^-{X_{\psi^{k,\epsi,\, \infty}}} &M \ar[r]^{0_M} &TM \, 
}   
\]
commute, and hence factoring through the vertical subbundle $VJ^\infty_MF\hookrightarrow TJ^\infty_M F$ by \eqref{VerticalJetbundleAsPullback}. Under the canonical isomorphism $VJ^\infty_M F \cong J^\infty_M VF$ from Lem. \ref{VerticalJetBundleAsJetVerticalBundle}, these maps correspond exactly to the commuting diagrams of the form
\[ 
\xymatrix@C=2.6pc @R=1.4pc { &&  VJ^\infty_M F \cong J^\infty_M V F \ar[d]
	\\ 
	\FR^k\times \DD \times M \ar[rru]^{X_{\psi^{k,\epsi,\, \infty}}} \ar[rr]^-{\pr_M} && M
    \mathrlap{\,,}
}   
\]
i.e., precisely the defining plots of $\mathbold{\Gamma}_M(J^\infty VF)$.
\end{proof}

\begin{proposition}[\bf Synthetic pushforward of jet prolongation]\label{SyntheticPushforwardOfJetProlongation} The synthetic pushforward 
\begin{align*}
\dd j^\infty_F\, : \, T\CF \xrightarrow{\quad \quad}  T \CF^\infty   
\end{align*}
of the jet prolongation map $j^\infty_F: \CF \rightarrow \CF^\infty$
is given, under the identifications of Prop. \ref{SyntheticTangentBundleOfFieldSpace} and Lem. \ref{SyntheticInfiniteJetTangentBundle}, precisely by the corresponding prolongation map
\begin{align*}
j^\infty_{VF} \, :\,  \mathbold{\Gamma}_M (VF)\xrightarrow{\quad \quad}  \mathbold{\Gamma}_M(J^\infty VF) \, . 
\end{align*}
That is, the following diagram commutes
\begin{align*}
\begin{gathered}
\xymatrix@R=.5em@C=2.6em  {T\CF \ar[dd]_{\sim} \ar[rr]^-{\dd j^\infty_F} &&  T\CF^\infty \ar[dd]^{\sim } 
	\\ \\
\mathbold{\Gamma}_M(VF) \ar[rr]^-{j^\infty_{VF}}  && \mathbold{\Gamma}_M(J^\infty VF)
\, . }
\end{gathered}
\end{align*}
\end{proposition}
\begin{proof}
This follows by chasing through the given identifications. More explicitly, the synthetic pushforward \eqref{SyntheticPushforward} $\dd j^\infty_F$ acts on
any plot $\CZ_{\phi^{k,\epsi}}: \FR^k\times \DD \rightarrow T \CF $ of the tangent bundle, i.e., equivalently on any section
\[ 
\xymatrix@C=2.6pc @R=1.2pc { &&  F \ar[d]^{\pi}
	\\ 
	\DD^1(1)\times \FR^k\times \DD \times M \ar[rru]^-{\CZ_{\phi^{k,\epsi}}} \ar[rr]^-{\pr_M} && M \, ,
}   
\]
by postcomposition, thus yielding a prolongated section
\[ 
\xymatrix@C=2.6pc @R=1.2pc { &&  J^\infty_M F \ar[d]^{\pi^\infty}
	\\ 
	\DD^1(1)\times \FR^k\times \DD \times M \ar[rru]^-{j^\infty_F \raisebox{-1pt}{(}\CZ_{\phi^{k,\epsi}}
 \raisebox{-1pt}{)}\;\;\;} \ar[rr]^-{\pr_M} && M \, ,
}   
\]
which is by definition a plot $j^\infty_F (\CZ_{\phi^{k,\epsi}}) : \FR^k \times \DD \rightarrow T \mathbold{\Gamma}_M(J^\infty F) $, covering the plot $j^\infty \phi^{k,\epsi} : \FR^k\times \DD \rightarrow \CF^\infty$. But by the identification of Lem. \ref{SyntheticTangentBundleOfJetBundleSections}, along with the defining expression of $j^\infty$ in local coordinates \eqref{InfinityJetOfSectionPlotLocally}, this is identified precisely with the section 
\[ 
\xymatrix@C=2.6pc @R=1.2pc { &&  J^\infty_M VF \cong V J^\infty_M F \ar[d]^{\pi^\infty}
	\\ 
	 \FR^k\times \DD \times M \ar[rru]^{j^\infty_{VF}  \raisebox{-1pt}{(}\CZ_{\phi^{k,\epsi}} \raisebox{-1pt}{)}\;\;\;\;} \ar[rr]^-{p_M} && M \, ,
}   
\]
where now $\CZ_{\phi^{k,\epsi}}$ is viewed as a plot of $\mathbold{\Gamma}_M(VF)$ under the identification of Prop. \ref{SyntheticTangentBundleOfFieldSpace}, and hence acted upon with the corresponding jet prolongation.
\end{proof}
Let us now turn to the matter of the prolongated evaluation map within thickened smooth sets. This is defined by the composition
\begin{align}\label{ProlongatedEvaluationMap}
\ev^\infty\; :\;  \CF \times M \xrightarrow{\quad (j^\infty,\, \id_M)\quad} \CF^\infty\times M \xlongrightarrow{\;\;\ev \;\;} J^\infty_M F\, . 
\end{align}
where the (canonical) evaluation map explicitly acts on $\FR^k\times \DD$-plots as
\begin{align}\label{EvaluationMap}
\mathrm{ev}\;:\;\mathbold{\Gamma}_M(J^\infty F) \times M & \; \xrightarrow{\quad \quad}  \; J^\infty_M F 
\\[-2pt]
\big(\psi^{k,\epsi,\, \infty}, \, p^{k,\epsi}\big) &\; 
\xmapsto{\quad \quad} \; 
\psi^{k, \epsi,\, \infty} \circ \big(\id_{\FR^k\times \DD}, \, p^{k,\epsi}\big)  \,, \nn 
\end{align}
where the pair of plots are maps $\psi^{k,\epsi,\, \infty}: \FR^k \times \DD \times M \rightarrow J^\infty_M F$ and $p^{k,\epsi}: \FR^k\times \DD \rightarrow M$, so that the composition above yields a plot $\FR^k\times \DD \rightarrow \FR^k\times \DD \times M \rightarrow J^\infty_M F$ of the jet bundle. At the level of $*$-plots, this of course encompasses the usual set-theoretic evaluation map
$$
(\psi^\infty,p) \xmapsto{\quad \quad} \psi^\infty(p) \quad \in \quad J^\infty_M F(*) \, .
$$
By expanding this definition carefully in local coordinates, it is easy to see that the plots in the image of the prolongated evaluation map \eqref{ProlongatedEvaluationMap} 
\begin{align*}
\CF\times M &\xrightarrow{\quad \quad} J^\infty_M F \\
\big(\phi^{k,\epsi},\,  p^{k,\epsi}\big)  
&\xmapsto{\quad \quad} j^\infty \phi^{k,\epsi} \circ 
\big(\id_{\FR^k\times \DD}, \, p^{k,\epsi}\big) 
\end{align*}
may be represented dually (and locally), using \eqref{InfinityJetOfSectionPlotLocally}, as
\begin{align*}
C^\infty(J^\infty_M F) &\xrightarrow{\quad \quad} \CO\big(\FR^k\times \DD \big) \\ 
x^\mu &\xmapsto{\quad \quad} p^\mu(c^k,\epsi) \nn \\
\big\{u^a_I\big\}_{0\leq |I|} &\xmapsto{\quad \quad} \big\{\partial_I \phi^a\big(p^\mu(c^k,\epsi),\,  c^k ,\epsi\big) \big\}_{0 \leq |I|} \nn \, . 
\end{align*}

Having defined prolongated evaluation map within $\ThickenedSmoothSets$, we may apply the synthetic tangent functor 
$T=[\DD^1(1), \, - ]$ to obtain the synthetic pushforward \eqref{SyntheticPushforward}
$$
\dd \ev^\infty \; : \;  T(\CF \times M) \cong T\CF \times TM \xrightarrow{\qquad} TJ^\infty_M F \, . 
$$
In order to derive a more concrete description plot-wise for this pushforward, we first describe the pushforward along the usual evaluation maps \eqref{EvaluationMap}.

\begin{proposition}[\bf Pushforward of evaluation map]\label{PushfowardOfEvaluationMap} 
$\, $ 
\begin{itemize} 
\item[{\bf (i)}] For the evaluation map $\ev : \CF \times M\rightarrow F$ of a finite-dimensional fiber bundle $F\rightarrow M$ (cf. \eqref{EvaluationMap}), the corresponding synthetic pushforward is given by 
\begin{align*}
\dd \ev \; : \quad T\CF \times & TM \quad \xrightarrow{\quad \quad} \quad  TF 
\\ 
\big(\CZ_{\phi^{k, \epsi}}, &\, X_{p^{k,\epsi}}\big) \xmapsto{\quad \quad} 
\; \dd_M \phi^{k,\epsi}\circ \big(\id_{\FR^k\times \DD}, X_{p^{k,\epsi}}\big) 
\;\; + \;\;  \CZ_{\phi^{k,\epsi}} \circ \big(\id_{\FR^k\times \DD} \times p^{k,\epsi}\big) \, .
\end{align*}

\newpage 
\vspace{1mm} 
\item[{\bf (ii)}] Similarly, for the evaluation map $\ev : \CF^\infty\times M \rightarrow J^\infty_M F$ of the infinite jet bundle $J^\infty_M F\rightarrow M$ \eqref{EvaluationMap}, the corresponding synthetic pushforward is given by 
\begin{align*}
\dd \ev \;:\quad \; T\CF^\infty \times & TM \quad \; \xrightarrow{\quad \quad} \quad  T J^\infty_M F \\ 
\big(\CZ_{\psi^{k,\epsi,\, \infty}},&\, X_{p^{k,\epsi}}\big) 
\xmapsto{\quad \quad} \; \dd_M \psi^{k,\epsi,\, \infty} \circ \big(\id_{\FR^k\times \DD}, X_{p^{k,\epsi}}\big) 
\;\;  + \;\; \CZ_{\psi^{k,\epsi,\, \infty}} \circ \big(\id_{\FR^k\times \DD} \times p^{k,\epsi} \big) \,. 
\end{align*}
\end{itemize} 
\end{proposition}
\begin{proof}
For the sake of transparency, we show this at the level of $*$-plots, with that of general $(\FR^k\times \DD)$-plots following similarly by (trivially) carrying along the dependence along the probe-legs. By definition, a tangent vector on $\CF\times M$ is a pair of maps $(\CZ_\phi, X_p)\in T\CF\times TM$ where 
$$
\CZ_\phi : \DD^1 (1)\times M \longrightarrow F
$$ 
covers $\pr_M : \DD^1(1)\times M \rightarrow M$, and 
$$
X_p: \DD^1(1)\longrightarrow M\, .
$$
Now, by definition of the synthetic  \eqref{SyntheticPushforward} acts via postcomposition and using \eqref{EvaluationMap}, we get  
$$
\dd \ev \big(\CZ_\phi, X_p\big) := \CZ_\phi \circ \big(\id_{\DD^1(1)}\times X_p\big) 
\;\; : \;\; \DD^1(1) \longrightarrow \DD^1(1)\times M \longrightarrow F\, .
$$
Each of the maps in this composition is given dually by
\begin{align*}
\big(\id_{\DD^1(1)}\times X_p\big)^* 
\;\; : \;\;  \CO\big(\DD^1(1)\big)\otimes C^\infty(M) & \; \longrightarrow \; \CO\big(\DD^1(1)\big) \\
f+\epsi\cdot g &\; \longmapsto \; f(p) + \epsi\cdot \big(X_p (f) + g(p)\big) \, ,
\end{align*}
and 
\begin{align*}
\CZ_\phi^* \; : \; C^\infty(F) &\; \longrightarrow \; \CO\big(\DD^1(1)\big)\otimes C^\infty(M) \\
h &\; \longmapsto \; h\circ \phi + \epsi\cdot \CZ_\phi(h) \;.
\end{align*}
Here $\CZ_\phi(h)$ is the function on $M$ given by $\CZ_\phi(h) (p):= \CZ_{\phi(p)}(h) $, i.e., by acting pointwise with the corresponding (vertical) tangent vector (equiv. derivation) $\CZ_{\phi(p)} \in V_p F$. It follows that $\dd \ev (\CZ_\phi,\, X_p):\DD^1(1) \rightarrow F$ is given dually as
\begin{align*}
C^\infty(F) & \; \longrightarrow \; \CO\big(\DD^1(1)\big) \\
h& \; \longmapsto \; (h\circ \phi)(p) +\epsi \cdot \big( X_p(h\circ \phi) + \CZ_\phi(h) (p)\big) \; =\; h\big(\phi(p)) + \epsi\cdot \big(\dd \phi_p X_p +\CZ_{\phi(p)}(h)\big) \, ,
\end{align*}
which is precisely the derivation corresponding to the tangent vector
$$
\dd \phi_p (X_p) \;\; + \;\; \CZ_{\phi}(p) \;\; \in \;\; T_{\phi(p)} F \, ,
$$
confirming the first claim.

Notice that by the first claim the formula holds for each finite-order jet evaluation map $\mathbold{\Gamma}_M(J^n F)\times M \rightarrow J^n_M F$, and hence the second claim follows by the limit property of $J^\infty_M F$ and its tangent bundle's $TJ^\infty_M F \cong \lim TJ^n_M F$ from Lem. \ref{SyntheticInfiniteJetTangentBundle},
along with the identification $VJ^\infty_M F\cong J^\infty_M V F$ from Lem. \ref{VerticalJetBundleAsJetVerticalBundle}.
\end{proof}

\begin{corollary}[\bf Pushforward of prolongated evaluation]\label{PushforwardOfProlongatedEvaluation}
$\,$
\begin{itemize}
\item[{\bf (i)}] The synthetic pushforward of the prolongated evaluation map $\ev^\infty : \CF\times M \rightarrow J^\infty_M F$ is given by 
\begin{align*}
\dd \ev^\infty \; : \quad \; T\CF \times & TM \quad \; \xrightarrow{\quad \quad} \quad  T J^\infty_M F \cong HJ^\infty_M F\oplus J^\infty_M VF  \\ 
\big(\CZ_{\phi^{k,\epsi}},&\, X_{p^{k,\epsi}}\big) 
\; \xmapsto{\quad \quad} \; \dd_M \big(j^\infty \phi^{k,\epsi} \big) \circ \big(\id_{\FR^k\times \DD}, X_{p^{k,\epsi}}\big) 
\;\; + \;\;  j^\infty \CZ_{\phi^{k,\epsi}} \circ \big(\id_{\FR^k\times \DD} \times p^{k,\epsi} \big) \, ,
\end{align*}

\vspace{1mm} 
\noindent with each of the summands factoring through the horizontal and vertical subbundles, respectively.

\vspace{1mm} 
\item[{\bf (ii)}] 
In (compatible) local coordinates where $\CZ_{\phi} = \CZ^a \cdot \frac{\delta}{\delta \phi^a}:=  \CZ_{\phi}^a \cdot \frac{\partial}{\partial u^a} \in T_\phi \CF \cong \Gamma_M(\phi^*VF) $ 
and $X_{p}=X^\mu \cdot \frac{\partial}{\partial x^\mu} |_{p}$ $\in T_p M$, the pushforward is represented by
\big(cf.  Cor. \ref{SyntheticPushfowardOfJetProlongatedSection}{\bf (iii)}\big)
\vspace{1mm}
\begin{align}\label{PushforwardOfProlongatedEvaluationLocalCoords}
\hspace{-4mm} 
\bigg(\CZ_{\phi}^a \cdot \frac{\partial}{\partial u^a}, \, X^\mu \cdot \frac{\partial}{\partial x^\mu} \Big|_{p} 
\bigg) \quad \longmapsto \quad    
  X^\mu \bigg(\frac{\partial}{\partial x^\mu} \Big\vert_{j^\infty_p \phi} 
+ \sum_{|I|=0}^{\infty} \frac{\partial \phi^a} {\partial x^{I+\mu}} (p)
\cdot    \frac{\partial}{\partial u^a_I} \Big\vert_{j^\infty_p \phi}\bigg) 
 + \; \sum_{|I|=0}^\infty \frac{\partial \CZ_\phi^a}{\partial x^I} (p)
\cdot \frac{\partial}{\partial u^a_I} \Big|_{j^\infty_p \phi} \, ,
\end{align}

\vspace{0mm}
\noindent
and similarly for arbitrary plots.
\end{itemize} 
\end{corollary}
\begin{proof}
By the functoriality of the synthetic pushforward \eqref{SyntheticPushforward}, we have
\begin{align*}
\dd \ev^\infty = \dd \big( \ev\circ (j^\infty \times \id_M) \big) = \dd \ev \circ \dd (j^\infty \times \id_M) = \dd \ev \circ 
\big(\dd j^\infty \times \id_{TM} \big) \, ,
\end{align*}
whereby the result follows by Prop. \ref{SyntheticPushforwardOfJetProlongation}
and Prop. \ref{PushfowardOfEvaluationMap}, with the first summand taking values in the horizontal subspace by 
Cor. \ref{SyntheticPushfowardOfJetProlongatedSection}. 
\end{proof}
\begin{remark}[\bf Pushforwards via derivatives of line plots]\label{PushforwardsViaDerivativesOfLinePlots}
We note that the above sequence of results may be equivalently proved by representing any plot of tangent vectors via a line plot $\phi^{k,\epsi}_t$ (Lem. \ref{LinePlotsRepresentTangentVectors}), and taking the corresponding line-parameter derivative of the jet prolongated line plot $j^\infty_F \phi^{k,\epsi}_t$, which is precisely the intuition used to motivate the (therein ad hoc) form of the pushforward in 
\cite[Ex. 4.7, Def. 7.4]{GS23}. Here, we have recovered it naturally as an example of the synthetic pushforward construction; hence being already in its extended thickened version. 
\end{remark}

\noindent {\bf The bicomplex of local forms and its Cartan calculus.} The traditional local bicomplex (of globally f.o. forms) on $\CF\times M$ 
(\cite[Def. 7.5]{GS23}) now appears naturally within $\ThickenedSmoothSets$, as the pullback of the variational bi-complex on $J^\infty_M F$ along the prolongated evaluation map, by employing its \textit{synthetic} pushforward (Cor. \ref{PushforwardOfProlongatedEvaluation})
\begin{align}\label{LocalBicomplex}
\Omega^{\bullet,\bullet}_\mathrm{loc}(\CF\times M)\, :&= \,  (\ev^\infty)^*\Omega^{\bullet,\bullet}_{\glb}(J^\infty_M F)
\;\; \longhookrightarrow \;\;
\Omega^\bullet(\CF\times M) \, , 
\end{align}
where
\begin{align*}
(\ev^\infty)^*\,:\, \Omega^{\bullet,\bullet}(J^\infty_M F) &\xrightarrow{\quad \quad}  \Omega^{\bullet}(\CF\times M) \\
\om &\xmapsto{\quad \quad} \om \circ \dd \ev^\infty \, . 
\end{align*}
By the local coordinate expressions of tangent vectors on $\CF\times M$ and the  from \eqref{PushforwardOfProlongatedEvaluationLocalCoords}, along with that of differential forms on $J^\infty_M F$ from \eqref{pqformInfinityJetLocalCoordinates}, it follows that any such local form be represented locally as 
$$(\ev^\infty)^* \om =  \sum_{I_1,\cdots, I_p=0} \om_{\mu_1\cdots \mu_p a_1 \cdots a_q}^{I_1\dots I_q}
\big(x, \{\partial_J \phi^b \}_{|J|\leq k}\big)\cdot  \dd_M x^{\mu_1}\wedge \cdots
\wedge \dd_M x^{\mu_p}\wedge \delta(\partial_{I_1} \phi^{a_1}) \wedge\cdots \wedge \delta(\partial_{I_q} \phi^{a_q}) \, . 
$$
The vertical and horizontal differentials are given by pulling back those of the variational bicomplex
$$
\delta \big( (\ev^\infty)^* \om\big)  \,:= \, (\ev^\infty)^* \dd_V \om
\quad \quad \mathrm{and} \quad \quad \dd_M \big( (\ev^\infty)^* \om \big) \, := \, (\ev^\infty)^* \dd_H \om
$$
and they act as traditionally expected in terms of local coordinate expressions.
We shall not repeat this discussion here, as it follows verbatim as in \cite[pp. 89-91]{GS23}. The same applies for the discussion regarding the corresponding local Cartan calculus (\cite[p. 92]{GS23}) with insertions of local vector fields $\CZ= Z\circ j^\infty \in \CX_\loc(\CF)\hookrightarrow \CX(\CF\times M)$
$$
\iota_\CZ \, : \, \Omega^{\bullet,\bullet}_\loc(\CF\times M) \xrightarrow{\quad \quad} \Omega^{\bullet-1,\bullet}_\loc(\CF\times M)
$$
and Lie derivative
$$
\CL_{\CZ}  :=  [\iota_\CZ, \delta] \;\; : \;\; \Omega^{\bullet,\bullet}_\loc(\CF\times M) \xrightarrow{\quad \quad} \Omega^{\bullet,\bullet}_\loc(\CF\times M) 
$$
and hence for the result of \cite[Prop. 7.7, cf. Rem. 7.8]{GS23} as well.

\begin{remark}
[\bf Local forms on $\CF\times M$ via the moduli space]\label{LocalFormsViaModuliSpace}
Just as in the case of forms on the jet bundle (Rem. \ref{FormsOnJetBundleViaModuliSpace}), the proof of \cite[Lem. 7.9]{GS23} holds for the same formal reasons, hence yielding a DGCA injection of local forms
$$
\Omega^\bullet_\loc(\CF\times M)\longhookrightarrow \Omega^\bullet_\mathrm{dR}(\CF\times M)
$$
into the de Rham forms \eqref{deRhamFormsOnThickenedSmoothSet} on $\CF\times M$.
\end{remark}

\subsection{On-shell fields as a thickened smooth critical set}
\label{Subsec-onshellSynth}

Smooth Lagrangian bundle maps $L:J^\infty_M F \rightarrow \wedge^d (T^*M) $, equivalently differential forms $L\in \Omega^{0,d}(J^\infty_M F$), may now be considered to be defined within thickened smooth sets. It follows that the (plot-wise) precomposition with the thickened jet prolongation $j^\infty: \CF\rightarrow \CF^\infty$ (Lem. \ref{ThickenedInfiniteJetProlongation}) yields the local Lagrangian densities of physical theories as per \cite[Def. 3.12]{GS23}  
$$
\CL\,  :=\,  L\circ j^\infty \;\; : \;\; \CF \xrightarrow{\quad \quad} \Omega^d_{\mathrm{vert}}(M)\, ,
$$
or equivalently as a horizontal local top-form
$
\CL \in  \Omega^{d,0}_\mathrm{loc}(\CF\times M).
$
Similarly, one obtains local currents on field space \cite[Def. 3.17]{GS23} by precomposing $(p,0)$-forms $P\in \Omega^{p,0}(J^\infty_M F)$ with the thickened jet prolongation
$$
\CP\,  :=\, P\circ j^\infty \;\; : \;\; \CF \xrightarrow{\quad \quad} \Omega^{p}_\mathrm{vert}(M)\, ,
$$
or equivalently $\CP \in \Omega^{p,0}_\loc(\CF\times M)$.

\medskip 
For any compact oriented submanifold $\Sigma^p \hookrightarrow M$, the notion of integration defines a map of thickened smooth sets 
\begin{align}\label{ThickenedIntegration}
\int_{\Sigma^p} \, : \;\; \Omega^p_{\mathrm{vert}}(M) \xrightarrow{\quad \quad} \FR 
\end{align}
extending that of plain smooth sets from \cite[(57)]{GS23}. This is defined by plot-wise integration, simply by carrying along the (finite) polynomial dependence on the nilpotent coordinates.
Explicitly, for any $\FR^k\times \DD$-parametrized p-form on $M$ 
$$
\om^{k,\epsi}\,= \,\sum_{|I|=0}^{l} \om^{k}_I\cdot \epsi^I \quad \in \quad  \Omega^{p}_{\mathrm{vert}}(M)(\FR^k\times \DD)\, \cong \, \Omega^{p}_{\mathrm{vert}}(M)(\FR^k)\otimes \CO(\DD)
$$
where for each $\om^{k}_I\in \Omega^p_{\mathrm{vert}}(M)(\FR^k)$, the value 
of the function $\int_{\Sigma^p} \om^{k,\epsi} \, \in \, y(\FR)(\FR^k\times \DD)\,\cong \, C^\infty(\FR^k)\otimes \CO(\DD) $ is given by  
$$
\int_{\Sigma^p} \om^{k,\epsi}\;\;:= \;\;  \sum_{|I|=0}^{l}  \; \bigg(\int_{\Sigma^p} \om^k_{I}\! \bigg) \cdot \epsi^I \quad 
$$
where $\big(\int_M \om^k_I\big)(x) :=  \int_M \iota_{x}^* (\om^k_I)$, with $\iota_x$ being the inclusion $M\xrightarrow{\sim}\{x\}\times M\subset \FR^k\times M$. With integration within $\ThickenedSmoothSets$ at hand, the definition of charges of local currents along submanifolds follows as in \cite[Def. 3.19]{GS23}  
$$
\CP_{\Sigma^p} \, := \, \int_{\Sigma^p} \circ  \;  \; \CP \quad : \quad \CF \xrightarrow{\quad \quad} \FR\, ,  
$$
and similarly for the algebra of \textit{local functionals} $C^\infty_\loc(\CF)\hookrightarrow C^\infty(\CF)$ generated by these \cite[Def. 3.20]{GS23}.

\medskip 
Assuming the spacetime manifold $M$ is compact\footnote{Otherwise one works over arbitrary compact, oriented $d$-dimensional submanifolds $\Sigma^d\hookrightarrow M$, with a slight modification of the ensuing discussion as per \cite[\S 5.3]{GS23}.
}, then the charge of the local Lagrangian $\CL$ of any theory is precisely its \textit{local action functional}
$$
S \, := \, \int_M \circ \; \; \CL\, .
$$
This allows to define the notion of criticality for arbitrary thickened (smooth) $\FR^k\times \DD$-plots of fields.

\begin{definition}[\bf  Critical plots of thickened smooth map]\label{CriticalRkDPoints}
Let $S:\CF \rightarrow \FR$ be a thickened smooth map. The \textit{critical $\FR^k\times \DD$-plots} of $S$ is the subset of $\FR^k\times \DD$-plots
\begin{align*}
		\mathrm{Crit}(S)\big(\FR^k\times \DD\big):=  
        \Big\{\phi^{k,\epsi}\in \CF(\FR^k\times \DD) \; \big{|}\; \partial_t S \big(\phi^{k,\epsi}_t\big)|_{t=0}=0, \hspace{0.3cm}
  \forall \, \phi^{k,\epsi}_t\in  \CF_{\phi^{k,\epsi}}\big(\FR^k\times\DD\times \FR^1_t\big) \Big\} ,		
\end{align*}
where $\CF_{\phi^{k,\epsi}}\big(\FR^k\times \DD\times\FR^1_t\big) 
= \big\{\phi^{k,\epsi}_t \in \CF\big(\FR^k\times \DD\times \FR^1_t\big) \; \big| \; \phi^{k,\epsi}_{t=0}= \phi^{k,\epsi} \in \CF(\FR^k\times \DD) \big\}$. 
\end{definition}
As with bare smooth sets, this yields  an assignment of sets of $\FR^k\times \DD$-critical plots
$$
\FR^k\times \DD \longmapsto \mathrm{Crit}(S)(\FR^k\times \DD)\, ,
$$

\noindent  for each $\FR^k \times \DD \in \ThickenedCartesianSpaces$, which for a generic $S$ might not be functorial. However, for local action functionals, this is indeed the case, generalizing \cite[Prop. 5.31]{GS23} to yield a thickened critical smooth set.

\begin{theorem}[\bf Functoriality of the thickened critical set]\label{FunctorialityOfTheThickenedCriticalSet}
Let $M$ be a compact oriented manifold without boundary, $F\rightarrow M$ a fiber bundle, $\CF$ its smooth set of sections and 
$S:\CF \rightarrow  \FR $ a local (action) functional given by 
\begin{align*}
	S = \int_M \circ \;\; \CL
\end{align*}  
\noindent
for some local Lagrangian $\CL: \CF \longrightarrow \Omega^{d}_\mathrm{vert}(M)$.
Then: 
\begin{itemize}
\item[\bf (i)] The assignment
\vspace{-1mm}
\begin{align*}
\FR^k\times \DD \longmapsto \mathrm{Crit}(S)(\FR^k\times \DD)
\end{align*}

\vspace{-1mm}
\noindent  defines a thickened smooth set.

\item[\bf (ii)]  The thickened critical smooth set is canonically identified with the on-shell space of fields $\CF_{\CE \CL }\hookrightarrow \CF$, defined as the intersection/pullback diagram of thickened smooth sets 
 \[
	\xymatrix@=1.6em  {\mathrm{Crit}(S) \cong \CF_{\CE \CL} \ar[d] \ar[rr] &&   \CF \ar[d]^{\mathcal{EL}} 
		\\ 
		\CF\ar[rr]^-{0_\CF}  && T^*_\mathrm{var} \CF
	\, , }
	\]
\noindent 	where $T^*_\var \CF$ is the variational cotangent bundle (Def. \ref{VariationalCotangentBundle}) and $\mathcal{EL}$ is the (thickened, smooth) Euler--Lagrange differential operator \eqref{EulerLagrangeOperator} of $\CL$, and $0_\CF : \CF \rightarrow T^*_\var \CF$ is the zero-section of thickened smooth sets, defined plot-wise by $\phi^{k,\epsi} \mapsto 0_{F}\circ \phi^{k,\epsi}$ where $0_F : F\rightarrow \wedge^d T^*M \otimes V^*F$ is the zero section of the vector bundle (of manifolds) over F.
\end{itemize}
\begin{proof}
Having appropriately imported all relevant notions related to the infinite jet bundle, its variational bicomplex, and the jet prolongation within 
$\ThickenedSmoothSets$ (see the later part of \cref{Subsec-thickmap}), the proof from \cite[Prop. 5.31]{GS23} in the purely smooth setting applies verbatim for $\FR^k\times \DD$-plots up to the point of yielding 
\begin{align}\label{VariationOfActionOfPlots}
\partial_{t}S\big(\phi^{k,\epsi}_t\big)\big|_{t=0}\, = \, \int_{M}\circ \, \, 
\Big\langle  \mathcal{E}\mathcal{L} \big(\phi^{k,\epsi}\big)\, , \,  \partial_t \phi^{k,\epsi}_t \big|_{t=0} \Big\rangle 
\quad \in \quad  C^\infty(\FR^k)\otimes \CO(\DD) \, \cong \, y(\FR)(\FR^k\times \DD)\, ,
\end{align}
now taking values in $C^\infty(\FR^k)\otimes \CO(\DD)$, for any 1-parameter variation $\phi^{k,\epsi}_t$ of a fixed $\FR^k\times \DD$-plot of fields $\phi^{k,\epsi}\equiv \phi^{k,\epsi}_{t=0}\in \CF(\FR^k \times \DD)$. Although we shall not repeat the former discussion here, the final step that follows requires some further technical justification in the thickened setting.

The right-hand side above employs the composition of maps of thickened smooth sets
$$
\int_M \circ \, \; \langle - , - \rangle \quad : \quad T\CF\times_\CF T^*_\var \CF \longrightarrow \Omega^d_{\mathrm{vert}}(M) \longrightarrow \FR
$$
where the former is the canonical pairing from \eqref{VariationalCotangentTangentPairing} and the latter is the integration map from \eqref{ThickenedIntegration}. More explicitly, this means that for any probe $\FR^k \times \DD \in \ThickenedCartesianSpaces$, this pairing is the one obtained by extending plot-wise the \textit{non-degenerate} pairing of $*$-plots (by the fundamental Lemma of variational calculus) 
$$
T_\phi \CF \times T^*_{\mathrm{var},\phi}\,\cong\, \Gamma_M(\phi^* VF)\times \Gamma_M\big(\phi^*(V^*F \otimes \wedge^d T^*M)\big) \;  \xrightarrow{\quad \quad} \; \FR 
$$
using the $C^\infty(\FR^k)\otimes \CO(\DD)$-module structure of each of 
$$
T_{\phi_{k,\epsi}}\CF(\FR^k\times \DD) \, \cong \,  \mathbold{\Gamma}_M(\phi_{k,\epsi}^*VF)(\FR^k\times \DD) \, \cong \,  \Gamma_M(\phi^*VF) \,\widehat{\otimes}\, C^\infty(\FR^k)\otimes \CO(\DD)
$$
and \footnote{Where $\widehat{\otimes}$ denotes the completed projective tensor product.} 
$$
T^*_{\mathrm{var},\phi_{k,\epsi}}\CF(\FR^k\times \DD) \, \cong \,  \mathbold{\Gamma}_M\big(\phi^*_{k,\epsi}(V^*F\otimes \wedge^d T^*M)\big)(\FR^k\times \DD) \, \cong \,  \Gamma_M\big(\phi^*(V^*F\otimes \wedge^d T^*M)\big)\, \widehat{\otimes}\, C^\infty(\FR^k)\otimes \CO(\DD)
$$
via the algebra multiplication on $C^\infty(\FR^k)\otimes \CO(\DD)$. 

In particular, since the original pairing of $*$-plots is non-degenerate as a map of $\FR$-modules, it follows that the pairing defined at each probe stage $\FR^k\times \DD$ is \textit{non-degenerate} as a map of $C^\infty(\FR^k)\otimes \CO(\DD)$-modules.
That is,
$$
\int_{M}\circ \, \; \big\langle \mathcal{B}_{\phi^{k,\epsi}} \, , \, \CZ_{\phi^{k,\epsi}} \big\rangle  \, = \, 0 \quad \in \quad  C^\infty(\FR^k)\otimes \CO(\DD) \, ,
$$
for all $\CZ_{\phi^{k,\epsi}} \in T_{\phi^{k,\epsi}}\CF(\FR^k\times \DD)$ if and only if
$$
\mathcal{B}_{\phi^{k,\epsi}} \, = \, 0 \;\; \in \;\; T^*_{\var,{\phi^{k,\epsi}}}\CF(\FR^k\times \DD) \, .
$$
Hence, since every plot of tangent vectors to fields is covered by a path of plots of fields (Lem. \ref{LinePlotsRepresentTangentVectors}), it follows by \eqref{VariationOfActionOfPlots} that an $\FR^k\times \DD$-plot is critical for the local action functional (Def. \ref{CriticalRkDPoints}) if and only if it satisfies the Euler--Lagrange equations:    
\begin{align*}
\hspace{-5mm} 
	\phi^{k,\epsi} \in \mathrm{Crit}(S)(\FR^k\times \DD ) \quad \iff  \quad \mathcal{E}\mathcal{L}(\phi^{k,\epsi})
 =0_{\phi^{k,\epsi}} \;\; \in \;\; & T^*_{\var,{\phi^{k,\epsi}}}\CF(\FR^k\times \DD)  \\  & \cong \;
 \Gamma_M\big( V^*F \otimes \wedge^d T^*M \big)\hat{\otimes} C^{\infty}(\FR^k)\otimes C^\infty(\DD)  \nn \, ,
\end{align*}
where $0_{\phi^{k,\epsi}}$ is the canonical section given by composing $\phi^{k,\epsi}: \FR^k\times \DD\times M \rightarrow F$ with the zero-section $0_F : F\rightarrow V^*F\otimes \pi^*_M(\wedge^d T^*M)$.
By the very construction of the Euler--Lagrange map within thickened smooth sets (Eqs. \eqref{EulerLagrangeOperator}-\eqref{ELOperatorOfSectionPlotAbusively}), the latter condition is \textit{functorial} under pullbacks along maps
of plots $f:\FR^{k'}\times \DD' \rightarrow \FR^{k}\times \DD$ , and so 
$$\mathcal{E}\mathcal{L}\big(\phi^{k,\epsi}\big)= 0_{\phi^{k,\epsi}} \implies 
\mathcal{E}\mathcal{L}\big(f^* \phi^{k,\epsi}\big)= 0_{f^*\phi^{k,\epsi}}\,.
$$
Hence if 
$\phi^{k,\epsi}\in \mathrm{Crit}(S)(\FR^k\times \DD)\subset \CF(\FR^k\times \DD)$, then $f^*\phi^{k,\epsi}\in \mathrm{Crit}(S)(\FR^{k'}\times \DD')\subset \CF(\FR^{k'}\times \DD')$ 
for all $f:\FR^{k'}\times \DD'\rightarrow \FR^k\times \DD$, and so the assignment of critical plots defines a presheaf on $\ThickenedCartesianSpaces$. 

Finally, since limits in sheaf categories are computed probe-wise, the probe-wise vanishing condition of the Euler--Lagrange operator is nothing but the defining condition of the corresponding pullback diagram, which also guarantees the sheaf condition and completes the proof.
\end{proof}
\end{theorem}

\begin{remark}[\bf de Rham differential of action functional]
It is in our current setting of synthetic differential geometry that Lawvere \cite{Law80} envisaged that the moduli space of one forms $\mathbold{\Omega}^1$ (Def. \ref{ModuliOfDifferentialForms}, cf. \cref{ModuliOfDifferentialForms}) might allow for ``a drastic simplification of the usual
differential form calculus''. Unfortunately, this does {\it not} seem to be the case for the purposes of variational calculus and the correct thickened smooth structure on the critical locus. Namely, one may  instead formally compute the composition 
$$
\dd_{\dd R} \circ S \;:\; \CF \longrightarrow \FR \longrightarrow \mathbold{\Omega}^1 \, ,
$$ 
whose (naive) vanishing condition, however, fails to even produce the correct smooth structure on the critical locus (see \cite[Rem. 5.41]{GS23}).
By abstract reasons \cref{ModuliOfDifferentialForms}, or by the explicit calculation above, this does encode a $1$-form in the sense of a fiberwise linear map
$$
\delta S \, : \, T\CF \longrightarrow \FR  \, .
$$
Still, however, its further characterization given by a pairing as in \eqref{VariationOfActionOfPlots} remains crucial for identifying the correct thickened smooth set structure on its critical locus (cf. Cor. \ref{OnshellSyntheticTangentBundle})!
\end{remark}

\begin{remark}[\bf Criticality on general spacetimes]\label{CriticalityOnGeneralSpacetimes}
The result holds similarly on non-compact spacetimes, along the lines of \cite[Prop. 5.39]{GS23}, by first by first defining criticality appropriately locally and jointly for all compact $d$-dimensional submanifolds \cite[Def. 5.37, Def. 5.38]{GS23}. The proof then follows verbatim using the petit sheaf nature of $\CF$ over $M$, given the extra technical justifications of Thm. \ref{FunctorialityOfTheThickenedCriticalSet}, so we shall not repeat this here. In particular, the critical locus of a local Lagrangian on any non-compact spacetime is still identified as the Euler--Lagrange zero locus, i.e, as the pullback diagram 
 \[
	\xymatrix@=1.6em  {\mathrm{Crit}(\CL) \cong \CF_{\CE \CL} \ar[d] \ar[rr] &&   \CF \ar[d]^-{\mathcal{EL}} 
		\\ 
		\CF\ar[rr]^-{0_\CF}  && T^*_\mathrm{var} \CF
	\,  }
	\]
of thickened smooth sets. The discussion further extends to spacetimes with boundary (Ex. \ref{ManifoldsWithCornersExample}) as per \cite[Rem. 5.40]{GS23}, now interpreted within $\ThickenedSmoothSets$.
\end{remark}

With the above results at hand, we are now able to recover the pushforward of the Euler--Lagrange operator $\CE \CL_*$ \cite[Def. 7.12]{GS23} and the corresponding on-shell tangent bundle to field space $T\CF_{\CE \CL}$ \cite[Def. 7.13]{GS23} as a direct application of the synthetic tangent functor (Def. \ref{SyntheticTangentBundle}). Crucially, both of these were defined in a rather ad hoc manner in the purely smooth setting of \cite{GS23}, and more so in all previous rigorous approaches to classical field theory we are aware of (if at all defined therein). The result below thus further justifies the intuition that infinitesimals, in the vein of synthetic differential geometry, are indeed necessary for the complete and rigorous description of field theory.
\begin{corollary}[\bf Synthetic tangent bundle of on-shell fields is the space of Jacobi fields]\label{OnshellSyntheticTangentBundle}
The \textit{synthetic tangent bundle} $T \CF_{\CE\CL}:= [\DD^1(1)\, , \, \CF_{\CE \CL}]$ to the subspace of on-shell fields is canonically identified as the pullback 	
 \[
	\xymatrix@=1.6em  {T \CF_{\CE \CL} \ar[d] \ar[rr] &&   T\CF \ar[d]^{\mathcal{EL}_*} 
		\\ 
		T\CF \ar[rr]^-{{0_\CF}_*}  && T(T^*_\var\CF)
	\, , }
	\]
\noindent 	
in $\ThickenedSmoothSets$. Explicitly, this is the thickened smooth set with $\FR^k\times\DD$-plots
$$
T\CF_{\CE\CL}\big(\FR^k\times \DD \big)= \big\{ \CZ_{\phi^{k,\epsi}} \in T \CF(\FR^k\times \DD \; \;   | \; \; \CE \CL_*\big(\CZ_{\phi^{k,\epsi}}\big)= 0_{\CF*}(\CZ_{\phi^{k,\epsi}})\big\} \, ,
$$
where in local coordinates the synthetic pushforward of the Euler--Lagrange operator $\CE \CL_* \equiv T(\CE \CL)$ is given by
\begin{align*}
\CE \CL_* \;\;:\;\; T\CF & \; \xrightarrow{\qquad} \; T (T^*_\var \CF)
\\[-2pt]
\CZ_{\phi^{k,\epsi}} & \; \xmapsto{\qquad} \;  \bigg( \CE \CL\big(\phi^{k,\epsi}\big)\,,\,\sum_{|I|=0} 
 \frac{\delta \CE \CL_a (\phi^{k,\epsi})}{\delta (\partial_I \phi^{k,\epsi,\,b})} \cdot 
\frac{\partial\CZ_{\phi^{k,\epsi}}^b}{\partial x^I} \cdot \dd u^a \bigg) \, .
\end{align*}
In theoretical physics terminology, the synthetic tangent bundle of the on-shell field space recovers precisely the (thickened, smooth) space of $\textit{``Jacobi fields''}$.
\end{corollary}
\begin{proof}
The synthetic tangent bundle of the on-shell field space $\CF_{\CE \CL}\cong Crit(\CL)$ is obtained by applying the tangent functor $T=[\DD^1(1), -]$ to the pullback diagram of Thm. \ref{FunctorialityOfTheThickenedCriticalSet} and Rem. \ref{CriticalityOnGeneralSpacetimes}. Using once again that this functor has a left adjoint \eqref{TangentFunctorHasLeftAdjoint}, it commutes with all limits and, in particular, pullbacks. Hence it follows that the on-shell synthetic tangent bundle $T\CF_{\CE \CL}$ is canonically isomorphic to the pullback diagram 
 \[
	\xymatrix@=1.6em  { T\mathrm{Crit}(\CL) \cong T \CF_{\CE \CL} \ar[d] \ar[rr] &&   T\CF \ar[d]^{\mathcal{EL}_*} 
		\\ 
		T\CF \ar[rr]^-{{0_\CF}_*}  && T(T^*_\var\CF)
	\, . }
	\]
    
It remains to demonstrate the explicit local form of the synthetic pushforward of the Euler--Lagrange operator $\CE \CL_* \equiv T(\CE \CL)$. Recall, this is defined simply by internal post-composition (i.e., plot-wise) 
$$
T\CF \, := \, \big[\DD^1(1), \CF\big] \xrightarrow {\quad \CE \CL \circ (-)\quad } \big[\DD^1(1), T^*_\var \CF\big]\, =:  \, T(T^*_\var \CF) \, , 
$$
and since the projection from path spaces of fields to fields is nothing but the restriction along $\DD^1(1)\hookrightarrow \FR^1_t $ (Eq. \eqref{DerivativeOfPathsIsRestriction}),
the following diagram commutes
 \[
	\xymatrix@R=1.6em@C=2.2em  { \mathbold{P}(\CF) \ar[d]_{\partial_t |_{t=0}}  \ar[rr]^{\mathcal{EL}_*} &&   \mathbold{P}(T^*_\var \CF) \ar[d]^{\partial_t |_{t=0}} 
		\\ 
		T\CF \ar[rr]^{\mathcal{EL}_*}  && T(T^*_\var\CF)
	\, . }
	\]
Thus it follows, since any plot of tangent vectors $\CZ_{\phi^{k,\epsi}}$ is covered by a path of plots of fields $\phi^{k,\epsi}_t$ (Lem. \ref{LinePlotsRepresentTangentVectors}), that the pushforward map $\CE \CL_*$ may be computed by first lifting to a path of plots, and then applying $\partial_t|_{t=0} \circ \CE \CL$. 
But this is precisely the calculation from $\cite[p. 96]{GS23}$, extended trivially to carry along the $\CO(\DD)$ dependence via \eqref{ELOperatorOfSectionPlotLocally}-\eqref{ELOperatorOfSectionPlotAbusively}, so that by the traditional chain rule (pointwise with respect to M)
\begin{align*}
 \bigg(\CE\CL \big(\phi^{k,\epsi}_{t=0}\big)\, ,\, \partial_t \CE \CL_a \big(\phi_t^{k,\epsi}\big) \Big\vert_{t=0} \bigg) 
 & = 
 \Bigg( \CE \CL\big(\phi^{k,\epsi}\big)\,,\,  \sum_{|I|=0} 
\bigg(\frac{\partial EL_a}{\partial u^b_I}\circ j^\infty\phi^{k,\epsi} \bigg) \cdot 
 \frac{
\partial}{\partial t} \frac{\partial \phi^{k,\epsi,\,b}_t}{\partial x^I}\bigg\vert_{t=0} \Bigg) 
\\
&= \Bigg(\CE \CL\big(\phi^{k,\epsi}\big),\,  \sum_{|I|=0} 
 \frac{\delta \CE \CL_a (\phi^{k,\epsi})}{\delta (\partial_I \phi^{k,\epsi,\,b})} \cdot 
\frac{\partial\CZ_{\phi^{k,\epsi}}^b}{\partial x^I}\Bigg)
\\
&= (\CE \CL_a)_* \big(\CZ_{\phi^{k,\epsi}}\big)
\; .
\end{align*}

\vspace{-5mm}
\end{proof}

\newpage 
\begin{remark}[\bf Variations via infinitesimals] In local coordinates and by the abuse of notation from Eq. \eqref{ELOperatorOfSectionPlotAbusively}, it is easily seen that the derivative at $t=0$ may be formally computed as the derivative of ``$\CE \CL_a \big(\phi^{k,\epsi} + \epsi^1 \cdot \CZ_{\phi^{k,\epsi}}\big)$'' along a nilpotent infinitesimal coordinate $\epsi$, i.e.,
$$
\epsi \;\; \in \;\; \CO(\DD^1(1))\cong \FR[\epsi]/(\epsi)^2,
$$
by Taylor expanding as prescribed above (cf. also Cor. \ref{PertubativeActionForVectorAffineBundles} later on). This is not surprising from our formalism, since under this abuse of notation $\phi^{k,\epsi} + \epsi \cdot \CZ_{\phi^{k,\epsi}}$ stands exactly for an element in $T\CF (\FR^k\times \DD) \, \cong  \, \big[\DD^1(1), \, \CF\big](\FR^k\times \DD) \, \cong \,  \CF\big(\FR^k\times \DD\times \DD^1(1)\big)$. The same holds for the derivative of the action functional from Thm. \ref{FunctorialityOfTheThickenedCriticalSet}. Thus, these results rigorously justify the formal (and somewhat ad hoc) appearance of the nilpotent infinitesimal $\epsi$ in standard variational manipulations of field theories, as nothing but a stand-in for the derivative along paths (of plots) of fields. We note, however, the consistency of these formal manipulations in this infinite-dimensional setting rests crucially upon the content of Lem. \ref{LinePlotsRepresentTangentVectors}. 
\end{remark}

\subsection{Perturbative field theory as infinitesimal restriction}
\label{PerturbationSection}

The discussion in \cref{JetBundlesSyntheticallySec} motivates the formalization of perturbation theory around a field configuration $\phi$ as restriction to its synthetic infinitesimal neighborhood inside field space $\CF$. Here, we establish this statement for the thickened smooth field spaces of sections (Def. \ref{ThickenedSmoothSetOfSections}). The result applies to smooth local functionals, local Lagrangian densities, and more generally differential operators. In particular, the usual perturbative formal polynomial action on the tangent space is recovered upon choosing fiberwise coordinates along the chosen background field configuration.

\subsubsection*{Infinitesimal fields and their Taylor coefficients}
\label{PerturbationTaylorDataSection}

Let $\pi:F\rightarrow M$ be a fiber bundle of finite-dimensional smooth manifolds, and let $\CF=\mathbold{\Gamma}_M(F)\in\ThickenedSmoothSets$ be its thickened smooth set of sections. Fix an ordinary configuration $\phi\in\CF(*)$. Recall that its infinitesimal neighborhood (Def. \ref{SyntheticInfinitesimalNeighborhood}) is the pullback of $\eta_\CF:\CF\rightarrow\frJ\CF$ along the point represented by $\phi$. Thus
\begin{align*}
\DD_{\phi,\CF}(\FR^k\times\DD)\,
=\, \big\{\phi^{k,\epsi}\in\CF(\FR^k\times\DD)\;\big|\;
\phi^{k,0}=\phi\circ p_M\big\}\, .
\end{align*}
Here $p_M:\FR^k\times M\rightarrow M$ is projection, and $\phi^{k,0}$ denotes restriction along the canonical embedding $\FR^k\hookrightarrow\FR^k\times\DD$. This neighborhood is intrinsically defined independently of any Taylor expansion.

\medskip 
Recall the definition of a thickened smooth set of sections from Def. \ref{ThickenedSmoothSetOfSections}. When the field bundle is a vector bundle $E\rightarrow M$, it is easily seen that the plots of fields are canonically identified as
\begin{align*}
\CE(\FR^k \times \DD )\, \cong \, \CE(\FR^k) \otimes \CO(\DD) \, \cong \, \CE(\FR^k)\oplus \big( \CE(\FR^k)\otimes V\big) , 
\end{align*}
by which it further follows that the plots of an infinitesimal neighborhood are given by
\begin{align}\label{PerturbationEqVectorplots}
\DD_{\phi,\CE}(\FR^k\times\DD)\, &\cong \, \Big\{\phi^{k,\epsi}\in \CE(\FR^k\times\DD)\;\big|\;
\phi^{k}=\phi + \psi^{k,\epsi^{>0}}\Big\} \nonumber
\\ 
&=: \,  \phi+\CE(\FR^k)\otimes_\FR V\, .
\end{align}
Here $\phi$ in the second line denotes its constant underlying $\FR^k$-plot, while $\psi^{k,\epsi^{>0}}\in \CE(\FR^k)\otimes V$ denotes the purely infinitesimal thickenings of the total $\FR^k\times \DD$-plot $\phi^{k,\epsi}$. When $E$ is an affine bundle, the (parametrized) sections on the right tensor product are instead of the model vector bundle. 

\begin{definition}[\bf Variations and formal Taylor data]\label{PerturbationDefTaylor}
Let \(E\to M\) be a vector or affine bundle, $G\rightarrow M$ a vector bundle, and $Q:\CE \rightarrow\CG$  a map between their thickened smooth sets of sections. For $n\geq1$, let $\phi^n:\FR^n\rightarrow\CF$ be an $\FR^n$-plot of the form
\[
\phi^n(t_1,\ldots,t_n)=\phi+\sum_{i=1}^n t_i h_i\, ,
\qquad h_i\, \in\, T_\phi \CE \, ,
\]
thus with its adjoint being the smooth section $(t,x)\mapsto\phi(x)+\sum_i t_i h_i(x)$ over $\FR^n\times M$. Define the corresponding $n$-order variation of $Q$ as
\begin{equation}\label{PerturbationEqVariations}
Q^{(n)}_\phi(h_1,\ldots,h_n)
\,:= \,\left.\partial_{t_1}\cdots\partial_{t_n}(Q\circ\phi^n)\right|_{t=0}
=\left.\partial_{t_1}\cdots\partial_{t_n}
Q\!\left(\phi+\sum_{i=1}^n t_i h_i\right)\right|_{t=0}, \,
\end{equation}
where we set $Q^{(0)}_\phi:=Q(\phi)$. These are derivatives of the ordinary (affine) $\FR^n$-plot $Q\circ\phi^n$, and by the chain rule they depend only on the corresponding tangent vectors (cf. Lem. \ref{DerivativesAlongLinesDependOnTangentVectors}).\footnote{The earlier Lemma shows this holds even for nonlinear field bundles, but \textit{at first-order variations}. The higher-order case requires prescribing fiberwise coordinates; see Def. \ref{FiberWiseCoordinates} and onwards.} For a section-valued target, they are computed pointwise on $M$. Since these are symmetric in the (tangent) vector arguments, we define the \emph{formal Taylor series} to be the sequence denoted symbolically by
\begin{equation*}
\mathrm{Tayl}_\phi Q(h)=\sum_{n\geq0}\tfrac{1}{n!}Q^{(n)}_\phi(h,\ldots,h).
\end{equation*}
No convergence or evaluation of this infinite sum at an ordinary nonzero (tangent) field is asserted.
\end{definition}

For a smooth \(\FR^k\)-family of arguments, \eqref{PerturbationEqVariations} is taken with \(s\in\FR^k\) held fixed. If \(\xi=\sum_a h_av_a\in \CF(\FR^k)\otimes V\), define its nilpotent evaluation by
\begin{equation}\label{PerturbationEqExtension}
Q^{(n)}_{\phi,\FR^k\times\DD}(\xi,\ldots,\xi)
=\sum_{a_1,\ldots,a_n}
Q^{(n)}_{\phi,\FR^k}(h_{a_1},\ldots,h_{a_n})\,
v_{a_1}\cdots v_{a_n}.
\end{equation}
Multilinearity makes this definition independent of the decomposition. Smoothness, symmetry and multilinearity follow by restricting to the ordinary finite-dimensional family spanned by any finite collection of the arguments and using ordinary multivariable differential calculus. It follows that for each $n\in \NN$, \eqref{PerturbationEqExtension} defines a map of thickened smooth sets
$$
Q^{(n)}_{\phi} \, : \, \DD_{\phi,\CE} \longrightarrow \CG \, .
$$
As one might naively expect, the collection of these polynomial maps compute exactly the restriction of the original $Q:\CE \rightarrow \CG$ to the infinitesimal neighborhood around $\phi$, to the order prescribed by each plot.
\begin{proposition}[\bf Formal Taylor expansion is actual restriction]\label{PerturbationPropMain}
Let \(E \to M\) be a vector or affine bundle, \(\CE=\mathbold{\Gamma}_M(E
)\), and \(\phi\in\CF(*)\). Let \(Q:\CF\to\CG\) be a morphism of thickened smooth sets, with \(\CG=\FR\) or \(\CG=\mathbold{\Gamma}_M(G)\) for a smooth vector bundle \(G\to M\). Then:
\begin{itemize}
\item[{\bf (i)}] For every probe \(\FR^k\times \DD\) with \(V^{N+1}=0\), and every \(\xi\in \CF(\FR^k)\otimes V\), the actual restriction satisfies
\begin{equation}\label{PerturbationEqMain}
\bigl(Q|_{\DD_{\phi,\CF}}\bigr)_{\FR^k\times\DD}(\phi+\xi)
=\sum_{n=0}^{N}\tfrac{1}{n!}
Q^{(n)}_{\phi,\FR^k\times\DD}(\xi,\ldots,\xi)\, ,
\end{equation}
with the equations being compatible with morphisms of probes.
\item[{\bf (ii)}] Every iterated variation is recovered from the restriction. Specifically, for
\[
\CO(\DD_n)=\FR[\epsi_1,\ldots,\epsi_n]/\big(\epsi_1^2,\ldots,\epsi_n^2\big)\, ,
\]
where \(\DD_n=(\DD^1(1))^{\times n}\), and for ordinary sections \(h_1,\ldots,h_n\),
\begin{equation}\label{PerturbationEqRecover}
[\epsi_1\cdots\epsi_n]\,
Q_{\DD_n}\!\left(\phi+\sum_{i=1}^n\epsi_i h_i\right)
=Q^{(n)}_\phi(h_1,\ldots,h_n)\, ,
\end{equation}
where the left-hand side means ``extract the coefficient'' of the monomial $\epsi_1\cdots \epsi_n$.
\item[{\bf (iii)}] For two such morphisms \(Q,R\) with the same target,
\begin{equation}\label{PerturbationEqEquivalence}
Q|_{\DD_{\phi,\CF}}=R|_{\DD_{\phi,\CF}}
\quad\Longleftrightarrow\quad
Q^{(n)}_\phi=R^{(n)}_\phi\quad  \text{for all }\quad n\geq0\,  .
\end{equation}
\end{itemize}
Thus the actual restriction and the full formal Taylor data of an existing smooth morphism determine one another.
\end{proposition}

\begin{proof}
We first factor an arbitrary infinitesimal field through a finite-parameter family.
Choose a basis \(v_1,\ldots,v_q\) of \(V\). By \eqref{PerturbationEqVectorplots}, every $\FR^k\times \DD$-plot in the neighborhood $\DD_{\phi,\CE}$ has a unique expression
\[
\phi^{k,\epsi}=\phi+\xi,\qquad \xi=\sum_{a=1}^q h_av_a,\qquad h_a\in \CE(\FR^k)\,.
\]
Define the ordinary smooth plot
\[
K:\FR^k\times\FR^q\xrightarrow{\qquad} \CE\,,
\qquad K(s,t)=\phi+\sum_{a=1}^q t_a h_a(s).
\]
There is a morphism \(i_\DD:\FR^k\times \DD\to \FR^k\times\FR^q\) that is identity on the \(\FR^k\)-coordinates and (dually) sends \(t_a\mapsto v_a\), by Hadamard's Lemma (Lem. \ref{PartialHadamardsLemma}). In particular,
\[
i_\DD^*K=\phi^{k,\epsi}.
\]

Then writing \(q_K=Q\circ 
K\in\CG(\FR^k\times\FR^q)\), naturality gives the exact identity
\begin{equation}\label{PerturbationEqNaturality}
Q_{\FR^k\times\DD}\big(\phi^{k,\epsi}\big)=i_\DD^*q_K\,.
\end{equation}
For a scalar-valued target $\CG = \FR$, \(q_K\) is an ordinary smooth function of \((s,t)\) which by partial Hadamard's in the \(t\)-variables is equivalently
\begin{equation}\label{PerturbationEqRemainder}
q_K(s,t)
=
\sum_{n=0}^{N}\tfrac1{n!}
Q^{(n)}_\phi
\bigg(\sum_a t_a h_a(s),\ldots,\sum_a t_a h_a(s)\bigg)
+R_{N+1}(s,t)\, ,
\end{equation}
where Hadamard’s lemma expresses the remainder as
\(
R_{N+1}(s,t)
=
\sum_{a_1,\ldots,a_{N+1}}
t_{a_1}\cdots t_{a_{N+1}}\,
R_{a_1\cdots a_{N+1}}(s,t)
\). Every remainder monomial pulls back to zero along $\iota_\DD$, since \(v^\alpha\in V^{N+1}=0\), so that equations \eqref{PerturbationEqNaturality} and \eqref{PerturbationEqRemainder}  imply
\begin{equation}\label{PerturbationEqMulti}
\bigl(Q|_{\DD_{\phi,\CF}}\bigr)_{\FR^k\times\DD}(\phi+\xi)
\equiv Q_{\mathbb R^k\times\mathbb D}(\phi^{k,\epsi})
=
\sum_{n=0}^{N}\tfrac1{n!}
Q^{(n)}_{\phi,\mathbb R^k\times\mathbb D}
(\xi,\ldots,\xi)\, .
\end{equation}
For a section-valued target $\CG=\mathbold{\Gamma}_M(G)$, the same argument applies to the components of \(q_K(s,t,x)\) in any local trivialization of \(G\), treating \((s,x)\) as smooth parameters. The resulting identities agree on overlaps and hence give \eqref{PerturbationEqMulti} globally.

This proves the identity at every probe and for every plot. Its left-hand side is already a natural transformation, so the right-hand side automatically satisfies naturality for arbitrary probe morphisms, including those that mix smooth coordinates with nilpotent corrections. Neither the chosen basis of \(V\) nor a chosen splitting of such a probe morphism is part of the resulting restricted morphism. This proves part $(i)$.

For part $(ii)$, note that 
in \(\CO(\DD_n)\) the square-free monomials form an \(\FR\)-basis and \(\epsi_1\cdots\epsi_n\neq0\). Setting \(\xi=\sum_i\epsi_i h_i\) in \eqref{PerturbationEqMain}, the terms of degree less than \(n\) cannot contribute to the coefficient of \(\epsi_1\cdots\epsi_n\), and terms of degree greater than \(n\) vanish. In degree \(n\), each permutation of the \(n\) distinct indices contributes the same iterated variation. Their \(n!\) contributions cancel the factor \(1/n!\), proving \eqref{PerturbationEqRecover}, and the constant coefficient recovers \(Q(\phi)\).

Finally for part $(iii)$, note that equality of restrictions implies equality of all variations. Conversely, equality of the variations on ordinary sections implies equality on smooth parameter families: for each \(s\), freeze the arguments at that value of \(s\). The resulting smooth coefficient functions, or coefficient sections on \(\FR^k\times M\), agree pointwise. Equation \eqref{PerturbationEqMain} then proves equality on every probe, which yields  \eqref{PerturbationEqEquivalence}.
\end{proof}

\subsubsection*{Local differential operators and action functionals}

We use the local field notation of \cite[Defs. 3.12, 7.5, Eqs. (172)--(174)]{GS23}, with its thickened interpretation from Cor. \ref{PushforwardOfProlongatedEvaluation} and Eq. \eqref{LocalBicomplex}, which we further recall below for convenience of the reader. Thus \(F\to M\) denotes the field bundle, \(\CF=\mathbold{\Gamma}_M(F)\) its field space, \(L\) a Lagrangian bundle map on jets, and \(\CL=L\circ j^\infty\) the induced local Lagrangian on fields. Ordinary variations (tangent vectors) are denoted by \(h,h_i \in T_\phi \CF \), and nilpotent inputs by \(\xi\). The symbol \(\frac{\delta}{\delta(\partial_I\phi^a)}\) below is the local field-coordinate derivative, defined by jet pullback (and not the distributional kernel notation).

\smallskip 
We will first treat the case where $E=F\rightarrow M$ is a vector bundle and the jet charts come from linear bundle coordinates, globally justifying the additive Taylor formulas. The following section will treat general smooth field bundles using fiberwise coordinates along a fixed background field configuration, which will justify the additive formulas locally in field space.

\paragraph{Local operators, Lagrangians and actions.}

A smooth differential operator of order at most \(r\) is induced by a smooth bundle map \(P_r:J^r_M F\to G \), where \(G\to M\) is a fiber bundle. Regarding \(P=P_r\circ\pi_r^\infty\) as a globally finite-order map on \(J^\infty_MF\), put
\begin{equation}\label{PerturbationEqDiffOperator}
\CP:=P\circ j^\infty:\CF\longrightarrow\mathbold{\Gamma}_M(E)\, ,
\qquad \CP(\phi)(x)=P_r(j_x^r\phi) = P(j^\infty_x\phi)\,.
\end{equation}
Similarly, for \(d=\dim M\), a globally finite-order Lagrangian bundle map
\[
L\;:\;J^\infty_MF\longrightarrow\wedge^dT^*M
\]
induces the local Lagrangian
\begin{equation*}
\CL:=L\circ j^\infty\;:\;\CF\longrightarrow\Omega^d_{\mathrm{vert}}(M).
\end{equation*}
For a compact oriented domain \(\Sigma=\Sigma^d\hookrightarrow M\), its action over $\Sigma$ is
\begin{equation}\label{PerturbationEqAction}
S_\Sigma:=\int_\Sigma\circ\; \CL\,,
\qquad
S_\Sigma(\phi)=\int_\Sigma L\circ j^\infty\phi\,.
\end{equation}

\paragraph{Local field derivatives.}

Let \(\ev^\infty:\CF\times M\to J^\infty_MF\) be the prolongated evaluation \eqref{ProlongatedEvaluationMap}, with \(\ev^\infty(\phi,x)=j_x^\infty\phi\). The local Lagrangian is, equivalently, the horizontal form \(\CL=(\ev^\infty)^*L\) on \(\CF\times M\). As in \cite{GS23}, the notation \(\partial_I\phi^a\) on this product denotes the local function
\[
\partial_I\phi^a:=u_I^a\circ\ev^\infty,
\qquad
\delta(\partial_I\phi^a):=(\ev^\infty)^*\dd_Vu_I^a.
\]
The local bicomplex differential is \(\delta\big((\ev^\infty)^*\om \big)=(\ev^\infty)^*\dd_V\om\) and the horizontal differential is 
\(\dd_M\big((\ev^\infty)^*\om \big)=(\ev^\infty)^*\dd_H\om\). These are the pullback conventions of Eq. \eqref{LocalBicomplex} and \cite[Def. 7.5, Eq. (172)]{GS23}.

\smallskip 
For a local function \(f_\loc=f\circ\ev^\infty \in C^\infty_\loc(\CF\times M)\), define in a jet chart
\begin{equation*}
\frac{\delta f_\loc}{\delta(\partial_I\phi^a)}
:=\frac{\partial f}{\partial u_I^a}\circ\ev^\infty.
\end{equation*}
For a horizontal density \(L=\bar L\,\dd^dx\), where \(\dd^dx=\dd x^1\wedge\cdots\wedge\dd x^d\), apply this definition to \(\bar L\) and retain the horizontal density factor. At a fixed field \(\phi\), this reads
\begin{equation*}
\frac{\delta\CL}{\delta(\partial_I\phi^a)}(\phi)
:=\left(\frac{\partial L}{\partial u_I^a}\right)\circ j^\infty\phi\,.
\end{equation*}
In particular, the undifferentiated-field case is exactly
\begin{equation}\label{PerturbationEqUsersconvention}
\frac{\delta\CL}{\delta\phi^a}(\phi)
=\left(\frac{\partial L}{\partial u^a}\right)\circ j^\infty\phi\,.
\end{equation}
For \(n\geq1\), iterated local field-coordinate derivatives mean
\begin{equation}\label{PerturbationEqIteratedlocal}
\frac{\delta^n\CL}
{\delta(\partial_{I_1}\phi^{a_1})\cdots\delta(\partial_{I_n}\phi^{a_n})}(\phi)
:=\left(
\frac{\partial^n L}{\partial u_{I_1}^{a_1}\cdots\partial u_{I_n}^{a_n}}
\right)\circ j^\infty\phi\,.
\end{equation}
The same notation applies component-wise to arbitrary differential operators \(\CP=P\circ j^\infty\).

\smallskip 
This is the local-coordinate convention of \cite[Eq. (174)]{GS23}. The apparent independent differentiation with respect to the fields and their spacetime derivatives is justified by the independent jet coordinates before pullback, as explained in \cite[Rem. 5.13]{GS23}. Individual coordinate derivatives depend on the jet chart, while the complete variation obtained by contraction below is intrinsic.

\paragraph{From local coordinate derivatives to an actual variation.}
In the notation of Cor. \ref{PushforwardOfProlongatedEvaluation}, a tangent vector is written
\[
h\equiv\CZ_\phi=h^a\frac{\delta}{\delta\phi^a}\;\; \in \; T_\phi\CF(*)
\cong\Gamma_M(\phi^*VF)\,.
\]
This denotes a section along \(\phi\), with spacetime-dependent coefficients \(h^a\). The pushforward of prolongated evaluation takes it to
\begin{equation*}
(\dd ev^\infty)_{\phi,x} \CZ_\phi \, = \, 
j^\infty h(x)
\, =\, \sum_{a,I}\partial_Ih^a(x)
\frac{\partial}{\partial u_I^a}\bigg|_{j_x^\infty\phi}\, ,
\end{equation*}
i.e., the vertical component of Eq. \eqref{PushforwardOfProlongatedEvaluationLocalCoords}. In particular, \cite[Eq. (173)]{GS23} reads
\begin{equation*}
\delta(\partial_I\phi^a)_{(\phi,x)}(h,0)
\,=\,\partial_Ih^a(x).
\end{equation*}
Thus the familiar formal rule \(\delta(\partial_I\phi^a)=\partial_I(\delta\phi^a)\) has a precise meaning upon evaluation on a tangent vector.

Writing \(\overline{\CL}=\bar L\circ\ev^\infty\), the first local variation is, with the vertical one-form written first as in \cite[Eq. (119)]{GS23},
\[
\delta\CL
=\sum_{a,I}\frac{\delta\overline{\CL}}
{\delta(\partial_I\phi^a)}\,
\delta(\partial_I\phi^a)\wedge\dd^dx.
\]
Contracting its vertical slot with \(h=\CZ_\phi \in T_\phi \CF\) gives the ordinary parameter derivative
\begin{equation*}
\CL^{(1)}_\phi(h)
=\partial_t\CL(\phi_t)\big|_{t=0}
=\sum_{a,I}\frac{\delta\CL}{\delta(\partial_I\phi^a)}(\phi)
\,\partial_Ih^a\, ,
\end{equation*}
where $\partial_t \phi_t \vert_{t=0} = h = \CZ_\phi$. Formula \eqref{PerturbationEqUsersconvention} is the \(I=\varnothing\) coefficient of this expression and keeping all jet indices supplies the full response to an arbitrary spacetime-dependent variation. 

\paragraph{The case of vector field bundles}

\begin{proposition}[\bf Local Taylor coefficients in field notation]\label{PerturbationCorLocal}\label{PerturbationThmLocal}\label{PerturbationPropLocal}
Let the source field bundle be a \textit{vector} bundle $E\rightarrow M$, and similarly for that of the target \(G\to M\). 
\begin{itemize}
\item[{\bf (i)}]
For \(\CP,\CL,S_\Sigma\) as in Eqs. \eqref{PerturbationEqDiffOperator}-\eqref{PerturbationEqAction}, their iterated variations are given by
\begin{equation}\label{PerturbationEqLocalcoeff}
\CP^{(n)}_\phi(h_1,\ldots,h_n)
=\sum_{\substack{a_1,\ldots,a_n\\|I_1|,\ldots,|I_n|\leq r}}
\!\frac{\delta^n\CP}
{\delta(\partial_{I_1}\phi^{a_1})\cdots\delta(\partial_{I_n}\phi^{a_n})}(\phi)
\prod_{\nu=1}^n\partial_{I_\nu}h_\nu^{a_\nu},
\end{equation}
with the same expression for \(\CL^{(n)}_\phi\). In particular, each coefficient is a symmetric multilinear differential operator of order at most \(r\) in each argument. The action coefficients are
\begin{equation}\label{PerturbationEqActioncoeff}
S_{\Sigma,\phi}^{(n)}(h_1,\ldots,h_n)
=\int_\Sigma\CL^{(n)}_\phi(h_1,\ldots,h_n) \,.
\end{equation}

\item[{\bf (ii)}]
The same formula applies for the Euler--Lagrange operator $\CE \CL$, here with fixed target bundle identified as
\[ 
V^*E \otimes \pi^*\big(\wedge^d T^*M\big) \, \cong \,  \ E^* \otimes \wedge^d T^*M  \; \xrightarrow{\qquad} M\, ,
\]
since for a vector field bundle $E\rightarrow M$ we have the canonical isomorphism $V^*E\cong\pi^*E^*$.
Their actual restrictions on infinitesimal neighborhoods satisfy Prop.~\ref{PerturbationPropMain} on every probe.
\end{itemize}
\end{proposition}

\begin{proof}
Substitute the ordinary parameter family \(\phi+\sum_\nu t_\nu h_\nu\) into \(\CP=P_r\circ j^r=P\circ j^\infty\). Since \(E\) is a vector bundle, its jet prolongation is
\[
j^r\phi+\sum_\nu t_\nu j^rh_\nu \,.
\]
Ordinary mixed differentiation of the finite-dimensional function \(P_r\) in its fiber coordinates gives \eqref{PerturbationEqLocalcoeff}, by the chain rule and \eqref{PerturbationEqIteratedlocal}, and the same calculation for \(L_r\) gives the density coefficients. Parameter differentiation commutes with integration over fixed compact \(\Sigma\), proving \eqref{PerturbationEqActioncoeff}.

The operators are morphisms of thickened smooth sets by thickened jet prolongation and composition. Integration is the coefficient-wise map of Eq. \eqref{ThickenedIntegration}:
\[
\int_\Sigma\!\bigg(\sum_b\om_bv_b\bigg)
=\sum_b\left(\int_\Sigma\om_b\right)v_b,
\]
and therefore Prop.~\ref{PerturbationPropMain} applies verbatim. In particular, for \(\xi=\sum_a h_av_a\) one has \(j^r\xi=\sum_a(j^rh_a)v_a\), i.e., spacetime differentiation ``preserves'' the nilpotent coefficients. All Taylor remainders vanish at the stated probe order, so that integration is applied only to a finite sum -- no exchange of an infinite sum and an integral is required.
\end{proof}
\begin{corollary}[\bf Perturbative actions for vector/affine field bundles]\label{PertubativeActionForVectorAffineBundles}
Let \(E\to M\) be a vector or affine field bundle and fix \(\phi\in\mathcal E(*)\) along with some action morphism \(S_\Sigma:\mathcal E\to\mathbb R\). fiberwise translation gives a canonical isomorphism of thickened smooth sets
\begin{align*}
 \tau_\phi\,:\,T_\phi\mathcal E&\xrightarrow{\quad \sim \quad }\mathcal E
 \\[-2pt]
 h&\quad \longmapsto\quad \phi+h\, ,
\end{align*}
and hence an induced isomorphism
\begin{align*}
\Psi_\phi: \mathbb D_{0,T_\phi\mathcal F} &\xrightarrow{\quad \sim\quad } \mathbb D_{\phi,\mathcal F} 
\\[-2pt]
\xi &\quad \longmapsto \quad \phi+\xi\, .
\end{align*}
For a vector/affine bundle, \(h\) is a section of its modeling vector bundle $T_\phi \CE \cong \mathbold{\Gamma}(\phi^*VE)$. Then:

\begin{itemize}
\item[\bf (i)] For
$S^{(n)}_{\Sigma,\phi}:=(S_{\Sigma}\circ\tau_\phi)^{(n)}_0$, the familiar ``formal perturbative action''
\begin{align*}
S_{\Sigma,\mathrm{pert},\phi}(h)\, & \equiv \,\text{``}S_\Sigma(\phi+h)\text{''}\, 
\\ &:=    \sum_{n\geq0}\tfrac1{n!}S^{(n)}_{\Sigma,\phi}(h,\ldots,h) \nn
\\[-3pt]
&=\, S_\Sigma(\phi)
+\int_\Sigma\sum_{a,I}\tfrac{\delta\CL}{\delta(\partial_I\phi^a)}(\phi)
\,\partial_Ih^a
+\tfrac12\int_\Sigma\sum_{a,b,I,J}
\frac{\delta^2\CL}{\delta(\partial_I\phi^a)\delta(\partial_J\phi^b)}(\phi)
\,\partial_Ih^a\,\partial_Jh^b+\cdots\, . \nn 
\end{align*}
encodes exactly the actual action restricted to the infinitesimal neighborhood. 

Explicitly, on every probe \(\mathbb R^k\times\mathbb D\) with \(V^{N+1}=0\),
\begin{align*}
\bigl(S_\Sigma |_{\mathbb D_{\phi,\mathcal F}}\circ\Psi_\phi\bigr)(\xi)\, &=\, \bigl(S_\Sigma|_{\DD_{\phi,\CF}}\bigr)_{\FR^k\times\DD}(\phi+\xi) \, \nn  \\ &= \,   \sum_{n=0}^{N}\tfrac1{n!} S^{(n)}_{\Sigma,\phi}(\xi,\ldots,\xi) \\
&=\, \sum_{n=0}^{N}\tfrac1{n!}\int_\Sigma
\sum_{\substack{a_1,\ldots,a_n\\|I_1|,\ldots,|I_n|\leq r}}
\frac{\delta^n\CL}
{\delta(\partial_{I_1}\phi^{a_1})\cdots\delta(\partial_{I_n}\phi^{a_n})}(\phi)\cdot\prod_{\nu=1}^n\partial_{I_\nu}\xi^{a_\nu}\, . \nn 
\end{align*}
The \(n=0\) term stands for \(S_\Sigma(\phi)\), and the integral/products are extended to \(\CO(\DD)\) coefficient-wise, with \(s\in\FR^k\) fixed.

\item[\bf(ii)] Conversely, the restricted action determines every coefficient by mixed nilpotent coefficient extraction: for ordinary tangent vectors $h_i\in(T_\phi\CF)(*)$ and $\DD_n$ as in Prop.~\ref{PerturbationPropMain},
\[
S^{(n)}_\phi(h_1,\ldots,h_n)
=
[\epsi_1\cdots\epsi_n]\,
S_{\DD_n}\!\bigg(\phi+\sum_i\epsi_i h_i\bigg).
\] 
\item[\bf(iii)] The same holds for the corresponding Lagrangian density $\CL$, meaning formally
$$
\text{``} \CL_{\mathrm{pert},\phi}\, = \, \CL|_{\mathbb D_{\phi,\mathcal F}}\circ\Psi_\phi \text{''} \, ,
$$
while replacing \(\CL\) by \(\CP\) gives back the corresponding general differential-operator identity.
\end{itemize}
\end{corollary}
\begin{proof} The translation isomorphism $\tau_\phi$, arising from the bundle isomorphism
\begin{align*}
a_\phi \, : \, \phi^* VE &\; \xrightarrow{\;\;\sim\;\;} \; E \\
(x,v) &\; \longmapsto \; \phi(x) + v\, ,
\end{align*}
identifies the zero section of the former $0\in T_\phi \CF$ with the chosen field \(\phi\in \CE \). Therefore, it restricts to an isomorphism of their infinitesimal neighborhoods. Thus, we may apply the restriction proposition (Prop. \ref{PerturbationPropMain}) to \(S\circ\tau_\phi\), which is a morphism on the section space of a vector bundle. The recovery assertion follows from the same proposition, while the explicit field derivative formulas follow from Prop. \ref{PerturbationPropLocal}.
\end{proof}

\begin{remark}[\bf Probing a theory infinitesimally]
The above result firmly grounds the intuition of the perturbative field theory of a Lagrangian theory 
$$(\CE,\CL)$$
at a fixed field configuration $\phi \in \CF(*)$ as \textit{probing the infinitesimal neighborhood} of the theory around the given field. The extent to which this has been established, so far, is for all local theories of connections (in a fixed topological sector, hence corresponding to an \textit{affine} field bundle modelled on the adjoint bundle tensored with the cotangent bundle $T^*M\otimes \mathrm{ad}(P)$), and for all theories of matter fields (hence corresponding to associated \textit{vector} field bundles). The case of theories with nonlinear field bundles -- such as generic nonlinear sigma models with field bundles $M\times N \rightarrow M$, and Einstein gravity with field bundle (say in the 2nd-order formalism) the open subbundle $\mathrm{Lor}(TM)\hookrightarrow \mathrm{Sym}^2(T^*M)$, are not yet covered by the above results. We come back to the issue of this generality in the following subsection.
\end{remark}

\begin{remark}[\bf Local field derivatives and the Euler--Lagrange operator]

Recall that for a general field bundle $\pi:F\to M$, the geometric target of the Euler--Lagrange bundle map is
\begin{equation}\label{PerturbationEqEltarget}
B\, :=\, V^*F\otimes\pi^*\big(\wedge^dT^*M\big) \xrightarrow{\qquad} F\,,
\end{equation}
so that the Euler--Lagrange operator evaluated on $\phi \in \CF(*)$ yields a section
$$
\CE\CL(\phi)\, \in \, \Gamma_M(\phi^*B) = T^*_{\var, \phi} \CF\,.
$$
That is, the Euler--Lagrange operator \eqref{EulerLagrangeOperator} is a section $\CE\CL:\CF\to T^*_{\mathrm{var}}\CF$, and in general not a map valued in some fixed vector-bundle section space over $M$. In particular, its values at different fields $\phi$ lie in different fibers $T^*_{\var, \phi} \CF$, so their subtraction and ordinary Taylor coefficients require a target identification. For a vector field bundle $E$, however, the canonical identification used in Prop.~\ref{PerturbationThmLocal} gives
\[
T^*_{\mathrm{var}}\CE
\cong\CF\times\mathbold{\Gamma}_M\big(E^*\otimes\wedge^dT^*M\big)\, ,
\]
while for an affine bundle the second factor is identified via the shift isomorphism as per Cor. \ref{PertubativeActionForVectorAffineBundles}. The Taylor formula then applies to the second component. The nonlinear field bundle case requires the Jacobian pullback in Thm.~\ref{PerturbationThmNonlinear} below, since replacing the varying target by $\phi^*B$ alone would not supply that identification.

The local coordinate derivative \(\delta\CL/\delta\phi^a\) is distinct from the Euler--Lagrange operator \(\CE\CL=EL\circ j^\infty\) of Eq. \eqref{EulerLagrangeOperator}. In a coordinate volume form, the latter reads \cite[Eq. (123)]{GS23} 
\begin{equation*}
\CE\CL_a(\phi)
=\sum_{|I|\leq r}(-1)^{|I|}\partial_I
\left(\frac{\delta\overline{\CL}}
{\delta(\partial_I\phi^a)}(\phi)\!\right) \dd^d x.
\end{equation*}
Here \(\partial_I\) is actual spacetime differentiation after jet evaluation, or equivalently, before evaluation it is total differentiation \(D_I\) on jets. We retain the symbol \(\CE\CL\) for this integrated-by-parts expression, rather than using the local-coordinate derivative symbol for it.

When boundary terms vanish, the action variation  for $h \in T_\phi \CE \cong \CE(*)$ is equivalently
\[
S_{\Sigma,\phi}^{(1)}(h)
=\int_\Sigma\sum_{a,I}
\frac{\delta\CL}{\delta(\partial_I\phi^a)}(\phi)\,\partial_Ih^a
=\int_\Sigma\CE\CL_a(\phi)h^a\,.
\]
In the presence of a boundary, the original first expression remains valid while the second requires the boundary term or further imposed boundary conditions. This is the integration-by-parts operation of \cite[Eq. (119)]{GS23} and the variational pairing of Eq. \eqref{VariationOfActionOfPlots}.
\end{remark}

We showcase how all the above theory fits the textbook perturbative  field theory algorithm by the following standard explicit example.

\begin{example}[\bf Perturbative scalar field theory on the flat torus]\label{PerturbativeScalarFieldTheoryExample}
Let the field bundle be
\[
F=M\times\FR,\qquad \CF(*)=C^\infty(M),
\]
on the torus $M=T^d$, say, with its flat metric and volume form $\dd^dx$. For a smooth potential $\mathsf Z:\FR\to\FR$, put
\begin{equation}\label{PerturbationEqTorusAction}
\begin{split}
L&=\bigg(\tfrac12\sum_\mu(u_\mu)^2+\mathsf Z(u)\bigg)\dd^dx\,,\\
\CL(\psi)&=\big(\tfrac12|\nabla\psi|^2+\mathsf Z(\psi)\big)\dd^dx\,,
\qquad S(\psi)=\int_M\CL(\psi)\,.
\end{split}
\end{equation}
Fix a background $\phi$, so that by Cor.~\ref{PertubativeActionForVectorAffineBundles}, translation gives $\psi=\phi+h$ where the ``fluctuation'' $h\in(T_\phi\CF)(*)$ is an ordinary smooth function. In the following Euler--Lagrange expressions we display the scalar coefficient in the canonical frame $\dd u\otimes\dd^dx$ of the target bundle.

The local field derivatives are
\[
\frac{\delta\CL}{\delta\phi}(\phi)=\mathsf Z'(\phi)\dd^dx\,,
\qquad
\frac{\delta\CL}{\delta(\partial_\mu\phi)}(\phi)
=\partial_\mu\phi\,\dd^dx\,.
\]
Thus, the formal perturbative Lagrangian is
\begin{equation*}
\CL_{\mathrm{pert},\phi}(h)
={}\CL(\phi)
+\bigl(\nabla\phi\cdot\nabla h+\mathsf Z'(\phi)h\bigr)\dd^dx\\
+\tfrac12\bigl(|\nabla h|^2+\mathsf Z''(\phi)h^2\bigr)\dd^dx
+\sum_{n\geq3}\tfrac{1}{n!}\mathsf Z^{(n)}(\phi)h^n\,\dd^dx \,.
\end{equation*}
The kinetic term contributes only in degrees one and two, while all higher ``vertices'' come from the potential. Its derivative contribution is precisely what the local field derivatives with nonempty jet index encode.

Writing
\[
\Delta:=\sum_\mu\partial_\mu^2\,,\qquad
\mathcal E \CL (\phi) :=-\Delta\phi+\mathsf Z'(\phi)\,,\qquad
K_\phi:=-\Delta+\mathsf Z''(\phi) \, ,
\]
and since the torus has no boundary, integration by parts gives the perturbative action
\begin{equation}\label{PerturbationEqTorusPerturbativeAction}
S_{\mathrm{pert},\phi}(h)
={}S(\phi)+\int_M\mathcal E\CL(\phi) h\,\dd^dx
+\tfrac12\int_M hK_\phi h\,\dd^dx\\
+\sum_{n\geq3}\tfrac1{n!}\int_M\mathsf Z^{(n)}(\phi)h^n\,\dd^dx \,.
\end{equation}
These are formal series for a general smooth potential, and so a well-defined evaluation at an ordinary nonzero $h$ is not asserted. Their full symmetric multilinear coefficients are
\begin{equation*}
\begin{split}
S^{(1)}_\phi(h)&=\int_M\mathcal E \CL (\phi) h\,\dd^dx\,,\\
S^{(2)}_\phi(h_1,h_2)
&=\int_M\bigl(\nabla h_1\cdot\nabla h_2+\mathsf Z''(\phi)h_1h_2\bigr)\dd^dx
=\int_M h_1K_\phi h_2\,\dd^dx\,,\\
S^{(n)}_\phi(h_1,\ldots,h_n)
&=\int_M\mathsf Z^{(n)}(\phi)h_1\cdots h_n\,\dd^dx
\qquad(n\geq3)\,.
\end{split}
\end{equation*}
As usual, if $\phi$ is on shell $\mathcal E \CL (\phi)=0$ and the linear term in $h$ vanishes. After subtracting $S(\phi)$, the quadratic term is the ``free action'' and the terms of degree at least three are the ``interactions''. Of course, the linearized Euler--Lagrange equation $K_\phi h=0$ is nothing but the Jacobi equation from Cor. \ref{OnshellSyntheticTangentBundle}.

\textit{Formal} variation of \eqref{PerturbationEqTorusPerturbativeAction} in a further direction $k$ gives
\[
\frac{\partial}{\partial t} S_{\mathrm{pert},\phi}(h +t k) \big\vert_{t=0}
=\int_M\bigg(
\mathcal E \CL (\phi) +K_\phi h+
\sum_{n\geq2}\tfrac{1}{n!}{\mathsf Z^{(n+1)}(\phi)}h^n
\bigg)k\,\dd^dx\, ,
\]
and so the full perturbative Euler--Lagrange equation is
\begin{equation*}
\mathcal \CE \CL (\phi)+K_\phi h+
\sum_{n\geq2}\tfrac{1}{n!}{\mathsf Z^{(n+1)}(\phi)}h^n=0\, .
\end{equation*}
Equivalently, it is the formal expansion perturbative  expansion of $\CE \CL(\phi +h) = -\Delta(\phi+h)+\mathsf Z'(\phi+h)$. The operator coefficients are $h\mapsto K_\phi h$ in degree one and
\[
(h_1,\ldots,h_n)\; \longmapsto \;
\mathsf Z^{(n+1)}(\phi)\,h_1\cdots h_n
\qquad(n\geq2)\, ,
\]
with constant coefficient $\mathcal E \CL(\phi)$.

We now show how these formulas encode -- and are themselves encoded -- in the \textit{actual infinitesimal restriction} of the action functional and Euler--Lagrange operator.
Let $\CO(\DD)=\FR\oplus V$ with $V^{N+1}=0$, and take a plot
\[
\xi\in(T_\phi\CF)(\FR^k)\otimes V
=C^\infty(\FR^k\times M)\otimes V\, .
\]
By Cor. \ref{PertubativeActionForVectorAffineBundles}, Eq. \eqref{PerturbationEqTorusPerturbativeAction} evaluates to the exact identity
\begin{equation*}
\bigl(S|_{\DD_{\phi,\CF}}\circ\Psi_\phi\bigr)_{\FR^k\times\DD}(\xi)
={}S(\phi)+\int_M\mathcal E\CL (\phi)\xi\,\dd^dx
+\tfrac12\int_M\xi K_\phi\xi\,\dd^dx\\
+\sum_{n=3}^{N}\tfrac1{n!}\int_M\mathsf Z^{(n)}(\phi)\xi^n\,\dd^dx\, .
\end{equation*}
The same applies to the Euler--Lagrange operator (fiber component):
\begin{equation}\label{PerturbationEqTorusExactEL}
\bigl(\CE \CL|_{\DD_{\phi,\CF}}\circ\Psi_\phi\bigr)_{\FR^k\times\DD}(\xi):= (\CE\CL)_{\FR^k\times\DD}(\phi+\xi)
=\mathcal E \CL (\phi)+K_\phi\xi+
\sum_{n=2}^{N}\tfrac{1}{n!}{\mathsf Z^{(n+1)}(\phi)}\xi^n\, .
\end{equation}
Note that if the infinitesimal probe is chosen to be of nilpotent order 1, then the quadratic term also vanishes. All operations are taken coefficient-wise with the smooth probe parameters held fixed. Spacetime derivatives preserve the nilpotent ideal, so each higher (finite-) order Taylor remainder vanishes without any analyticity assumption. At an on-shell background, setting the right side of \eqref{PerturbationEqTorusExactEL} to zero gives exactly the actual Euler--Lagrange zero locus restricted to the infinitesimal neighborhood of $\phi$.
\end{example}

\subsubsection*{The case of general field bundles}\label{PerturbationGeneralBundlesSection}

For nonlinear fibers, ordinary addition \(\phi+h\) is replaced by a fiberwise coordinate choice. The following, albeit still somewhat technical formulation, avoids assuming a full infinite-dimensional manifold structure on the full space of sections.

\begin{definition}[\bf fiberwise coordinates along a field configuration]\label{FiberWiseCoordinates}
Let \(\pi:F\to M\) be a smooth fiber field bundle and consider its tangent field bundle
$$
 \phi^* VF \longrightarrow M
$$
over a fixed field configuration $\phi \in \CF(*)$.
A \textit{fiberwise coordinate choice} along \(\phi\) is a diffeomorphism over \(M\)
\[
a\,:\, \CU\xrightarrow{\;\; \sim\;\; }\CV\,,
\qquad \CU\subset \phi^* VF,\quad\CV\subset F\,,
\]
between open neighborhoods of the zero section of the former bundle and the image \(\phi(M)\) of the latter, with \(a(0_x)=\phi(x)\) and \(\dd_{\mathrm{vert}}a|_{0_x}=\id_{(\phi^*VF)_x}\).
\end{definition}
In actual physical examples of field bundles of interest, there exist special (case-dependent) choices of fiberwise coordinates, such as the affine fiberwise coordinates from earlier (Cor. \ref{PertubativeActionForVectorAffineBundles}), but also from nonlinear examples of interest (cf. Ex. \ref{PerturbationExLieGroup}, Ex. \ref{PerturbationExSigmaMetric}, and Ex. \ref{PerturbationExGravity}). Nevertheless, the fact that these exist for any fiber bundle is generic, by choosing a smooth vertical Riemannian metric and restricting its fiberwise exponential map to a sufficiently small neighborhood of the zero section. We provide a proof of this for the sake of generality.
\begin{lemma}[\bf Existence of fiberwise coordinates]\label{ExistenceoOfFiberwiseCoordinates} fiberwise coordinates along a field configuration \textit{exist}. In particular, each smooth vertical Riemannian metric on $F$ yields one such example.
\end{lemma}
\begin{proof}
Choose a smooth (vertical) Riemannian fiberwise pairing on $VF$. Smooth dependence of its fiber geodesics gives
\[
a(x,v)=\exp^{F_x}_{\phi(x)}(v)
\]
on an open subset of $\pi^*VF$ containing every $0_x$. Indeed, the zero initial velocity gives a constant geodesic defined up to time one, and nearby initial data also have solutions up to that time (completeness of the metric is unnecessary). At every $x$, we have $a(0_x)=\phi(x)$ and a linear isomorphism
\begin{align*}
\dd a_{0_x}\;:\;T_xM \oplus (\phi^*VF)_x&\; \xrightarrow{\qquad} \;\;
T_{\phi(x)}F \\
\qquad (u,v)\;\; &\;\xmapsto{\qquad\; } \dd\phi_x(u)+v
\end{align*}
 since $\pi\circ\phi=\id_M$ splits the tangent space into the image of $\dd\phi_x$ and the vertical tangent space. The inverse function theorem therefore supplies an injective local-diffeomorphism neighborhood of each $0_x$.

More explicitly, choose an open cover $\{U_i\}$ of $M$ and constants $\rho_i>0$ such that the fiber balls $\|v\|<\rho_i$ over $U_i$ lie in such neighborhoods. Take a locally finite subordinate partition of unity $\{\lambda_i\}$ and put $\rho(x)=\tfrac12\sum_i\lambda_i(x)\rho_i$. At each $x$, an active index of maximal $\rho_i$ shows that the ball of radius $\rho(x)$ lies in an injectivity neighborhood for that fiber, and in the local-diffeomorphism locus. The resulting open set
\[
\CU=\big\{(x,v):\|v\|<\rho(x)\big\}
\]
contains the entire zero section. Since $a$ preserves $x$, injectivity in each fiber implies injectivity on $\CU$. Hence $a:\CU\to\CV:=a(\CU)$ is a diffeomorphism onto an open subset and $\phi(M)=a(0_M)\subset\CV$. No uniform positive radius on $M$ is assumed, even when $M$ is noncompact.
\end{proof}

These coordinates are defined along (and around) the entire image of background field $\phi$ over $M$. Passing to sections of these field bundles parametrizes exactly those fields whose values lie in $\CV$, using sections of $\pi^*VF$ with values in $\CU$. 
In other words, such coordinates yield in particular an injection of some neighborhood inside the tangent space at $\phi \in \CF$, 
$$
\CO \subset T_\phi \CF(*) \longhookrightarrow \CF(*) 
$$
whose image includes $\phi \in \CF(*)$, hence supplying an interpretation for ``local addition formulas'' inside $\CF(*)$! The plausible 
infinite-dimensional manifold theory implications of such coordinates are immaterial for our purposes. The crucial point is that these induce isomorphisms of thickened smooth sets between the corresponding infinitesimal neighborhoods.

\begin{proposition}[\bf Isomorphism on the field infinitesimal neighborhood]\label{PerturbationLemChart}
$\,$
\begin{itemize}
\item[\bf(i)]
A fiberwise coordinate choice induces an isomorphism
\begin{align}\label{PerturbationEqFormalchart}
\Psi_a\,:\, \DD_{0,T_\phi \CF}&\xrightarrow{\quad   \sim\quad }\DD_{\phi,\CF}
\\[-2pt]
 \xi &\quad \longmapsto \quad  a\circ\xi\, , \nn
\end{align}
which depends only on the germ of \(a\) along zero.
\item[\bf (ii)]
After potentially shrinking \(\CU\), if need be, the germ of \(a\) has a smooth extension \(\bar a:\phi^*VF\to F\) over \(M\).
\end{itemize}
\end{proposition}

\begin{proof}
A plot
\(
\xi\in
\mathbb D_{0,T_\phi \CF}
(\mathbb R^k\times\mathbb D)
\) is a thickened section represented by a map over \(M\),
\[
\xi:\mathbb R^k\times\mathbb D\times M\longrightarrow \phi^*VF\, ,
\] whose restriction to the reduced probe is
$
\xi^{k,0}(s,x)=0_x. $ Because \(\mathcal{U} \hookrightarrow \phi^*VF\) contains every \(0_x : x\in M\) and $\CU$ is open, this map factors uniquely through \(\mathcal U\)
\[
\xi:\mathbb R^k\times\mathbb D\times M
\longrightarrow\mathcal U\longhookrightarrow \phi^*VF\, .
\]
Indeed, restricting the domain of a smooth function to an \textit{open} neighborhood imposes a condition on the underlying points; nilpotent thickening introduces no additional underlying points.

We may therefore compose with \(a\) and define
\[
(\Psi_a)_{\mathbb R^k\times\mathbb D}(\xi):=a\circ\xi\ \,
\] with underlying $\FR^k$-plot
\(
(a\circ\xi)^{k,0}(s,x)=a(0_x)=\phi(x),
\) and so it indeed belongs to
\(
\mathbb D_{\phi,\mathcal F}(\mathbb R^k\times\mathbb D).
\) To make the composition explicit, we may choose bundle coordinates over an open subset \(U\subset M\), so that the components of \(\xi\) are given by elements
\(
\xi^b\in C^\infty(\mathbb R^k\times U)\otimes V.
\) Writing the fiber components of \(a\) as \(a^c(x,v)\), smooth Weil evaluation gives
\begin{align*}
(a\circ\xi)^c(s,x) &
=
\sum_{n=0}^{N}\tfrac1{n!}
\big(d_v^na^c\big)_{(x,0)}
\bigl(\xi(s,x),\ldots,\xi(s,x)\bigr) 
\\[-3pt]
&= 
\phi^c(x)
+\sum_{n=1}^{N}\tfrac1{n!}
\sum_{b_1,\ldots,b_n}
\frac{\partial^na^c}
{\partial v^{b_1}\cdots\partial v^{b_n}}(x,\phi(x))\,
\xi^{b_1}\cdots\xi^{b_n}\, .
\end{align*}
The constant term is \(a^c(x,\phi(x))=\phi^c(x)\), while every other term is nilpotent. This is an exact finite expression, since the remainder vanishes because \(V^{N+1}=0\).

Conversely, by the same argument, we may safely define
$$
\big(\Psi^{-1}_a\big)_{\FR^k\times \DD}(\psi) \, := \, a^{-1} \circ \psi \, ,
$$
where $\psi$ is any infinitesimal neighborhood plot of $\phi$, so that $\psi^{k,0}= \phi$. Indeed, this map is manifestly the plot-wise inverse of $\Psi_a$. Again, locally, we may write
\[
\psi^c=\phi^c+\zeta^c,
\qquad
\zeta^c\in C^\infty(\mathbb R^k\times U)\otimes V\, ,
\] 
so that 
\begin{align*}
\big(a^{-1}\circ\psi\big)^c
&=
\sum_{n=1}^{N}\tfrac1{n!}
\big(d_u^n (a^{-1})^c\big)_{(x,\phi(x))}
(\zeta,\ldots,\zeta) 
\\
&= \sum_{n=1}^{N}\tfrac1{n!}
\sum_{d_1,\ldots,d_n}
\frac{\partial^n(a^{-1})^c}
{\partial u^{d_1}\cdots\partial u^{d_n}}(x,0)\,
\zeta^{d_1}\cdots\zeta^{d_n} \, ,
\end{align*}
where the constant term vanishes because \(a^{-1}(x,\phi(x))=0\).

Next, if two coordinate maps agree on an open neighborhood of the zero section, every \(\xi\) above factors through that neighborhood. Their compositions with \(\xi\) therefore agree on every probe, which proves germ dependence.

For the extension, choose a bundle metric on \(\phi^*VF\) and a positive smooth radius \(r(x)\) such that the fiberwise ball of radius \(r(x)\) lies in \(\CU\). Such a radius exists by paracompactness. Choose a smooth function \(\chi:[0,\infty)\to[0,1]\) equal to one near zero and equal to zero for \(s\geq1/4\). The map
\[
c(x,v)=\chi\!\big(\|v\|^2/r(x)^2\big)\,v
\]
has image in \(\CU\) and equals the identity near the zero section, and so \(\bar a=a\circ c:\phi^*VF \to F\) is the required globally defined extension.\footnote{This construction also works when \(M\) is noncompact.} Plot-wise, every infinitesimal \(\xi\) factors through the neighborhood where \(c=\mathrm{id}\). Therefore
\[
\bar a\circ\xi=a\circ c\circ\xi=a\circ\xi\, ,
\] and so the global extension induces exactly the same map on the entire infinitesimal neighborhood, independently of its behavior away from zero.
\end{proof}

We are now sufficiently equipped to prove the nonlinear generalization of Prop. \ref{PerturbationPropMain} and Prop. \ref{PerturbationPropLocal}.
\begin{theorem}[\bf Nonlinear fields and infinitesimal restriction]
\label{PerturbationThmNonlinear}\label{PerturbationCorNonlinear}
Let $F\to M$ be a smooth fiber bundle and choose fiberwise coordinates $a:\CU\subset \pi^* VF \xrightarrow{\sim}\CV\subset F$ along $\phi\in \CF(*)$, with a germ extension $\bar a:\phi^*VF \to F$ as in Prop.~\ref{PerturbationLemChart}. Then:
\begin{itemize}

\item[\bf (i)] For a morphism $Q:\CF\to\CG$, with $\CG=\FR$ or $\CG=\mathbold{\Gamma}_M(G)$ for a fixed vector bundle $G\to M$, every probe $\FR^k\times\DD$, $V^{N+1}=0$, and every $\xi\in\DD_{0,T_\phi \CF}(\FR^k\times\DD)$ satisfy the identity
\begin{equation}\label{PerturbationEqNonlinear}
\bigl(Q|_{\DD_{\phi,\CF}}\circ \Psi_a\bigr)_{\FR^k\times\DD}(\xi)
=\sum_{n=0}^{N}\tfrac{1}{n!}
\bigl(Q\circ\mathbold{\Gamma}_M(\bar a)\bigr)^{(n)}_{0,\FR^k\times\DD}(\xi,\ldots,\xi).
\end{equation}
The Taylor coefficients depend only on the germ of $a$ and determine, and are determined by, the restricted morphism in these coordinates.

\item[\bf (ii)] Local differential operators retain their order under this precomposition. In particular, this applies to finite-order differential operators $\CP=P_r\circ j^r$ with $P_r:J^r_MF\to G$, to local Lagrangians $\CL=L_r\circ j^r$ with values in $\wedge^dT^*M$, and to their scalar-valued action functionals $S_\Sigma=\int_\Sigma\CL$ on compact oriented domains. For these cases the transformed expressions are
\[
\CP_a=\big(P_r\circ J^r_M\bar a\big)\circ j^r\,,\qquad 
\CL_a=\big(L_r\circ J^r_M\bar a\big)\circ j^r\,,\qquad
S_{\Sigma,a}=\int_\Sigma\CL_a \,.
\]
The differential operators $\CP_a$ and $\CL_a$ have order at most $r$; $S_{\Sigma,a}$ is a scalar-valued functional, and their Taylor coefficients are those of Prop.~\ref{PerturbationThmLocal} on the vector bundle $\phi^*VF$.
\end{itemize}
\end{theorem}

\begin{proof}

The extension $\bar a$ gives a globally defined morphism
\[
\mathbold{\Gamma}_M(\bar a)\;:\;
T_\phi \CF\longrightarrow \CF 
\]
whose restriction to $\DD_{0,T_\phi \CF}$ is $\Psi_a$: 
 \[
	\xymatrix@R=1.3em@C=2em  {\DD_{0,T_\phi \CF} \ar[d] \ar[rr]^{\Psi_a} &&   \DD_{\phi, \CF} \ar[d] 
		\\ 
		T_\phi \CF \ar[rr]^-{\mathbold{\Gamma}_M(\bar{a})}  && \CF
	\, , }
	\]
\noindent
Applying then Prop.~\ref{PerturbationPropMain} to $Q\circ\mathbold{\Gamma}_M(\bar a)$ at zero proves \eqref{PerturbationEqNonlinear}. Two extensions of the same germ give the same formal restriction by Lem.~\ref{PerturbationLemChart}, so their Taylor coefficients agree by \eqref{PerturbationEqRecover}. For an operator induced by $P_r:J^r_MF\to G$, the composite is induced by
\[
J^r_M \phi^* VF\xrightarrow{\quad J^r_M\bar a \quad}J^r_MF\xrightarrow{\;\;P_r\;\;}G,
\]
since jet prolongation commutes with composition by a bundle map. The resulting differential operator
thus has order at most $r$, and Prop.~\ref{PerturbationThmLocal} applies. The same holds for Lagrangians and their compact-domain integrals.
\end{proof}

\begin{remark}[\bf Extent of  ``coordinates'' on field space]
We stress that so far the only coordinate choice in \cref{PerturbationGeneralBundlesSection} is a \emph{finite-dimensional fiberwise} identification \(a:\CU\subset \phi^*VF\xrightarrow{\sim}\CV\subset F\). Passing to sections parametrizes fields as \(h\mapsto a\circ h\), for sections whose values lie in \(\CU\). Prop.~\ref{PerturbationLemChart} guarantees that this induces an isomorphism
\begin{equation*}
\Psi_a:\DD_{0,T_\phi\CF}\xrightarrow{\;\;\sim\;\;}\DD_{\phi,\CF}\, ,
\end{equation*}
identifying infinitesimal neighborhoods in infinite-dimensional spaces. It \textit{does not} require an ordinary manifold chart between open subsets of \(T_\phi\CF\) and \(\CF\). To assert such an open-chart interpretation, one would additionally specify an appropriate infinite-dimensional topology and smooth structure. Our construction uses only the fiberwise coordinates and their induced maps of thickened smooth sets. The coordinate choice is global along $\phi(M)$ over spacetime, and covers the entire infinitesimal neighborhood at $\phi$, on every probe and at every finite nilpotence order. It does \textit{not} identify the whole field space $\CF$ with $T_\phi\CF$. The extension $\bar a$ is defined on the whole tangent field bundle, but need not be invertible there. For vector or affine field bundles, translation by \(\phi\) naturally provide a canonical corresponding identification of infinitesimal neighborhoods (Cor. \ref{PertubativeActionForVectorAffineBundles}).
\end{remark}

\begin{example}[\bf Lie group sigma-model targets]\label{PerturbationExLieGroup}
Let $F=M\times G$, where $G$ is a Lie group with Lie algebra $\frg$, and identify a section $\phi\in \Gamma_M(M\times G)$ with the corresponding map $\phi:M\to G$. The vertical bundle is $VF=M\times TG$, and pushforward by left multiplication gives
\begin{align*}
\phi^*VF\cong\phi^*TG\;&\xrightarrow{\quad \sim\quad }\;\;\; M\times\frg \\
(x,\,v)&\quad \longmapsto\;\; \bigl(x,\,\dd L_{\phi(x)^{-1}}v\bigr)\, ,
\end{align*}
where $L_g$ denotes left multiplication by $g\in G$. The group exponential then supplies the following globally defined fiberwise map
\begin{align}\label{PerturbationEqLieGroupChart}
a_\phi \;:\; \phi^*VF& \; \xrightarrow{\qquad} \quad F
\\
(x,\,v)& \;\;\; \longmapsto \bigl(x,\,\phi(x)\cdot \exp_G(\dd L_{\phi(x)^{-1}}v)\bigr)\, . \nn
\end{align}
Since $\dd(\exp_G)_0=\id_{\frg}$, one has $a_\phi(0_x)=\phi(x)$ and $\dd_{\mathrm{vert}}a_\phi|_{0_x}=\id_{T_{\phi(x)}G}$. Choose an open neighborhood $\Omega\subset\frg$ of zero on which $\exp_G$ is a diffeomorphism onto its image, so that in the left trivialization the restriction of \eqref{PerturbationEqLieGroupChart} is a diffeomorphism from $M\times\Omega$ onto
\[
\CV_\phi=\big\{(x,g)\in M\times G\;\; :\;\; \phi(x)^{-1}\cdot g\in\exp_G(\Omega)\big \}\, ,
\]
with inverse $(x,g)\mapsto(x,\log_G(\phi(x)^{-1}g))$. This gives fiberwise coordinates along the entire image of $\phi$, with the same $\Omega$ at every $x$, even when $M$ is noncompact, and formula \eqref{PerturbationEqLieGroupChart} already provides the global extension required in Thm.~\ref{PerturbationThmNonlinear}.
Thus \eqref{PerturbationEqLieGroupChart} is canonical relative to the given group structure on $G$ and the background $\phi$, without invoking an auxiliary vertical metric. 
\footnote{The resulting map on $\phi^*VF$ does not depend on whether left or right trivialization is used. Indeed, for $v\in T_gG$, put $X=\dd L_{g^{-1}}v$ and $Y=\dd R_{g^{-1}}v=\operatorname{Ad}_gX$. Equivariance of the exponential under conjugation gives $
g\cdot \exp_G X=\exp_G(\operatorname{Ad}_gX)\cdot g=\exp_G Y\cdot g$.}

Under left trivialization an ordinary tangent section $h\in \Gamma_M(\phi^*VF)$ is represented by $X(x)=\dd L_{\phi(x)^{-1}}h(x)$, and similarly the ``fluctuating'' field $\psi$ is represented by $\phi\cdot \exp_G X$. On a probe $\FR^k\times\DD$, with $\CO(\DD)=\FR\oplus V$, the same trivialization identifies nilpotent inputs plots with
\[
\xi\in C^\infty(\FR^k\times M,\frg)\otimes V.
\]
The induced isomorphism of infinitesimal neighborhoods and its inverse are
\begin{align}\label{PerturbationEqLieGroupRestriction}
\Psi_{a_\phi}:\DD_{0,T_\phi\CF}&\xrightarrow{\quad \sim\quad } \;\; \DD_{\phi,\CF} \nn
\\[-2pt] 
\qquad \xi \quad &\;\;\; \longmapsto\;\; \phi\exp_G(\xi)\,,
\end{align}
and $\Psi_{a_\phi}^{-1}(\psi)=\log_G(\phi^{-1}\psi)$ (using the actual thickened extensions of the smooth maps, by multilinearity).  For a local action $S_\Sigma$, the corresponding formal series is the familiar expansion
\[
S_{\Sigma,\mathrm{pert},\phi}^{a_\phi}(h)
=\sum_{n\geq0}\tfrac1{n!}
\left.\partial_t^n S_\Sigma\bigl(\phi\exp_G(tX)\bigr)\right|_{t=0}.
\]
Its nilpotent evaluation is exactly the restricted action by Thm.~\ref{PerturbationThmNonlinear}. A target metric may enter the action, such as the most studied case of the Principal Chiral Model (see e.g. \cite{Lac22}), but is unnecessary for this coordinate construction. 
\end{example}

\begin{example}[\bf Semi-Riemannian sigma-model targets]\label{PerturbationExSigmaMetric}
Let $F=M\times N$ for a target $(N,\gamma)$ with a (semi-)Riemannian metric, regarded as part of the sigma-model data. For a background $\phi:M\to N$, one has $\phi^*VF\cong\phi^*TN$. The Levi--Civita connection of $\gamma$ supplies
\begin{equation}\label{PerturbationEqSigmaMetricChart}
a_\phi(x,v)=\bigl(x,\exp^\gamma_{\phi(x)}v\bigr)\,,
\qquad
\Psi_{a_\phi}(\xi)=\exp^\gamma_\phi(\xi)\,.
\end{equation}
The geodesic equation is a smooth ordinary differential equation for either signature. Its exponential is therefore defined on an open neighborhood of the zero section, satisfies $a_\phi(0_x)=\phi(x)$ and $\dd_{\mathrm{vert}}a_\phi|_{0_x}=\id$, and restricts to fiberwise coordinates as in Def.~\ref{FiberWiseCoordinates}.

For positive-definite $\gamma$, this is exactly the existence construction above with the vertical metric $\gamma$ induced by the target data. For indefinite $\gamma$, an auxiliary positive-definite bundle metric may be used to perform the shrinking and cutoff in Lem.~\ref{PerturbationLemChart}, since the indefinite quadratic form itself cannot define the required small fiber balls. These auxiliary choices do not change the germ of \eqref{PerturbationEqSigmaMetricChart}, hence do not change its induced infinitesimal isomorphism or perturbative coefficients. The coordinates are thus \textit{canonical} relative to the given target metric. Fluctuations are sections of $\phi^*TN$, which need not be a trivial bundle. More generally, the same construction applies to a target equipped with some affine connection.
\end{example}

Two different choices of fiberwise coordinates $a,b$ yield distinct Taylor coefficients from the right-hand side of Eq. \eqref{PerturbationEqNonlinear}. But since the left-hand side is expressed via an intrinsic restriction of $Q$ to the infinitesimal neighborhood, naturality allows us to compute the functional relation between two different (coordinate) perturbative expansions.

\begin{corollary}[\bf Change of fiberwise coordinates]
\label{PerturbationCorTaylorCoordinates}
In the setting of Thm.~\ref{PerturbationThmNonlinear}, choose two fiberwise coordinates $a,b$ along $\phi$, both normalized to have identity vertical derivative at zero. With germ extensions $\bar a,\bar b$, write
\[
Q_a:=Q\circ\mathbold{\Gamma}_M(\bar a),\qquad
Q_b:=Q\circ\mathbold{\Gamma}_M(\bar b),\qquad
Q_{a,n}:=(Q_a)^{(n)}_0,\quad Q_{b,n}:=(Q_b)^{(n)}_0 .
\]
Put $c=a^{-1}\circ b$, with original domain $b^{-1}(\CV_a\cap \CV_b)\subset \phi^*VF$, as a germ along the zero section. Denote by 
$$
C\, :=\,\Psi_a^{-1}\circ \Psi_b \, : \DD_{0,T_\phi \CF} \xrightarrow{\quad \sim \quad} \DD_{0,T_\phi \CF}
$$
the corresponding automorphism on infinitesimal neighborhoods. Choose a smooth extension $\bar c:\phi^*VF\to \phi^*VF$ of this germ, and set
\[
C_m:=\bigl(\mathbold{\Gamma}_M(\bar c)\bigr)^{(m)}_0,\qquad
C_0=0,\qquad C_1=\id\, .
\]
Then:
\begin{itemize}
\item[{\bf (i)}]
These coefficients are independent of the extensions. More explicitly, they are pointwise, symmetric multilinear maps on sections of $\phi^*VF$
\begin{align*}
C_m(h_1,\ldots,h_m)(x)
\, &=\, \big(\dd_v^m c\big)_{(x,0)}\bigl(h_1(x),\ldots,h_m(x)\bigr) \\
&=\, \frac{\partial^m c}
{\partial v^{\alpha_1}\cdots\partial v^{\alpha_m}}(x,0)\,
h_1^{\alpha_1}(x)\cdots h_m^{\alpha_m}(x)\, ,
\end{align*}
and their internal extensions act on arbitrary parametrized sections as well.
\item[{\bf (ii)}]
 The actual restricted maps satisfy
\begin{equation}\label{PerturbationEqTransition}
Q|_{\DD_{\phi,\CF}}\circ \Psi_b
=\bigl(Q|_{\DD_{\phi,\CF}}\circ \Psi_a\bigr)\circ C \, ,
\end{equation}
which in terms of the formal series
\[
\widehat Q_a(h):=\sum_{n\geq0}\tfrac1{n!}Q_{a,n}(h,\ldots,h),
\qquad
\widehat C(h):=\sum_{m\geq1}\tfrac1{m!}C_m(h,\ldots,h),
\]
and similarly $\widehat Q_b$, is precisely the identity
\begin{equation}\label{PerturbationEqTaylorSubstitution}
\widehat Q_b(h)
=\widehat Q_a\bigl(\widehat C(h)\bigr)
=\sum_{k\geq0}\tfrac1{k!}
Q_{a,k}\Bigl(\widehat C(h),\ldots,\widehat C(h)\Bigr).
\end{equation}
\item[{\bf (iii)}]
The degree-zero term is $Q_{b,0}=Q_{a,0}=Q(\phi)$. For $n\geq1$, the full (polarized) coefficient formula is recovered by
\begin{equation}\label{PerturbationEqTaylorPartitions}
Q_{b,n}(h_1,\ldots,h_n)
=
\sum_{\pi\in\operatorname{Part}(\{1,\ldots,n\})}
Q_{a,|\pi|}
\Bigl(\bigl(C_{|B|}(h_B)\bigr)_{B\in\pi}\Bigr).
\end{equation}
Here $\operatorname{Part}(\{1,\ldots,n\})$ is the set of partitions into nonempty blocks, and $h_B$ is the list of arguments $h_i$ with $i\in B$. Each block contributes one argument to the outer coefficient. Symmetry makes the order within each block and the order of blocks immaterial. These are identities of the symmetric multilinear morphisms on $\mathbold{\Gamma}_M(\phi^*VF)$, hence hold on every probe.
\end{itemize}
\end{corollary}

\begin{proof}
The germ $c$ fixes zero $c(0_x)=a^{-1}\big(b(0_x)\big)=a^{-1}(\phi(x))=0$ and has vertical derivative $\id$, by the chain rule and the normalization of $a,b$. Its extension $\bar{c}: \phi^* VF\rightarrow \phi^*VF$ can be constructed by the same fiberwise cutoff as in Lem.~\ref{PerturbationLemChart}: shrink to a neighborhood on which $c$ is defined, then compose with a smooth map $\phi^*VF$ into that neighborhood which is the identity near zero. Lem.~\ref{PerturbationLemChart} and the recovery formula \eqref{PerturbationEqRecover} show that its Taylor coefficients depend only on the germ. Since $c$ is a fiberwise map over $M$, ordinary differentiation at each $x$ gives the displayed pointwise formula for $C_m$.

The identity $\Psi_b=\Psi_a\circ C$ gives \eqref{PerturbationEqTransition}, while on a probe $\FR^k\times\DD$ with $V^{N+1}=0$, Thm.~\ref{PerturbationThmNonlinear} and Prop.~\ref{PerturbationPropMain}, applied also to $\mathbold{\Gamma}_M(\bar c)$, give
\[
\begin{split}
C(\xi)&=\sum_{m=1}^{N}\tfrac1{m!}C_m(\xi,\ldots,\xi),\\
\sum_{n=0}^{N}\tfrac1{n!}Q_{b,n}(\xi,\ldots,\xi)
&=\sum_{k=0}^{N}\tfrac1{k!}Q_{a,k}\bigl(C(\xi),\ldots,C(\xi)\bigr).
\end{split}
\]
Because $C_0=0$, every term of $\widehat C$ has positive degree. Hence each degree in the substituted series receives only finitely many contributions. On this particular probe every term of total degree greater than $N$ vanishes. This establishes the nilpotent meaning of \eqref{PerturbationEqTaylorSubstitution}, without requiring convergence or an identity between the chosen global extensions away from zero.

For the coefficient formula, first take ordinary tangent vectors $h_1,\ldots,h_n$ and the probe
\[
\CO(\DD_n)=\FR[\epsi_1,\ldots,\epsi_n]/(\epsi_1^2,\ldots,\epsi_n^2),
\qquad \xi=\sum_{i=1}^n\epsi_i h_i\,.
\]
Write $\epsi_I=\prod_{i\in I}\epsi_i$. Multilinearity and symmetry give the exact expansion
\[
C(\xi)=
\sum_{\varnothing\neq I\subseteq\{1,\ldots,n\}}
\epsi_I\,C_{|I|}(h_I).
\]
Indeed, each set of $m$ distinct indices occurs $m!$ times in the degree-$m$ ordered expansion, cancelling each $1/m!$ factor, while repeated indices vanish. In the degree-$k$ outer term, the coefficient of $\epsi_1\cdots\epsi_n$ therefore comes from ordered lists of $k$ nonempty, disjoint blocks whose union is $\{1,\ldots,n\}$. Every unordered partition into $k$ blocks occurs $k!$ times, cancelling the outer $1/k!$. The left-hand coefficient is then $Q_{b,n}(h_1,\ldots,h_n)$ by \eqref{PerturbationEqRecover}, which proves \eqref{PerturbationEqTaylorPartitions} and hence the equality of formal Taylor data in \eqref{PerturbationEqTaylorSubstitution}.

For arbitrary $h_i\in\mathbold{\Gamma}_M(\phi^*VF)(\FR^k\times\DD)$, perform the same extraction on the product with $\DD_n$, using its independent square-zero generators. The calculation takes coefficients in $\CG(\FR^k\times\DD)$ and gives the same formula, while naturality of the maps and coefficient extraction proves compatibility with all probe pullbacks.
\end{proof}

\paragraph{First coefficients and the quadratic action.}
Since $C_1=\id$, the first terms of \eqref{PerturbationEqTaylorPartitions} are
\[
\begin{split}
Q_{b,1}(h)\, &=\, Q_{a,1}(h)\,,\\
Q_{b,2}(h_1,h_2)&=\,Q_{a,2}(h_1,h_2)
+Q_{a,1}\bigl(C_2(h_1,h_2)\bigr)\,,\\
Q_{b,3}(h,h,h)&=\, Q_{a,3}(h,h,h)
+3Q_{a,2}\bigl(h,C_2(h,h)\bigr)
+Q_{a,1}\bigl(C_3(h,h,h)\bigr)\, , \\
Q_{b,4}(h,h,h,h)&=\, Q_{a,4}(h,h,h,h) +6Q_{a,3}\bigl(h,h,C_2(h,h)\bigr)+3Q_{a,2}\bigl(C_2(h,h),C_2(h,h)\bigr)+4Q_{a,2}\bigl(h,C_3(h,h,h)\bigr)\\
&\hspace{2.7cm}+Q_{a,1}\bigl(C_4(h,h,h,h)\bigr)\, .
\end{split}
\]
Thus higher coefficients generally depend on the nonlinear coordinate choice, despite describing the same restricted morphism. For an action $S$ at a critical background, $S_{a,1}=0$, and hence
\[
S_{b,2}=S_{a,2}.
\]
The quadratic vertex is therefore independent of the fiberwise coordinate choice for on-shell backgrounds. The cubic vertex can still change by $3S_{a,2}(h,C_2(h,h))$, and for off-shell backgrounds, even the quadratic vertex acquires the term involving $S_{a,1}$. \footnote{If one allows coordinate choices with a non-identity linear part, the quadratic vertex at a critical background instead transforms by that linear change in its two arguments.}

\begin{example}[\bf Higher perturbative interactions from a free field]\label{PerturbationExFreeReparametrization}
Recall the case of scalar field theory on the flat torus $M$ (Ex. \ref{PerturbativeScalarFieldTheoryExample}), and fix an arbitrary background $\phi\in\CF(*)$, identified with its scalar function on $M$. The corresponding vertical fiber bundle is then
\[
VF\cong M\times T\FR\,,
\qquad
\phi^*VF=\big\{(x,v_x)\; :\; v_x\in T_{\phi(x)}\FR\cong \FR\big\}\,.
\]
In Ex. \ref{PerturbativeScalarFieldTheoryExample}, we used the canonical \textit{affine} fiberwise coordinates, but there exist alternative nonlinear fiberwise coordinates. An example of such a polynomial alternative is $b_\phi :\phi^*VF\longrightarrow F$ below, compared with the original affine identification $a_\phi$:
\[
\begin{aligned}
a_\phi(x,v_x)&=(x,\phi(x)+v_x)\, ,\\
b_\phi(x,v_x)&=\bigl(x,\phi(x)+v_x+\tfrac{1}{2}v_x^2\bigr)\,,
\end{aligned}
\]
Restricting $b_\phi$ to $|v_x|<1/2$ makes it a fiberwise coordinate map, while its displayed full polynomial supplies a global extension of the corresponding germ. Both maps send $0_x$ to $\phi(x)$ and have identity vertical derivative there. Their transition acts on the ``fluctuation'' bundle as
\begin{align*}
c=a_\phi^{-1}\circ b_\phi \;:\; \phi^*VF&\longrightarrow\phi^*VF\, ,
\\
(x,v_x)&\longmapsto \bigl(x,v_x+\tfrac{1}{2}v_x^2\bigr)\, .
\end{align*}
On a probe $\FR^k\times\DD$, the induced maps on infinitesimal neighborhoods are
\begin{equation*}
\begin{aligned}
\Psi_{a_\phi}\, :\, \DD_{0,T_\phi\CF}&\xrightarrow{\quad \sim\quad }\DD_{\phi,\CF}\, ,
&\xi&\longmapsto\phi+\xi\, ,\\
\Psi_c\, :\, \DD_{0,T_\phi\CF}&\xrightarrow{\quad \sim\quad }\DD_{0,T_\phi\CF}\, ,
&\xi&\longmapsto\xi+\tfrac{1}{2}\xi^2\, ,\\
\Psi_{b_\phi}\equiv\Psi_{a_\phi}\circ\Psi_c\, :\, \DD_{0,T_\phi\CF}&\xrightarrow{\quad \sim\quad }\DD_{\phi,\CF}\, ,
&\xi&\longmapsto\phi+\xi+\tfrac{1}{2}\xi^2\, ,
\end{aligned}
\end{equation*}
where $\Psi_c$ is the map $C$ of Cor.~\ref{PerturbationCorTaylorCoordinates}. \footnote{Note that on square-zero probes $\Psi_c$ is the identity. For $\CO(\DD)=\FR[\epsi]/(\epsi^3)$, however, $\Psi_{b_\phi}(\epsi h)=\phi+\epsi h+\tfrac{1}{2}\epsi^2h^2$.}

Consider now the case of $Z(u)= \frac{1}{2} m^2 u^2$ from Ex. \ref{PerturbativeScalarFieldTheoryExample}, i.e., the free non-perturbative action on $\CF= C^\infty(M)$ is 
\[
S(\psi)=\tfrac12\int_M\bigl(|\nabla\psi|^2+m^2\psi^2\bigr)\dd^dx
=\tfrac12\int_M\psi K\psi\,\dd^dx\, ,
\qquad 
\]
 where $K=-\Delta+m^2$ with $ m^2>0$. We write $S_{a_\phi}(h)=S(\phi+h)$ and $S_{b_\phi}(h)=S\big(\phi+h+\tfrac{1}{2}h^2\big)$ for the two \textit{formal perturbative}\footnote{The fact that these specific formulas are convergent when evaluated directly on actual elements $h\in T_\phi \CF$ is due to the polynomial nature of the choice of potential $V(\phi)$. Recall from Ex. \ref{PerturbativeScalarFieldTheoryExample} that this is not necessarily the case, and in general the conceptual distinction is important.} actions around $\phi$, respectively. Integration by parts gives
\[
S_{a_\phi}(h)=S(\phi)+\int_M hK\phi\,\dd^dx+\tfrac12\int_M hKh\,\dd^dx\, ,
\]
which has polynomial degree at most two in $h$. On the other hand, polynomial coordinates give the exact identity
\begin{equation*}
S_{b_\phi}(h)
={}S(\phi) \;+\; \int_M hK\phi\,\dd^dx
\;\;+\;\;\underbrace{\tfrac12\int_M\bigl(hKh+ h^2K\phi\bigr)\dd^dx}_{\text{quadratic}}
\;\;+\;\;\underbrace{\tfrac{1}{2}\int_M h^2Kh\,\dd^dx}_{\text{cubic}}
\;\;+\;\;\underbrace{\tfrac{1}{8}\int_M h^2K(h^2)\,\dd^dx}_{\text{quartic}}\,.
\end{equation*}
In particular, the quadratic vertices satisfy
\[
S_{b_\phi,2}(h_1,h_2)-S_{a_\phi,2}(h_1,h_2)
=\int_M(K\phi)h_1h_2\,\dd^dx\,.
\]
Thus an off-shell background contributes both a linear term and a change of quadratic vertex. If the background is on-shell, $K\phi=0$, these contributions vanish, while the cubic and quartic vertices remain: their diagonal values are respectively $3\int_M h^2Kh\,\dd^dx$ and $3\int_M h^2K(h^2)\,\dd^dx$. See also \cite{KV12} for more on scalar ``field redefinitions'' of free theories and their relation to perturbative quantization.
\end{example}

\begin{remark}[\bf Intrinsic perturbative field theory]\label{IntrinsicPerturbativeFieldTheory}
We stress that on nilpotent input probes, the identity $S_{b_\phi}=S_{a_\phi}\circ\Psi_c$ encodes exactly the fact that both perturbative actions represent the same \textit{actual} infinitesimal neighborhood restricted (non-perturbative) action $S|_{\DD_{\phi,\CF}}$, for any Lagrangian field theory. In other words, the following diagram commutes,
\[
\xymatrix@R=2.1em@C=9em{
  \DD_{0,T_\phi\CF}
    \ar[d]^{\Psi_{a_\phi}}
    \ar[dr]^{S_{a_\phi}}
  &
  \\
  \DD_{\phi,\CF}
    \ar[r]^<<<<<<<<<<<{\left.S\right|_{\DD_{\phi,\CF}}}
  & \FR \, ,
  \\
  \DD_{0,T_\phi\CF}
    \ar[u]_{\Psi_{b_\phi}}
    \ar[ur]_{S_{b_\phi}}
    \ar@/^3pc/[uu]^{\Psi_c}
  &
}
\]
which exhibits the intrinsic / invariant object encoding perturbative field theory to be the restriction of the original local action functional to an infinitesimal neighborhood. Different choices of fiberwise coordinates yield \textit{equivalent} (isomorphic) explicit presentations of the same classical perturbative field theory.
\end{remark}

\begin{example}[\bf Lorentzian metric fields and general relativity]\label{PerturbationExGravity}
Let
\[
E=\operatorname{Sym}^2(T^*M)\,,
\qquad
F=\operatorname{Lor}(M)\subset E
\]
be the open subbundle of symmetric covariant $2$-tensors over $M$ which have a fixed Lorentzian signature, and let $g\in\CF(*)=\Gamma_M(F)$ be a background metric. Openness gives canonical identifications
\[
VF\cong\pi^*E\,,
\qquad g^*VF\cong E\,,
\qquad T_g\CF\cong\mathbold{\Gamma}_M(E)\, ,
\]
so that tangent vector ``fluctuations'' are arbitrary symmetric covariant $2$-tensor fields.
\begin{itemize}
\item[\bf(i)]
The ambient vector-bundle structure supplies the \textit{additive fiberwise coordinates}
\begin{align}\label{PerturbationEqMetricAdditive}
a_g^+ \;:\;\; \CU_g \;\; &\; \longhookrightarrow \quad F\,  
\\[-2pt]
 (x,h)&\; \longmapsto \; (x,g_x+h)\, , \nn 
\end{align}
The domain here is $\CU_g:=\big\{(x,h)\in E \, \, :\, \, g_x+h\in F_x\big\}$, hence is open and contains the entire zero section. The inverse on the image is $(x,q)\mapsto(x,q-g_x)$, so that  in particular, $a_g^+(0_x)=g_x$ and $\dd_{\mathrm{vert}}a_g^+|_{0_x}=\id_{E_x}$. Again, no auxiliary connection or metric is needed here. Although the additive formula need not take all of $E$ into $F$, Lem.~\ref{PerturbationLemChart} supplies a global extension of its germ when needed for Thm.~\ref{PerturbationThmNonlinear}.

For every probe $\FR^k\times\DD$, an infinitesimal symmetric tensor $g+\xi$ has body $g$ and therefore factors through the open subbundle $F$. Thus
\begin{align}\label{PerturbationEqMetricRestriction}
\Psi_{a_g^+}\;:\;\DD_{0,T_g\CF}&\xrightarrow{\quad \sim \quad }\;\; \DD_{g,\CF} 
\\[-2pt]
 \xi&\quad \longmapsto \quad  g+\xi\, . \nn 
\end{align}
Note also that, if $\CO(\DD)=\FR\oplus V$ and $V^{N+1}=0$, the inverse metric is the \textit{exact finite} expression
\begin{equation*}
(g+\xi)^{-1}
=\sum_{j=0}^{N}(-g^{-1}\xi)^j g^{-1}.
\end{equation*}
Here pairings are viewed as maps $TM\to T^*M$, so that $g^{-1}\xi$ is an endomorphism of $TM$, with all formulas evaluated on the same probe. This illustrates the literal nilpotent meaning of the usual inverse-metric expansion.

\item[\bf(ii)] There is also a natural choice of \textit{exponential fiberwise coordinates}, familiar in gravity (cf. \cite{DN15}), whose formula is defined for every ordinary symmetric tensor:
\begin{equation}\label{PerturbationEqMetricExponential}
a_g^{\exp}:E\longrightarrow F\,,
\qquad a_g^{\exp}(x,h)=\bigl(x,g_x\exp(g_x^{-1}h)\bigr)\,,
\end{equation}
where again we view the $(0,2)$ tensors as maps $TM \rightarrow T^*M$, and reversely for $(0,2)$-tensors, so $g_x^{-1} h : T_xM\rightarrow T^*_xM \rightarrow T_xM$. The exponential here is given by the usual series, so that $\exp(g_x^{-1} h) : T_xM\rightarrow T_xM$, and composition with $g_x$ on the left means the pairing 
$$
g_x \exp(g_x^{-1}h) (u,v) \, = \, g_x \big( \exp(g_x^{-1}h) u, v \big)
$$
This is symmetric since $A=g_x^{-1}h$ is $g$-self-adjoint \footnote{By standard linear algebra $g(g^{-1} h u,v) = h(u,v)=h(v,u)= g(Av,u)= g(u,Av)$, where the first $h$ is meant as a map $TM\rightarrow T^*M$, while the latter two as a symmetric pairing.} and so the corresponding pairing may be written
\[
\bigl(g\exp A\bigr)(u,v)
=g\bigl(\exp(A/2)u,\exp(A/2)v\bigr).
\]
Since the endomorphism $\exp(A/2)$ is invertible, $\bigl(g\exp A\bigr)$ necessarily has the signature of $g$. Moreover $a_g^{\exp}(0_x)=g_x$ and
\[
\left.\partial_t\bigl(g\exp(tg^{-1}h)\bigr)\right|_{t=0}=h, 
\]
hence \eqref{PerturbationEqMetricExponential} restricts to fiberwise coordinates near zero and already provides their global extension. No claim of global invertibility is needed. 
\end{itemize}
 The coordinate change from exponential to additive fluctuations \footnote{As an aside, note that both the additive and exponential constructions commute with pullback under spacetime diffeomorphisms.} is
\begin{equation}\label{PerturbationEqMetricTransition}
c=\big(a_g^+\big)^{-1}\circ a_g^{\exp},
\qquad
c(h)=g\bigl(\exp(g^{-1}h)-\id\bigr)
=h+\tfrac12hg^{-1}h+\tfrac1{3!}hg^{-1}hg^{-1}h+\cdots.
\end{equation}
Its first two polarized coefficients are
\[
C_1=\id,
\qquad
C_2(h_1,h_2)=\tfrac12\bigl(h_1g^{-1}h_2+h_2g^{-1}h_1\bigr)\, ,
\]
and thus, for any fixed-target local expression $Q$ in Thm.~\ref{PerturbationThmNonlinear}, Cor.~\ref{PerturbationCorTaylorCoordinates} gives
\[
\widehat Q_{\exp}=\widehat Q_+\circ\widehat C,
\qquad
Q_{\exp,2}(h_1,h_2)
=Q_{+,2}(h_1,h_2)+Q_{+,1}\bigl(C_2(h_1,h_2)\bigr).
\]
In particular, for a gravitational action at an on-shell background, the quadratic vertices in these two parametrizations agree, as expected. Higher vertices are related by the same substitution law, and on nilpotent inputs the series \eqref{PerturbationEqMetricTransition} terminates. Hence, as stressed in Rem. \ref{IntrinsicPerturbativeFieldTheory}, both fiberwise coordinates represent the \textit{same actual infinitesimal restriction}, including for the local Lagrangian and compact-domain action of general relativity under the hypotheses already stated. The Euler--Lagrange section is treated by the covector transformation in the next corollary.
\end{example}

Cor. \ref{PerturbationCorTaylorCoordinates} applies to the fixed-target operators such as $\CP$, $\CL$ and $S_\Sigma$. Operators which are sections of bundles of field spaces, such as the Euler--Lagrange operator, however, are not covered by this case. This is because its target is not a fixed vector field bundle for each field input $\psi$, and hence the naive Taylor expansion is a priori ill-defined. The choice of fiberwise coordinates $a$ addresses the problem by identifying the varying target field bundles with the fixed background field target bundle $(\phi^* VF)^*\otimes \wedge^d T^*M$.

\begin{corollary}[\bf Euler--Lagrange sections and infinitesimal restriction]
\label{PerturbationCorEulerLagrange}
In the setting of Thm.~\ref{PerturbationThmNonlinear}, retain the fiberwise coordinates $a$, the extension $\bar a : \phi^*VF \rightarrow F$ and the vector bundle $\phi^*VF$.
Let $L:J^\infty_M F\to\wedge^d T^*M$ be a Lagrangian and let $\CE\CL:\CF\to T^*_{\mathrm{var}}\CF$ be its Euler--Lagrange section, with target bundle $B$ from \eqref{PerturbationEqEltarget}. Put
\[
E_\phi:=(\phi^*VF)^*\otimes\wedge^dT^*M\cong\phi^*B\,,\qquad
J_a(v):=(\dd_{\mathrm{vert}}a)_v\;:\;(\phi^*VF)_x\xrightarrow{\;\;\sim\;\;}V_{a(v)}F
\quad(v\in\CU_x)\,.
\]
\begin{itemize}
\item[{\bf (i)}]
The Jacobian pulls the varying covector fiber back to the fixed bundle $E_\phi$. On ordinary tangent sections with values in $\CU\subset \phi^*VF$, define
\begin{equation}\label{PerturbationEqElpullback}
\widetilde{\CE\CL}_a(h)(x)
\,:=\, J_a(h(x))^*\bigl(\CE\CL(a\circ h)(x)\bigr)
\quad\in \quad (E_\phi)_x ,
\end{equation}
where the dual Jacobian acts trivially on the spacetime top-form factor. The same construction with $\bar a$ gives a globally defined differential operator
\[
\widetilde{\CE\CL}_{\bar a}\;:\;
\mathbold{\Gamma}_M(\phi^*VF)\longrightarrow\mathbold{\Gamma}_M(E_\phi)
\]
extending \eqref{PerturbationEqElpullback} as a germ. With
\[
L_a:=L\circ J^\infty_M\bar a\, ,\qquad \CL_a:=L_a\circ j^\infty\, ,
\]
it is exactly the fixed-target fiber component of the Euler--Lagrange operator of the transformed Lagrangian:
\begin{equation}\label{PerturbationEqEltransformed}
\widetilde{\CE\CL}_{\bar a}\,=\, \CE(\CL_a)\, .
\end{equation}
\item[{\bf (ii)}]
Writing $\widetilde{\CE\CL}_{a,0}^{(n)}$ for its Taylor coefficients at zero, every probe $\FR^k\times\DD$ with $V^{N+1}=0$ and every $\xi\in\DD_{0,T_\phi \CF}(\FR^k\times\DD)$ satisfy
\begin{equation}\label{PerturbationEqElrestriction}
\begin{split}
J_a(\xi)^*\,(\CE\CL)_{\FR^k\times\DD}\bigl(\Psi_a(\xi)\bigr)
&=\bigl(\widetilde{\CE\CL}_{\bar a}\bigr)_{\FR^k\times\DD}(\xi)\\
&=\sum_{n=0}^{N}\tfrac1{n!}
\Bigl(\widetilde{\CE\CL}_{a,0}^{(n)}\Bigr)_{\FR^k\times\DD}
(\xi,\ldots,\xi)\,.
\end{split}
\end{equation}
These coefficients depend only on the germ of $a$. Since $J_a(\xi)^*$ is invertible, the expansion recovers the actual Euler--Lagrange section restricted to $\DD_{\phi,\CF}$, including its varying target fiber. Conversely, that restricted section, together with $a$, recovers every coefficient by the mixed-nilpotent extraction formula.

\item[{\bf (iii)}]
For another coordinate choice $b$, put $c=a^{-1}\circ b$ as a germ and $C=\Psi_a^{-1}\circ \Psi_b$. The coordinate-change formula is the covector transformation
\begin{equation}\label{PerturbationEqEltransition}
\widetilde{\CE\CL}_b(\xi)
=\bigl(\dd_{\mathrm{vert}}c_\xi\bigr)^*
\widetilde{\CE\CL}_a\bigl(C(\xi)\bigr).
\end{equation}
Here the operators are restricted to their formal neighborhoods and all expressions are evaluated on the same probe. Thus a change of field coordinates changes the formal Euler--Lagrange expression by both substitution and the dual Jacobian. The underlying restricted section is intrinsic.
\end{itemize}
\end{corollary}

\begin{proof}
The Jacobian (vertical derivative) defines a vector-bundle isomorphism over $\CU$,
$$
J_a\, : \, p_{\phi^*VF}^*  (\phi^*VF)\vert_{\CU} \longrightarrow a^*VF\, ,
$$
where $p_{\phi^*VF}:\phi^*VF\to M$ is the projection and $V(\phi^*VF) \cong p_{\phi^*VF}^*  (\phi^*VF) $. Thus dually, after tensoring with the top-form bundle, an isomorphism
\[
J_a^* \, : \, a^*B\longrightarrow p_{\phi^*VF}^*(E_\phi)|_{\CU}\, .
\]
This explicitly identifies the varying target fiber of $\CE \CL (\phi)$. For $\bar a$, the corresponding map $J_{\bar a}^*:\bar a^*B\to p_{\phi^*VF}^*E_\phi$ exists on all of $\phi^*VF$, but need not be invertible away from zero. Composing it with the Euler--Lagrange jet map after prolongation of $\bar a$ defines $\widetilde{\CE\CL}_{\bar a}$.

To prove \eqref{PerturbationEqEltransformed}, work in a relatively compact domain in a bundle chart and take a variation $h_t= h +tk \in T_\phi \CF(\FR^1)$ supported in its interior. For $\psi_t=\bar a\circ h_t$ the chain rule in a coordinate chart gives \(\left.\partial_t\psi_t^c\right|_{t=0}
=
\frac{\partial\bar a^c}{\partial v^b}(x,h(x))\,k^b(x)\), i.e., in terms of the Jacobian 
$$
\partial_t\psi|_{t=0}=J_{\bar a}(h)k\, .
$$
The variation formula therefore gives
\[
\begin{split}
\left.\partial_t\int L_a\circ j^\infty(h+tk)\right|_{t=0}
&= \left.\partial_t\int \CL\big( \bar{a}\circ h_t \big)\right|_{t=0}
\\
&=\int\big\langle\CE\CL(\bar a\circ h),\,J_{\bar a}(h)k\big\rangle\\
&=\int\big\langle J_{\bar a}(h)^*\CE\CL(\bar a\circ h),\, k\big\rangle\, .
\end{split}
\]
Applying instead the same variation formula directly to $\CL_a = L_a \circ j^r$ gives the same left side with $\CE(\CL_a)(h)$ paired against $k$, so that 
$$
\int\big\langle \CE(\CL_a)(h),\, k\big\rangle\, =\, \int\big\langle J_{\bar a}(h)^*\CE\CL(\bar a\circ h),\, k\big\rangle\, .
$$
Since the equality holds for any variation $k$, it follows that
$$
\widetilde{\CE\CL}_{\bar a}\,=\, \CE(\CL_a) \, .
$$
As with Rem. \ref{CriticalityOnGeneralSpacetimes}, compact support is used only in this proof to remove the boundary term and no boundary hypothesis is imposed on the resulting local identity.

The transformed operator now has a fixed vector-bundle target, so Thm.~\ref{PerturbationThmNonlinear} applies to it on the vector field bundle $\phi^*VF\to M$, using the identity coordinates at zero which yields \eqref{PerturbationEqElrestriction}. The germ of $J_{\bar a}$ agrees with that of $J_a$, so the formal restriction and its coefficients are independent of the extension. Every nilpotent $\xi$ has body zero and factors through $\CU$ and the inverse bundle map to $J_a^*$ therefore evaluates at $\xi$. Equivalently, in adapted frames its matrix has invertible body and a nilpotent correction, whose inverse is a finite geometric series. Applying this inverse to \eqref{PerturbationEqElrestriction} recovers $\CE\CL(\Psi_a(\xi))$. Extracting mixed nilpotent coefficients after applying $J_a^*$ proves the converse.

For the change-of-coordinates assertion, $b=a\circ c$ gives the vertical chain rule
$
J_b(v)=J_a(c(v))\circ\dd_{\mathrm{vert}}c_v$, and thus dually
$$
J_b(v)^*
=
\bigl(\mathrm d_{\mathrm{vert}}c_v\bigr)^*
J_a(c(v))^*\, .
$$
Therefore, for an ordinary fluctuation \(h\in T_\phi \CF \) in the common domain,
\[
\begin{aligned}
\widetilde{\CE\CL}_b(h)
&=
J_b(h)^*\mathcal E\mathcal L(b\circ h)
\\
&=
\bigl(\mathrm d_{\mathrm{vert}}c_h\bigr)^*
J_a(c\circ h)^*
\mathcal E\mathcal L(a\circ c\circ h)
\\
&=
\bigl(\mathrm d_{\mathrm{vert}}c_h\bigr)^*
\widetilde{\CE\CL}_a(c\circ h).
\end{aligned}
\]
On infinitesimal fluctuations $\xi \in \DD_{0,T_\phi \CF}$, \(c\circ\xi=C(\xi)\), so this yields  \eqref{PerturbationEqEltransition}. All identities are induced by smooth bundle maps and commute with probe pullbacks, including nilpotent corrections to smooth parameters.
\end{proof}

\paragraph{Field-derivative form of the Euler--Lagrange expansion.}
In coordinates $u^c=a^c(x,v)$, \eqref{PerturbationEqElpullback} reads
\[
\bigl(\widetilde{\CE\CL}_a(h)\bigr)_b
=
\frac{\partial a^c}{\partial v^b}\big(x,h(x)\big)\,
\bigl(\CE\CL(a\circ h)\bigr)_c 
\, .
\]
Write
\[
\widetilde{\CE\CL}_a(h)=\mathcal K_b(h)\,\dd v^b\otimes\dd^dx\,,
\]
using the fixed covector frame of $(\phi^*VF)^*$, so that the coefficients $\mathcal K_b$ are local functions of $h$. The convention of Eq. \eqref{PerturbationEqIteratedlocal} then gives
\[
\begin{split}
\Bigl(\widetilde{\CE\CL}_{a,0}^{(n)}{}(h_1,\ldots,h_n)\Bigr)_b
\;={}&\sum_{\substack{c_1,\ldots,c_n\\|I_1|,\ldots,|I_n|\leq2r}}
\frac{\delta^n\mathcal K_b}
{\delta(\partial_{I_1}v^{c_1})\cdots\delta(\partial_{I_n}v^{c_n})}(0)
\prod_{\nu=1}^n\partial_{I_\nu}h_\nu^{c_\nu}.
\end{split}
\]
Thus the physics notation automatically differentiates the complete transformed coefficient, including the field-dependent Jacobian. Simply differentiating the untransformed Euler--Lagrange components would omit these terms. Because $J_a(\xi)^*$ is invertible, \eqref{PerturbationEqElrestriction} also identifies the formal on-shell equations with the actual Euler--Lagrange zero locus restricted around $\phi$.

\subsubsection*{Formal polynomial actions as functions on the tangent space}

Fix a local action morphism \(S:\CF\to\FR\) and a fiberwise coordinate choice \(a\). Using an extension \(\bar a\) from Lem.~\ref{PerturbationLemChart}, set
\begin{equation}\label{PerturbationEqVertices}
S_n^a:=\bigl(S\circ\mathbold{\Gamma}_M(\bar a)\bigr)^{(n)}_0,
\qquad P_n^a(h):=\tfrac1{n!}S_n^a(h,\ldots,h)\,.
\end{equation}
Each \(S_n^a\) is a symmetric multilinear functional of ordinary tangent vectors $h\in T_\phi \CF(*)$ and each \(P_n^a\) is a homogeneous polynomial functional of degree \(n\). The formal polynomial action, or more precisely the formal power series in homogeneous polynomial functionals, is
\begin{equation}\label{PerturbationEqPert}
S_{\mathrm{pert}}^a(h):=\sum_{n\geq0}P_n^a(h)
=\sum_{n\geq0}\tfrac1{n!}S_n^a(h,\ldots,h)\,\, ,
\end{equation}
where each summand separately is defined on ordinary tangent vectors, even if their infinite sum has no numerical value there.

The local field-derivative notation of \eqref{PerturbationEqIteratedlocal} applies in these fluctuations as follows. If \(S=S_\Sigma\) is a local action, set
\[
L_a:=L_r\circ J^r_M\bar a\,,
\qquad \CL_a(v):=L_a\circ j^rv\,,
\qquad \widetilde S_a(v)=\int_\Sigma\CL_a(v)\,.
\]
The ``vertices'' in \eqref{PerturbationEqVertices} are then explicitly
\begin{equation}\label{PerturbationEqNonlinearphysics}
S_n^a(h_1,\ldots,h_n)
=\int_\Sigma\sum_{\substack{b_1,\ldots,b_n\\I_1,\ldots,I_n}}
\frac{\delta^n\CL_a}
{\delta(\partial_{I_1}v^{b_1})\cdots\delta(\partial_{I_n}v^{b_n})}(0)
\prod_{\nu=1}^n\partial_{I_\nu}h_\nu^{b_\nu}.
\end{equation}
Here the local derivatives are defined by the jet coordinates of the vector bundle \(\phi^*VF\), exactly as in \eqref{PerturbationEqIteratedlocal}. Thus the physicists' formal polynomial on the tangent space is obtained by contracting local jet-pullback derivatives with the prolonged variations and then integrating. For nonlinear fibers, this expression uses the chosen fluctuation coordinates. A different fiberwise choice changes the coefficients by the formal substitution \eqref{PerturbationEqTransition}, while the restricted action remains intrinsic.

\begin{corollary}[\bf Tangent-space interpretation of the restricted action]\label{PerturbationCorDictionary}
In the fiberwise coordinates \(a\), the actual restriction \(S|_{\DD_{\phi,\CF}}\) encodes (and is encoded by) the nilpotent evaluation of the formal polynomial action \eqref{PerturbationEqPert}. Explicitly, the ordinary $\FR^k$-plots are $(T_\phi\CF)(\FR^k)=\Gamma_{\FR^k\times M}(p_M^*(\phi^*VF))$. For every probe $\FR^k\times\DD$,
\begin{equation}\label{PerturbationEqTwoinputs}
\begin{split}
(T_\phi\CF)\big(\FR^k\times\DD\big)&=(T_\phi\CF)(\FR^k)\otimes \CO(\DD),\\
\DD_{0,T_\phi\CF}\big(\FR^k\times\DD\big)&=(T_\phi\CF)(\FR^k)\otimes V.
\end{split}
\end{equation}
For \(\xi\in(T_\phi\CF)(\FR^k)\otimes V\), with \(V^{N+1}=0\),
\begin{equation}\label{PerturbationEqDictionary}
\bigl(S|_{\DD_{\phi,\CF}}\circ \Psi_a\bigr)_{\FR^k\times\DD}(\xi)
=\sum_{n=0}^{N}\tfrac1{n!}\big(S_n^a\big)_{\FR^k\times\DD}(\xi,\ldots,\xi).
\end{equation}
Conversely, this restricted morphism determines all the homogeneous polynomial functionals \(P_n^a\) and all their polarizations \(S_n^a\).
\end{corollary}

\begin{proof}
Equation \eqref{PerturbationEqTwoinputs} is \eqref{PerturbationEqVectorplots} applied to \(\phi^*VF\) at its zero section, while equation \eqref{PerturbationEqDictionary} is Thm.~\ref{PerturbationThmNonlinear} with the notation \eqref{PerturbationEqVertices}. Recovery of the multilinear coefficients follows from \eqref{PerturbationEqRecover}, applied after pullback along \(\Psi_a\). Restricting each coefficient to the diagonal gives \(P_n^a\) and, conversely, polarization recovers \(S_n^a\) from its diagonal polynomial.
\end{proof}

For a fixed ordinary tangent vector \(h\in(T_\phi\CF)(*)\), the formal parameter can be displayed separately as
\begin{equation*}
\sum_{n\geq0}t^nP_n^a(h)\quad \in\quad \FR[[t]]\, .
\end{equation*}
Indeed, take \(\CO(\DD)=\FR[t]/(t^{N+1})\) and \(\xi=th\).\, then equation \eqref{PerturbationEqDictionary} gives exactly
\begin{equation*}
\bigl(S|_{\DD_{\phi,\CF}}\circ \Psi_a\bigr)_{\DD}(th)
\, =\, \sum_{n=0}^{N}t^nP_n^a(h) \, .
\end{equation*}
These compatible truncations are the formal Taylor series of \(t\mapsto S(\bar a\circ(th))\) at \(t=0\). The parameter \(t\) records polynomial degree and -- again -- setting \(t=1\) in the infinite formal series is not generally defined.

\subsubsection*{Relation to perturbative BV/BRST and the $L_\infty$ brackets}\label{PerturbationRemBv}
The preceding description already contains the ingredients entering degrees $0$ and $1$ of the classical perturbative BV construction: ordinary field variations, their variational cotangent space, and the Taylor coefficients of the Euler--Lagrange operator. Our restriction results identify the infinitesimal evaluations of these ingredients. They do not, by themselves, construct the corresponding derived critical locus or identify its full local $L_\infty$ algebra. We first describe the algebraic interpretation of the data already present, and then distinguish the derived enhancement of the equations from the infinitesimal gauge quotient.

\paragraph{The ingredients in degrees $0$ and $1$.}
Fix a background field configuration $\phi\in\CF(*)$ and fiberwise coordinates $a$, as in Cor.~\ref{PerturbationCorEulerLagrange}. Algebraically, these ingredients can be placed in the graded vector space\footnote{We use the shifted symmetric $L_\infty[1]$ convention, that is,  the brackets are graded symmetric and have cohomological degree $1$, as in the symmetric-coalgebra formulation of \cite[App.~A]{JRSW19}. Thus, if $L$ denotes the actual graded space underlying the $L_\infty$-algebra structure, our graded space is $L[1]$. For us,  ordinary fields and their antifields have degrees $0$ and $1$,  respectively, while in the standard $L_\infty$-structure they have degrees $1$ and $2$.}
\begin{equation}\label{PerturbationEqBvFields}
\begin{split}
\CB_\phi^0&:=(T_\phi\CF)(*)=\Gamma_M(\phi^*VF)\,,\\
\CB_\phi^1&:=(T^*_{\var,\phi}\CF)(*)
=\Gamma_M\!\big((\phi^*VF)^*\otimes\wedge^dT^*M\big)\,,\\
\CB_\phi^j&:=0\qquad (j\neq0,1)\,.
\end{split}
\end{equation}
The degree $1$ component is called the \emph{space of antifields}. The distinguished element $\CE\CL(\phi)$ lies in this space, but is not an independently chosen antifield -- it will be the curvature of the perturbative brackets. The degrees of linear coordinate functions on these field spaces are opposite to degrees of the corresponding graded vector-space components. In particular, coordinates dual to $\CB_\phi^1$ have degree $-1$, in agreement with the usual ``ghost number'' of antifield coordinates.

The Euler--Lagrange operator must first be expressed in the fixed variational cotangent fiber, as described in Cor. \ref{PerturbationCorEulerLagrange}, 
\begin{equation}\label{PerturbationEqBvElSeries}
\widetilde{\CE\CL}_a(h)
=J_a(h)^*\CE\CL(a\circ h)\,,
\qquad
\widehat{\CE\CL}_a(h)
:=\sum_{n\geq0}\tfrac1{n!}q_n^a(h,\ldots,h)\, ,
\end{equation}
where the second expression is its formal Taylor series, with
\begin{equation}\label{PerturbationEqBvBrackets}
q_n^a(h_1,\ldots,h_n)
:=\widetilde{\CE\CL}_{a,0}^{(n)}(h_1,\ldots,h_n),
\qquad
q_0^a=\CE\CL(\phi).
\end{equation}
Each $q_n^a$ for $n\geq1$ is a symmetric multilinear differential operator from $(\CB_\phi^0)^n$ to $\CB_\phi^1$, which we may extend to a graded symmetric\footnote{Here graded symmetry means that interchanging homogeneous inputs $z_i,z_{i+1}$ contributes $(-1)^{|z_i||z_{i+1}|}$.} map
\[
q_n^a:\operatorname{Sym}^n(\CB_\phi)\longrightarrow\CB_\phi,
\qquad |q_n^a|=1,
\]
by declaring it zero whenever an input has degree $1$.

It is immediate to see that these maps satisfy the $L_\infty[1]$-algebra relations. Indeed, every possibly nonzero inner bracket takes values in $\CB_\phi^1$, so every composite appearing in an identity vanishes separately
\begin{equation}\label{LinftyRelations}
q_i^a\bigl(q_j^a(z_1,\ldots,z_j),z_{j+1},\ldots\bigr)=0
\qquad (i\geq1,\ j\geq0)\, ,
\end{equation}
including insertions of $q_0^a$. Hence this algebraic construction gives a two-term \textit{curved $L_\infty[1]$-algebra}, uncurved precisely when $\CE\CL(\phi)=0$. Since the coefficients are differential operators in every input, it restricts to open subsets of $M$ and is local in the sense of the shifted version of \cite[Def.~7.2.1]{Gwi12}. We stress, however, that constructing this two-term model from the available algebraic data is distinct from establishing its geometric interpretation through a derived infinitesimal restriction theorem.

At a critical background, $q_1^a$ is the Jacobi operator of Cor.~\ref{OnshellSyntheticTangentBundle}. Explicitly, differentiation of the first formula in \eqref{PerturbationEqBvElSeries} gives
\[
q_1^a(h)
\, =\, \bigl(\dd J_a^*\bigr)_0(h)\,\CE\CL(\phi)
+J_a(0)^*\bigl(\dd(\CE\CL)\bigr)_\phi(h)
\, =\, \bigl(\dd(\CE\CL)\bigr)_\phi(h)\, ,
\]
where the last expression is identified with the fixed variational cotangent fiber along the zero section. The corresponding \textit{Maurer--Cartan equation} is therefore nothing but the formal perturbative field equation
\begin{equation}\label{PerturbationEqBvMc}
\sum_{n\geq1}\tfrac1{n!}q_n^a(h,\ldots,h)\, =\, 0\, .
\end{equation}
For an arbitrary background, the same statement holds, but with the \textit{curvature term $q_0^a$} included. Thus the linearized field equation occurs in the unary bracket, while the higher brackets encode the nonlinear terms of the perturbative field equation.

To relate these brackets to action ``vertices'', we assume that the variation boundary term vanishes for the variations under consideration. For example, take $\Sigma=M$ compact and without boundary, or fix appropriate boundary/asymptotic conditions. Differentiating the variation formula \eqref{PerturbationEqPert} in the fixed fluctuation variables yields, for $n\geq0$,
\begin{equation}\label{PerturbationEqBvCyclicity}
S_{n+1}^a(h_1,\ldots,h_n,k)
=\int_\Sigma\big\langle q_n^a(h_1,\ldots,h_n),k\big\rangle\,.
\end{equation}
In particular, the right-hand side is symmetric in all $n+1$ ordinary variations. This is the cyclicity condition in this two-term shifted setting, for the degree-$-1$ BV pairing supplied by the integrated variational pairing. Conversely, the formal action can be written, up to its constant term, as
\begin{equation}\label{PerturbationEqBvActionFromBrackets}
S_{\mathrm{pert}}^a(h)-S(\phi)
=\sum_{n\geq0}\tfrac1{(n+1)!}
\int_\Sigma\big\langle q_n^a(h,\ldots,h),h\big\rangle\,.
\end{equation}
At a critical background, this becomes
\begin{equation*}
S_{\mathrm{pert}}^a(h)-S(\phi)
=\underbrace{\tfrac12S_2^a(h,h)}_{\text{quadratic free action}}
+\underbrace{\sum_{n\geq3}\tfrac1{n!}S_n^a(h,\ldots,h)}_{\text{interaction vertices}}.
\end{equation*}
Thus $q_1^a$ is paired with variations by the quadratic vertex, $q_2^a$ by the cubic vertex, and generally $q_n^a$ by the $(n+1)$-st vertex. This is the field-equation/action correspondence discussed in \cite[Secs.~4.2 and~5.1]{JRSW19}\cite[Sec.~2.3]{CGRS20}\cite[Sec.~2.5]{GS22}, expressed here in the fixed fiber provided by Cor.~\ref{PerturbationCorEulerLagrange}.

One can also see the antifield differential without introducing distributional kernels. Write $h^\dagger$ for the formal antifield variable, and for an admissible ordinary variation $k$ let $\chi_k \in C^\infty_\loc(\CB^\bullet_\phi)$\footnote{A certain formal completion of the algebra generated by local functionals, whose details we do not need for the present informal discussion.} given by
\[
\chi_k(h^\dagger)\,:=\, \int_\Sigma\langle h^\dagger,k\rangle\, ,
\]
thus of degree $|\chi_k|=-1$. On the formal variational algebra $C^\infty_\loc(\CB^\bullet_\phi)$, the degree-$1$ Koszul differential is then specified by
\begin{equation}\label{PerturbationEqBvKt}
\begin{split}
s_{\mathrm{KT}}f(h)&=0,\\
s_{\mathrm{KT}}\chi_k
&=\int_\Sigma\big\langle\widehat{\CE\CL}_a(h),k\big\rangle
=\sum_{n\geq0}\tfrac1{n!}S_{n+1}^a(h,\ldots,h,k)\\
&=\left.{\partial_t}
S_{\mathrm{pert}}^a(h+tk)\right|_{t=0},
\end{split}
\end{equation}
where $f$ depends only on the ordinary fields, and the differential extends by the graded Leibniz rule. Its square vanishes because the right-hand side contains no antifields, which is directly equivalent to the dual $L_\infty[1]$-relations \eqref{LinftyRelations}. Yet equivalently, with the BV bracket sign fixed by \eqref{PerturbationEqBvKt},
\[
s_{\mathrm{KT}}=\big\{S_{\mathrm{pert}}^a\,,-\big\}_{\mathrm{BV}}\,,
\qquad
\big\{S_{\mathrm{pert}}^a\,,\,S_{\mathrm{pert}}^a\big\}_{\mathrm{BV}}=0\,,
\]
the second identity following from the absence of antifield dependence in this action. These are the algebraic formulas of the Koszul model for a derived critical locus, before hard-coding gauge symmetry in the geometry. A closely related account is \cite[Sec.~2.1]{Gwi12}. Here, we make no regularity assumption that would force this Koszul complex to resolve the ordinary, underived critical locus.

As above, the field-only formal action is an element of
\begin{equation}\label{PerturbationEqBvalgebra}
S_{\mathrm{pert}}^a\in\widehat{\CO}(T_\phi\CF),
\qquad
\widehat{\CO}(T_\phi\CF)
:=\prod_{n\geq0}\operatorname{Pol}_{\mathrm{cont}}^n(T_\phi\CF),
\end{equation}
where homogeneous polynomials mean diagonal evaluations of continuous symmetric multilinear forms on $(T_\phi\CF)(*)$, once a topology is fixed. For the finite-order compact-domain actions above, their continuity follows from \eqref{PerturbationEqNonlinearphysics} in the usual smooth-section topology. In infinite dimensions, writing this algebra as a completed symmetric algebra of a dual requires specifying that dual and the tensor completions; it does not identify all continuous polynomial functionals with ordinary algebraic symmetric powers. See also \cite[Def.~7.2.6]{Gwi12} for perturbative local functionals as formal series of multilinear functionals.

The precise content of the present restriction results is the identification of these Taylor coefficients with actual ungraded infinitesimal evaluations. Namely, if $\CO(\DD)=\FR\oplus V$ with $V^{N+1}=0$, then for every probe $\FR^k\times\DD$ and every $\xi\in\DD_{0,T_\phi\CF}(\FR^k\times\DD)$, Eq.~\eqref{PerturbationEqElrestriction} reads
\begin{equation}\label{PerturbationEqBvNilpotent}
J_a(\xi)^*(\CE\CL)_{\FR^k\times\DD}\bigl(\Psi_a(\xi)\bigr)
=\sum_{n=0}^{N}\tfrac1{n!}
(q_n^a)_{\FR^k\times\DD}(\xi,\ldots,\xi)\,.
\end{equation}
The $n=0$ term denotes the constant section $q_0^a$. The corresponding statement for the action is \eqref{PerturbationEqDictionary}, where mixed nilpotent coefficients determine the displayed multilinear maps with ordinary-field inputs and the action vertices. Consequently, no convergence of the infinite formal series on ordinary variations is required. Under a change of fiberwise coordinates, the dual-Jacobian law \eqref{PerturbationEqEltransition} makes \eqref{PerturbationEqBvCyclicity} compatible with the action's formal substitution law.

\paragraph{Derived critical locus and nonnegative degrees.}
The geometric interpretation of the nonnegative-degree model belongs to a derived extension of our setting. There the ``critical locus'', or better, the Euler--Lagrange zero-locus is formed as the \textit{homotopy intersection} of the Euler--Lagrange section with the zero section of the variational cotangent bundle:
\begin{equation}\label{PerturbationEqBvDerivedCritical}
	\xymatrix@R=1.6em@C=2em  {\operatorname{Crit}^{\mathrm{der}}(S)
:=\CF\mathop{\times}\limits^{h}_{T^*_{\var}\CF} \CF \ar[d] \ar[rr] &&   \CF \ar[d]^-{\mathcal{EL}} 
		\\ 
		\CF\ar[rr]^-{0_\CF}  && T^*_\mathrm{var} \CF
	\, . }
\end{equation}

This diagram stands for the intended derived construction, where the \textit{homotopy} fiber product is \textit{not} being taken in our present category of set-valued sheaves over $\ThickenedCartesianSpaces$ (cf. Thm \ref{FunctorialityOfTheThickenedCriticalSet}), but rather over a derived enhancement of $\ThickenedCartesianSpaces$.  That is, a site of appropriately thickened \textit{differential graded} Cartesian probes whose underlying graded vector directions are concentrated in \emph{nonnegative cohomological degrees}.\footnote{The extension of the notion of the infinitesimal neighborhood in that setting is straightforward, by requiring the purely finite/underived plots to degenerate to the corresponding point.} Their coordinate algebras are therefore concentrated in \emph{nonpositive cohomological degrees} and carry a differential of degree $+1$. \footnote{This is the same convention as in \eqref{PerturbationEqBvFields}: vector degrees $0,1$ give coordinate generators of degrees $0,-1$, while both the codifferential on the symmetric coalgebra and the differential on its dual coordinate algebra have degree $+1$. }For perturbation theory at a critical configuration $\phi$, the relevant object is its infinitesimal neighborhood
\[
\DD_{\phi,{\operatorname{Crit}^\mathrm{der}(S)}} \longhookrightarrow \; \operatorname{Crit}^{\mathrm{der}}(S) ,
\]
rather than the entire derived critical locus. An extension of our restriction theorem then should identify the local $L_\infty[1]$ model with this derived infinitesimal neighborhood, including its homological information. Establishing this derived restriction statement is beyond the ungraded results above, and is left for future installments of the series.

\paragraph{Infinitesimal gauge quotient and negative degrees.}
Gauge symmetry is a further input. In the shifted convention, degree $-1$ in the deformation complex consists of \textit{infinitesimal gauge parameters}, and lower degrees describe higher gauge transformations or reducibility data. Their dual coordinate functions are the positive-ghost-number ghosts and higher ghosts.

After specifying the gauge structure and a BV presentation, denote the full graded space by $\CB_{\mathrm{BV},\phi}$, with $\CB_{\mathrm{BV},\phi}^0=\CB_\phi^0$ and $\CB_{\mathrm{BV},\phi}^1=\CB_\phi^1$. At a critical background its linear part typically has the form
\[
\cdots\longrightarrow
\CB_{\mathrm{BV},\phi}^{-1}\xrightarrow{\ R_\phi\ }
\CB_\phi^0\xrightarrow{\ q_1^a\ }
\CB_\phi^1\xrightarrow{\ R_\phi^\dagger\ }
\CB_{\mathrm{BV},\phi}^{2}\longrightarrow\cdots,
\]
where $R_\phi$ is the linearized gauge action and $R_\phi^\dagger$ is its variational adjoint \cite[Sec~5.1]{JRSW19}\cite[Sec.~2.3]{CGRS20}\cite[Sec.~2.5]{GS22}. Gauge invariance gives
\[
q_1^aR_\phi=0,
\qquad R_\phi^\dagger q_1^a=0\, ,
\]
and terms and brackets involving them encode the \emph{infinitesimal gauge quotient}. The degree $-1$ BV pairing also supplies their partners in degrees $2$ and above: a (higher) gauge parameter in degree $-r$ is paired with an antifield in degree $r+1$. Thus the full BV construction includes the corresponding Noether identities as well as the gauge generators, while neither these additional components nor their brackets are determined by the original Euler--Lagrange Taylor coefficients alone.

\textit{Nonperturbatively}, the gauge quotient should instead be presented by a sheaf valued in $\infty$-groupoids, for instance
\[
\mathcal X:(\ThickenedCartesianSpaces)^\op
\longrightarrow\infty\text{-}\mathrm{Grpd}\, ,
\]
satisfying homotopy descent. Its objects are field configurations, its morphisms are \textit{finite} gauge transformations, and its higher morphisms encode \textit{finite} gauge transformations between gauge transformations etc., when present. For an ordinary group action this is the stacky quotient and so ordinary groupoids are included as a special case. The perturbative gauge structure describes the infinitesimal behavior near a chosen configuration and gauge transformations infinitesimally near the identity. It does not, by itself, determine disconnected gauge transformations or global bundle sectors -- which live solely in the nonperturbative realm.

For a gauge-invariant action descending to $\mathcal X$, the full BV interpretation is the derived critical locus of that descended action, formally completed at $\phi$. This combines the derived equations with the gauge quotient and its dual antifield data. Enlarging the probe site to detect derived directions and retaining groupoid-valued gauge information are therefore distinct parts of the proposed extension. Neither has been supplied by the present ungraded restriction theorem, and is left for further future works.

\bigskip 
\paragraph{Outlook.} Having provided a solid foundation for synthetic constructions in field theory via infinitesimally thickened smooth sets, in the next installment of this series, we will further enlarge our probe-site inside super-commutative algebras
$$
\CartesianSpaces\longhookrightarrow \ThickenedCartesianSpaces \longhookrightarrow \mathrm{CAlg}^{\op} \longhookrightarrow \mathbb{Z}_2\mbox{-}\mathrm{CAlg}^\op \, ,
$$
to what we call super (thickened) Cartesian spaces $\mathrm{SupThCartSp}\hookrightarrow \mathbb{Z}_2\mbox{-}\mathrm{CAlg}^\op$.
As we will show therein, sheaves over this site, i.e., ``\textit{super thickened smooth sets}'', provide a natural setting for a rigorous formalization of local fermionic field theory alongside bosonic fields and their smooth and infinitesimal structures. See \cite{Gi25} for an overall exposition towards that goal, and further parts of the series.

\vspace{1cm} 
\noindent {\bf Acknowledgements.} We thank Urs Schreiber for discussions relating to various aspects of the paper, and David Jaz Myers for comments on an early version of the Appendix.
 We acknowledge the usage of ChatGPT-6 Astra for checking calculations and arguments in parts of the newly added section \cref{PerturbationSection}, especially the nonlinear field bundle generalization.
 

\newpage

\section{Appendix}
This Appendix includes purely mathematical and technical background. In particular, we prove that that it is indeed sufficient to work with purely $\FR$-algebraic probes, explain the internal relation of infinitesimal neighborhoods and jets, and further detail the classifying space of de Rham forms. The first is necessary to justify our usage of the term ``synthetic differential geometry'' throughout. The latter, although not strictly necessary for the main text, we believe do clarify the intuition behind these objects in our context of thickened smooth sets, in such a  manner that is not available in the existing literature. 
\label{App-Synth}

\subsection{The $\FR$-algebraic nature of the Cahiers topos}
\label{App-CahierCorners}

In this manuscript, our aim has been to make contact with existing practice of infinitesimals in the physics literature, while using the least amount of mathematical background, but nevertheless not compromising on the level of rigor. This appendix serves to justify our usage of the $\FR$-algebraic version of Dubuc's ``Cahiers topos'' comprised of sheaves over infinitesimally thickened Cartesian probes $\ThickenedCartesianSpaces\hookrightarrow \mathrm{CAlg}^\op$, rather than the original $C^\infty$-algebraic version from \cite{Dubuc79}, via Thm. \ref{ThickenedCartesianSpacesAreSmoothLoci} and Cor. \ref{CahiersToposIsRalgebraic}. Of course, this by no means compromises the usefulness of the $C^\infty$-algebraic technology in more technical considerations of variants of thickened probes sites (cf. Rem. \ref{VariantsOfThickenedSites}, ftn. 3), especially in showing that these do indeed yield \textit{``well-adapted''} models for synthetic differential geometry \cite{Dubuc79}\cite{MoerdijkReyes}\cite{Kock06} -- including the case of the Cahiers topos we employ.
\begin{lemma}[{\bf Hadamard's Lemma for partial derivatives}] \label{PartialHadamardsLemma}
For any smooth function $f(x,y)\in C^\infty(\FR^k\times\FR^m)$ and any $l\in \NN$, there exist smooth functions $\{h_{i_{1}\cdots i_{l+1}}\}_{i_1,\cdots, i_{l+1}=1,\cdots, m}\subset C^\infty(\FR^k\times \FR^m)$ such that 
\vspace{-1mm} 
\begin{align*}
f&=  f(x,0)+ y^i \cdot \partial_{y^i} f(x,0) + \tfrac{1}{2}y^i y^j\cdot \partial_{y^i}\partial_{y^j}f(x,0) +\cdots +\; \tfrac{1}{l!} \sum_{i_{1},\cdots,i_{l}=1}^m y^{i_{1}}\cdots y^{i_{l}}\cdot \partial_{y^{i_{1}}}\cdots \partial_{y^{i_{l}}}f(x,0)
\\
 &  \quad \quad \qquad 
 + \sum_{i_{1},\cdots,i_{l+1}=1}^m y^{i_{1}}\cdots y^{i_{l+1}}\cdot h_{i_{1}\cdots i_{l+1}} (x,y) \\
&= \sum_{|I|\leq l} k_I(x) \cdot y^I + \sum_{|I|=l} y^{I+1} \cdot h_{I+1}(x,y)\, ,
\end{align*}
where $\{y^i\}_{i=1,\cdots,m}\subset C^\infty(\FR^m)$ are the canonical (linear) coordinate functions, and $k_{i_1\cdots i_{l+1}}(x) \subset C^\infty(\FR^k)$ are (uniquely) defined by the first equality.
\end{lemma}
\begin{proof}
This follows by a slight generalization of the proof for the traditional Lemma by Hadamard (e.g. \cite[Lem. 2.8]{Nestruev03}). Consider first the case of $k,m=1$ and $l=0$ and define
\begin{align*}
g: [0,1]_t \times \FR_x &\longrightarrow \FR\\
(t,x) &\longmapsto f(x,t\cdot y) \, ,
\end{align*}
so that
$$
\partial_t g(x,t) = {\partial_y f}(x, t\cdot y) \cdot y\, .
$$
It follows that
$$ 
  g(x,1) - g(x,0) 
  = 
  \int_{0}^1 
    \partial_t g(x,t) 
  \dd t 
  = 
  y 
    \cdot 
  \int_0^1  
    \partial_y f (x, t\cdot y) 
  \dd t
  \,,
$$
and, by definition,
$$
g(x,1)-g(x,0)\equiv f(x,y)-f(x,0)\,,
$$
hence
$$
f(x,y) = f(x,0) + y \cdot h(x,y)\,,
$$
for 
$$
  h(x,y)
    := 
  \int_0^1  
    {\partial_y f}(x, t\cdot y) 
  \dd t
  \, .
$$
Notice that this implies that $\partial_y f (x,0)=h(x,0)$, and so iterating the above on $h(x,y)$ yields the claim for $m=n=1$ and $l\in \NN$. The case of general $m,n$ follows analogously.
\end{proof}

\begin{corollary}[\bf Hadamard on Manifolds/Cartesian product]
\label{HadamardOnManifoldCartesianProduct}
Let $M\in \SmoothManifolds$ be a smooth manifold. For any smooth function $f\in C^\infty(M\times \FR^m)$ and any $l\in \NN$, there exist smooth functions $\{h_{i_{1}\cdots i_{l+1}}\}_{i_1,\cdots, i_{l+1}=1,\cdots, m}\subset C^\infty(
  M
\times \FR^m)$ such that 
\vspace{-2mm} 
\begin{align*}
f(x,y)&=  f(x,0)+ y^i \cdot \partial_{y^i} f(x,0) + \tfrac{1}{2}y^i y^j\cdot \partial_{y^i}\partial_{y^j}f(x,0) +\cdots +\; \tfrac{1}{l!} \sum_{i_{1},\cdots,i_{l}=1}^m y^{i_{1}}\cdots y^{i_{l}}\cdot \partial_{y^{i_{1}}}\cdots \partial_{y^{i_{l}}}f(x,0)
\\
 &  \quad \quad \qquad 
 + \sum_{i_{1},\cdots,i_{l+1}=1}^m y^{i_{1}}\cdots y^{i_{l+1}}\cdot h_{i_{1}\cdots i_{l+1}} (x,y) \\
&= \sum_{|I|\leq l} k_I(x) \cdot y^I + \sum_{|I|=l} y^{I+1} \cdot h_{I+1}(x,y)\, ,
\end{align*}
where $\{y^i\}_{i=1,\cdots,m}\subset C^\infty(\FR^m)$ are the canonical (linear) coordinate functions, and $k_{i_1\cdots i_{l+1}}(x) \subset C^\infty(M)$ are (uniquely) defined by the first equality.
\end{corollary}
\begin{proof}
Working locally in product charts $\big\{\FR^n\times \FR^m \xrightarrow{\,\, \phi_j\times \id_{\FR^m}\, \,}U_i\times \FR^m \hookrightarrow M\times \FR^m\big\}_{j\in J}$, the result follows as above, since all of the (higher) partial derivative operators appearing are globally defined with respect to the canonical global chart on $\FR^m$, and the corresponding expansions must agree on overlaps at each polynomial order in $\{y^i\}_{i=1,\cdots, m}$.
\end{proof}

The above implies that the function algebras of our thickened probes $\ThickenedCartesianSpaces$ are all, in fact, quotients of smooth Cartesian algebras $C^\infty(\FR^{k'})$. 

\begin{corollary}[\bf Thickened function algebras as quotient algebras]\label{ThickenedFunctionAlgebrasAsQuotientAlgebras}
There exist the following isomorphisms of $\FR$-algebras:

\noindent {\bf (i)}
The function algebra of any thickened Cartesian space $\FR^k\times \DD^m(l) \in \ThickenedCartesianSpaces$ is canonically identified with the quotient
$$
C^\infty(\FR^k)\otimes_\FR \CO\big(\DD^m(l)\big)\; \cong_{\mathrm{CAlg}_\FR} \; C^\infty(\FR^{k+m}) / (y_1,\cdots, y_m)^{l+1}\, .
$$

\noindent {\bf (ii)}
The function algebra of any thickened Manifold $M\times \DD^m(l) \in \ThickenedSmoothManifolds$ is canonically identified with the quotient
$$
C^\infty(M)\otimes_\FR \CO\big(\DD^m(l)\big)\; \cong_{\mathrm{CAlg}_\FR} \; C^\infty(M\times\FR^{m}) / (y_1,\cdots, y_m)^{l+1}\, .
$$

\noindent {\bf (iii)} Similarly, the corresponding thickened versions of algebras $C^\infty(\FR^k)$ and $C^\infty(M)$ by any Weil algebra $W\cong_{\mathrm{CAlg}_\FR} \CO\big(\DD^m(l)\big)/I$ may be identified with the quotients
$$
C^\infty(\FR^k)\otimes_\FR W\;  \cong_{\mathrm{CAlg}_\FR}\;
C^\infty(\FR^{k+m}) / \, \tilde{I}
$$
and 
$$
C^\infty(M)\otimes_\FR W\; \cong_{\mathrm{CAlg}_\FR} \; C^\infty(M\times \FR^m) 
/ \, \tilde{I}\, ,
$$
where $\tilde{I}$ is the `lifted' \footnote{Any $I\subset \CO\big(\DD^m(l)\big)$ is finitely presented by polynomials in the nilpotent coordinates. These lift to polynomials in $C^\infty(\FR^m)$.} ideal corresponding to $I\subset \CO\big(\DD^m(l)\big)$, so in particular $(y_1 \cdots y_m)^{l+1}\subset \tilde{I}$.
\end{corollary}
\begin{proof} The first identification follows by expanding any (representative) $f\in C^\infty(\FR^{k+m})$ as 
$$
f= \sum_{|I|\leq l} k_I(x) \cdot y^I + \sum_{|I|=l} y^{I+1} \cdot h_{I+1}(x,y)\, ,
$$ 
by the uniqueness of the $\{k_I(x)\}_{|I| \leq l}$ family (Lem. \ref{PartialHadamardsLemma}). The second identification by the expansion 
of representatives via Cor. \ref{HadamardOnManifoldCartesianProduct}. The final identifications follow completely analogously, keeping track of the extra quotient by polynomials in $\CO\big(\DD^m(l)\big)$ and $C^\infty(\FR^{k+m})$, respectively, with respect to any identification 
$W\cong \CO\big(\DD^m(l)\big)/I$ ; see Lem. \ref{InfinitesimalPointsAsSubspacesOfInfinitesimalDisks}.
\end{proof}

The case {\bf (iii)} of thickened Cartesian spaces is already enough to show that they are, in fact, dual to \textit{finitely generated $C^\infty$-algebras}, i.e., dual to objects of a particularly well-behaved full subcategory of $C^\infty$-algebras \cite{MoerdijkReyes}, and not just $\FR$-algebras. We shall not need to use full definition of a $C^\infty$-algebra here, since for our purposes we may use the characterization of finitely generated $C^\infty$-algebras and $C^\infty$-algebra morphisms between them from \cite[p. 29]{MoerdijkReyes} as a definition for the former: 

\begin{definition}[\bf Finitely generated $C^\infty$-algebras]\label{FinitelyGeneratedCinftyAlgebras}
The category $C^\infty_{f.g.}$-Alg of finitely generated $C^\infty$-algebras is equivalent to the category whose
$\,$ 

\noindent {\bf (i)} objects are $\FR$-algebras of the form
$$
  A 
    \; \cong \;  
  C^\infty(\FR^k)/I 
  \,
  , 
$$

\noindent {\bf (ii)} morphisms
$$
  \phi
  \;:\;
  A 
   \cong 
  C^\infty(\FR^k)/I 
    \xrightarrow{\quad \quad}
  B
   \cong 
  C^\infty(\FR^{k'})/I'
$$ 
between a pair of f.g. $C^\infty$-algebras are given by (equivalence classes of)
$\FR$-algebra morphisms
$$
\psi^* \;:\; C^\infty\big(\FR^k\big) \xrightarrow{\quad \quad} C^\infty\big(\FR^{k'}\big) 
$$
 (equiv. smooth $\psi: \FR^{k'}\rightarrow \FR^k$, by Prop. \ref{CartSptoAlgebras})
 such that the ideal (generated by the image) $\psi^*(I)$ is a subset of the latter target ideal $\psi^*(I)\subset I'$, and where two morphisms 
$\psi^*, \psi'^{*} : C^\infty(\FR^k) \longrightarrow C^\infty(\FR^{k'})$ are equivalent iff the maps
$$
\big(\pr_i \circ\psi - \pr_i \circ \psi'\big)\;\; : \;\; \FR^{k'}\longrightarrow \FR^k \longrightarrow \FR 
$$ 
lie in the ideal $I'\subset C^\infty(\FR^{k'})$ for all $i\in\{1,\cdots, k\}$, where
$\pr_i: \FR^k\rightarrow \FR$ are the canonical coordinate projection maps.
\end{definition}
\begin{remark}[\bf Determined by value on generators]
A most useful property of maps between finitely generated $C^\infty$-algebras is that they are completely determined by their values on the (linear coordinate) ``generators'' of the source $C^\infty$-algebra 
$$
\big[x^j\big]\equiv \big[\pr_j\big]  \; \longmapsto \; \phi \big(\big[x^j\big]\big) \equiv \big[x^j \circ \psi\big] \, .
$$
Evidently, in the characterization of Def. \ref{FinitelyGeneratedCinftyAlgebras}, this extra property is encoded by the fact that any such map descends (dually) from a smooth map of manifolds $\psi:\FR^{k'}\rightarrow \FR^k$. Crucially, this might not be true for \textit{arbitrary} $\FR$-algebra morphisms between $\FR$-algebras of the general form $C^\infty(\FR^k)/I$. However, as we shall show, it is true for thickened Cartesian algebras (Prop. \ref{RAlgebraMapsAreCinfty}).
\end{remark}

\begin{remark}[\bf General $C^\infty$-algebras]
While an $\FR$-algebra structure on $A$ induces, for any polynomial $p\in \mathrm{Pol}(\FR^n, \, \FR^m)$, a map $\widehat{p}:A^{\times n} \rightarrow A^{\times m}$ compatible with the algebra multiplication, a $C^\infty$-algebra structure on $A$ further induces such a map $\widehat{f}: A^{\times n} \rightarrow A^{\times m}$ for any smooth function $f\in C^\infty(\FR^n,\FR^m)$. That is, a general $C^\infty$-algebra is nothing but a product-preserving functor $A:\CartesianSpaces \rightarrow \mathrm{Set}$.
\end{remark}
The original Cahiers topos of Dubuc \cite[\S 4]{Dubuc79}\cite{KockReyes87} is defined as sheaves over site with (Weil) thickened Cartesian spaces as the opposite of a subcategory of $C^\infty_{f.g.}$-Alg with respect to the (good open) coverage \footnote{To be precise, the original definition by Dubuc employs the site of thickened manifolds as an opposite subcategory of $C^\infty_{f.g.}\mbox{-}\mathrm{Alg}$ (rather than thickened Cartesian spaces), but for the given (good open) coverage the two corresponding sheaf categories are equivalent, precisely as in \eqref{SheavesOverManifoldsAndSheavesOverCart}. Hence, we focus solely on the thickened Cartesian $C^\infty$-algebraic site.} as in \eqref{GoodOpenCoversonThCartSp}, that is, with $C^\infty$-algebra morphisms rather than $\FR$-algebra morphisms as we have done. We denote this site by
$$
  C^\infty\mbox{-}\ThickenedCartesianSpaces
   \;  \xhookrightarrow{\quad \quad} 
   \; 
  C^\infty_{f.g.}\mbox{-Alg}^{\mathrm{op}} 
  \;
    \longhookrightarrow  
    \; 
  C^\infty\mbox{-}\mathrm{Alg}^{\mathrm{op}}
  \,  .
$$
In a bit more detail, we note that $C^\infty\mbox{-}\mathrm{Alg}$ enjoys a monoidal (tensor) product $\otimes_\infty$ which is in fact a pushout construction 
and, in particular, satisfies
$$ 
C^\infty(\FR)\otimes_\infty C^\infty(\FR)\; \cong \;  C^\infty(\FR\times \FR)\, ,
$$
and similarly for any two manifolds $M,N \in \SmoothManifolds$.\footnote{By which it immediately follows that the (fully faithful) embedding $\SmoothManifolds\hookrightarrow C^\infty\mbox{-}\mathrm{Alg}^{op}$ preserves products (and in fact furthermore transverse pullbacks/intersections), in contrast to that into $\mathrm{CAlg}_\FR^\mathrm{op}$ from Prop. \ref{CartSptoAlgebras}. } This is in stark contrast to the $\FR$-algebraic tensor product $\otimes_\FR$ which only recovers part of the smooth function on products, $C^\infty(\FR)\otimes_\FR C^\infty(\FR) \subset C^\infty(\FR\times \FR)$. The existence of the $C^\infty$-tensor product is one of the main motivating reasons for considering the notion of $C^\infty$-algebras.

\medskip 
This means that, in particular, when one considers the $C^\infty$-algebraic version of the site, the objects $\FR^k\times \DD
\in C^\infty\mbox{-}\ThickenedCartesianSpaces$ must be given by as formal duals of
$$
C^\infty(\FR^k)\otimes_\infty \CO(\DD)
$$
However, a standard result (\cite[Thm. 5.3]{Kock06}) shows that the two tensor products coincide on products with Weil algebras
$$
C^\infty(\FR^k)\otimes_\infty \CO(\DD) \;  \cong \; 
C^\infty(\FR^k)\otimes_\FR \CO(\DD) \, ,
$$
where the isomorphism is as sets, $\FR$-algebras, or even $C^\infty$-algebras. Hence, one may think of the two versions of the site as having the same objects but potentially different Hom-sets. 

\medskip 
We will show that considering the notion of $C^\infty$-algebra morphisms is, in a sense, redundant when working with the subcategory of thickened Cartesian spaces, in that these coincide precisely with the set of $\FR$-algebra maps between any two finitely generated $C^\infty$-algebras of such spaces. Notice that by Def. \ref{FinitelyGeneratedCinftyAlgebras}, any $C^\infty$-algebra map between (function algebras of) thickened Cartesian space is, in particular, an $\FR$-algebra map, i.e., there is a canonical injection
$$
  \mathrm{Hom}_{C^\infty\mbox{-}\ThickenedCartesianSpaces}
  \big(
    \FR^k\times \DD , \, \FR^{k'}\times \DD'
  \big) 
    \, \xhookrightarrow{\quad \quad} \,\mathrm{Hom}_{\ThickenedCartesianSpaces}\big(\FR^k\times \DD , \, \FR^{k'}\times \DD'\big) \, .
$$
Thus we only need to show that an arbitrary $\FR$-algebra map 
$\phi^*: \CO(\FR^{k'}\times \DD') \longrightarrow \CO(\FR^k\times \DD)$
is in fact a $C^\infty$-algebra map, thus making the above inclusion surjective and so a bijection.

\medskip 
To that end, we note that Hadamard's Lemma \ref{HadamardsLemma} (and Lem. \ref{PartialHadamardsLemma}) may be centered around points different from the origin, precisely as with Taylor expansions. A minor modification of the proof of Lem. \ref{PartialHadamardsLemma} for the case of $k=0$ yields that for any $y_0 \in \FR^m$, any smooth $f\in C^\infty(\FR^m)$ may be expanded as
$$
f(y) = \sum_{|I|\leq l} k_I(y_0) \cdot (y-y_0)^I + \sum_{|I|>l} (y-y_0)^{I+1} \cdot h_{I+1}(y)\,,
$$
for some uniquely defined smooth functions $\{k_I\}_{|I|\leq l} \subset C^\infty (\FR^m)$. Choosing $y=y_0+y_1$ then yields that for any two points $y_0,y_1\in \FR^m$
\begin{align}\label{HadamardOnSum}
f(y_0+y_1)= \sum_{|I|\leq l} k_I(y_0) \cdot (y_1)^I + \sum_{|I|>l} (y_1)^{I+1} \cdot h_{I+1}(y_0+y_1)\,.
\end{align}

\begin{proposition}[\bf $\FR$-algebra maps into thickened Cartesian algebras]\label{RAlgebraMapsAreCinfty}
$\,$

\noindent {\bf (i)} Any $\FR$-algebra map
$$
\phi \;:\; C^\infty(\FR^{k'}) \xrightarrow{\quad \quad} C^\infty(\FR^k)\otimes_\FR W\cong C^\infty(\FR^{k+m})/\,\tilde{I}
$$ 
is a $C^\infty$-algebra map, where the latter identification is via Cor. \ref{ThickenedFunctionAlgebrasAsQuotientAlgebras}.

\noindent {\bf (ii)} More generally, any $\FR$-algebra map 
$$
\phi \;:\; C^\infty(\FR^{k'})/{I'}\xrightarrow{\quad \quad} C^\infty(\FR^k)\otimes_\FR W\cong C^\infty(\FR^{k+m})/\,\tilde{I}\, ,
$$ 
where $I'=(p_1,\cdots,p_n)$ is an ideal finitely generated by \textit{polynomials}, is a $C^\infty$-algebra map.
\end{proposition}
\begin{proof}
Consider first the case for $k'=1$, $I'={0}$, and recall Cor. \ref{ThickenedFunctionAlgebrasAsQuotientAlgebras} that 
\begin{align*}
C^\infty(\FR^k)\otimes W & \;\cong \; C^\infty(\FR^k) \otimes \CO\big(\DD^m(l)\big)/\big(r_1(\epsi),\cdots, r_q(\epsi)\big)\\
&\;\cong \;
C^\infty\big(\FR^{k}_x\times \FR^m_y\big)\Big/ 
\Big\langle \! \big(y^1,\cdots, y^m\big)^{l+1}, \big(r_1(y),\cdots, r_q(y)\big)\! \Big\rangle 
\end{align*}
where $\{y^i\}_{1,\cdots, m}$ are the coordinates on $\FR^m$ and $\{r_j\}_{1,\cdots,q}$ are polynomials in these coordinates. Now any $\FR$-algebra map
$$
\phi: C^\infty(\FR)  \longrightarrow C^\infty(\FR^k) \otimes \CO\big(\DD^m(l)\big)\big/
\big(r^1(\epsi),\cdots, r^q(\epsi)\big) 
$$
acts in particular on the (to be) generator of $t\equiv \id_\FR\in C^\infty(\FR)$
$$
t \quad \longmapsto \;\; \Bigg[\sum_{|I|\leq l} f_I(x) \cdot \epsi^{I}\Bigg]
=  
\Bigg[\sum_{|I|\leq l} f_I(x) \cdot \epsi^{I} + h_i(x,\epsi)\cdot r^i(\epsi)\Bigg] \quad \in \quad C^\infty(\FR^k) \otimes \CO\big(\DD^m(l)\big)/\big(r^1(\epsi),\cdots, r^q(\epsi)\big) \, .
$$
Notice the target element defines a smooth function
$$
f:= \sum_{|I|\leq l} f_I(x) \cdot y^{I} \quad \in \quad C^\infty(\FR^{k+m})\cong \mathrm{Hom}_{\SmoothManifolds}(\FR^{k+m}, \,\FR)
$$
which hence determines a unique $C^\infty\mbox{-}$algebra map 
$$
[f]^* \;:\; C^\infty(\FR) \; \longrightarrow \; C^\infty\big(\FR^{k}_x\times \FR^m_y\big)
\Big/ \Big\langle\! \big(y^1,\cdots, y^m\big)^{l+1}, \big(r_1(y),\cdots, r_q(y)\big)\!\Big\rangle
$$
via its equivalence class 
$$
\big[f\big] = \bigg[f + \sum_{|I|>l} m_I(x,y) \cdot y^I + n_i(x,y)\cdot r^i(y) 
\bigg].
$$ 
The corresponding action may be explicitly determined, for instance, by the pullback (of any representative)
\begin{align*}
f^* \;:\; C^\infty(\FR) &\longrightarrow C^\infty(\FR^{k+m})\\
  g &\longmapsto g \circ f 
    = 
  g\bigg(f_0(x)+
    \sum_{1\leq |I|\leq l} f_I(x) \cdot y^{I}
  \bigg)
\end{align*} 
since the target may be expanded, uniquely up to $\mathcal{O}(y^l)$ for any representative, 
$$
g\big(f_0(x)\big) +   \sum_{1\leq |J|\leq l} \widehat{k}_J\Big(f_0(x),\{f_I(x)\}_{1\leq |I|\leq l}\Big) \cdot  y^J  +\CO\big(y^{l+1}\big)\,.
$$
Here the terms $\widehat{k}_J\big(f_0(x),\{f_I(x)\}_{1\leq |I|\leq l}\big)$ are uniquely determined by the (partial) Hadamard Lemma on sums \eqref{HadamardOnSum} on $g$, together with the corresponding multinomial expansions. The latter, hence, clearly descends to the $C^\infty$-algebra map
$$
  [f]^* \;:\; C^\infty(\FR)\longrightarrow  C^\infty \big(\FR^{k}_x\times \FR^m_y\big)\Big/ 
  \Big\langle \! \big(y^1,\cdots, y^m\big)^{l+1}, \big(r_1(y),\cdots, r_q(y)\big) \!
  \Big\rangle 
$$
claimed above.

Since $[f]^*$ is an algebra map defined via the value $\phi(t)$ it is immediate to see it coincides with
$\phi$ on the polynomial subalgebra of $C^\infty(\FR)$. To witness further that the two $\FR$-algebra maps coincide on all of $C^\infty(\FR)$, notice that for any $x_0 \in \FR_x^k$ we may compose with the evaluation algebra map
$$
\phi_{x_0} := \big(\ev_{x_0}\otimes \id_W\big) \circ \phi \;\; : \;\; 
C^\infty(\FR)  \longrightarrow C^\infty\big(\FR^k_x\big) \otimes W \longrightarrow W 
$$
and the value of $\phi(g)$ is completely determined by the family of evaluations $\phi_{x_0}(g)$, by expanding via the $\FR$-tensor product structure. In other words, the family $\{(\ev_{x_0}\otimes \id_W)\}_{x_0\in \FR^k}$ is ``jointly monic'' (cf. \cite[(8.6)]{Kock06}). The same applies for $[f]^*$ and its evaluations. But each of the $\phi_{x_0}$ evaluation maps coincides precisely with the corresponding evaluation maps $[f]^*_{x_0}$
$$
\phi_{x_0} (g) = [f]^*_{x_0}(g) = g\big(f_0(x_0)\big) +   
\sum_{1\leq |J|\leq l} \widehat{k}_J\Big(f_0(x),\{f_I(x)\}_{1\leq |I|\leq l}\Big)(x_0) \cdot \epsi^J \, ,
$$
since any $\FR$-algebra map into a Weil algebra is necessarily a $C^\infty$-algebra map, and hence completely determined by its value on the coordinate $t$ (\cite[Cor. 3.19]{MoerdijkReyes}\cite[Step 1, prf of Thm. 35.14]{KMS93}.

The case of $\FR$-algebra maps $\phi \,:\, C^\infty(\FR^{k'}) \xrightarrow{\quad \quad} C^\infty(\FR^k)\otimes_\FR W$ for $k'\geq 2$ follows analogously by evaluating $\phi$ on each of the (to be generators) linear functions $t^i = \pr_i \in C^\infty(\FR^{k'})$.

Lastly, consider the case of a source algebra $C^\infty(\FR^{k'})/I'$ with $I'=(p_1,\cdots p_n)$ where each $p_j=p_j (t^1,\cdots,t^{k'})$ is a polynomial in the coordinates of $\FR^{k'}$. The proof proceeds precisely as above, by defining $f:\FR^{k+m}\rightarrow \FR^{k'}$ via the values of $\phi$ on each of the (equivalence classes of) linear coordinates $t^i$, so that in particular $f^*(t^i) = \phi(t^i)$. The only extra thing to show is that 
$$
f^*(I') \, \subset \, \widetilde{I} \, ,
$$
\noindent 
so that $f^*$ indeed (dually) descends and defines a $C^\infty$-algebra morphism, which by the previous argument will coincide with $\phi$ on all of $C^\infty(\FR^{k'})/{I'}$. To see this, it suffices to notice that since $I'$ is generated by \textit{polynomials}, and $\phi$ takes values in $C^\infty(\FR^{k+m})/\widetilde{I}$, we have
$$
p\big(\phi(t^1),\cdots, \phi(t^{k'}) \big)= \phi\big( p(t^1,\cdots, t^{k'})\big) \quad \in \quad \widetilde{I}   
$$
and hence also
\begin{align*}
f^*p:&= p\circ f \equiv p\big(f^*(t^1), \cdots , f^*(t^{k'})\big) \\ 
&=
p\big(\phi(t^1),\cdots, \phi(t^{k'}) \big) = \phi\big( p(t^1,\cdots, t^{k'})\big)  \quad \in \quad \widetilde{I} \, . 
\end{align*}

\vspace{-5mm} 
\end{proof}

\begin{theorem}[\bf{Thickened Cartesian spaces are smooth loci}]\label{ThickenedCartesianSpacesAreSmoothLoci}
The categories $C^\infty\mbox{-}\ThickenedCartesianSpaces$ and $\ThickenedCartesianSpaces$ are \textit{canonically} equivalent
$$
C^\infty\mbox{-}\ThickenedCartesianSpaces \, \cong \, \ThickenedCartesianSpaces \, .
$$
\end{theorem}
\begin{proof}Since the two categories may be viewed as having the `same' objects, it is sufficient to show that the two kinds of Hom-sets are in canonical bijection. This follows directly by {\bf (ii)} of Prop. \ref{RAlgebraMapsAreCinfty} together with the (polynomially finitely generated) quotient identification of thickened Cartesian algebras from Cor. \ref{ThickenedFunctionAlgebrasAsQuotientAlgebras} as a source algebra. 
Alternatively, it can also be seen via the sequence of canonical bijections
\begin{align*}
\mathrm{Hom}_{\mathrm{CAlg}}\big(C^\infty(\FR^{k'})\otimes_\FR W' , \, C^\infty(\FR^k)\otimes W\big) \, &\cong  \mathrm{Hom}_{\mathrm{CAlg}}\big(C^\infty(\FR^{k'}) , \, C^\infty(\FR^k)\otimes W\big) \times \mathrm{Hom}_{\mathrm{CAlg}}\big( W' , \, C^\infty(\FR^k)\otimes W\big) \\
& \cong  \mathrm{Hom}_{C^\infty\mbox{-}\mathrm{Alg}}\big(C^\infty(\FR^{k'}) , \, C^\infty(\FR^k)\otimes_\infty \! W\big) \times \mathrm{Hom}_{C^\infty \mbox{-}\mathrm{Alg}}\big( W' , \, C^\infty(\FR^k)\otimes W\big)\\
&\cong  \mathrm{Hom}_{C^\infty\mbox{-}\mathrm{Alg}}\big(C^\infty(\FR^{k'})\otimes_\infty W' , \, C^\infty(\FR^k)\otimes W\big) \\
&\cong  \mathrm{Hom}_{C^\infty\mbox{-}\mathrm{Alg}}\big(C^\infty(\FR^{k'})\otimes_\FR W' , \, C^\infty(\FR^k)\otimes W\big)\, ,
\end{align*}
where the first uses the colimit property of $\otimes_\FR$ in $\mathrm{CAlg}$, the second uses  Prop. \ref{RAlgebraMapsAreCinfty}{\bf (i)} and the analogous result that can be proven similarly for the source algebra being a Weil algebra (in fact with target any $C^\infty$-algebra, see  \cite[Cor. 3.19]{MoerdijkReyes}\cite[Thm. 5.3]{Kock06}). Finally, we use the colimit property of the $C^\infty$-tensor product in $C^\infty\mbox{-}\mathrm{Alg}$, which in fact coincides with the $\FR$-algebraic one when tensoring with Weil algebras (\cite[Thm. 5.3]{Kock06}) (hence the reason of being unspecified in the second target of the Hom set). 
\end{proof}
Although we do not need to go into the trouble of proving the analogous statement for target algebras being those of thickened manifolds here, it should nevertheless follow similarly -- by first identifying $C^\infty(M)\cong C^\infty(\FR^k)/ I_{s(M)}$, where $s:M\hookrightarrow \FR^k$ is some smooth closed embedding (by Whitney's embedding theorem, see e.g. \cite[p.29]{MoerdijkReyes}) and $I_{s(M)}$ is the ideal of smooth functions that vanish along $s(M)\subset \FR^k$.
\begin{corollary}[\bf Cahiers topos is $\FR$-algebraic]\label{CahiersToposIsRalgebraic}
  When both categories are endowed with the good open coverage from Eq. \eqref{GoodOpenCoversonThCartSp} they are also canonically isomorphic as sites. It follows that the Cahiers topos of Dubuc from \cite{Dubuc79} is equivalent to the topos of thickened smooth sets from Def. \ref{ThickenedSmoothSets}
  $$
  \mathrm{Sh}(C^\infty\mbox{-}\ThickenedCartesianSpaces) \, \cong \, \ThickenedSmoothSets \, .
  $$
\end{corollary}

\begin{remark}[\bf Preservation of products]
By virtue of the Cahiers topos being a well-adapted model for synthetic differential geometry \cite{Dubuc79}\cite{MoerdijkReyes}\cite{Kock06}, the embedding
$$
y\, : \, \SmoothManifolds \xhookrightarrow{\quad \quad} \ThickenedSmoothSets
$$
preserves, in particular, products $y(M\times N)\cong y(M)\times y(N)$.\footnote{In fact, it also preserves transverse pullbacks (intersections).} Notice, this is not immediate from the $\FR$-algebraic incarnation of the (dual) site, if considered with its $\FR$-algebraic monoidal structure, but it follows immediately if instead supplied with the $C^\infty$-algebraic monoidal structure.
\end{remark}

\subsection{Synthetic jet bundle and prolongation}
\label{App-SyntheticJetProlongation}
In this part of the Appendix we explain how the jet bundle itself, along with its thickened smooth structure, may be identified as a natural functorial construction within $\ThickenedSmoothSets$. The same can be said for the jet prolongation operation of sections of jet bundles.
Towards this, we straightforwardly extend the construction of the infinitesimal neighborhoods (Defs. \ref{SyntheticInfinitesimalNeighborhood}, \ref{SyntheticInfinitesimalNeighborhoodOfSubmanifold}) to the following ``\,infinitesimal neighborhood along \textit{morphism} diagonal\,'' $\CG\rightarrow \CF$.

\begin{definition}[\bf Infinitesimal neighborhood along a morphism diagonal]\label{InfinitesimalNeighborhoodOfDiagonal}
 Let any map $\pi: \CG\rightarrow \CF$ be any morphism of thickened smooth sets. The ``\textit{infinitesimal neighborhood along the diagonal of} $\pi$\,'', $\frT_\CF \CG$, is defined as the pullback
\[
\xymatrix@C=1.8em@R=1.4em  { \frT_\CF \CG  \ar[rr] \ar[d] && {\CG} \ar[d]^{\eta_\CF \circ \pi} 
	 \\ 
  \CF \ar[rr]^{\eta_\CF}  &&  \frJ \CF \, .
}
\]
\end{definition}
Explicitly, this is the thickened smooth set with plots 
\begin{align}\label{plotsofFormalDiskBundleofBundle}
\frT_\CF \CG (\FR^k \times \DD) \cong
\Big\{ \big(\phi^k_\epsi, \psi^k_\epsi \big) \in \CF(\FR^k\times \DD)\times \CG(\FR^k\times \DD) \; \big\vert \; \phi^k_{\epsi=0}= \pi\big(\psi^k_{\epsi=0}\big)\in \CF(\FR^k) \Big\} \, ,
\end{align}
i.e., pairs of plots in $\CF$ and $\CG$, which agree in their finite extension upon mapping the latter in $\CF$, but are free to explore infinitesimal directions otherwise.
Denoting the left vertical map by
$$
\pi^{\frT} \,:\, \frT_\CF \CG\longrightarrow \CF\, ,
$$
we see that its fibers  in $\CF$ are precisely all plots in $\CG$, which agree upon projecting to finite directions in $\CF$. For $\CG=*$ a point, this recovers directly Def. \ref{SyntheticInfinitesimalNeighborhood}, while that of an infinitesimal neighborhood around a submanifold $\Sigma$ from Def. \ref{SyntheticInfinitesimalNeighborhoodOfSubmanifold} is recovered via $\CG = \frJ \Sigma$.

\smallskip 
Moreover, consider $\ThickenedSmoothSets_{/\CF}$ as the slice category over a fixed thickened smooth set $\CF$, thought of as `bundles over $\CF$'. By the universal property of the pullback, it follows that  
\begin{align}\label{formalDiskbundleFunctor}
   \frT_\CF \;:\; \ThickenedSmoothSets_{/\CF} & \;\longrightarrow \; \ThickenedSmoothSets_{/\CF}\\
   (\pi:\CG \rightarrow \CF)&\;\longmapsto \; ( \pi^{\frT} : \frT_\CF \CG \rightarrow \CF ) \nn 
\end{align}
is a functor. Explicitly, if $h:\CG\rightarrow \CG'$ is a map of bundles over $\CF$, there is an induced map of corresponding morphism infinitesimal neighborhoods over $\CF$ acting (naturally) on plots as
\begin{align*}
\frT_{\CF} (h) \;:\; \frT_\CF \CG \big(\FR^k\times \DD\big) &\; \longrightarrow \; \frT_\CF \CG' \big(\FR^k\times \DD \big) \\
\big(\phi^k_\epsi, \psi^k_\epsi\big)&\; \longmapsto \; \big (\phi^k_\epsi, h(\psi^k_\epsi) \big) \, .
\end{align*}
\begin{example}[\bf Diagonal neighborhood]\label{Ex-DiagonalNeighborhood}
Consider the case where $\CG=\CF$ and the right vertical map is $\eta_\CF:\CF\rightarrow \frJ \CF$. Then the above construction defines the \textit{(synthetic) infinitesimal neighborhood of the diagonal} (see e.g. \cite[\S 4]{Kock16}), or simply \textit{diagonal neighborhood}, $\mathfrak{T} \CF$ 
as the pullback
\[
\xymatrix@C=1.8em@R=1.4em  { \frT \CF  \ar[rr] \ar[d] && \CF \ar[d]^{\eta_\CF} 
	 \\ 
  \CF \ar[rr]^{\eta_\CF}  &&  \frJ \CF \, .
}
\]
The name is justified since  $\FR^n\times \DD$-plots of $\frT \CF$ are given by
\begin{align*}
\frT \CF(\FR^k\times \DD) \cong
\Big\{ \big(\phi^k_{1,\epsi},\phi^k_{2,\epsi}\big)\in \CF(\FR^k\times \DD)\times\CF(\FR^k\times \DD) \; \big\vert \; \phi^k_{1,\epsi=0}=\phi^k_{2,\epsi=0} \, \in \CF(\FR^k) \Big\} \, ,
\end{align*}
that is, pairs of plots of $\CF$ that agree on their purely finite directions, but are free to differ in their infinitesimal extensions.
\end{example}

It follows that the set of $\infty$-jets of sections of $\pi: \CG\rightarrow \CF$ of Def. \ref{syntheticJets} may be equivalently written as
\begin{align*}
    J^\infty_\CF\CG(*) = \underset{p\in \CF(*)}{\bigcup}\!\mathrm{Hom}_{{\ThickenedSmoothSets}_{/\CF}}/\big(\frT_\CF(*)\, ,\, \CG \big) \, .
\end{align*}
The above formula is suggestive of a way to define the $\FR^k\times \DD$-plots of a would-be thickened smooth set $J^\infty_\CF \CG$. Indeed, for any plot $\phi^k_{\epsi}\in \CF(\FR^k\times \DD)$,i.e., equivalently by Yoneda lemma a morphism $\phi^k_{\epsi {*}}\FR^k\times \DD\rightarrow \CF$ we may consider
\begin{align}\label{formalDiskofplot}
\DD_{\phi^k_\epsi,\CF} \;:&=\;  
\frT_\CF\big(\phi^k_{\epsi{*}}: y(\FR^k\times \DD) \rightarrow \CF \big) \, ,
\end{align}
the \textit{ infinitesimal neighborhood around the $\FR^n\times \DD$-plot $\phi^k_\epsi$}. Sections over all of these are hence deserving to be considered as $\FR^k\times \DD$-plots of the corresponding the $\infty$-jet bundle.

\begin{definition}[\bf Synthetic jet bundle]\label{syntheticJetbundle}
Let $\pi: \CG\rightarrow \CF$ be a map of thickened smooth sets. Define the \textit{synthetic $\infty$-jet bundle} $J^\infty_\CF \CG$ as the assignment  
\begin{align}\label{PlotsofSyntheticJetbundle}
    J^\infty_\CF\CG(\FR^k\times \DD) \, := \underset{\phi^k_\epsi \in \CF(\FR^k\times \DD )}{\bigcup}\!\!\mathrm{Hom}_{\mathrm{\ThickenedSmoothSets}_{/\CF}}\big(\DD_{\phi^k_\epsi,\CF}\, ,\, \CG \big) \, ,
\end{align}
for any probe $\FR^k\times \DD \in \mathrm{ThCartSp}$.
\end{definition}
The definition is in line with intuition. The points of $J^\infty_\CF \CG$ are given by sections over infinitesimal neighborhoods of points in the base, while similarly $\FR^k\times \DD$-plots of $J^\infty_\CF \CG$ are given by sections over infinitesimal neighborhoods around $\FR^k\times \DD$-plots in the base. For this to define a thickened smooth set, it remains to extend the assignment functorially to maps between probe spaces. This is canonical.

\begin{lemma}[\bf Synthetic jet bundle is a thickened smooth set]
\label{SyntheticJetbundleisThickenedSmoothSet}
The assignment of the synthetic $\infty$-jet bundle $J^\infty_\CF\CG$ defines canonically a functor $\mathrm{ThCartSp}^{op}\rightarrow \mathrm{Set}$, which furthermore satisfies the sheaf condition. 
\begin{proof}
 Let $f: \FR^{k'}\times \DD'\rightarrow \FR^k\times \DD$ be a map of thickened Cartesian spaces, denote by $f_{*}: y(\FR^{k'}\times \DD')\rightarrow y(\FR^k\times \DD)$ its image under the Yoneda embedding and by $f^{*} :\CF(\FR^k\times \DD) \rightarrow \CF(\FR^{k'}\times \DD')$ the corresponding map of plots of any formal smooth set $\CF$. There is a map
 \begin{align*}
   \mathrm{Hom}\big(y(\FR^k\times \DD), \CF) \big) &\; \longrightarrow \; \mathrm{Hom}\big(y(\FR^{k'}\times \DD'), \CF \big) \\
   \phi^k_{\epsi{*}}&\; \longmapsto \; \phi^k_{\epsi{*}}\circ f_{*} = \big(f^*\phi^k_\epsi \big)_{*}
 \end{align*}
 for any $\phi^k_{\epsi{*}}$ corresponding to a plot $\phi^k_{\epsi}\in \CF(\FR^k\times \DD)$, by the Yoneda Lemma. 
 Furthermore, for each pair of a plot $\phi^k_{\epsi}$ and a map of Cartesian spaces $f$, there exists a canonical  smooth map
 \begin{align*}
f_{!}:\DD_{f^*\phi^k_\epsi,\CF}&\; \longrightarrow \; \DD_{\phi^k_\epsi,\CF} \\
\big(\phi^m_{\epsi''}, \psi^m_{\epsi''}\big) &\; \longmapsto \; \big(\phi^m_{\epsi''}, f_{*}( \psi^{m}_{\epsi''})\big)
 \end{align*}
 for each $\FR^m\times \DD''$-plot of $\DD_{f^*\phi^k_{\epsi},\CF}$, identified with an element of $\CF(\FR^m\times \DD'')\times \mathrm{Hom}(\FR^m\times \DD'', \FR^{k'}\times \DD')$ by \eqref{plotsofFormalDiskBundleofBundle}. It follows that the image is indeed a plot of $D_{\phi^k_\epsi,\CF}$, since by definition $\phi^m_{\epsi''=0}=(\phi^k_\epsi\circ f) \circ \psi^{m}_{\epsi''=0}$. Abstractly, the map $f_!$ is uniquely determined by the universal property of the pullback $\DD_{\phi^k_\epsi,\CF}$, since the two pullbacks paste as 
 \[
\xymatrix@C=1.8em@R=1.4em  { \DD_{f^* \phi^k_\epsilon,\CF}  \ar[rr] \ar[d]_{f_!} && y(\FR^{k'}\times \DD') \ar[d]^{f_*} 
	 \\ 
  \DD_{\phi^k_\epsi,\CF} \ar[rr] \ar[d] &&y(\FR^k\times \DD) \ar[d]^{\eta_\CF\circ \phi^k_{\epsi{*}}}\\
  \CF \ar[rr]^{\eta_\CF}  &&  \frJ \CF  \, .
}
\]
With the map $f_!$ at hand, we may precompose any section $\sigma^k_{\epsi}$ over $\DD_{\phi^k_{\epsi},M}$ as
\[ 
\xymatrix@R=1em@C=3em  { &&&  \CG\ar[d]^{\pi}
	\\ 
	\DD_{f^*\phi^k_\epsi,\CF}\ar[r]^{f_!} &\DD_{\phi^k_{\epsi},\CF} \ar[rru]^{\sigma^k_\epsi} \ar@{^{(}->}[rr] && \CF
}   
\]
yielding a section $\sigma^{k}_\epsi \circ f_!$ over $\DD_{f^*\phi^k_\epsi,\CF}$. This construction works for every $\phi^k_\epsi  \in J^\infty_\CF \CG(\FR^k\times\DD) $, and so yields a map
\begin{align}\label{pullbackofJetBundleplots}
(-)\circ f_! \;\; :\;\;  J^\infty_\CF \CG\big(\FR^k\times \DD\big)&
\; \longrightarrow \; J^\infty_\CF(\CG) \big(\FR^{k'}\times \DD'\big)\, .
\end{align}
Finally, by construction $(g\circ f)_! = g_!\circ f_!$ and so $J^\infty_\CF \CG$ defines a functor $\mathrm{ThCartSp}^{op}\rightarrow \mathrm{Set}$. Lastly, it is not hard to check the sheaf condition explicitly.
\end{proof}
\end{lemma}
The $\infty$-jet bundle comes equipped with a canonical (smooth) projection
\begin{align}
\label{SyntheticinftyJetprojection}
\pi^\infty \;\; :\qquad\;\; J^\infty_\CF \CG  \qquad  \;\; &\longrightarrow \;\; \CF 
\\[-2pt]
 \Big(\DD_{\phi^k_\epsi,\CF} \xrightarrow{\sigma^k_\epsi} \CG \Big) &\longmapsto \;\; \phi^k_\epsi \nn
\end{align}
sending any section $\sigma^k_\epsi$ over the infinitesimal neighborhood around a plot $\phi^k_\epsi$ to the plot itself. It follows that for any fixed thickened $\CF$ smooth set, we get a functor of bundles over $\CF$
\begin{align}\label{inftyJetFunctor}
   J^\infty_\CF\;:\; \ThickenedSmoothSets_{/\CF} &\; \longrightarrow \; \ThickenedSmoothSets_{/\CF}\\
   \big(\pi:\CG \rightarrow \CF\big)& \;\longmapsto \;
   \big(\pi^\infty: J^\infty_\CF \CG \rightarrow \CF \big)\, , \nn 
\end{align}
as for any bundle map $h:\CG\rightarrow \CG'$ over $\CF$, there is an induced map of $\infty$-jet bundles over $\CF$ 
\begin{align}\label{inftyJetfunctoronMaps}
J^\infty_\CF(h)\;:\;J^\infty_\CF \CG &\; \longrightarrow \; J^\infty_\CF \CG'\, \\
\sigma^k_\epsi &\; \longmapsto \; h\circ \sigma^k_\epsi\,  , \nn
\end{align}
acting by  post-composition on the corresponding sections  
$\sigma^k_\epsi$ over infinitesimal nbds of $\FR^k\times \DD$-plots. Naturality under pullbacks of plots $\sigma^k_{\epsi}\mapsto \sigma^k_\epsi \circ f_!$ follows as in Lemma \ref{SyntheticJetbundleisThickenedSmoothSet}, by associativity of composition,
$$
J^\infty_\CF h\big(\sigma_\epsi^k \circ f_!\big)= 
h\circ \big(\sigma^k_\epsi \circ f_!\big) = 
\big(h\circ \sigma^k_\epsi \big)\circ f_!= 
J^\infty_\CF \big(h (\sigma^k_\epsi)\big)\circ f_!\;.
$$ 

\begin{remark}[\bf Applicability to non-concrete smooth sets]
The construction of the synthetic $\infty$-jet bundle is very general and, beyond bundles of smooth manifolds, it applies to maps between arbitrary thickened smooth sets, even if they do not actually have any underlying points (cf.  \cref{ModuliOfDifferentialForms}). Such spaces are particular examples ``non-concrete'' spaces.\footnote{Namely, thickened smooth sets whose $\FR^k\times \DD$-plots are not determined by the underlying thickened point ${*}\times \DD$-plots.} 
\end{remark}

Nevertheless, this more general construction still satisfies many properties of the traditional infinite jet bundle. For example, we have the following property.

\begin{example}[\bf Jet bundle of identity bundle]
\label{JetBundleofIdentitybundle}
Let $\CF$ be a thickened smooth set and consider $\id_\CF:\CF\rightarrow \CF$ as a bundle over itself. There is a canonical isomorphism
$$
J^\infty_\CF \CF  \;\cong\; \CF\, .
$$
Indeed, to see this
let $\phi^k_\epsi\in \CF(\FR^k\times \DD)$ be a plot, a section $\sigma^k_\epsi$ of $\id:\CF\rightarrow \CF$ over the infinitesimal disk $\DD_{\phi^k_\epsi,\CF}$ around the plot is the same as a lift
\[ 
\xymatrix@R=1.6em@C=3em   { &&  \CF\ar[d]^{\id}
	\\ 
	\DD_{\phi^k_\epsi,\CF} \ar[rru]^{\sigma^k_\epsi} 
 \ar@{^{(}->}[rr]^>>>>>>>>{\iota_{\phi^k_\epsi}\;\;} && \CF\, .
}   
\]
Commutativity of the diagram gives $ \id\circ \sigma^k_\epsi=\sigma^k_\epsi=\iota_{\phi^k_\epsi}$, and so fixes the $\FR^k\times \DD$-plots of $J^\infty_\CF \CF$ to be the set of canonical inclusions of infinitesimal disks around $\FR^k\times \DD$ plots of $\CF$,
\begin{align*}
J^\infty_\CF \CF (\FR^k\times \DD) = \underset{\phi^k_\epsi \in \CF(\FR^k\times \DD )}{\bigcup}
\!\!\!\Big\{\DD_{\phi^k_\epsi,\CF}\xhookrightarrow{\iota_{\phi^k_\epsi}}\CF \Big\} \, .
\end{align*}
Define the canonical map on plots to be the obvious one
\begin{align*}
g \;\;:\quad J^\infty_\CF \CF (\FR^k\times \DD) & \; \longrightarrow \; \CF(\FR^k\times \DD) 
\\[-2pt]
\Big(\DD_{\phi^k_\epsi}\xhookrightarrow{\; \iota_{\phi^k_\epsi}\;}\CF \Big)&
\; \longmapsto \; \phi^k_\epsi \, ,
\end{align*}
which is bijective since, by construction of $\DD_{\phi^k_\epsi}$ in \eqref{formalDiskofplot}, there exists a unique such canonical inclusion $\iota_{\phi^k_\epsi}$ for each plot $\phi^k_\epsi$. What remains to be shown is that the assignment is functional under maps of probes, but this follows by similar arguments. 

\end{example}

\begin{example}[\bf Synthetic jet bundle of trivial bundle]
\label{SyntheticJetoftrivialbundle}
Consider the `trivial bundle' $\pr_2: \CH\times \CF\rightarrow \CF$, then an $\FR^k \times \DD$-plot of $J^\infty_\CF(\CH\times \CF)$ is a map $\big(\sigma^k_{\epsi,1},\sigma^k_{\epsi,2}\big) :\DD_{\phi^k_\epsi}\rightarrow \CH\times \CF$ that satisfies the section condition. This fixes the form of the map to be $\big(\sigma^k_{\epsi,1},\,\iota_{\phi^k_\epsi}\big):\DD_{\phi^k_\epsi}\rightarrow \CH \times \CF$. Equivalently, we have
$$
J^\infty_\CF(\CH\times \CF) (\FR^k\times \DD) \cong \underset{\phi^k_\epsi\in \CF(\FR^k\times\DD)}{\bigcup}
\!\!\!\big\{\DD_{\phi^k_\epsi} \rightarrow \CH\big\} \, .
$$
Following the manifold intuition, this is to be thought of as the space of `jets of maps' $\CF\rightarrow \CH$. Notice, there might be many maps $\DD_{\phi^k_\epsi}\rightarrow \CH$ for each plot $\phi^k_\epsi$ of $\CF$ and so  $J^\infty_\CF(\CH\times \CF)$ need not necessarily be isomorphic to $\CH \times \CF$. Nevertheless, there is still a canonical `inclusion' map 
$$
\CH\times \CF \, \longhookrightarrow \, J^\infty_\CF(\CH\times \CF)
$$
that sends a plot $\big(\psi^k_\epsi, \phi^k_\epsi\big)$ to the section
$$
\Big(\psi^k_{\epsi*}\circ \pi_2 \, ,\, \iota_{\phi^k_\epsi}\Big) \;:\;
\DD_{\phi^k_\epsi} \, \longrightarrow \, \CH\times \CF
$$
over $\DD_{\phi^k_{\epsi}}\xrightarrow{\iota_{\phi^k_\epsi}} \CF$, where $\pi_2:\DD_{\phi^k_\epsi}\rightarrow y(\FR^k\times \DD)$ denotes the canonical map of the pullback definition \eqref{formalDiskofplot}\, .
Note that for manifolds $U\times M\rightarrow M$, the map corresponds to $(u,p)\mapsto j^\infty_p(\mathrm{const}_u)$. 
\end{example}

By construction, the infinitesimal neighborhood bundle functor $\frT_\CF$ and the $\infty$-jet bundle $J^\infty_\CF$ functor 
are closely related, in the precise sense of forming an adjunction. An abstract proof of this relying on differential cohesion may be found in \cite{KS17}. Here, we provide an equivalent proof that relies on elementary but more explicit facts.

\begin{proposition}[\bf Infinitesimal neighborhood / Jet bundle adjunction]
\label{FormaldiskbundleJetbundleAdjunction}
For each $\CF \in \ThickenedSmoothSets$,
the infinitesimal neighborhood bundle functor and the $\infty$-jet bundle functor form an adjunction 
$$
\frT_\CF \;\; \dashv \;\; J^\infty_\CF \quad : \quad \ThickenedSmoothSets_{/\CF} \longrightarrow \ThickenedSmoothSets_{\CF}\,.
$$
\begin{proof}
The claim is that there exists a bijection
\begin{align*}
\mathrm{Hom}_{{\ThickenedSmoothSets}_{/\CF}}\big(\frT_\CF \CH, \, \CG\big)
\cong 
\mathrm{Hom}_{{\ThickenedSmoothSets}_{/\CF}}\big(\CH, \, J^\infty_\CF \CG\big)
\end{align*}
which is natural in both $\CH$ and $\CG$, with the fixed maps to $\CF$ left implicit. Considering first the case where
$\CH=\FR^k\times \DD$ is a representable object equipped with a plot $\phi^k_{\epsi}: \FR^n\times \DD\rightarrow \CF$, 
then 
\begin{align*}
\underset{\phi^k_\epsi \in \CF(\FR^k\times \DD )}{\bigcup} \!\! \mathrm{Hom}_{\mathrm{\ThickenedSmoothSets}_{/\CF}}\Big(\frT_\CF\big(\FR^k\times \DD\big)\, , \, \CG\Big)& \; =: \; J^\infty_\CF \CG(\FR^k\times \DD) 
\\[-12pt] 
 & 
 \; \cong \; \underset{\phi^k_\epsi \in \CF(\FR^k\times \DD )}{\bigcup} \!\! \mathrm{Hom}_{{\ThickenedSmoothSets}_{/\CF}}\Big( \FR^k\times \DD
 ,\, J^\infty_\CF \CG \Big) ,
\end{align*}
where the first line is by definition, while the second is by Yoneda lemma. In particular, this implies that 
\begin{align*}
\mathrm{Hom}_{\ThickenedSmoothSets_{/\CF}}\Big(\frT_\CF\big(\FR^k\times \DD\big)\, , \, \CG\Big)
\; \cong \; 
\mathrm{Hom}_{{\ThickenedSmoothSets}_{/\CF}}\Big( \FR^k\times \DD
 ,\, J^\infty_\CF \CG \Big)
\end{align*}
naturally in $\CG$ and in the representable space $\FR^k\times \DD$, and so the adjunction holds when the source is an arbitrary representable. 

But any source sheaf $\CH$ may be written as the colimit of representables, $\CH\cong \underset{i\in I}{\mathrm{colim}}\, (\FR^{k_i}\times \DD_i)$ and since, by stability of colimits in sheaf topoi, $\frT_\CF$ preserves colimits, it follows that
\begin{align*}
\mathrm{Hom}_{\mathrm{\ThickenedSmoothSets}_{/\CF}}(\frT_\CF \CH , \CG) & \; \cong  \;
\mathrm{Hom}_{\mathrm{\ThickenedSmoothSets}_{/\CF}}\bigg(\underset{i\in I}{\mathrm{colim}}
\, \frT_\CF (\FR^{k_i}\times \DD_i) \, , \, \CG \bigg)\\
&\cong \; \underset{i\in I}{\lim} \,\mathrm{Hom}_{{\ThickenedSmoothSets}_{/\CF}}\big(\frT_\CF (\FR^{k_i}\times \DD_i)\, , \, \CG\big) \\
&\cong \; \underset{i\in I}{\lim}\, \mathrm{Hom}_{\ThickenedSmoothSets_{/\CF}}\big(\FR^{k_i}\times \DD_i\,, \, J^\infty_\CF \CG \big) \\
&\cong \; \mathrm{Hom}_{\mathrm{\ThickenedSmoothSets}_{/\CF}} \big( \CH, J^\infty_\CF \CG)\, .
\end{align*}

\vspace{-4mm} 
\end{proof}
\end{proposition}

The following intuitively expected result is originally due to \cite{KS17}, which shows that the traditional notion of an infinite jet bundle together with its (thickened) smooth structure coincides with the internal synthetic definition.

\begin{theorem}[\bf Synthetic jet bundle of fiber bundle]\label{SyntheticJetBundleOfFiberBundle} For any fiber bundle $F\rightarrow M$ of finite-dimensional manifolds, the synthetic jet bundle construction (Def. \ref{syntheticJetbundle}) recovers precisely  the corresponding infinite jet bundle limit $J^\infty_M F$ formed in thickened smooth sets (Def. \ref{InfiniteJetBundleFormalSmoothLimit}), i.e., 
$$
J^{\infty}_{y(M)} y(F) \, \cong \,  J^\infty_M F \, ,
$$
where $J^\infty_M F = \lim_{\ThickenedSmoothSets} y\big(J^k_M F\big)$.
\begin{proof}
(Sketch) The points of both smooth sets coincide, since by construction
 $$
 J^\infty_{y(M)}y(F)(*) \; := \; \underset{p\in M}{\bigcup} \mathrm{Hom}_{\ThickenedSmoothSets/y(M)}\big(\DD_{p,M}\, , \, y(F) \big) 
 \; \cong \; \underset{p\in M}{\bigcup}J^\infty_p F=: y(J^\infty_MF)(*) \, ,
 $$
due to Lem. \ref{JetsofSections=InfinitesimalJets} and Lem. \ref{SyntheticInfinitesimalNeighborhoodOfManifold}. The rest of the proof can be found in \cite{KS17}, following similarly as an application of the above adjunction and the local triviality of finite-dimensional manifolds. 
\end{proof}
\end{theorem}

In view of the former theorem, Ex. \ref{JetBundleofIdentitybundle} applied to the case of $\id:M\rightarrow M$ for a finite-dimensional manifold translates to the statement that the identity map has vanishing traditional jets. Similarly, Ex. \ref{SyntheticJetoftrivialbundle} applied to the case of $F\times M\rightarrow M$ recovers the (smooth) set of infinity jets of maps $M\rightarrow F$. Furthermore, the theorem reassures that the synthetic infinite bundle is the natural generalization of the usual jet bundle of a finite-dimensional fiber bundle, in the vast realm of maps between (thickened) smooth spaces, potentially infinite-dimensional, with no local triviality and even non-concrete spaces. Indeed, the jet prolongation of sections may be defined internally and naturally in this context too.

\begin{definition}[\bf Synthetic infinite jet prolongation]\label{SyntheticInfiniteJetProlongation} Let $\pi:\CG\rightarrow \CF$ be a map of smooth sets and $\pi^\infty : J^\infty_\CF \CG\rightarrow \CF$ its infinite jet bundle. The \textit{infinite jet prolongation} is the map (of sets)
$$
j^\infty\;:\;
\Gamma_\CF(\CG)\longrightarrow \Gamma_\CF(J^\infty \CG)
$$
defined by the composition
\begin{align*}
(\sigma:\CF \rightarrow \CG) \;\; \longmapsto \;\;\Big(\CF\xrightarrow{\sim}J^\infty_\CF \CF \xrightarrow{J^\infty_\CF( \sigma)} J^\infty_\CF \CG \Big)\,,
\end{align*}
where the first identification on the right is given in Ex. \ref{JetBundleofIdentitybundle} and the second map is the application of the infinity jet functor \eqref{inftyJetfunctoronMaps} on the smooth map $\sigma$.
\end{definition}

Explicitly, in terms of plots the section $j^\infty \sigma$ is the assignment that fits in the diagram
\[ 
\xymatrix@R=1.2em@C=3em   { &&  J^\infty_\CF \CG(\FR^k\times \DD)\ar[d]^{\pi}
	\\ 
	\CF(\FR^k\times \DD) \ar[rru]^{j^\infty \sigma} \ar[rr]^>>>>>>>>{\id} && \CF(\FR^k\times \DD)\, ,
}   
\]
naturally in $\FR^k\times \DD$, given by
\begin{align*}
j^\infty \sigma \;:\; \CF(\FR^k\times \DD) & \; \longrightarrow \; J^\infty_\CF \CG (\FR^k\times \DD) 
\\ 
\phi^k_\epsi &\; \longmapsto \; 
\big(\sigma \circ \iota_{\phi^k_\epsi}: \DD_{\phi^k_\epsi,\CF}\rightarrow \CG \big) \, .
\end{align*}
The synthetic infinite jet prolongation naturally extends to a smooth map between the corresponding smooth sets of sections of the corresponding bundles,
\begin{align*}
j^\infty_\CF \;:\; \mathbold{\Gamma}_\CF(\CG) \longrightarrow \mathbold{\Gamma}_\CF(J^\infty \CG)\, ,
\end{align*}
where the (thickened)  smooth set structure of the spaces of sections is exactly as in the finite-dimensional case 
\eqref{FieldSpaceAsPullback}, i.e.,
$$
\mathbold{\Gamma}_\CF(\CG) \big(\FR^{k'}\times \DD' \big)= 
\Big\{\sigma^{k'}_{\epsi'} \,:\, y\big(\FR^{k'}\times \DD' \big) \times \CF \rightarrow \CG \;\; \big\vert \;\; \pi \circ \sigma^{k'}_{\epsi'} = \pr_2 \Big\} \, , $$
with $\FR^{k'}\times \DD'$ plots of smooth sections defined as $\FR^{k'}\times \DD'$-parametrized sections of $\CG\rightarrow \CF$, and similarly for $\mathbold{\Gamma}_\CF (J^\infty \CG)$. Indeed, 
$$
j^\infty_\CF \;:\; \mathbold{\Gamma}_\CF(\CG)\big(\FR^{k'}\times \DD'\big)
\; \longrightarrow \; \mathbold{\Gamma}_\CF(J^\infty \CG)\big(\FR^{k'}\times \DD'\big) 
$$
is defined by the composition
\begin{align*}
\Big(\sigma^{k'}_{\epsi'}\,:\, y\big(\FR^{k'}\times \DD'\big)\times \CF \rightarrow \CG\Big)
\;\; \longmapsto \;\;
\bigg(y\big(\FR^{k'}\times \DD'\big)\times \CF \hookrightarrow J^\infty_{\CF}\big(y(\FR^{k'}\times \DD')\times \CF\big)\xrightarrow{\;J^\infty_\CF(\sigma^{k'}_{\epsi'})\;} J^\infty_\CF(\CG) \! \bigg)    
\end{align*}
with the first inclusion on the right being that of Ex. \ref{SyntheticJetoftrivialbundle}. Explicitly, the smooth map on the right is given on plots by
$$
\big(\psi^k_\epsi,\, \phi^k_\epsi\big) \; \longmapsto\;
\Big( \sigma^{k'}_{\epsi'}\circ \big(\psi^k_{\epsi *} \circ \pi_2 \, , \, \iota_{\phi^k_\epsi}\big)
\;:\, \DD_{\phi^k_\epsi}\rightarrow y\big(\FR^{k'}\times \DD'\big) \times \CF \rightarrow \CG \Big) \, .
$$
\begin{lemma}[\bf Synthetic jet prolongation for fiber bundles]
\label{SyntheticJetprolongforFiberBundle}
Let $\pi: F\rightarrow M$ be a fiber bundle of finite-dimensional manifolds. The synthetic prolongation 
$$
j^\infty_{y(M)}\;:\;\mathbold{\Gamma}_{y(M)}\big(y(F)\big) 
\; \longrightarrow \; 
\mathbold{\Gamma}_{y(M)}\big(J^\infty y(F) \big)\cong \mathbold{\Gamma}_{y(M)}\big(y(J^\infty F)\big)$$
coincides with  the traditional (smooth) jet prolongation
$$j^\infty_M\;:\; \mathbold{\Gamma}_M (F)\longrightarrow \mathbold{\Gamma}_M (J^\infty F)$$
from Lem. \ref{ThickenedInfiniteJetProlongation}, under the natural identifications of the corresponding smooth sets. That is, the diagram 
\[
\xymatrix@C=2.6em@R=1.4em  { \mathbold{\Gamma}_{y(M)}\big(y(F)\big) \ar[rr]^{j^\infty_{y(M)}}  && \mathbold{\Gamma}_{y(M)}\big(J^\infty y(F) \big) \ar[r]^{\sim} &\mathbold{\Gamma}_{y(M)}\big(y(J^\infty F)\big) 
	 \\ 
  \mathbold{\Gamma}_M(F) \ar[rr]^{j^\infty_M} \ar[u]^{\sim}  &&  \mathbold{\Gamma}_M(J^\infty F) \ar[ru]_{\sim} 
}
\]commutes, where the (canonical) vertical isomorphisms are those of Yoneda Lemma. 
\end{lemma}
\begin{proof}
Denote by $\sigma_*$ the Yoneda embedding of a section $\sigma : M\rightarrow F$. The synthetic jet prolongation gives a section $j^\infty \sigma_*: y(M) \rightarrow J^\infty_{y(M)}\big(y(F)\big)$ which acts on points by
\begin{align*}
y(M)(*) &\;\longrightarrow \; J^\infty_{y(M)}y(F) (*)\cong y(J^\infty_M F)(*) \\
p& \; \longmapsto \; \big(\sigma_{*}\circ \iota_p : \DD_{p,y(M)}\rightarrow F\big)\, , 
\end{align*}
where the isomorphism on the right-hand side is by the previous theorem. By Lem. \ref{JetsofSections=InfinitesimalJets} and Lem. \ref{SyntheticInfinitesimalNeighborhoodOfManifold}, this corresponds to $p\mapsto j^\infty_p \sigma \in J^\infty_M(F)$, i.e.,  the traditional prolongation $j^\infty_M \sigma : M\rightarrow J^\infty_M(F)$. Analogously, $j^\infty \sigma_*$ acts on $\FR^k\times \DD$-plots by
\begin{align*}
y(M)(\FR^k\times \DD) &\;\longrightarrow \; J^\infty_{y(M)}y(F) (\FR^k\times \DD)\cong y(J^\infty_M F)\big(\FR^k\times \DD\big) \\
\phi^k_{\epsi *}&\; \longmapsto \; \big(\sigma_{*}\circ \iota_{\phi^k_\epsi} : \DD_{\phi^k_\epsi,y(M)}\rightarrow F\big)\, . 
\end{align*}
Arguing similarly to Lem. \ref{JetsofSections=InfinitesimalJets} and Lem. \ref{SyntheticInfinitesimalNeighborhoodOfManifold}, this corresponds to $\phi^k_\epsi\mapsto j^\infty_M\sigma \circ \phi^k_\epsi \in y(J^\infty_M F)(\FR^k\times \DD)$. Thus overall, the synthetic prolongation $j^\infty \sigma_*$ of a smooth section $\sigma:M\rightarrow F$ coincides with (the Yoneda embedding of) the traditional prolonged smooth section $j^\infty_ M \sigma $.

The above shows that the two jet prolongation maps agree on sections of the bundle, i.e., points of $\mathbold{\Gamma}_M (F)\cong \mathbold{\Gamma}_{y(M)}\big(y(F)\big)$. It remains to show they agree on $\FR^{k'}\times \DD'$-parametrized sections of the bundle, i.e. $\FR^{k'}\times \DD'$-plots of $\mathbold{\Gamma}_M (F)$. As before, for $\sigma^{k'}_{\epsi'}\in \mathbold{\Gamma}_M(F)(\FR^{k'}\times \DD')$ denote by $\sigma^{k'} _{\epsi' *} \in\mathbold{\Gamma}_{y(M)}\big(y(F)\big)(\FR^{k'}\times \DD')$ the corresponding plot. Then its synthetic jet prolongation is a smooth map $j^\infty \sigma^{k'} _{\epsi' *}:y(\FR^{k'}\times \DD')\times y(M) \rightarrow J^\infty_{y(M)}\big(y(F)\big)$ which acts on points by 
\begin{align*}
y\big(\FR^{k'}\times \DD'\big) \times y( M)(*) &\;\longrightarrow \; 
J^\infty_{y(M)}y(F) (*)\cong y\big(J^\infty_M F\big)(*) 
\\
(x,p)&\; \longmapsto \; \Big(\sigma^{k'}_{\epsi' *}\circ (x\circ \pi_2 , \iota_p) : \DD_{p,y(M)}\rightarrow \big(y(\FR^{k'}\times \DD') \times y(M)\big) \rightarrow F\Big)\, , 
\end{align*}

In a similar vein as before, it can be checked that this corresponds to the (vertical) traditional jet prolongation $j^\infty_M \sigma^{k'}_{\epsi'}\in  \mathbold{\Gamma}_{M}(J^\infty F)(\FR^{k'}\times \DD')$, when acting on the underlying points of $\FR^{k'}\times \DD' \times M$. Modulo explicit complicated formulas, it follows that $j^\infty \sigma^{k'}_{\epsi' *}$ corresponds to the traditional jet prolongation acting on $\FR^k\times \DD$-plots of $\FR^{k'}\times \DD' \times M$.
\end{proof}

\subsection{The full classifying nature of the de Rham moduli space}
\label{ModuliOfDifferentialForms}
In this manuscript we have defined the moduli space of de Rham forms (Ex. \ref{ThickenedModuliSpaceOfdeRhamForms}) over the site of thickened Cartesian spaces, by assigning an algebraic version of forms as its plots along infinitesimal probes. By the fully faithful embedding of manifolds into thickened smooth sets (Ex. \ref{ManifoldsWithCornersExample}), the fact that $\mathrm{Sh}(\ThickenedSmoothManifolds)\cong \ThickenedSmoothSets$ and the Yoneda Lemma, this does classify forms \textit{on smooth manifolds} in the traditional sense 
$$
\Omega^1(M) := \mathrm{Hom}^{\mathrm{fib.lin.}}_{\SmoothManifolds}(TM, \FR)  \; \cong \;  \mathrm{Hom}_{\ThickenedSmoothSets}(M,\mathbold{\Omega}^1)\, ,
$$
and similarly for $n$-forms. 

\medskip 
However, the full classifying nature of $\mathbold{\Omega}^1$ as originally envisioned in the setting of synthetic differential geometry \cite{Law80}, where it should modulate / classify forms defined as maps out of arbitrary synthetic tangent bundles $T\CF =[\DD^1(1), \, \CF] $, requires a slightly different extension in the thickened setting rather than the ``Kähler'' extension of Ex. \ref{ThickenedModuliSpaceOfdeRhamForms}. That is, instead of defining its infinitesimal plots by generalizing the algebraic incarnation of forms on the tangent bundle, one defines its plots via the -- maps out of the (synthetic) tangent bundle -- picture. For brevity, and for reasons of further generality, in this Appendix we shall denote $\CE := \ThickenedSmoothSets$ and any thickened Cartesian probe as $e\in E := \ThickenedCartesianSpaces$.

\begin{definition}[\bf Fully classifying moduli space of de Rham forms]\label{FullyClassifyingThickedModuliSpaceOf1Forms}
The fully classifying \textit{thickened moduli space of de Rham $1$-forms} $\widehat{\mathbold{\Omega}}^1$ is defined as the sheaf $\widehat{\mathbold{\Omega}}^1 \in \ThickenedSmoothSets$ with $e$-plots given by
$$
\widehat{\mathbold{\Omega}}^1 (e)\, := \, \Hom_{\CE}^{\mathrm{fib.lin.}}\big(Te, \, \FR\big)\, .  
$$ 
\end{definition}
\begin{remark}[\bf Original form moduli definitions]\label{OriginalDefinitions}
 The original axiomatic definition (cf. \cite[\S 3]{Law80}\cite[\S I.20]{Kock06}) identifies the above object instead as a certain equalizer, which can be shown to be equivalent to the above when expressed plot-wise in our setting. On the other hand, the original source \cite[\S 8]{DK84}, which firstly identifies the classifying space in an explicit sheaf topos, works over a larger category of probes whose (different) definition may still be shown to coincide with an extension of the above in such a larger sheaf topos (Cor. \ref{ModuliSpacesAreIsomorphic}).
\end{remark}

Of course, this version of the moduli space still classifies differential forms on ordinary manifolds for exactly the same reason, i.e., by the Yoneda Lemma. Moreover, as we shall see, if $Te$ were representable in $\CE$ (which is not the case for thickened smooth sets), then this definition would coincide with that of Ex. \ref{ThickenedModuliSpaceOfdeRhamForms} (Cor. \ref{ModuliSpacesAreIsomorphic}). Otherwise, as in our sheaf topos, the relation between the two is that the original $\mathbold{\Omega}$ sits as a subspace inside the latter $\widehat{\mathbold{\Omega}}^1$ (Lem. \ref{RelationOfTwoModuliSpaces}). The reason for the latter's more ``internal'' definition will become clear from the following discussion.

\medskip 
\noindent {\bf The amazing right adjoint.}
Recall, by the properties of the internal hom functor, the synthetic tangent functor $T(-):=[\DD^1(1),-]:\CE \rightarrow \CE $ enjoys a left adjoint \eqref{TangentFunctorHasLeftAdjoint}, 
$$\big(\DD^1(1)\times -\big ) \quad \dashv \quad  T(-)\, ,$$
a property that has proven to be extremely useful in thickened smooth sets (cf. Lem. \ref{ManifoldMappingSpaceTangentBundle}, Prop. \ref{SyntheticTangentBundleOfFieldSpace}, Lem. \ref{SyntheticInfiniteJetTangentBundle}, Cor. \ref{OnshellSyntheticTangentBundle}).
The fact that the synthetic tangent functor should also have a right adjoint was first pointed out in the context of synthetic differential geometry by \cite{Law80}.\footnote{Since then, an object $K\in C$ of a Cartesian closed category $C$ whose mapping space functor $[K,-]$ has a right adjoint is called ``tiny'' or ``atomic''.} 
If such a right adjoint exists, then $[\DD^1(1),-]$ must preserve colimits, which is indeed the case (see \cite[Appendix 4]{MoerdijkReyes}). It turns out that for any sheaf topos, such as $\ThickenedSmoothSets$ and $\CE$, preserving colimits is also sufficient for the existence of a right adjoint, denoted by $(-)_{\DD^1(1)}:\CE \rightarrow \CE$,
\footnote{As proven in \cite{MoerdijkReyes}, \textit{any} infinitesimal point $\DD\in \mathrm{ThCartSp}\hookrightarrow E$ in our site has the property that $[\DD,-]$ preserves colimits, and hence has also a right adjoint $(-)_{\DD}$ given by the same formula.}
$$
\big(\DD^1(1)\times -\big ) \quad \dashv \quad T(-) \quad  \dashv\quad  (-)_{\DD^1(1)} \, .
$$
In other words, there is a canonical bijection 
\begin{align}\label{RightAdjointtoTangentFunctor} 
\mathrm{Hom}_{\CE}(T\CF,\, \CG) \;\; \cong \;\; 
\mathrm{Hom}_{\CE}\big(\CF, \CG_{\DD^1(1)}\big)
\end{align}
for any two smooth spaces $\CF,\CG$, naturally in both entries. Explicitly for any $\CG\in \CE$, the space $\CG_{\DD^1(1)}$ is defined via \footnote{The formula 
defines a \textit{presheaf} $\CG_U$ for any $U\in E$. It is the colimit-preserving property of $[\DD^1(1),-]$ that guarantees the sheaf condition.} 
\begin{align*}
\CG_{\DD^1(1)}\,:\, \mathrm{E}^{op}&\xrightarrow{\quad \quad} \quad \mathrm{Set} 
\\[-3pt]
   e \;\; &\xmapsto{\quad \quad}  \mathrm{Hom}_{\CE} \big(\big[\DD^1(1),\, e\big]\, , \, \CG \big)\, ,
\end{align*}
and so by definition of the synthetic tangent bundle, the $e$-plots consist of all smooth maps $Te\rightarrow \CG$ from the tangent bundle of $e$ \eqref{FunctionAlgebraOnTangentProbe}. Of particular importance for us is 
the value of this right adjoint on the real line, 
$$ 
\FR_{\DD^1(1)}\;\; \in \;\;  \CE
\, ,$$ whose $e$-plots are identified with \textit{all real-valued smooth maps}, $\mathrm{Hom}_\CE(Te,\FR)$, thus potentially \textit{nonlinear} in the fibers. It follows by the adjoint bijection \eqref{RightAdjointtoTangentFunctor} that this thickened space \textit{classifies real valued maps out of the synthetic tangent bundle} $\CF$ of any space sheaf $\CF\in \CE$, in that
\begin{align}\label{ClassifyingFunctionsOfTangentBundle}
\mathrm{Hom}_{\CE}(T\CF,\, \FR) \; \cong \;
\mathrm{Hom}_{\CE}\big(\CF, \FR_{\DD^1(1)}\big) \, .
\end{align}
With the fully classifying moduli space 
$\widehat{\mathbold{\Omega}}^1\in \ThickenedSmoothSets$ defined as in Def. \ref{FullyClassifyingThickedModuliSpaceOf1Forms}, it is immediately seen to be a subobject of $\FR_{\DD^1(1)}$, as its $e$-plots are simply the subset of \textit{linear} (smooth) maps $Te\rightarrow \FR$ in $\CE$.
\begin{corollary}[\bf Form moduli as linear subobject]\label{1FormModuliSpaceIsSubobject}
The moduli space of 1-forms $\widehat{\mathbold{\Omega}}^1$ is canonically identified with a (thickened smooth) subspace of $\FR_{\DD^1(1)}$
\begin{align*}
\widehat{\mathbold{\Omega}}^1\longhookrightarrow \FR_{\DD^1(1)} \quad \in \;\; \CE \, . 
\end{align*}
\end{corollary}

Our approach to the definition of the moduli space now allows for a direct \footnote{At least, somewhat more direct than the original source 
\cite{DK84}.} extraction of its classifying property. Indeed, given the above embedding and the fact that $\FR_{\DD^1(1)}$ classifies \textit{all real-valued} maps out of synthetic tangent bundles $T\CF$ (cf. \eqref{ClassifyingFunctionsOfTangentBundle}), the classifying nature of the moduli space of 1-forms $\widehat{\mathbold{\Omega}}^1 \in \CE$ becomes apparent: It classifies a subset of 
smooth maps $T\CF \rightarrow \FR$ out of the synthetic tangent bundle, which in the case of (an infinitesimally linear) $\CF$ with a fiberwise $\FR$-linear structure coincides with 
those smooth maps that are furthermore linear in the fibers. 

\begin{proposition}[\bf Classifying nature of form moduli]\label{ClassifyingNatureOfFormModuli}
$\,$
\begin{itemize}[leftmargin=19pt]
\item[\bf (i)] For any generalized thickened smooth space $\CF\in \CE$, there is a canonical injection
\begin{align*}
 \mathrm{Hom}_{\CE}\big(\CF, \widehat{\mathbold{\Omega}}^1\big) \longhookrightarrow \mathrm{Hom}_{\CE}\big(T\CF,\,\FR \big) \, .
\end{align*}
\item[\bf (ii)] Furthermore, for any (infinitesimally linear) $\CF$ whose tangent bundle $T\CF$ has a fiberwise $\FR$-linear structure, the image of this inclusion corresponds to fiberwise linear maps
\begin{align*}
 \mathrm{Hom}_{\CE}\big(\CF, \widehat{\mathbold{\Omega}}^1\big)\; \cong \; \mathrm{Hom}_{\CE}^{\mathrm{fib.lin.}}\big(T\CF, \, \FR \big) \, .
\end{align*}
\end{itemize} 
\begin{proof}
Let $\CF$ be any thickened smooth space. Postcomposing any smooth map $\CF\rightarrow \widehat{\mathbold{\Omega}}^1$ with the monomorphism of Cor. \ref{1FormModuliSpaceIsSubobject} yields the 
injection 
$$
\mathrm{Hom}_{\CE}\big(\CF, \widehat{\mathbold{\Omega}}^1\big) \longhookrightarrow
\mathrm{Hom}_{\CE}\big(\CF, \,\FR_{\DD^1(1)} \big) \, .
$$
Under the canonical bijection of the adjoint relation \eqref{RightAdjointtoTangentFunctor}, this yields
$$ 
\mathrm{Hom}_{\CE}\big(\CF, \,\widehat{\mathbold{\Omega}}^1\big) 
\longhookrightarrow
\mathrm{Hom}_{\CE}\big(\CF, \,\FR_{\DD^1(1)} \big) \cong \mathrm{Hom}_{\CE}(T\CF, \,\FR )\, , 
$$
proving the first statement. 

Now notice that when $\CF$ is considered as a colimit on representable probes, $\CF\cong \mathrm{colim}_{i} (e^i)$, then we also have
\begin{align*}
 \mathrm{Hom}_{\CE}\big(\CF, \,\widehat{\mathbold{\Omega}}^1\big) &\cong 
 \mathrm{Hom}_{\CE}\big(\mathrm{colim}_{i} (e^i),\,  \widehat{\mathbold{\Omega}}^1\big)  \\ &\cong  \lim_{i,\mathrm{Set}}\mathrm{Hom}_{\CE}\big(e^i,\, \widehat{\mathbold{\Omega}}^1 \big) \\ 
 &\cong \lim_{i,\mathrm{Set}}\mathrm{Hom}_{\CE}^{\mathrm{fib.lin.}}\big(Te^i , \, \FR\big)\, ,
\end{align*}
where we used that Hom functors commute with (co)limits in the first entry, and then Def. \ref{FullyClassifyingThickedModuliSpaceOf1Forms} along with the Yoneda Lemma. 
Next, consider the case where $T\CF$ has a fiberwise linear structure. Since fiberwise linearity is a functorial property, i.e., a composition of fiberwise linear maps is also such, it follows that the image of the inclusion 
\begin{align*}
\lim_{i,\mathrm{Set}}\mathrm{Hom}_{\CE}^{\mathrm{fib.lin.}}(Te^i, \,\FR) 
\;\; \longhookrightarrow \;\; \lim_{i,\mathrm{Set}}\mathrm{Hom}_{\CE}(Te^i, \,\FR)
& \cong   \mathrm{Hom}_{\CE}(\mathrm{colim}_{i}(Te^i), \,\FR) 
\\[-3pt]
&\cong \mathrm{Hom}_{\CE}(T\CF, \,\FR)  
\end{align*}
is the set of the fiberwise linear maps $T\CF\rightarrow \FR$, due to the universal cocone property of the colimit.
That is,
$$
\Hom_\CE\big(\CF, \, \widehat{\mathbold{\Omega}}^1\big)  \, \cong \, \lim_{i,\mathrm{Set}}\mathrm{Hom}_{\CE}^{\mathrm{fib.lin.}}\big(Te^i , \, \FR\big) \, \cong \,  \Hom_\CE^{\mathrm{fib.lin.}}(T\CF,\,\FR) \, \longhookrightarrow \, \Hom_\CE(T\CF ,\FR) \, .  
$$ 
The detailed justification of the latter intuition \footnote{This seems to be only implicit in the original proof from \cite[Thm. 8.2]{DK84}.} is the fact that the colimit $\mathrm{colim}_{i}(Te^i)$ may be equivalently computed in the slice category,  $\mathrm{FibLin}_\CF$, of fiberwise linear ``bundles'' over $\CF$. Indeed, this follows since corresponding pushforward maps $\{Te^i \rightarrow T\CF\}_{i\in I}$ of the diagram $\{e^i \rightarrow \CF\}_{i\in I}$ are all fiberwise linear and they form a covering of $T\CF$ (in the ambient topos, i.e., a jointly epic family), by virtue of $\{e^i \rightarrow \CF\}_{i\in I}$ being such a covering of $\CF$ and $T=[\DD^1(1)\, , \, - ]$ preserving jointly epic families (since it preserves colimits, being a left adjoint). The latter implies precisely that $T\CF$ satisfies the corresponding universal cocone property in $\mathrm{FibLin}_\CF\hookrightarrow \CE_\CF$, namely given a compatible family of fiberwise linear maps $f_i : T e^i \rightarrow \FR$, there exists a unique fiberwise linear $f: T\CF \rightarrow \FR$ such that $f_i = f\circ T p_i$. \footnote{This argument is the generalization, in some (well-adapted) topos, of the following simple fact from (Cartesian) vector spaces: Consider a colimit of vector spaces $p_i : \FR^{m_i} \rightarrow \FR^m$ and some \textit{smooth} map $f: \FR^m \rightarrow \FR $ with the property that each $f\circ p_i$ is linear. Then $f$ is also necessarily linear, since each $f|_{\mathrm{Im}(p_i)}$ is linear and $\coprod \mathrm{Im}(p_i)$ covers $\FR^m$, in the set theoretic sense.} 
\end{proof}
\end{proposition}

\begin{remark}[\bf Vector-valued form classifying spaces]
The above results generalize straightforwardly to the case of maps $T\CF \rightarrow \CV$, where $\CV$ has a $\FR$-module structure, for instance if $\CV=y(V)$ is the embedding of a 
finite-dimensional vector space $V\in \mathrm{Vect}_\FR^{\mathrm{f.d.}}$ (viewed as a smooth manifold), or even a thickened smooth set of sections $\mathbold{\Gamma}_M(V)$ of vector bundle $V\rightarrow M$. In such a situation, working as above, one may
identify a sub-object $\mathbold{\Omega}^1(-,\CV) \hookrightarrow \CV_{\DD^1(1)}$ which classifies fiberwise linear maps into $\CV$.   
\end{remark}

Let us now make explicit the relation of the above to the algebraic definition of the moduli space $\mathbold{\Omega}^1$ from Ex. \ref{ModuliOfDifferentialForms}.
\begin{lemma}[\bf Relation of the two moduli spaces]\label{RelationOfTwoModuliSpaces}
The algebraic moduli space is naturally a subspace of the fully classifying moduli space 
$$
\mathbold{\Omega}^1 \longhookrightarrow \widehat{\mathbold{\Omega}}^1 \, .
$$
\end{lemma}
\begin{proof}
Recall that $\Hom_{\CE}(Te, \FR)$ is the set of natural transformations between the sheaves $Te:= [\DD^1(1), \, e]$ and $\FR$ over $E$. By Ex. \ref{TangetBundleOfInfinitesimalPoint}, we have that $Te(e') := \Hom_{\mathrm{CAlg}}\big(\CO(Te), \, \CO(e') \big)$, while by the Yoneda embedding $y(\FR)(e') \cong \Hom_{\mathrm{CAlg}}\big(\CO(\FR),\, \CO(e')\big)$ $\cong \CO(e')$.
Thus,
$$
\Hom_\CE(Te,\, \FR) \, \cong \, \Big\{ \big\{\Hom_{\mathrm{CAlg}}\big(\CO(Te), \, \CO(e') \big) \longrightarrow \CO(e')\big\}_{\mathrm{nat.}\, \mathrm{in} \,e' \in E}\Big\}\, ,
$$
where the latter set ranges over families of maps which are natural under pullback along probes $e'\in E$. 

Now, any element $v\in \CO(Te)$ defines a family of (natural) maps on the right-hand side via its evaluation on algebra maps
\begin{align*}
\ev_v \, : \,  \Hom_{\mathrm{CAlg}}\big(\CO(Te), \, \CO(e') \big) &\xlongrightarrow{\quad \quad} \CO(e') \\
\phi_{e'}\quad  & \longmapsto \quad  \phi_{e'}(v) \, ,
\end{align*}
which yields a map of sets
\begin{align}\label{TangentFunctionAlgebraIntoInternalRealValuedMaps}
\CO(Te) \longrightarrow \Hom_\CE(Te,\, \FR) \, ,
\end{align}
even in the current situation where $Te$ is not representable in $\CE$.

Moreover, in analogy with the canonical embedding 
$$
\widehat{\mathbold{\Omega}}^1 (e)\, := \, \Hom_{\CE}^{\mathrm{fib.lin.}}\big(Te, \, \FR\big)\xhookrightarrow{\quad \quad} \Hom_{\CE}\big(Te, \, \FR\big) 
$$ 
of the plots of the fully classifying moduli space into the set of \textit{internal} real-valued maps on $TE$, there is a canonical embedding of the plots of the algebraic moduli space into the function algebra of $Te$,
$$
\mathbold{\Omega}^1(e) \xhookrightarrow{\quad \quad} \CO(Te)
$$
given via any quotient representative \footnote{Note that the inclusion is independent of the quotient representative when both sides are pulled back to $\mathbold{\Omega}^1(e)$ and $\CO(Te)$.} \eqref{KahlerFormsInQuotientIdentification} as
$$
f_i(x) \cdot \dd x^i \longmapsto f_i(x) \cdot y^i \, .
$$ The latter's image under \eqref{TangentFunctionAlgebraIntoInternalRealValuedMaps} can be checked to yield (internal) fiberwise linear maps, hence we are to show that the composed map
$$
\mathbold{\Omega}^1(e)\hookrightarrow \CO(Te) \longrightarrow \widehat{\mathbold{\Omega}}^1(e)  
$$
is also injective, naturally for each probe $e$.

To see this, note that for any $e=\FR^k\times \DD$ with 
$$ \Omega^1(\FR^k\times \DD)
\cong C^\infty( \FR^k)\otimes_{\FR}P,
\qquad
P:=
\bigoplus_{i=1}^k\mathcal{O}(\DD)\,dx^i
\oplus\Omega^1(\DD)\, ,
$$
we may construct a new corresponding (nilpotent) \textit{Weil} algebra
$$\mathcal{O}(\DD')\, 
:= \, \mathcal{O}(\DD)\oplus P
$$
where $
(a,\alpha)\cdot(b,\beta):=(ab,\, a\beta+b\alpha)$. Thus, $e':= \FR^k\times \DD'$ is another thickened probe space with function algebra
\begin{align*}
\CO(e') \, &\cong \, C^\infty(\FR^k)\otimes \CO(\DD') \, \cong \, \big(C^\infty(\FR^k)  \otimes W\big) \oplus \big( C^\infty(\FR^k)\otimes P\big) \\ 
&\cong \, \CO(e) \oplus \Omega^1(e) \, ,
\end{align*}
which admits a canonical algebra map from the original probe algebra
\begin{align*}
\CO(e)&\longrightarrow
\CO(e')\otimes_{\FR}\mathcal \CO\big(\DD^1(1)\big)
\\
f&\longmapsto(f,\,0)+\epsi\,(0,\,\dd_{\mathrm{dR}}f)
\end{align*}
where $\epsi \in \CO\big(\DD^1(1)\big)$ is the square-zero coordinate on the infinitesimal line. Dually, this is a map of thickened Cartesian spaces
$$
e'\times \DD^1(1) \longrightarrow e \, ,
$$
and hence upon embedding in $\ThickenedSmoothSets$, by adjunction, a map
$$
e' \longrightarrow [\DD^1(1), e] =: Te \, ,
$$
i.e., a $e'$-plot of Te.

Finally, for any $\omega=\sum_j h_j\,\dd_{\mathrm{dR}}g_j\in\mathbold{\Omega}^1(e)$, evaluation on this plot gives
$$ 
\omega(v) =\sum_j(h_j,0)\big(0,d_{\mathrm{dR}}g_j\big) =(0,\omega)\quad \in\quad \mathcal{O}(e'). 
$$
Hence, an algebraic form whose evaluation vanishes on the plot of $Te$ must vanish for the particular chosen $e'$-plot, which means that
$$
(0,0) = \omega(v) = (0,\omega)\, ,
$$
i.e., itself vanishes $\omega=0$ thus yielding that 
$$
\mathbold{\Omega}^1(e)\longhookrightarrow \widehat{\mathbold{\Omega}}^1(e)
$$ 
is injective, naturally in e.
\end{proof}
For the purposes of local field theory, the above lemma may be interpreted as the consistency of \textit{both} $\mathbold{\Omega}^1$ and $\widehat{\mathbold{\Omega}}^1$ serving as classifying spaces for \textit{local} differential forms on the product $\CF\times M$ (i.e., a field space with the spacetime the fields are defined over). Namely, by Prop. \ref{ClassifyingNatureOfFormModuli} the latter classifies \textit{all} differential forms, i.e. fiberwise linear maps out of $T(\CF\times M)$, and hence \textit{in particular} those of local nature (factoring through $T(J^\infty_M F)$). But by our earlier results (Rem. \ref{LocalFormsViaModuliSpace}), the algebraic moduli space $\mathbold{\Omega^1}$ also classifies such local forms. By Lem. \ref{RelationOfTwoModuliSpaces}, it must be any such two modulating maps, $\widehat{\CP}_c$ and  $\CP_c$, classifying the \textit{same} local form $\CP\in \Omega^{1}_\loc(\CF\times M)$ are the ``same'', in that they commute under the canonical embedding of the classifying spaces
	\[ 
	\xymatrix@R=1.2em@C=2.6em  { && \widehat{\mathbold{\Omega}}^1
		\\ 
		\CF\times M \ar[rru]^{\widehat{\CP}_c} \ar[rr]^>>>>>>>>{\CP_c} &&  \mathbold{\Omega}^1 \ar[u] \, .
	}   
	\]

\begin{remark}[\bf Differential n-form classifying spaces]
The prior discussion generalizes with minor modifications for the case of the classifying space(s) of differential n-forms on infinitesimally linear spaces (Rem. \ref{InfinitesimallyLinearSpaces}). In this case, the appropriate ``amazing right adjoint'' is instead $(-)_{\DD^n(1)}:\CE \rightarrow \CE $,
$$
\big(\DD^n(1)\times -\big ) \quad \dashv \quad (-)^{\DD^n(1)} \quad  \dashv\quad  (-)_{\DD^n(1)} \, ,
$$
since for infinitesimally linear spaces $T\CF^{\times^{n}_\CF} \cong \big[\DD^n(1),\, \CF\big]$. The above results then generalize accordingly to yield the sequence of subspace identifications
$$
\mathbold{\Omega}^n \longhookrightarrow \widehat{\mathbold{\Omega}}^n \longhookrightarrow \FR_{\DD^n(1)}\, ,
$$
where $\mathbold{\Omega}^n$ is the (algebraic) moduli space (Ex. \ref{ThickenedModuliSpaceOfdeRhamForms}) and $\widehat{\mathbold{\Omega}}^n$ is the fully classifying space with $e$-plots given by alternating, fiberwise multilinear maps $(Te)^{\times^n_e}\rightarrow \FR$ in $\CE$ (cf. Def. \ref{FullyClassifyingThickedModuliSpaceOf1Forms}).
\end{remark}

\medskip 
\noindent {\bf The case of tangent bundles being representable.}
As a final aside, we note that it is possible to make the two different definitions of moduli spaces necessarily coincide, if one is willing to work over a larger site of thickened probes. We briefly show how this works here, but we note that due to the above factorization, this is not strictly necessary for the purposes of local field theory.
As we shall see, it is sufficient to consider instead some larger site\footnote{Supplied with some coverage, whose explicit form is not relevant for what follows. We assume it is subcanonical, and that the formal duals of \eqref{FunctionAlgebraOnTangentProbe} again represent the synthetic tangent bundles of the corresponding probes.} $E$ of smooth loci (duals of finitely generated $C^\infty$-algebras)
$$
\ThickenedCartesianSpaces\longhookrightarrow E \longhookrightarrow C^\infty_{\mathrm{f.g.}}\mbox{-}\mathrm{Alg}^{\op}\, , 
$$
which is \textit{closed under the formation of synthetic tangent bundles}. More precisely, this means that in the category of sheaves
$$
\CE \, := \, \mathrm{Sh}(E)
$$
the synthetic tangent bundle of any object $e\in E$ is representable
$$
\big[\DD^1(1), y(e)\big ] \, \cong \, y(Te)
$$
where
$Te\in E$ is the formal dual of some function algebra $\CO(Te)\in C^\infty_{\mathrm{f.g.}}\mbox{-}\mathrm{Alg}^{\op}$. For instance, this property is indeed implicit in the original source of \cite{DK84} on the classifying space of differential forms, where the classifying space is formalized explicitly within a sheaf topos for the first time.

\medskip 
In particular, this site should include all  tangent bundles to thickened Cartesian spaces (Ex. \ref{TangetBundleOfInfinitesimalPoint}) defined by the function algebras $\CO\big(T(\FR^k\times\DD^m(l))\big)\in C^\infty_{\mathrm{f.g.}}\mbox{-}\mathrm{Alg}^{\op}$
$$
T\big(\FR^k\times \DD^m(l)\big)\, \in \, E ,
$$
and similarly for all other objects in the site. More generally, by \cite[Prop. 1.11]{MoerdijkReyes}, this translates to the following condition on the site $E$:
$$
e\, \in\, E \quad \implies \quad Te \, \in \, E \,,
$$
where for any 
$e$ being the formal dual of 
$$
\CO(e) \, \cong \, C^\infty\big(\FR^k_x\big)/ I \, ,
$$
the corresponding ``tangent bundle'' is defined as the dual of
\begin{align}\label{FunctionAlgebraOnTangentProbe}
\CO(Te) \, : = \, C^\infty\Big(\FR^k_x \times \FR^k_y\Big) \Big/ 
\bigg(I(x), \, \bigg\{\sum_i y_i \frac{\partial f}{\partial x^i} \, | \, f\in I\bigg\} \bigg)\, . 
\end{align}
Examples of such sites are for instance the whole of $C^\infty_{\mathrm{f.g.}}\mbox{-}\mathrm{Alg}^{\op}$, or the smaller sites ``$\mathbb{G}$ and $\mathbb{F}$'' from \cite{MoerdijkReyes}, whose details we do not need to recall for our considerations below\footnote{These sites are in fact closed under arbitrary infinitesimal exponentiations (\cite[Cor. 1.16]{MoerdijkReyes}).}. Of course, the ``minimal'' among such sites towards our goal here, extending $\ThickenedCartesianSpaces$ inside $C^\infty_{\mathrm{f.g.}}\mbox{-}\mathrm{Alg}^{\op}$, is given by including all (iterated) tangent bundles of thickened Cartesian spaces $\FR^k\times \DD$, namely the formal duals of
\begin{align*}
\CO\Big(T\big(\FR^k\times \DD^m(l)\big)\!\Big) & \cong 
C^\infty\Big(\FR^k\times \FR^k \times \FR^m_{x_0} \times \FR^m_{x_1}\Big) 
\Big/ \Big( (x_0)^{l+1}, \, x_1\cdot (x_0)^l \Big)   \, , 
\\
\CO\Big(T^2\big(\FR^k\times \DD^m(l)\big)\!\Big) &\cong 
C^\infty\Big(\big(\FR^k\big)^{\times 4} \times \FR^m_{x_0} \times \FR^m_{x_1} \times \FR^m_{x_2} \times \FR^m_{x_3}\Big) 
\\
& \phantom{AAAAAAAAAAAAA} \Big/ \Big( (x_0)^{l+1}, \,  x_1\cdot (x_0)^l, \, x_2\cdot (x_0)^l,\, x_3\cdot (x_0)^l +  x_2\cdot x_1 \cdot l \, (x_l)^{l-1}\Big)    , 
\end{align*}
and so on for any $\FR^k \times \DD \in \ThickenedCartesianSpaces$.

\medskip 
The first task, then, is to provide an appropriate working definition of the \textit{algebraic} classifying space as a sheaf over this more general category of thickened probe-spaces, extending that of Ex. \ref{ThickenedModuliSpaceOfdeRhamForms}. Since these are duals of \textit{finitely-generated} $C^\infty$-algebras, one can define this directly as follows (cf. \cite[Prop. 5.6]{Joyce19}).

\begin{definition}[\bf Thickened moduli space of de Rham forms]\label{nFormsFormalSmoothSet}
$\,$ 
\begin{itemize}

\item[\bf (i)] The \textit{thickened moduli space of de Rham $1$-forms} $\mathbold{\Omega}^1$ is defined as the sheaf $\mathbold{\Omega}^1 \in \CE$ with $e$-plots given by
$$
\mathbold{\Omega}^1 (e)\, := \, \mathrm{Span}_{\CO(e)}\big\{ \dd v \, | \, v \in \CO(e)\big\} \big/ \dd \big( h(v^1,\cdots,v^q)\big) = \sum_i \frac{\partial h}{\partial z^i} (v^1,\cdots, v^q) \cdot \dd v^i\, ,  
$$ 
where $h\in C^\infty(\FR_z^q)$ ranges through any \textit{smooth} function in $q$-variables, for any $q\in \NN$. Under any identification  $\CO(e)\cong C^\infty(\FR^k_x)/I$, this can be equivalently expressed as
\begin{align}\label{1FormsOnSmoothLoci}
\mathbold{\Omega}^1 (e)\, \cong \, \CO\big(e\big)\cdot\big(\dd x^1,\cdots, \dd x^k\big)\,  \big/ \, \dd I\, ,  
\end{align}
\item[\bf (ii)] Similarly, the \textit{moduli space of $n$-forms} $\mathbold{\Omega}^n$ is the sheaf $\mathbold{\Omega}^n\in \CE$ with $e$-plots 
$$
\mathbold{\Omega}^{n}\big(e\big)\; := \;  \bigwedge^{n}_{\CO\left(e\right)} \mathbold{ \Omega}^1\big(e\big) \, .
$$
\end{itemize} 
\end{definition}
It is immediately obvious that for $e=\FR^k$, this recovers the standard differential forms on $\FR^k$
$$
\mathbold{\Omega}^{1}(\FR^k)\, \cong\,  \Omega^1(\FR^k)\, , 
$$
while, by Hadamard's Lemma \ref{HadamardsLemma}, for any infinitesimal point $\DD$, this recovers the algebraic ``Kähler'' differential forms of the underlying $\FR$-algebra\footnote{Namely, by Hadamard's lemma one can check that the quotient relation is equivalent to that of quotienting only by polynomials.} from \eqref{AlgebraicFormsOnInfinitesimalPoint}.
More generally, by the partial Hadamard's Lemma \ref{PartialHadamardsLemma}, one sees similarly that this definition reproduces 
$$
\mathbold{\Omega}^1(\FR^k\otimes \DD) \; \cong \; \Omega^1(\FR^k)\otimes \CO(\DD)\oplus C^\infty(\FR^k)\otimes \mathbold{\Omega}^1(\DD)
$$
on any thickened Cartesian space $\FR^k\times \DD$. The following immediate result justifies this definition on such more general thickened smooth probe spaces $e \in E$.
\begin{lemma}[\bf 1-forms inject into tangent algebras]
\label{1formsInjectIntoTangentAlgebras}
There is a \textit{canonical} injection 
\begin{align*}
\mathbold{\Omega}^1(e) \; \longhookrightarrow  \; \CO(Te) \, \cong \,  \Hom_E(Te, \, \FR)\, \cong \, \Hom_\CE(Te,\, \FR)\, 
\end{align*}
of $1$-forms on $e\in E$ into the (smooth) functions on $Te \in E$, whose image is the fiberwise linear maps internal to $\CE$.
\end{lemma}
\begin{proof} This follows similarly to the corresponding part of Lem. \ref{RelationOfTwoModuliSpaces}, but now noting that $Te$ is representable in the new topos. More explicitly, the first inclusion is immediately obvious under any representation $\CO(e)\cong C^\infty(\FR^k_x)/I$ and the induced ones for $\CO(Te)$ \eqref{FunctionAlgebraOnTangentProbe} and $\mathbold{\Omega}^1(e)$ \eqref{1FormsOnSmoothLoci}, whereby it is identified with the mapping 
$$
f_i(x) \cdot \dd x^i \longmapsto f_i(x) \cdot y^i \, .
$$
Note that the inclusion is independent of the quotient representative when both sides are pulled back to $\mathbold{\Omega}^1(e)$ and $\CO(Te)$.
The latter isomorphism above follows since in any site $E$ consisting of duals of finitely generated algebras one has 
$$
\Hom_E(Te,\, \FR) \, \cong \, \Hom_{C^\infty\mbox{-}\Alg}\big(C^\infty(\FR), \, \CO(Te)\big) \, \cong \, \CO(Te) \, .
$$

This justifies thinking of $1$-forms on $e$ as the subset of \textit{smooth and fiberwise linear} functions $Te \rightarrow \FR $ in the site $E$, since the inclusion maps onto elements linear in ``the coordinate along the fibers'' $f_i(x)\cdot y^i \in \CO(Te)$ from \eqref{FunctionAlgebraOnTangentProbe}. That is,
\begin{align}\label{FiberwiseLinearFunctionsOnTangentBundle}
\Hom_E^{\mathrm{fib.lin.}}(Te, \, \FR) \, := \, \mathrm{Im}\big(\Omega^1(e)\big) \;\; \subset \;\; \CO(Te)\cong \Hom_E (Te,\, \FR)\, . 
\end{align}
Indeed, under the Yoneda embedding, the above image is identified with those maps that are \textit{internally fiberwise linear}, so that
$$
\mathbold{\Omega}^1(e) \, \cong \, \Hom_\CE^{\mathrm{fib.lin.}}(Te, \, \FR) \, .
$$

\vspace{-5mm} 
\end{proof}

Since the fully classifying space of differential forms  may be defined in this larger sheaf topos precisely as in Def. \ref{FullyClassifyingThickedModuliSpaceOf1Forms}, the following is immediate from Lem. \ref{1formsInjectIntoTangentAlgebras}. 
\begin{corollary}[\bf Moduli spaces are isomorphic]\label{ModuliSpacesAreIsomorphic}
In a larger sheaf topos where tangent bundles are representable, the algebraic moduli space  is canonically isomorphic to the corresponding fully classifying moduli space
$$
\mathbold{\mathbold{\Omega}}^1 \, \cong \, \widehat{\mathbold{\Omega}}^1\, .
$$
\end{corollary}



\begin{thebibliography}{100}


 

\bibitem[An92]{An91}
I.~Anderson, {\it Introduction to the variational bicomplex}, 
in {
  Mathematical Aspects of Classical Field Theory}, M.~Gotay, J.~Marsden, and
  V.~Moncrief (eds.), {Contemporary Math.} {\bf 132}, 51--73,
Amer. Math. Soc. (1992),
[\href{https://bookstore.ams.org/conm-132}{\tt 
ams:conm-132}]. 


 


\bibitem[AD80]{AD80}
I.~Anderson and T.~Duchamp, {\it On the existence of global variational
  principles}, Amer.\ J.\ Math. {\bf 102} (1980), 781-868,
[\href{https://doi.org/10.2307/2374195}{\tt 
doi:10.2307/2374195}]. 


\bibitem[AF97]{AF97}
I. M. Anderson and M. E. Fels, 
{\it Symmetry reduction of variational bicomplexes and the principle of symmetric criticality},
Amer. J. Math. {\bf 119} (1997), 609-670,
[\href{https://www.jstor.org/stable/25098547}{\tt jstor:25098547}]. 

  


\bibitem[Be00]{Bell00}
J. L. Bell, {\it An invitation to smooth infinitesimal analysis}, 
[\href{https://publish.uwo.ca/~jbell/invitation%20to%20SIA.pdf}
{\tt uwo.ca/$\sim$jbell/invitationtoSIA.pdf}].  


\bibitem[Be98]{Bell98}
J. L. Bell, 
{\it A Primer of Infinitesimal Analysis}, Cambridge University Press (1998),  
[\href{https://www.cambridge.org/de/universitypress/subjects/mathematics/logic-categories-and-sets/primer-infinitesimal-analysis-2nd-edition?format=HB&isbn=9780521887182}{\tt 
ISBN:9780521887182}]. 






\bibitem[BS17]{BS17}
M. Benini and A. Schenkel,
{\it Poisson algebras for non-linear field theories in the Cahiers topos},
 Ann. Henri Poincar\'e {\bf  18} (2017), 1435-1464,
[\href{https://doi.org/10.1007/s00023-016-0533-2}{\tt 
doi:10.1007/s00023-016-0533-2}], 
[\href{https://arxiv.org/abs/1602.00708}{\tt 
	arXiv:1602.00708}].


 


\bibitem[Be80]{Bergeron}
F. Bergeron, 
{\it Objet infinit\'esimalement lin\'eaire dans un mod\'ele
bien adapt\'e de G.D.S}, 
in : G. E. Reyes (ed.), Géométrie Différentielle Synthétique, Rapport de Recherches, DMS 80-12, 
Universit\'e de Montreal (1980).
 



\bibitem[Bl24b]{Blohmann23b} 
C. Blohmann, {\it Lagrangian field theory}, \newline 
[\href{https://people.mpim-bonn.mpg.de/blohmann/assets/pdf/Lagrangian_Field_Theory_v24.0.pdf} 
{\tt people.mpim-bonn.mpg.de/blohmann/assets/pdf/Lagrangian\_Field\_Theory\_v24.0.pdf}]. 



\bibitem[BDK90]{BDK}
F.~Brandt, N.~Dragon, and M.~Kreuzer, 
{\it Completeness and nontriviality of the
  solutions of the consistency conditions},  Nucl. Phys. {\bf B332}
  (1990), 
224--249,
[\href{https://doi.org/10.1016/0550-3213(90)90037-E}{\tt 
doi:10.1016/0550-3213(90)90037-E}]. 

 



\bibitem[BD86]{BD86}
M. Bunge and E. Dubuc,
{\it Archimedian local $C^\infty$-rings and models of synthetic differential geometry},  Cah. Top. G\'eom. Diff. Cat.  {\bf XXVII-3} (1986), 3-22,
[\href{https://www.numdam.org/item/?id=CTGDC_1986__27_3_3_0}{\tt numdam:CTGDC\_1986\_\_27\_3\_3\_0}]. 

\bibitem[BGS18]{BGS18}
M. Bunge, F. Gago, and A. M. San Luis, 
{\it Synthetic Differential Topology}, 
Cambridge University Press (2018), \newline 
[\href{https://doi.org/10.1017/9781108553490}{\tt 
doi:10.1017/9781108553490}]. 


\bibitem[CGRS20]{CGRS20}
M. Dimitrijevi\'c \'Ciri\'c, G. Giotopoulos,
V. Radovanovi\'c and R. J. Szabo,
{\it $L_\infty$-Algebras of Einstein--Cartan--Palatini Gravity},
J. Math. Phys. {\bf 61} (2020), 112502, [\href{https://doi.org/10.1063/5.0011344}
{\tt doi:10.1063/5.0011344}],
[\href{https://arxiv.org/abs/2003.06173}{\tt arXiv:2003.06173}].
 
 
\bibitem[CR13]{CarchediRoytenberg13}
D. Carchedi and D. Roytenberg, 
{\it  On Theories of Superalgebras of Differentiable Functions},
	Theor. Appl. Categ. {\bf 28} (2013),  1022-1098,
[\href{http://www.tac.mta.ca/tac/volumes/28/30/28-30abs.html}{\tt 
 tac.mta.ca/28/30/28-30abs}],
[\href{https://www.arxiv.org/abs/1211.6134}{\tt arXiv:1211.6134}]. 
 
 
\bibitem[Ch24]{Cher24}
F. Cherubini, {\it  Synthetic $G$-jet-structures in modal homotopy type theory}, Math.  Structures Comput. Sci. {\bf 34} (2024), 834-868,
[\href{https://doi.org/10.1017/S0960129524000355}{\tt 
doi:10.1017/S0960129524000355}]. 


 



\bibitem[Du79]{Dubuc79}
E.~Dubuc,
{\it Sur les mod{\`e}les de la g{\'e}om{\'e}trie diff{\'e}rentielle synth{\'e}tique},
Cah. Top. G{\'e}om.\ Diff. Cat{\'e}g.
{\bf 20} (1979), 231-279, [\href{http://www.numdam.org/item?id=CTGDC_1979__20_3_231_0}{\tt CTGDC\_1979\_\_20\_3\_231\_0}].


 
\bibitem[DK84]{DK84}
E. Dubuc and A. Kock, {\it  On 1-form classifiers},
Commun. Algebra {\bf 12} (1984), 1471-1531, \newline 
[\href{https://doi.org/10.1080/00927878408823064}{\tt 
doi:10.1080/00927878408823064}]. 

\bibitem[DN15]{DN15}
M. Demmel and A. Nink,
{\it Connections and geodesics in the space of metrics},
[\href{https://arxiv.org/abs/1506.03809}
{\tt arXiv:1506.03809}].
 

\bibitem[Fl53]{Flanders}
H. Flanders, {\it Development of an extended exterior differential calculus}, 
Trans. Amer. Math. Soc. {\bf 75} (1953), 311-326,
[\href{https://doi.org/10.1090/S0002-9947-1953-0057005-8}{\tt 
 doi:10.1090/S0002-9947-1953-0057005-8}]. 

\bibitem[FSJ24]{FSJ24}
K. Francis-Staite and D.~D. Joyce, {\it $C^\infty$-algebraic geometry with corners}, 
 Cambridge Univ. Press, Cambridge (2024), 
[\href{https://doi.org/10.1017/9781009400190}{\tt 
doi:10.1017/9781009400190}], 
[\href{https://arxiv.org/abs/1911.01088}{\tt 
 arXiv:1911.01088}]. 


 

\bibitem[GMS00]{GMS00}
G. Giachetta, L. Mangiarotti, and G. Sardanashvily,
{\it Cohomology of the variational bicomplex on the infinite-order jet space}, 
[\href{https://arxiv.org/abs/math/0006074}{\tt arXiv:math/0006074}]. 

\bibitem[GMS09]{GMS09}
G. Giachetta, L. Mangiarotti, and G. Sardanashvily, {\it Advanced classical field theory}, World Scientific, Singapore (2009), [\href{https://doi.org/10.1142/7189}{\tt doi:10.1142/7189}].



\bibitem[Gi25]{Gi25}
G. Giotopoulos, {\it Sheaf Topos Theory: A powerful setting for Lagrangian Field Theory}, 
[\href{https://arxiv.org/abs/2504.08095}{\tt arXiv:2504.08095}].

\bibitem[Gi26]{Gi26}
G. Giotopoulos, {\it
Manifolds with corners in the $\FR$-Algebraic Cahiers topos}, in preparation. 

\bibitem[GKSS26]{GKSS25}
G. Giotopoulos, 
I. Khavkine,
H. Sati,
 and U. Schreiber, {\it Synthetic differential jet bundles are reduced}, \newline 
 [\href{https://arxiv.org/abs/2601.15862}{\tt arXiv:2601.15862}].
 
\bibitem[GS22]{GS22}
G. Giotopoulos and R. J. Szabo,
{\it Braided Symmetries in Noncommutative Field Theory},
J. Phys. A: Math. Theor. {\bf 55} (2022), 353001, [\href{https://doi.org/10.1088/1751-8121/ac5dad} 
{\tt doi:10.1088/1751-8121/ac5dad}], [\href{https://arxiv.org/abs/2112.00541}{\tt arXiv:2112.00541}].

\bibitem[Part I]{GS23}
G. Giotopoulos and H. Sati, {\it Field Theory via Higher Geometry I: Smooth Sets of Fields}, 
J. Geom. Phys. {\bf 213} (2025) 105462, 
[\href{https://doi.org/10.1016/j.geomphys.2025.105462}{\tt 
doi:10.1016/j.geomphys.2025.105462}],
[\href{https://arxiv.org/abs/2312.16301}{\tt arXiv:2312.16301}].

       

\bibitem[Gr06]{Gruenwald09}
J. Gr{\"u}nwald, {\it \"Uber duale Zahlen und ihre Anwendung in der Geometrie}, Monatsh. F. Math. Phys. {\bf 17} (1906), 81–136, 
[\href{https://doi.org/10.1007/BF01697639}{\tt doi:10.1007/BF01697639}].

\bibitem[Gwi12]{Gwi12}
O. Gwilliam,
{\it Factorization algebras and free field theories},
Ph.D. thesis, Northwestern University (2012), [\href{https://ncatlab.org/nlab/files/GwilliamThesis.pdf}
{\tt thesis PDF}]. 





 





 


\bibitem[JM22]{Myers22}
D. Jaz Myers, {\it Orbifolds as microlinear types in synthetic differential cohesive homotopy type theory}, 
\newline 
[\href{https://arxiv.org/abs/2205.15887}{\tt 
arXiv:2205.15887}].

\bibitem[Jo19]{Joyce19}
D.~D. Joyce, {\it  Algebraic geometry over $C^{\infty}$-rings}, 
Mem. Amer. Math. Soc. {\bf 260} (2019), no.~1256, v+139 pp., \newline 
[\href{https://www.ams.org/books/memo/1256}{\tt 
ams.org/books/memo/1256/}], 
[\href{https://arxiv.org/abs/1001.0023}{\tt arXiv:1001.0023}]. 

\bibitem[JRSW19]{JRSW19}
B. Jur\v{c}o, L. Raspollini, C. S\"amann and M. Wolf,
{\it $L_\infty$-Algebras of Classical Field Theories and the
Batalin--Vilkovisky Formalism},
Fortschr. Phys. {\bf 67} (2019), 1900025, [\href{https://doi.org/10.1002/prop.201900025}
 {\tt doi:10.1002/prop.201900025}], \newline [\href{https://arxiv.org/abs/1809.09899}{\tt arXiv:1809.09899}].



\bibitem[KS17]{KS17}
I. Khavkine and U. Schreiber,
{\it Synthetic geometry of differential equations: I. Jets and comonad structure}, to appear in J. Geom. Phys.,  
[\href{https://arxiv.org/abs/1701.06238}{\tt arXiv:1701.06238}]. 

 


\bibitem[Ko80]{Kock80}
A. Kock, {\it Formal manifolds and synthetic theory of jet bundles}, Cah. Topol. G\'eom. Diff. {\bf  21} (1980), 227-246, 
\newline 
[\href{https://www.numdam.org/item/CTGDC_1980__21_3_227_0/}{\tt 
numdam:CTGDC\_1980\_\_21\_3\_227\_0/}]. 


\bibitem[Ko81]{Kock81}
A. Kock, {\it Properties of well-adapted models for synthetic differential geometry}, 
J. Pure Appl. Algebra {\bf 20} (1981), 55-70,
[\href{https://doi.org/10.1016/0022-4049(81)90048-7}{\tt 
doi:10.1016/0022-4049(81)90048-7}]. 


\bibitem[Ko06]{Kock06}
A. Kock, {\it Synthetic Differential Geometry}, 
Cambridge University Press (2006), \newline 
[\href{https://doi.org/10.1017/CBO9780511550812}{\tt 
doi:10.1017/CBO9780511550812}]. 


\bibitem[Ko10]{Kock10}
A. Kock, {\it Synthetic geometry of manifolds}, Cambridge University Press (2010),
\newline 
[\href{https://doi.org/10.1017/CBO9780511691690}{\tt 
doi:10.1017/CBO9780511691690}]. 


\bibitem[Ko16]{Kock16}
A. Kock, {\it New methods for old spaces: synthetic differential geometry}, 
New Spaces in Mathematics and Physics - Formal and Philosophical Reflections Vol 2 
(eds. M. Anel and G. Cartren), Cambridge University Press (2016), \newline 
[\href{https://www.cambridge.org/de/universitypress/subjects/mathematics/geometry-and-topology/new-spaces-mathematics-and-physics-formal-and-conceptual-reflections?format=WX&isbn=9781108854368#ressourcen}{\tt 
ISBN:9781108854368}], 
[\href{https://arxiv.org/abs/1610.00286}{\tt 
arxiv/1610.00286}]

\bibitem[Ko25]{Kock25}
A. Kock, 
{\it Two theorems of Lie on infinitesimal symmetries of differential equations}, 
Theory Appl. Categ. {\bf 43} (2025), 93–107,
[\href{http://www.tac.mta.ca/tac/volumes/43/5/43-05abs.html}
{\tt tac.mta.ca/tac/volumes/43/5/43-05abs.html}].





\bibitem[KR87]{KockReyes87}
A. Kock and G. Reyes, 
{\it Corrigendum and addenda to: Convenient vector spaces embed into the Cahiers topos}, Cahiers de Top. G{\'e}om. Diff. Cat{\'e}g., {\bf 28} 2 (1987), 99-110, 
[\href{http://www.numdam.org/item?id=CTGDC_1987__28_2_99_0}{\tt numdam:CTGDC\_1987\_\_28\_2\_99\_0}].






\bibitem[Kol08]{Kolar08}
I. Kolar, {\it Weil bundles as generalized jet spaces}, 
 Handbook of Global Analysis, Elsevier (2008), 625-664, \newline 
[\href{https://doi.org/10.1016/B978-044452833-9.50013-9}{\tt 
doi:10.1016/B978-044452833-9.50013-9}]. 




\bibitem[KMS93]{KMS93}
I. Kolar, P. W. Michor, and J. Slovak, 
{\it Natural Operations in Differential Geometry}, Springer-Verlag, Berlin (1993),
[\href{https://doi.org/10.1007/978-3-662-02950-3}{\tt 
doi:10.1007/978-3-662-02950-3}]. 


\bibitem[Kos09]{Kostecki09}
R. P. Kostecki, {\it Differential Geometry in Toposes}, 
University of Warsaw (2009),  \newline 
[\href{https://www.fuw.edu.pl/~kostecki/sdg.pdf}{\tt 
fuw.edu.pl/$\sim$kostecki/sdg.pdf}]. 








\bibitem[Kr15]{Kr15}
D. Krupka,
{\it Introduction to Global Variational Geometry}, 
Atlantis Press, Paris (2015), \newline 
[\href{https://doi.org/10.2991/978-94-6239-073-7}{\tt 
doi:10.2991/978-94-6239-073-7}]. 

\bibitem[KV12]{KV12}
D. Kreimer and A. Velenich,
{\it Field diffeomorphisms and the algebraic structure
of perturbative expansions}, \newline 
[\href{https://arxiv.org/abs/1204.3790}
{\tt arXiv:1204.3790}].

\bibitem[Lac22]{Lac22}
 S. Lacroix,
{\it 4-dimensional Chern--Simons theory and integrable field theories},
J. Phys. A: Math. Theor. {\bf 55} (2022), 083001, [\href{https://doi.org/10.1088/1751-8121/ac48ed}
{\tt doi:10.1088/1751-8121/ac48ed}], [\href{https://arxiv.org/abs/2109.14278}{\tt arXiv:2109.14278}].


\bibitem[Lav96]{Lavendhomme96}
R. Lavendhomme, {\it Basic concepts of synthetic differential geometry}, 
 Springer, New York (1996), \newline 
[\href{https://doi.org/10.1007/978-1-4757-4588-7}{\tt 
doi:10.1007/978-1-4757-4588-7}].



\bibitem[La67]{Law67}
F. W. Lawvere,
{\it Categorical dynamics}, lecture, Chicago 1967);
published as pp.1-28 in 
A. Kock (ed.), {\it Topos Theoretic Methods in Geometry},
Aarhus Univ. Var. Pub. Ser. {\bf 30} (1979), 
[\href{https://dmitripavlov.org/scans/kock.pdf}{\tt 
dmitripavlov.org/scans/kock.pdf}]. 



\bibitem[La80]{Law80}
F. W. Lawvere, {\it Toward the description in a smooth topos of the dynamically possible motions and deformations of a continuous body}, Cah. Top. Géom. Diff. Cat. 
{\bf 21} (1980), 377-392, 
[\href{https://www.numdam.org/article/CTGDC_1980__21_4_377_0.pdf}{\tt 
numdam:CTGDC\_1980\_\_21\_4\_377\_0}]. 

\bibitem[La97]{Law97}
F. W. Lawvere, {\it Toposes of laws of motion}, 
transcript of a talk in Montreal (1997), \newline 
[\href{https://ncatlab.org/nlab/files/LawvereToposesOfLawsOfMotions.pdf}
{\tt nlab/files/LawvereToposesOfLawsOfMotions.pdf}]. 

\bibitem[La98]{Law98}
F. W. Lawvere, {\it Outline of synthetic differential geometry}, 
lectures in Buffalo (1998), \newline 
[\href{https://ncatlab.org/nlab/files/LawvereSDGOutline.pdf}{\tt 
nlab/files/LawvereSDGOutline.pdf}]. 


\bibitem[La02]{Law02}
F. W. Lawvere, {\it Categorical algebra for continuum microphysics}, 
J. Pure Appl. Algebra {\bf 175} (2002), 267-287, \newline 
[\href{https://doi.org/10.1016/S0022-4049(02)00138-X}{\tt 
doi:10.1016/S0022-4049(02)00138-X}]. 








\bibitem[MLM94]{MacLaneMoerdijk}
S. MacLane and I. Moerdijk,
{\it Sheaves in geometry and logic: A first introduction to topos theory},
Springer-Verlag, New York (1994), 
[\href{https://link.springer.com/book/10.1007/978-1-4612-0927-0}{\tt doi:10.1007/978-1-4612-0927-0}].









\bibitem[Me96]{Melrose}
R.B. Melrose, {\it  Differential Analysis on Manifolds with Corners},
book draft, [\href{http://math.mit.edu/~rbm/book.html}{\tt math.mit.edu/$\sim$rbm/book.html}].  






\bibitem[MR91]{MoerdijkReyes}
I. Moerdijk and G. E. Reyes, 
{\it Models for Smooth Infinitesimal Analysis},
Springer, Berlin (1991), \newline  
[\href{https://link.springer.com/book/10.1007/978-1-4757-4143-8}{\tt doi:10.1007/978-1-4757-4143-8}].



\bibitem[Mu88]{Mumford88}
D. Mumford, {\it Red book of varieties and schemes}, Lecture Notes in Mathematics {\bf 1358}, Springer (1988, 1999), \newline  
[\href{https://doi.org/10.1007/b62130}{\tt doi:10.1007/b62130}]


\bibitem[MH16]{MH16}
J. Musilov\'a and  S. Hronek, 
{\it The calculus of variations on jet bundles as a universal approach for a variational formulation of fundamental physical theories}, 
Commun. Math. {\bf 24} (2016),
[\href{https://cm.episciences.org/9458}{\tt  doi.org/10.1515/cm-2016-0012}]. 


\bibitem[Ne03]{Nestruev03}
J. Nestruev, 
{\it Smooth manifolds and observables}, Graduate Texts in Mathematics {\bf 218}, Springer, New York (2003),
\newline 
[\href{https://doi.org/10.1007/978-3-030-45650-4}{\tt doi:10.1007/978-3-030-45650-4}].

\bibitem[Ni01]{Nish01}
H. Nishimura, {\it Synthetic differential geometry of jet bundles}, Bull. Belg. Math. Soc. Simon Stevin {\bf 8} (2001), no.~4, 639--650, 
[\href{https://projecteuclid.org/journals/bulletin-of-the-belgian-mathematical-society-simon-stevin/volume-8/issue-4/Synthetic-differential-geometry-of-jet-bundles/10.36045/bbms/1102714793.full}{\tt doi:10.36045/bbms/1102714793}]. 


\bibitem[Pe85]{Penon85}
J. Penon, {\it De l’infinit\'esimal au local}, 
PhD Thesis Universit\'e Paris, Diagrammes {\bf S13} (1985), 1-191,
\newline 
[\href{https://www.numdam.org/article/DIA_1985__S13__1_0.pdf}{\tt 
numdam:DIA\_1985\_\_S13\_\_1\_0.pdf}]. 








\bibitem[Re86]{Reyes86}
G. E. Reyes, {\it  Synthetic reasoning and variable sets}, in: Categories in Continuum Physics. Lecture Notes in Mathematics {\bf 1174}, Springer, Berlin (1986), [\href{https://doi.org/10.1007/BFb0076935}{\tt 
doi:10.1007/BFb0076935}]. 

\bibitem[Re07]{Reyes07}
G. E. Reyes, {\it Embedding manifolds with boundary in smooth toposes}, 
Cah. Top. Géom. Diff. Catég. {\bf 48} (2007), 83-103,
[\href{https://eudml.org/doc/91717}{\tt eudml.org/doc/91717}]. 

\bibitem[Re11]{Reyes11}
G.~E. Reyes, {\it A derivation of Einstein's vacuum field equations}, in: Models, Logics, and Higher-dimensional Categories, 245--261, CRM Proc. Lecture Notes, 53, Amer. Math. Soc. (2011),
[\href{https://www.ams.org/books/crmp/053/}{\tt 
ams.org/books/crmp/053}],
[\href{https://reyes-reyes.com/wp-content/uploads/2009/12/a-derivation-of-einsteins-vacuum-field-equations1.pdf}{\tt pdf}].



 

\bibitem[Sa02]{Sardanashvily02}
G. Sardanashvily,
{\it Ten lectures on jet manifolds in classical and quantum field theory}, 
[\href{https://arxiv.org/abs/math-ph/0203040}{\tt  arXiv:math-ph/0203040}]. 
 	




 




\bibitem[Sc13a]{dcct}
U. Schreiber,
{\it Differential cohomology in a cohesive infinity-topos},
[\href{https://arxiv.org/abs/1310.7930}{\tt arXiv:1310.7930}].


\bibitem[Sch23]{Urs23}
U. Schreiber, {\it Geometry of physics – supergeometry}, 
[\href{https://ncatlab.org/nlab/show/geometry+of+physics+--+supergeometry}{\tt 
nlab/show/geometry+of+physics+--+supergeometry}]. 



\bibitem[Sc24]{Schreiber24}
U. Schreiber, 
{\it  Higher Topos Theory in Physics},
{\it Encyclopedia of Mathematical Physics}, 2nd ed., Elsevier (2024),
62-76, 
[\href{https://doi.org/10.1016/B978-0-323-95703-8.00210-X}{\tt 
doi:10.1016/B978-0-323-95703-8.00210-X}], 
[\href{https://arxiv.org/abs/2311.11026}{\tt arXiv:2311.11026}].










\bibitem[Shu06]{Shulman06}
M. Shulman, {\it Synthetic Differential Geometry}, 
Chicago Pizza-Seminar (2006),  \newline 
[\href{https://home.sandiego.edu/~shulman/papers/sdg-pizza-seminar.pdf}{\tt 
sandiego.edu/~shulman/papers/sdg-pizza-seminar.pdf}]. 




\bibitem[Ta79]{Takens79}
F. Takens, 
{\it A global version of the inverse problem of the calculus of variations}, 
J. Differential Geom. {\bf 14} (1979), 543-562,
[\href{https://projecteuclid.org/journals/journal-of-differential-geometry/volume-14/issue-4/A-global-version-of-the-inverse-problem-of-the-calculus/10.4310/jdg/1214435235.full}{\tt doi:10.4310/jdg/1214435235}]. 



\bibitem[We53]{Weil53}
A. Weil, {\it Th\'eorie des points proches sur les vari\'et\'es diff\'erentiables}, in {\it G\'eom\'etrie diff\'erentielle. Colloques Internationaux du Centre National de la Recherche Scientifique, Strasbourg (1953)}, pp. 111--117, CNRS, Paris. 



\bibitem[We17]{Wellen17}
F. Wellen (now: F. Cherubini),
{\it Formalizing Cartan Geometry in Modal Homotopy Type Theory}, 
PhD Thesis, Karlsruhe Institute of Technology (2017), 
[\href{https://publikationen.bibliothek.kit.edu/1000073164}{\tt 
doi:10.5445/IR/1000073164}]. 












\end{thebibliography}
\end{document}